\documentclass[11pt]{article}
\usepackage{fullpage}

\usepackage{amsmath,amssymb}
\usepackage{bm}
\usepackage[numbers]{natbib}

\usepackage{amssymb}
\usepackage{mathrsfs}
\usepackage{graphicx}
\usepackage{subcaption}
\usepackage{pdfpages}
\usepackage{afterpage}
\usepackage{caption}
\usepackage{multirow}
\usepackage{rotating}
\usepackage[flushleft]{threeparttable}
\usepackage{multirow}
\usepackage{enumitem} 
\usepackage{algorithm2e}
\SetAlgorithmName{Algorithm}{Algorithm}{List of Algorithms}
\usepackage{svg}
\usepackage{stackengine}
\usepackage{float}
\RestyleAlgo{ruled}
\usepackage{tikz} 
\usetikzlibrary{matrix}

\usepackage{preamble/star}
\usepackage[group-separator={,},group-minimum-digits={3}]{siunitx}

\usepackage[colorlinks,linkcolor=blue,anchorcolor=blue,citecolor=blue,urlcolor=black]{hyperref}
\usepackage{cleveref}

\def\T{{ \mathrm{\scriptscriptstyle T} }}

\newcolumntype{P}[1]{>{\centering\arraybackslash}p{#1}}

\def\T{{ \mathrm{\scriptscriptstyle T} }} 

\newcommand{\Rom}[1]{\text{\uppercase\expandafter{\romannumeral #1\relax}}}
\def\T{{ \mathrm{\scriptscriptstyle T} }}

\def\##1\#{\begin{align}#1\end{align}}
\def\$#1\${\begin{align*}#1\end{align*}}

\newcommand{\var}{\textrm{var}}

\newcommand\blfootnote[1]{%
  \begingroup
  \renewcommand\thefootnote{}\footnote{#1}%
  \addtocounter{footnote}{-1}%
  \endgroup
}

\renewcommand{\max}{\mathop{\mathrm{max}}}

\usepackage{color}

\usepackage{subcaption}
\usepackage{float}
\usepackage{multirow}
\usepackage{setspace}
\newcommand{\ex}{\mathrm{e}}

\definecolor{wjs}{RGB}{200,0,50}

\begin{document}


\def\scititle{
    A Sieve-Accelerated Quadrature Method for Exact Privacy Accounting in the 2020 U.S.\ Decennial Census
}

\title{\LARGE \scititle}


\author{Buxin Su\thanks{University of Pennsylvania; Email: \texttt{subuxin@sas.upenn.edu}.} \and Weijie Su\thanks{University of Pennsylvania; Email: \texttt{suw@wharton.upenn.edu}.} \and Chendi Wang\thanks{Xiamen University; Email: \texttt{chendi.wang@xmu.edu.cn}.}
\blfootnote{Authors are listed in alphabetical order.}
}



\date{\today}

\maketitle

\begin{abstract}

In 2020, the U.S.\ Census Bureau adopted differential privacy for the Decennial Census by injecting integer-valued Gaussian noise into published census tabulations. Exactly evaluating the privacy guarantees of these data releases would enable the Bureau to determine the absolute minimum noise required to satisfy a given privacy budget, preventing the injection of unnecessary excess noise and thereby substantially enhancing the statistical utility of the data for downstream applications such as federal funding allocation and political redistricting. In this paper, we introduce a computationally efficient and mathematically rigorous quadrature method to evaluate the exact privacy profile of practical, large-scale census releases under the composition of heterogeneous discrete Gaussian mechanisms. Mathematically, this problem reduces to evaluating the tail probabilities of high-dimensional convolutions of integer-valued random variables sampled from heterogeneous discrete Gaussian distributions under exceptionally stringent numerical error tolerances (e.g., $10^{-35}$). By recasting the exact privacy accounting as a numerical integration problem via the discrete Fourier transform, we explicitly exploit the exponential convergence of the trapezoidal rule for complex analytic, periodic characteristic functions. Furthermore, to overcome the computational bottleneck of evaluating highly oscillatory integrands in high dimensions, we develop a sieve algorithm that identifies and prunes negligible quadrature nodes, accelerating the computation by three orders of magnitude. Taken together, these numerical innovations enable the first exact, assumption-free privacy accounting for the 2020 Census Demographic and Housing Characteristics File, achieving a 1,824-fold speedup over prior methods while maintaining census-mandated error tolerances.

\end{abstract}

\section{Introduction}
\label{sec:1}

The U.S.\ Decennial Census is among the most influential statistical programs in the world, shaping federal funding allocation \citep{hotchkiss2017uses, funding2023census}, political redistricting \citep{cohen2021census, Kenny2023Comment}, congressional apportionment \citep{eckman2021apportionment,bureau2021approtionment}, and a wide range of labor and economic research \citep{autor2003rise,bureau2023guidance}. Decennial Census products, such as the 2020 Census Demographic and Housing Characteristics (DHC) File, contain highly sensitive demographic information \citep{haney2025safetab}. Consequently, the direct release of exact census counts creates severe privacy risks through reconstruction and re-identification attacks \citep{Abowd_2019,dick2023confidence,Hawes2022census,kulynych2025unifying}. To rigorously mitigate these concerns, the Census Bureau adopted differential privacy \citep{dwork2006our,dwork2006calibrating} as the official confidentiality protection framework for the 2020 Census. Differential privacy provides a mathematical guarantee that changing any single individual's record has only a strictly bounded effect on the released output. For implementation, the Census Bureau utilized the TopDown algorithm \citep{abowd20222020,abowd2022Census,cumings2023disclosure} within its Disclosure Avoidance System (DAS). The DAS injects integer-valued noise, sampled from discrete Gaussian distributions, into a massive collection of counting queries across hierarchical geographic levels, followed by post-processing to ensure data consistency and non-negativity.

The Census Bureau currently provides formal privacy guarantees for the 2020 Census DHC File using zero-concentrated differential privacy (zCDP) \citep{bun2016concentrated, mironov2017renyi}. However, the privacy guarantees derived from zCDP are upper bounds on the (unknown) exact privacy loss. The gap between the bounds and the exact privacy loss leads to a suboptimal privacy--utility trade-off, resulting in the injection of unnecessarily high levels of noise into the census counts for a given privacy budget. This degrades the statistical accuracy of the data and, in turn, impacts downstream applications such as redistricting and social science research \citep{Phil2023bureau,Kenny2021use,Kenny2024census,anderson2015american,boyd2022Differential}. In addition, zCDP parameters lack a straightforward statistical interpretation, often requiring technical expertise to inform public policy discussions \citep{kifer2022bayesian}.

A crucial fact that is frequently overlooked in the literature is that evaluating the \textit{exact} privacy loss for the Census DHC File can be cast as a well-defined mathematical problem. Indeed, recent work by \cite{su20242020} demonstrates that the total privacy loss can be expressed exactly as a composition of trade-off curves associated with the census queries across hierarchical geographic levels under the $f$-differential privacy ($f$-DP) framework \citep{Dong2022Gaussian}. This equivalently reduces the exact accounting problem to computing the tail probabilities of high-dimensional convolutions of heterogeneous discrete Gaussian distributions, that is, discretized counterparts of the continuous Gaussian distribution supported solely on the integers.

While the exact privacy level is conceptually straightforward to express from a mathematical perspective, numerically evaluating it via quadrature presents significant computational challenges. To achieve numerical tractability, prior methods \citep{su20242020} relied on theoretical simplifications of the dataset's sensitivity by adopting the add-remove model instead of the replacement model for neighboring datasets \citep{Cumings2024geographic} and assuming homogeneity of the injected integer noise under composition. While this effort resolved an open question posed by the Census Bureau \citep{kifer2022bayesian}, demonstrating that noise variances could be reduced while maintaining (nearly) the same level of privacy protection, the computational cost proved prohibitive. Specifically, evaluating a single privacy budget allocation required approximately 1{,}000 CPU hours.\footnote{All CPU times in this paper are measured on an AWS EC2 c5.metal instance.} Even worse, the 2020 Census DHC File involves 946 distinct pairs of privacy budget allocations. Requiring 1{,}000 CPU hours for each of the 946 evaluations renders comprehensive, full-scale accounting computationally impractical.

From a mathematical viewpoint, the difficulty of this numerical evaluation stems from several intertwined issues. Unlike continuous Gaussian mechanisms, where convolutions remain Gaussian and admit closed-form expressions, convolutions of discrete Gaussian random variables yield probability mass functions that oscillate rapidly around their continuous approximations. This oscillatory behavior is compounded by high-dimensional, multi-fold compositions (e.g., $80$-fold compositions along geographic paths in the DHC File).\footnote{When the compositional dimension is extremely large, continuous approximations suffice; when the dimension is very small, naive numerical approximations remain viable. An intermediate dimension such as $80$ falls into a difficult regime where neither approach succeeds without substantial error.} This renders standard continuous analytic composition approaches \citep{kairouz2021distributed,zhu2022optimal} inapplicable. Moreover, the discrete Gaussian mechanism used in the 2020 Census is heterogeneous, because the injected noise variances are not identical across queries and geographic levels. This heterogeneity makes accounting substantially more difficult, as it destroys the algebraic structure that would otherwise permit closed-form or low-complexity approximations. Furthermore, the design of the 2020 Census DHC File demands exceptionally stringent numerical error tolerances (e.g., bounded by $10^{-25}$ or $10^{-35}$). Such extreme precision requirements render standard numerical approaches, including fast Fourier transform (FFT)-based quadrature methods \citep{Koskela2020computing, gopi2021numerical}, infeasible due to accumulated truncation errors.

\subsection{Contributions}

In this paper, we develop a highly efficient, mathematically rigorous quadrature-based numerical integration method for the exact privacy accounting of the 2020 Census DHC File that explicitly accommodates fully practical settings. Under identical problem settings, including the same numerical precision requirements, our proposed method requires only 0.5 CPU hours to evaluate a single privacy budget allocation, compared to the roughly 1{,}000 CPU hours required by the prior state-of-the-art \citep{su20242020}.

This acceleration by over three orders of magnitude enables the first exact, assumption-free privacy accounting for all 946 pairs of privacy budget allocations in the 2020 Census DHC File. Consequently, we provide a comprehensively improved privacy--utility trade-off compared to the Bureau's zCDP-based approach \citep{abowd20222020}, demonstrating that exact accounting permits noise variance reductions of 15.08\% to 24.82\% without sacrificing privacy guarantees.
This addresses the open problem posed by the Census Bureau in a substantially more comprehensive manner \citep{kifer2022bayesian}.

Our methodological and computational advances are enabled by three key numerical innovations, each focused on resolving the complexities inherent in composing heterogeneous mechanisms:

\vspace{-1em}
\paragraph{Exponentially convergent quadrature via the trapezoidal rule.}
By reformulating the discrete privacy accounting problem as a numerical integration task via the discrete Fourier transform, we explicitly leverage the exponential convergence of the trapezoidal rule for complex analytic, periodic integrands. This approach achieves the required precision accuracy using exponentially fewer quadrature nodes. Notably, this exponential convergence rate is not automatic; we establish it only after a careful, rigorous complex-analytic analysis of the discrete Gaussian characteristic function, explicitly bounding its growth within a strip in the complex plane.

\vspace{-1em}
\paragraph{Sub-Gaussian tail bounds via lattice mapping for heterogeneous weights.}
To exploit the exponential convergence of the trapezoidal rule, we must first truncate the infinite sum involved in the discrete Fourier transform. An obstacle in doing so arises from the heterogeneity of discrete Gaussian mechanisms under composition, where the weighted sum of the integer-valued noise no longer naturally resides on a simple lattice support. We overcome this by recognizing that the privacy budgets allocated by the Census Bureau are \textit{rational} numbers. By deriving a unified scaling factor that maps the heterogeneous weighted sum onto a common integer lattice via B\'ezout's identity, we preserve the discrete probability mass function without resorting to lossy binning or continuous interpolation. This enables us to establish precise sub-Gaussian concentration inequalities for heterogeneous discrete Gaussian mechanisms, with bounds that guarantee the tail mass falls below the prescribed error margins.

\vspace{-1em}
\paragraph{Accelerated integration via a sieve algorithm.}
Even with exponential convergence, achieving a $10^{-35}$ tolerance in high-dimensional settings demands the evaluation of millions of quadrature nodes. Inspired by sieve methods in number theory, we develop a novel algorithmic procedure that exploits the periodic peak structure of the characteristic functions to rigorously identify and discard nodes where the highly oscillatory integrand is guaranteed to be negligible. This sieve drastically reduces the number of required integrand evaluations from over ${\sim}\,10^6$ nodes down to just $203$.

As an aside, we also provide a theoretical explanation, grounded in the central limit theorem, for why exact privacy accounting for compositions of the discrete Gaussian mechanism strictly diverges from that of the continuous Gaussian mechanism. Specifically, using an Edgeworth expansion, we prove that a central-limit-theorem-type approximation for the convolution of discrete Gaussian distributions involves a non-negligible, higher-order oscillatory discrepancy term. We show that this term is explicitly bounded by $O(m^{-3/2})$, where $m$ denotes the number of folds in the composition.
Our code is available at \url{https://github.com/BuxinSu/Exact-Privacy-Accounting-for-2020-U.S.-Census.git}.

\subsection{Organization of the Paper}
The remainder of the paper is organized as follows.
Section \ref{sec:prelim} reviews preliminaries on differential privacy and the 2020 U.S. Census, with particular emphasis on the 2020 Census DHC File.
Section~\ref{sec:foundation} reformulates the privacy accounting problem as a numerical computation problem.
Section \ref{sec:method} presents our new privacy accounting method, which enables exact computation of the $(\varepsilon, \delta)$-DP and $f$-DP curves for the 2020 Census DHC File; the corresponding empirical results are reported in Section \ref{sec:results}.
Finally, Section \ref{sec:theory} provides theoretical insight into why the discrete Gaussian mechanism differs from the continuous one from the perspective of the central limit theorem.

\section{Preliminaries}
\label{sec:prelim}
In this section, we review the differential privacy frameworks used in our analysis and in the 2020 Census DHC File. 
We then outline how differential privacy is implemented in the Census Bureau's DAS via the discrete Gaussian mechanism and hierarchical privacy budget allocation.

\subsection{Differential Privacy}
\label{sec:dp}

We first introduce three different variants of differential privacy. The most well-known and standard definition is $(\varepsilon, \delta)$-DP \citep{dwork2006our, dwork2006calibrating}. For ease of computing compositions, the method currently adopted by the Census Bureau uses $\rho$-zCDP \citep{bun2016concentrated, dwork2016concentrated, mironov2017renyi, Canonne2020discrete, abowd20222020} to determine the privacy level.
Our improved approach is largely based on $f$-DP \citep{Dong2022Gaussian}, which is defined from a statistical perspective and provides a natural interpretation of the trade-off between privacy level and the power of re-identification inferences \citep{Dong2022Gaussian, kifer2022bayesian}.

\paragraph{$(\varepsilon, \delta)$-DP.} The $(\varepsilon, \delta)$-DP notion, introduced by \cite{dwork2006our, dwork2006calibrating}, is the first formal definition of differential privacy. A randomized mechanism $\widetilde{M}$ is said to satisfy $(\varepsilon, \delta)$-DP for $\varepsilon \geq 0$ and $0 \leq \delta \leq 1$ if, for any pair of neighboring datasets $D$ and $D'$—which differ in a single individual record—and any event $S$, the following holds:
\begin{equation}\label{eq:dp_def}
\PP(\widetilde{M}(D)\in S)\leq \ex^{\varepsilon}\cdot \PP(\widetilde{M}(D')\in S)+\delta.
\end{equation}
$\widetilde{M}$ is regarded as more private when $\varepsilon$ and $\delta$ are small. In particular, when $\varepsilon = \delta = 0$, the distributions of $\widetilde{M}(D)$ and $\widetilde{M}(D')$ are identical, meaning that the mechanism achieves perfect privacy.

\paragraph{$\rho$-zCDP.}
The privacy budget for the 2020 Census DHC File is measured using zero-Concentrated Differential Privacy \citep{bun2016concentrated}, which is based on Rényi divergence. For two distributions $P$ and $Q$ with probability density functions $p$ and $q$, respectively, the Rényi divergence of order $\alpha > 1$ is defined as
$
R_{\alpha}(P\|Q) = \frac{1}{\alpha - 1} \log \int p(x)^{\alpha} q(x)^{1 - \alpha}dx.
$
The quantities $R_{1}(P\|Q)$ and $R_{\infty}(P\|Q)$ are defined as the limits of $R_{\alpha}(P\|Q)$ as $\alpha \to 1$ and $\alpha \to \infty$, respectively.
Based on Rényi divergence, one obtains the definition of zCDP, where the divergence between two random variables is understood as the divergence between their corresponding distributions. 
\begin{definition}[zCDP, \citep{bun2016concentrated}]
    A randomized mechanism $\widetilde{M}$ is said to satisfy $\rho$-zCDP if 
    \begin{align*}
        R_{\alpha}(\widetilde{M}(D)\|\widetilde{M}(D')) \leq \rho \alpha, \qquad \text{for all } \alpha > 1,
    \end{align*}
    for any neighboring datasets $D$ and $D'$.
\end{definition}
When $\rho$ is small, the distributions of $\widetilde{M}(D)$ and $\widetilde{M}(D')$ are closer to each other.
The privacy budget allocated by the Census Bureau \citep{privacyallocation2022} is $\rho = 4.9622$ for the 2020 Census DHC File.

\paragraph{$f$-DP.}
In this paper, we employ the recently developed $f$-DP framework \citep{Dong2022Gaussian}, which has been shown to be well suited for privacy analysis under composition \citep{bu2020deep, wang2024unified, su2025statistical, li2025mitigating}.

To define $f$-DP, consider formulating the problem of distinguishing between a pair of neighboring datasets $D$ and $D'$ as a hypothesis testing problem:
\[
H_0: \text{the true dataset is } D \qquad\text{ versus}\qquad H_1: \text{the true dataset is } D'.
\]
Let $0 \le \phi \le 1$ be any (possibly randomized) rejection rule, and denote by $\alpha_{\phi} = \mathbb{E}_{H_0}[\phi]$ and $\beta_{\phi} = 1 - \mathbb{E}_{H_1}[\phi]$ the type I and type II errors, respectively. The trade-off function $T(\widetilde{M}(D), \widetilde{M}(D')): [0,1] \to [0,1]$ between $D$ and $D'$ is defined as
\[
T(\widetilde{M}(D), \widetilde{M}(D'))(\alpha) = \inf_{\phi}\{\beta_{\phi}: \alpha_{\phi} \le \alpha\}
\]
for any $0 \le \alpha \le 1$ \citep{Dong2022Gaussian}.\footnote{Let $P$ and $Q$ denote the probability distributions of $\widetilde{M}(D)$ and $\widetilde{M}(D')$, correspondingly. Formally, the trade-off function $T(\widetilde{M}(D), \widetilde{M}(D'))$ should be defined through $P$ and $Q$, thereby being expressed as $T(P, Q)$.}
That is, it characterizes the minimal achievable type II error for a given level of type I error.
We say that a mechanism $\widetilde{M}$ satisfies $f$-DP if 
\[
T(\widetilde{M}(D), \widetilde{M}(D'))(\alpha) \geq f(\alpha)
\]
for any neighboring $D$ and $D'$ and any $\alpha \in [0,1]$.
Typically, we require $f:[0,1] \to [0,1]$ to be a valid trade-off function for some pair of distributions. This holds if and only if $f$ is continuous, convex, non-increasing, and satisfies $f(\alpha) \le 1 - \alpha$.

A larger trade-off function $f$ indicates that it is more difficult to distinguish between $H_0$ and $H_1$, and hence the mechanism provides stronger privacy.
It is worth noting that $f$-DP not only enables exact privacy analysis, but also offers a more meaningful semantic interpretation of differential privacy \citep{kifer2022bayesian, gomez2025varepsilon} compared with $(\varepsilon,\delta)$-DP.
Mathematically, an $f$-DP guarantee is equivalent to an infinite collection of guarantees specified by the $(\varepsilon, \delta)$ privacy curve for all $\varepsilon > 0$ \citep{Dong2022Gaussian}.
However, in practice, computing both the exact trade-off function $f$ and the corresponding $(\varepsilon, \delta)$-DP curve is generally challenging, as the distribution of $\widetilde{M}$ may be complex, especially under composition.
In this paper, we present an efficient and practical method for computing the exact $f$-DP curve and the corresponding $(\varepsilon,\delta)$-DP curve.

\paragraph{Composition.}
Composition describes the cumulative privacy loss incurred when multiple mechanisms are applied to the same dataset. For $(\varepsilon,\delta)$-DP, if mechanisms $\{\widetilde{M}_i\}_{i=1}^{k}$ satisfy $(\varepsilon_i,\delta_i)$-DP individually, then their sequential composition satisfies $(\sum_{i=1}^{k}\varepsilon_i,\sum_{i=1}^{k}\delta_i)$-DP.
For zCDP, R\'enyi divergences are additive under composition, so the total privacy parameter $\rho$ is the sum of the individual RDP parameters $\rho_i$ of each mechanism, i.e., $\rho = \sum_{i=1}^k\rho_i$. For $f$-DP, composition corresponds to the tensor product of privacy trade-off functions.
Precisely, we have
\[
T(\{\widetilde{M}_i(D)\}_{i=1}^{k}, \{\widetilde{M}_i(D')\}_{i=1}^{k})
= T\left(\prod_{i=1}^kP_i,\prod_{i=1}^kQ_i\right),
\]
where $P_i$ and $Q_i$ are the underlying distributions of $\widetilde{M}(D)$ and $\widetilde{M}(D')$, correspondingly.
This yields an exact characterization of cumulative privacy through hypothesis-testing trade-offs.
Parallel composition captures the privacy loss when multiple mechanisms are applied to
disjoint subsets of the data.
Suppose the dataset $D$ is partitioned into
$k$ disjoint blocks $D^{(1)},\ldots,D^{(k)}$, and each mechanism $\widetilde{M}_i$
accesses only $D^{(i)}$. The composition of the
sequence
$\widetilde{M}_1(D^{(1)}),\ldots,\widetilde{M}_k(D^{(k)})$ is characterized in Proposition \ref{prop:parallel_comp}.

\paragraph{Sensitivity.}
For a deterministic query $M$ taking values in $\mathbb{R}^{d}$, the $\ell_2$-sensitivity of $M$ is defined as
\begin{align*}
\Delta_2 = \sup_{D,D'} \left\{ \left\| M(D) - M(D') \right\|_{\ell_2} \right\},
\end{align*}
where $\|\cdot\|_{\ell_2}$ denotes the $\ell_2$ norm of a vector, and the supremum is taken over all neighboring datasets $D$ and $D'$ that differ in at most one data record.

In the 2020 Census DHC File, a \emph{bounded} (replacement) sensitivity is used for each marginal query. This reflects the difference between two neighboring datasets under a single-record replacement. 
The DAS \citep{Census2023implementation, abowd20222020} considers the sensitivity of coarsened counting queries with binary categories (e.g., ``18 and older'' vs.\ ``17 and younger''), as detailed in Section~\ref{sec:2020_census}. Replacing one individual with another whose attributes fall into different categories can change each of the counts (``18 and older'' and ``17 and younger'') by 1.

\subsection{The 2020 Census DHC File}
\label{sec:2020_census}

\paragraph{Data structure.}
The 2020 Census DHC File is a major official Census product released by the U.S.\ Census Bureau \citep{census2020CensusDHC}. It is a collection of integer-valued counting statistics, with each corresponding to an item in \{geographic units\} $\times$ \{marginal queries\} \citep{Cumings2024geographic}.

Geographically, the 2020 Census DHC File partitions the United States into eight geographic levels:
\begin{align} \label{eqn:geolevel}
{\text{US},\ \text{State},\ \text{County},\ \text{PRIM},\ \text{Tract Subset},\ \text{Tract Subset Group},\ \text{Block Group},\ \text{Block}}.
\end{align}
Each node, or unit, within a geographic level is referred to as a geographic unit. For example, Pennsylvania is a geographic unit at the State level.
The smallest geographic unit in the 2020 Census DHC File is the Block, which typically corresponds to a very small area, such as a few street blocks in a city. All other geographic units are defined as aggregations of blocks. Therefore, for any fixed block, there is at most one geographic unit at each geographic level \eqref{eqn:geolevel} that contains this block.
Using this property, we define the directed path $\mathcal{P}_b$ of a block $b$ as the collection of all unique geographic units containing this block across the geographic levels \eqref{eqn:geolevel}. For example, the directed path of a block near the University of Pennsylvania would begin with ${\text{US},\ \text{Pennsylvania},\ \text{Philadelphia},\ \cdots}$.

For each geographic unit, the 2020 Census DHC File applies the same set of ten marginal queries, including topics on age, sex, race, Hispanic or Latino origin, household type, family type, relationship to householder, group quarters population, housing occupancy, and housing tenure. 
The collection of integer-valued counting statistics for all items in \{geographic units\} $\times$ \{marginal queries\} constitutes the unprivatized data. 
An example of the counting statistics for items in \{geographic units\} $\times$ \{marginal queries\} is the number of Hispanic people aged 18 or older in Pennsylvania.
The final step in producing the 2020 Census DHC File is to add noise to each entry in the unprivatized data to protect privacy using the DAS~\citep{abowd20222020}, which we introduce later in this section.

\paragraph{Discrete Gaussian Mechanism.}
The noise distribution used in the DAS is the discrete Gaussian distribution \citep{micciancio2007worst, Canonne2020discrete}, a distribution on the integers $\mathbb{Z}$ whose probability mass function is proportional to the Gaussian density. The discrete Gaussian distribution with noise parameter $\sigma^2$, denoted by $\mathcal{N}_{\mathbb{Z}}(0,\sigma^2)$, has probability mass function
\begin{align} \label{eqn:DG}
p_{\sigma}(x) = \frac{\ex^{-x^2/2\sigma^2}}{{\sum_{i\in\mathbb{Z}}\ex^{-i^2/2\sigma^2}}}
\end{align}
for any $x \in \mathbb{Z}$. For any parameter vector $\bsigma^2 = \left[ \sigma_1^2, \cdots, \sigma_n^2 \right]^{\T}$, we define the multivariate discrete Gaussian distribution with parameters $\bsigma^2$ as
$\Nb(0, \bsigma^2) = \left[ {\cal N}_{\ZZ}(0, \sigma_1^2), \cdots, {\cal N}_{\ZZ}(0, \sigma_n^2) \right]^{\T}$.

Let $\mathcal{X}$ denote the sample space, and let $D \subset \mathcal{X}$ be a dataset.
Consider a $d$-dimensional deterministic query $\Mb:\mathcal{X} \rightarrow \mathbb{Z}^d$ that takes integer values.
The discrete Gaussian mechanism adds noise to $\Mb(D)$ as
\begin{align}
\label{eq:DG-mechanisms}
  \widetilde{\Mb}(D) = \Mb(D) + \Nb_{\mathbb{Z}}(0,\bsigma^2).
\end{align}
where $\Nb_{\mathbb{Z}}(0,\bsigma^2)$ denotes the $d$-dimensional multivariate discrete Gaussian distribution.

\paragraph{Disclosure Avoidance System.}
Recall that the unprivatized data form a collection of integer-valued counting statistics for each item in \{geographic units\} $\times$ \{marginal queries\}.
Mathematically, for any fixed directed path $\mathcal{P}_b$, we denote
\begin{align} \label{eqn:counting}
    \mathbf{M}_b = \{M_{(g, q)_1}, M_{(g, q)_2}, \cdots, M_{(g, q)_n}\}
\end{align}
as the collection of integer-valued counting statistics for each $(g,q)$, where $g$ denotes the unique geographic unit containing block $b$ at each geographic level, $q$ denotes the marginal query, and \begin{align} \label{eqn:mech_dim}
    n=|\text{geographic units containing block } b| \times |\text{number of marginal query}|=80.
\end{align}
A typical pair $(g,q)$ could be $(\text{the state containing } b,\ \text{``Hispanic people aged 18 and older''})$, in which case $M_{(g,q)}$ denotes the number of people aged 18 and older in the tract containing block $b$.

The DAS protects data privacy by adding independent discrete Gaussian noise with noise parameter $\sigma_i^2$ to each $M_{(g,q)_i}$ in \eqref{eqn:counting}. The noise parameter $\sigma_i^2$ may vary across different $(g,q)_i$ and is determined as follows.
Along each directed path $\mathcal{P}_b$, the 2020 Census DHC File allocates a privacy budget $\rho_i$ to each item in \eqref{eqn:counting}, such that the sum of all $\rho_i$ equals a fixed overall budget $\rho$. Mathematically, along each directed path $\mathcal{P}_b$, the 2020 Census DHC File assigns a sequence of privacy budgets corresponding one-to-one with \eqref{eqn:counting}:
\begin{align}\label{eqn:rho_i}
    {\boldsymbol{\rho}}_b = \{\rho_1, \rho_2, \cdots, \rho_n\},
\end{align}
where $\rho = \sum_i \rho_i = 4.9622$, and the sequence ${\boldsymbol{\rho}}_b$ is referred to as the privacy budget allocation along $\mathcal{P}_b$.
Note that the allocation ${\boldsymbol{\rho}}_b$ may depend on $b$ and differ across directed paths. One example of ${\boldsymbol{\rho}}_b$ is shown in Table~\ref{tab:PLB_typical}. Among millions of directed paths \citep{abowd20222020}, a total of 43 distinct privacy budget allocations are used in the 2020 Census DHC File.

Given the privacy budget $\rho_i$ for each item in \eqref{eqn:counting}, the 2020 Census DHC File injects independent discrete Gaussian noise with noise parameter $\sigma_i^2$ into each $M_{(g,q)_i}$ as
\begin{align} \label{eqn:noise}
    \widetilde{M}_{(g,q)_i} = M_{(g,q)_i} + \cN_{\ZZ}(0, \sigma_i^2), \quad \sigma_i^2 = \frac{1}{\rho_i}.
\end{align} 
This produces the privatized counting statistics
\begin{align} \label{eqn:private_counting}
    \widetilde{\Mb}_b = \{\widetilde{M}_{(g, q)_1}, \widetilde{M}_{(g, q)_2}, \cdots, \widetilde{M}_{(g, q)_n}\} \quad \text{with noise parameters } \bsigma_{b}^2 = [\sigma_1^2, \sigma_2^2, \cdots, \sigma_n^2].
\end{align}
Under the above mechanism, the 2020 Census DHC File satisfies $\rho$-zCDP, as established in Theorem 2 of \cite{abowd20222020}, due to the composition properties of the discrete Gaussian mechanism \citep{Canonne2020discrete}.\footnote{A post-processing step is applied before releasing the 2020 Census DHC File (see \cite{abowd20222020}). However, since privacy guarantees are invariant under post-processing \citep{dwork2014algorithmic}, we omit the details to avoid confusion.}
In the remainder of the paper, we use $\widetilde{\Mb}_0$ to denote the privatized counting statistics with $\sigma_i^2$ determined by Table~\ref{tab:PLB_typical} and \eqref{eqn:noise}.

\begin{table}[!htbp]
\centering
\caption{One of the $\rho_i$ allocations used in the 2020 Census DHC File is shown here. Each row represents a marginal query, and each column corresponds to a geographic level. The remaining allocations in \eqref{eqn:rho_i} are provided in Section \ref{sec:supp_fig} of the Supplementary Material.}
\label{tab:PLB_typical}
\resizebox{\textwidth}{!}{%
\begin{tabular}{lccccccc}
\hline
Block & Block Group & County & PRIM & State & Tract Subset & Tract Subset Group & US \\
\hline
$11/10000$ & $43/1000$ & $31/1000$ & $239/5000$ & $999/10000$ & $217/2500$ & $239/5000$ & $73/10000$ \\
$11/10000$ & $43/1000$ & $31/1000$ & $239/5000$ & $999/10000$ & $217/2500$ & $239/5000$ & $73/10000$ \\
$11/10000$ & $43/1000$ & $31/1000$ & $239/5000$ & $999/10000$ & $217/2500$ & $239/5000$ & $73/10000$ \\
$11/10000$ & $43/1000$ & $31/1000$ & $239/5000$ & $999/10000$ & $217/2500$ & $239/5000$ & $73/10000$ \\
$11/10000$ & $43/1000$ & $31/1000$ & $239/5000$ & $999/10000$ & $217/2500$ & $239/5000$ & $73/10000$ \\
$11/10000$ & $43/1000$ & $31/1000$ & $239/5000$ & $999/10000$ & $217/2500$ & $239/5000$ & $73/10000$ \\
$11/10000$ & $43/1000$ & $31/1000$ & $239/5000$ & $999/10000$ & $217/2500$ & $239/5000$ & $73/10000$ \\
$11/10000$ & $43/1000$ & $31/1000$ & $239/5000$ & $999/10000$ & $217/2500$ & $239/5000$ & $73/10000$ \\
$11/10000$ & $43/1000$ & $31/250$  & $239/5000$ & $999/10000$ & $217/625$  & $239/5000$ & $73/10000$ \\
$11/10000$ & $43/1000$ & $31/250$  & $239/1250$ & $999/2500$  & $217/625$  & $239/1250$ & $73/2500$ \\
\hline
\end{tabular}%
}
\end{table}

\subsection{Notation}
We summarize here the notation used throughout the paper.

We denote by $\PP$ the probability measure, and by $\ZZ$, $\NN$, $\QQ_{+}$, and $\RR$ the sets of integers, natural numbers, positive rational numbers, and real numbers, respectively.
Bold symbols such as $\Xb$ or $\xb$ denote matrices or vectors, while plain symbols such as $x$ denote scalar quantities.
${\cal N}_{\ZZ}(0, \sigma^2)$ denotes the univariate discrete Gaussian distribution \eqref{eqn:DG} with noise parameter $\sigma^2$.
${\bf N}_{\ZZ}(0, \bsigma^2) = \left[ {\cal N}_{\ZZ}(0, \sigma_1^2), \cdots, {\cal N}_{\ZZ}(0, \sigma_n^2) \right]^{\T}$ denotes the multivariate discrete Gaussian distribution with parameter vector $\bsigma^2 = \left[ \sigma_1^2, \cdots, \sigma_n^2 \right]^{\T}$.

\paragraph{The 2020 Census DHC File.}
We consider the 2020 Census DHC File \citep{census2020CensusDHC}.
Let $b$ denote an individual census block, the smallest geographic unit in the 2020 Census DHC File.
The 2020 Census DHC File includes eight geographic levels, as listed in \eqref{eqn:geolevel}, where at each level there is at most one geographic unit containing block $b$.
We define the directed path $\mathcal{P}_b$ as the collection of all geographic units that contain block $b$.
For any fixed directed path $\mathcal{P}_b$, the same set of marginal queries is applied to every geographic unit, producing the integer-valued counting statistics $\Mb_b$ in \eqref{eqn:counting}.
The 2020 Census DHC File assigns a privacy budget allocation $\boldsymbol{\rho}_b$ as in \eqref{eqn:rho_i} along $\mathcal{P}_b$, with the total budget equal to the overall privacy level $\rho = 4.9622$.
Finally, the 2020 Census DHC File injects independent discrete Gaussian noise with noise parameter $\sigma_i^2$ into each $M_{(g,q)_i}$ as in \eqref{eqn:noise}, producing the privatized counting statistics $\widetilde{\Mb}_b$ in \eqref{eqn:private_counting}.

\paragraph{Privacy accounting.}
For a mechanism $\mathbf{M}$, we write $\widetilde{\mathbf{M}}$ for its privatized (noise-added) version and $\mathbf{M}$ for the corresponding unprivatized mechanism.
We use $\rho$ to denote the overall privacy budget under the zCDP framework.
The pair $(\varepsilon, \delta)$ denotes the privacy parameters in $(\varepsilon, \delta)$-DP.
The pair $(\alpha, \beta)$ denotes the type I and type II errors on the $f$-DP (trade-off) curve.

\section{Privacy Accounting Formulation}
\label{sec:foundation}

In this section, we formulate privacy accounting for the 2020 Census DHC File as a statistical hypothesis testing and numerical computation problem. We first express the overall privacy guarantee through pairwise compositions of neighboring datasets, showing that exact privacy levels reduce to evaluating hypothesis-testing trade-off functions under composed discrete Gaussian mechanisms. We then transform both the $(\varepsilon,\delta)$-DP and $f$-DP privacy profiles into a unified numerical task—accurately computing tail probabilities of weighted convolutions of discrete Gaussian distributions—which serves as the foundation for our methodology.

\subsection{From preliminaries to privacy accounting}
We now reformulate the preliminaries as a privacy accounting problem.
Consider two neighboring datasets $D$ and $D'$ obtained by replacing a record $x$ on directed path $\mathcal{P}_k$ with another record $x'$ on directed path $\mathcal{P}_l$.
For example, this corresponds to replacing an individual living in (US, Pennsylvania, Philadelphia, $\cdots$, the University of Pennsylvania) with another individual living in (US, Massachusetts, Cambridge, $\cdots$, Harvard University). 
Let $\widetilde{\mathbf{M}}_k$ and $\widetilde{\mathbf{M}}_l$ denote the privatized counting statistics along directed path $\mathcal{P}_k$ and $\mathcal{P}_l$, respectively.
Let $m = 2n = 160$ denote the number of dimensions of the composed mechanism $[\widetilde{\mathbf{M}}_k, \widetilde{\mathbf{M}}_l]$, with $n$ as defined in \eqref{eqn:mech_dim}.
\begin{align*}
\begin{split}
    \begin{bmatrix}
        \widetilde{\mathbf{M}}_k(D) \\
        \widetilde{\mathbf{M}}_l(D)
    \end{bmatrix}
    =
    \begin{bmatrix}
        {\mathbf{M}}_k(D) \\
        {\mathbf{M}}_l(D)
    \end{bmatrix}
    +
    \begin{bmatrix}
        \cN_{\ZZ}(0, \sigma_{1}^2) \\
        \cN_{\ZZ}(0, \sigma_{2}^2) \\
        \cdots \\
        \cN_{\ZZ}(0, \sigma_{m}^2)
    \end{bmatrix},
\end{split}
\end{align*}
where $[\sigma_1^2, \sigma_2^2, \cdots, \sigma_m^2] = [\bsigma_{k}^2, \bsigma_{l}^2]$ is the concatenation of the two sequences of noise parameters used in $\widetilde{\mathbf{M}}_k$ and $\widetilde{\mathbf{M}}_l$, as specified in \eqref{eqn:noise} and \eqref{eqn:private_counting}.

Let $\mu = 1$ denote the sensitivity of each coordinate of $\Mb_k$ and $\Mb_l$. Then the composed mechanism $[\widetilde{\mathbf{M}}_k, \widetilde{\mathbf{M}}_l]$ satisfies $f_{kl}$-DP, where $f_{kl}$ is the trade-off function associated with the following hypothesis testing problem:
\begin{align} \label{eqn:hypo_test}
\begin{split}
    H_0: \begin{bmatrix}
        \cN_{\ZZ}(0, \sigma_{1}^2) 
        \\
        \cN_{\ZZ}(0, \sigma_{2}^2)
        \\
        \cdots, 
        \\
        \cN_{\ZZ}(0, \sigma_{m}^2)
    \end{bmatrix}, \qquad H_1: \begin{bmatrix}
        \cN_{\ZZ}(\mu, \sigma_{1}^2) 
        \\
        \cN_{\ZZ}(\mu, \sigma_{2}^2)
        \\
        \cdots, 
        \\
        \cN_{\ZZ}(\mu, \sigma_{m}^2)
    \end{bmatrix}.
\end{split}
\end{align}
The following proposition characterizes the overall privacy level of the 2020 Census DHC File. Its proof follows directly from the fact that the overall privacy guarantee of the 2020 Census DHC File is given by the parallel composition (cf., \cite{Mcsherry2010privacy} and \cite{Smith2022making}) over all pairs of composed mechanisms $[\widetilde{\mathbf{M}}_k, \widetilde{\mathbf{M}}_l]$.
Note that the parallel composition requires taking minimum over all $f_{kl}$ while taking maximum over $\varepsilon_{kl}$. 
\begin{proposition}
    \label{prop:parallel_comp}
    The overall privacy level of the 2020 Census DHC File is $(\min_{k,l} f_{kl})^{**}$ where $(\min_{k,l} f_{kl})^{**}$ is the lower convex envelope \citep{vu:tel-02965421} of $\min_{k,l} f_{kl}$.
    Moreover, suppose the composition of two mechanisms $[ \widetilde{\mathbf{M}}_k, \widetilde{\mathbf{M}}_l ]$ satisfies $(\varepsilon_{kl}, \delta)$-DP, then overall privacy budget in the 2020 Census DHC File is $(\max_{k,l} \varepsilon_{kl}, \delta)$-DP.
\end{proposition}

Proposition \ref{prop:parallel_comp} indicates that, in order to obtain the exact privacy guarantees for the 2020 Census DHC File, it suffices to compute the exact privacy accounting quantities $f_{kl}$ and the $(\varepsilon_{kl}, \delta)$ curve for all pairs of directed paths $k$ and $l$. To this end, we leverage the hypothesis testing formulation \eqref{eqn:hypo_test} of differential privacy \citep{Dong2022Gaussian, wang2022analytical}.

\subsection{From privacy accounting to numerical computing}
\label{sec:to_numerical}
For any given $k$ and $l$, we now cast the privacy accounting of $\varepsilon_{kl}$ and $f_{kl}$ as a numerical computation problem, as stated in Propositions \ref{prop:eps_delta_curve} and \ref{prop:f_dp_curve}. This formulation serves as the cornerstone of our overall privacy accounting analysis for the 2020 Census DHC File.
Section \ref{sec:method}, which contains the main contribution of this paper, introduces a novel method for computing the $(\varepsilon_{kl}, \delta)$ and $f_{kl}$ curves; the corresponding results and practical recommendations are presented in Section \ref{sec:results}.

We first leverage Proposition 3.2 in \cite{wang2022analytical} to characterize the $(\varepsilon, \delta)$-DP curve in terms of the distribution functions of the privacy-loss log-likelihood ratios.
\begin{lemma}[Proposition 3.2 in \cite{wang2022analytical}]
\label{lemma:privacy-profile}
    Let $X_i \sim P_i = \mathcal{N}_{\ZZ}(0,\sigma_i^2)$ and $Y_i \sim Q_i = \mathcal{N}_{\ZZ}(\mu,\sigma_i^2).$ 
    The composition of two mechanisms $[ \widetilde{\mathbf{M}}_k, \widetilde{\mathbf{M}}_l ]$ satisfy  $(\varepsilon_{kl}, \delta)$-DP curve satisfying 
    \begin{align*}
        \delta(\varepsilon_{kl}) = 
        \mathbb{P}\left[\sum_{i=1}^{m} \log \frac{\text{d} Q_i(Y_i)}{\text{d} P_i(Y_i)} > \varepsilon_{kl}\right] - \ex^{\varepsilon_{kl}} \cdot \mathbb{P}\left[\sum_{i=1}^{m} \log \frac{\text{d} Q_i(X_i)}{\text{d} P_i(X_i)} > \varepsilon_{kl}\right]. 
    \end{align*}
\end{lemma}
A straightforward calculation yields the following proposition, which reduces the privacy accounting problem to a numerical computation problem involving the accurate evaluation of tail probabilities of weighted convolutions of discrete Gaussian distributions. Section \ref{sec:method} focuses on addressing this numerical challenge.
\begin{proposition}[Formulation of the exact $(\varepsilon, \delta)$-DP curve]
\label{prop:eps_delta_curve}
    Recall that $\rho=4.9622$ is the total privacy budget in the 2020 Census DHC File. Let $a_i = \rho_i/\rho$. Then, the $(\varepsilon_{kl}, \delta)$-DP curve is exactly characterized by the following equation.
    \begin{align*}
        \delta(\varepsilon_{kl}) =\ & \PP_{X_{i} \sim \cN_{\ZZ}(0, \sigma_{i}^2)} \left( \sum_{i=1}^{m} a_{i} X_{i} > \frac{\varepsilon_{kl}}{\rho} - 1 \right) 
        - \ex^{\varepsilon_{kl}} \cdot \PP_{X_{i} \sim \cN_{\ZZ}(0, \sigma_{i}^2)} \left( \sum_{i=1}^{m} a_{i} X_{i} > \frac{\varepsilon_{kl}}{\rho} + 1 \right).
    \end{align*}
\end{proposition}

By the Neyman--Pearson lemma (see \cite{lehmann2005testing} or Lemma A.1 in \cite{Dong2022Gaussian}), we can reduce the privacy accounting of $f_{kl}$ to a similar numerical computation problem without loss of optimality.
\begin{proposition}[Numerical formulation of the exact $f$-DP curve]
\label{prop:f_dp_curve}
The trade-off function $f_{kl}$ is uniquely determined by the type I error $\alpha_{kl}(\zeta)$ and the type II error $\beta_{kl}(\zeta)$.
\begin{align*}
    \alpha_{kl}(\zeta) =\ & \PP_{X_{i} \sim \cN_{\ZZ}(0, \sigma_{i}^2)} \left( \sum_{i=1}^m a_{i} X_{i} > \zeta + 1 \right) + c \cdot \PP _{X_{i} \sim \cN_{\ZZ}(0, \sigma_{i}^2)} \left( \sum_{i=1}^m a_{i} X_{i} = \zeta + 1 \right)
    \\
    \beta_{kl}(\zeta) =\ & \PP_{X_{i} \sim \cN_{\ZZ}(0, \sigma_{i}^2)} \left( \sum_{i=1}^m a_{i} X_{i} \leq \zeta - 1 \right) - c \cdot \PP _{X_{i} \sim \cN_{\ZZ}(0, \sigma_{i}^2)} \left( \sum_{i=1}^m a_{i} X_{i} = \zeta - 1 \right)
\end{align*}
for some constant $c$. 
\end{proposition}
The proofs of Propositions \ref{prop:eps_delta_curve} and \ref{prop:f_dp_curve} are provided in Section \ref{sec:proof_foundation}.
A direct consequence of Propositions \ref{prop:eps_delta_curve} and \ref{prop:f_dp_curve} is that privacy accounting for the composition of discrete or continuous Gaussian mechanisms is mathematically equivalent to computing the tail probability of a convolution of discrete Gaussian distributions, namely
\begin{align} \label{eqn:conv}
    \mathbb{P}_{X_{i}\sim\mathcal{N}_{\mathbb{Z}}(0,\sigma_i^2)}\left(\sum_{i=1}^{m} a_i X_{i} > t_0 \right)
\end{align}
for any threshold $t_0$. In the next section, we propose an efficient method that can compute the $(\varepsilon_{kl}, \delta)$-DP and $f_{kl}$-DP curves under a prescribed error tolerance.

Before introducing our method, we emphasize the difficulty of numerically computing \eqref{eqn:conv}.
The first challenge arises from the stringent accuracy requirement. In the 2020 Census DHC File \citep{abowd20222020}, an illustrative privacy level is given by $(\varepsilon, \delta) = (26.34, 10^{-10})$. According to Proposition \ref{prop:eps_delta_curve}, an effective numerical error tolerance should be much smaller than $\delta / \ex^{\varepsilon}$.
In Section \ref{sec:epsilondp}, we therefore set the tolerance level to $\Delta = 10^{-35}$ to ensure that numerical errors do not affect the accuracy of the resulting $(\varepsilon, \delta)$-DP curve for the 2020 Census DHC File.
Although the $f$-DP curve imposes less stringent accuracy requirements, we set $\Delta = 10^{-25}$ in Section \ref{sec:fdp} to maintain high precision, particularly near the extreme points of the curve.

The second difficulty arises from the distinction between discrete and continuous Gaussian distributions. Empirically, Figure \ref{fig:discrete-continuous} illustrates the probability mass function of the convolution of two discrete Gaussian distributions and compares it with the corresponding continuous Gaussian density, both in the form of \eqref{eqn:conv}. Notably, the probability mass function of the discrete Gaussian convolution oscillates around the continuous Gaussian density, and a visible discrepancy can be observed between the two curves.
Adopting $\rho$-zCDP for privacy accounting, as in the 2020 Census DHC File, is intuitively equivalent to approximating the discrete Gaussian mechanism by a continuous Gaussian one, which generally yields only an upper bound on the true privacy level \citep{kifer2022bayesian}. Section \ref{sec:theory} provides a more detailed discussion from a theoretical perspective.

\begin{figure}[!htp] 
\centering 
\includegraphics[width=0.7\textwidth]{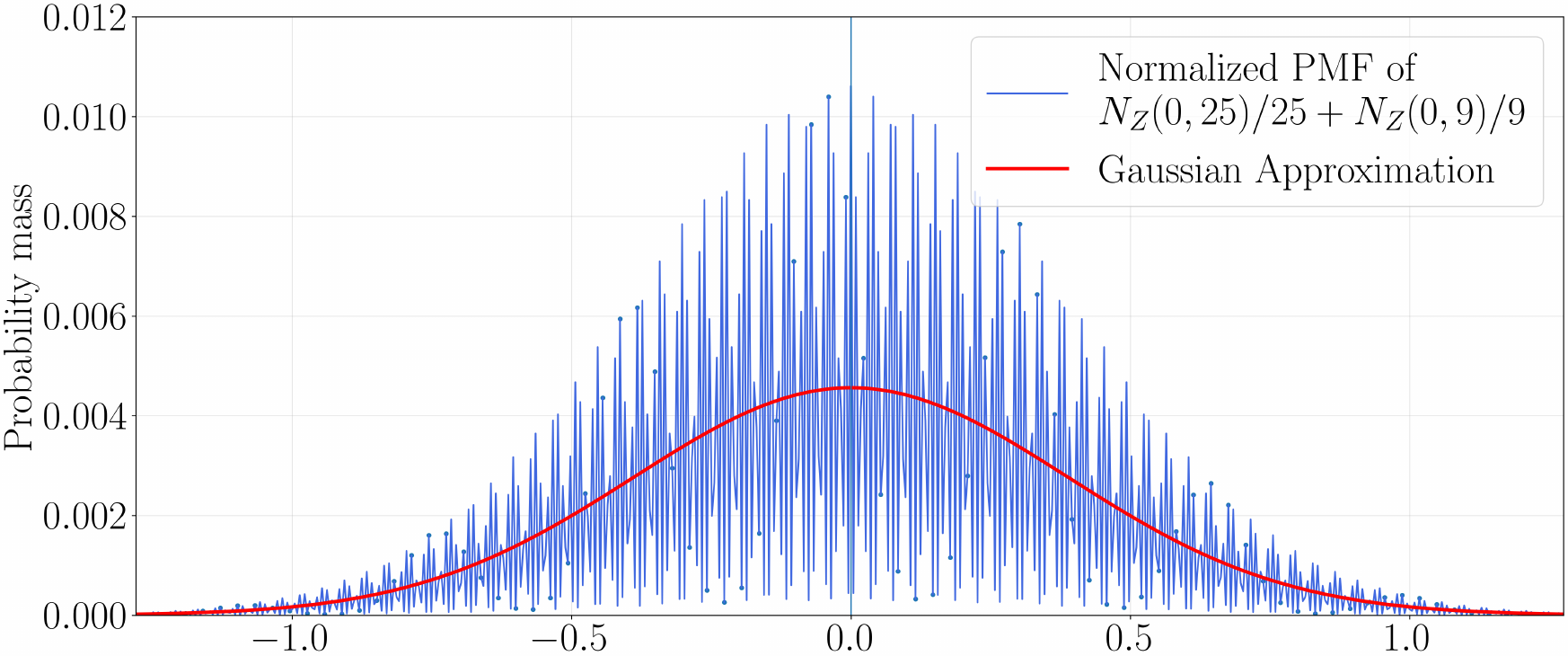}
\caption{Probability mass function of the convolution of two discrete Gaussian distributions and its continuous Gaussian approximation.} 
\label{fig:discrete-continuous} 
\end{figure}

\section{Method}
\label{sec:method}

We now propose a new method for computing the $(\varepsilon_{kl}, \delta)$-DP curve and the $f_{kl}$-DP curve to any prescribed error tolerance for fixed $k$ and $l$. By Propositions \ref{prop:eps_delta_curve} and \ref{prop:f_dp_curve}, this is equivalent to developing a method that computes the upper tail probability of a weighted sum of independent, heterogeneous discrete Gaussian random variables with a uniform error tolerance $\Delta$, as in \eqref{eqn:conv}, for any threshold $t_0$.
Our method consists of four steps, each with a provable guarantee that the overall error does not exceed $\Delta$.
At the end of this section, we summarize the complete method in Theorem \ref{thm:method}.

\paragraph{Step 1: Tail bound.}
The goal of this step is to truncate $\sum_{i=1}^{m} a_i L X_i$ so that the tail probability beyond the truncation is at most $\Delta/4$.
To this end, note that all $a_i \in \QQ_+$ in the 2020 Census DHC File. Hence, there exists a unique $L \in \QQ_+$ such that $a_i L \in \NN$ for all $i$ and $\gcd(a_1 L, \cdots, a_m L) = 1$.
By Bézout's identity,
it follows that
\begin{align*}
\text{support} \left( \sum_{i=1}^{m} a_i L X_i \right) = \ZZ.
\end{align*}
We now state a practical choice of the truncation level $U$ depending on $L$. The proof, given in Section \ref{sec:proof_choice_U}, is based on a sub-Gaussian concentration inequality.
\begin{proposition}[Practical choice of truncation $U$]
\label{prop:choice_U}
    Let $U$ be the smallest integer larger than
    \begin{align} \label{eqn:choice_U}
         L \cdot \sqrt{ - 2 \left( \sum_{i=1}^{m} a_i^2 \sigma_{i}^2 \right) \cdot \log \left( \Delta/8 \right)},
    \end{align}  
    we have
    \begin{align*}
        \left| \mathbb{P}_{X_{i}\sim\mathcal{N}_{\mathbb{Z}}(0,\sigma_i^2)}\left(\sum_{i=1}^{m} a_i X_{i} > t_0 \right) - \PP_{X_{i}\sim\mathcal{N}_{\mathbb{Z}}(0,\sigma_i^2)} \left( U \geq \sum_{i=1}^{m} a_i L X_{i} > t_0 \cdot L \right) \right| \leq \Delta/4,
    \end{align*}
\end{proposition}
With this truncation, instead of computing the probability mass function of $\sum_{i=1}^{m} a_i L X_i$ over an infinite support beyond $t_0$, it suffices to consider only the finite set of values for which $\sum_{i=1}^{m} a_i L X_i$ lies between $t_0 L$ and $U$.

\paragraph{Step 2: Discrete Fourier transform.}

The goal of this step is to express the truncated cumulative probability of $\sum_{i=1}^{m} a_i L X_i$ as a single integral, so that it can be evaluated numerically in later steps. We begin with the following decomposition:
\begin{align*}
\begin{split}
    \PP_{X_{i}\sim\mathcal{N}_{\mathbb{Z}}(0,\sigma_i^2)} \left( U \geq \sum_{i=1}^{m} a_i L X_{i} > t_0 \cdot L \right) =\ & \sum_{ t = \lceil t_0 \cdot L \rceil}^{U} \PP_{X_{i}\sim\mathcal{N}_{\mathbb{Z}}(0,\sigma_i^2)} \left( \sum_{i=1}^{m} a_i L X_{i} = t \right) 
    \\
    \overset{(a)}{=}\ & \frac{1}{2 \pi} \sum_{ t = \lceil t_0 \cdot L \rceil}^{U} \int_{- \pi}^{\pi} \ex^{- i \zeta t} \prod_{i=1}^{m} f_{a_i L X_{i}}(\zeta) d\zeta,
\end{split}
\end{align*}
where $f_{a_i L X_i}(\zeta) = \EE \ex^{i\zeta a_i L X_i}$ denotes the characteristic function of $a_i L X_i$, and step (a) follows from the inverse discrete Fourier transform (see Exercise 3.3.2(iii) in \cite{durrett2019probability}).
The characteristic function has the closed-form expression
\begin{align*}
    f_{a_i L X_{i}}(\zeta) = \EE \ex^{i\zeta \cdot a_i L X_{i}}  = f_{X_{i}}(a_i L \zeta) = \ & \frac{\sum_{u=-\infty}^{\infty} \ex^{i u \cdot a_i L \zeta} \cdot \ex^{ - \frac{u^2}{2 \sigma_{i}^{2}} }}{ \sum_{u=-\infty}^{\infty}  \ex^{ - \frac{u^2}{2 \sigma_{i}^{2}} }}.
\end{align*}
Let $X = \sum_{i=1}^{m} a_i L X_i$. Since $\EE \ex^{i\zeta X} = \prod_{i=1}^{m} f_{a_i L X_i}(\zeta)$ is an even function of $\zeta$, we obtain
\begin{align*}
    \frac{1}{2 \pi} \sum_{t = \lceil t_0 \cdot L \rceil}^{U} \int_{- \pi}^{\pi} \ex^{- i \zeta t} \prod_{i=1}^{m} f_{a_i L X_{i}}(\zeta) d\zeta 
    =\ & \int_{- \pi}^{\pi} \frac{1}{2 \pi} \sum_{t = \lceil t_0 \cdot L \rceil}^{U} \ex^{- i \zeta t} \prod_{i=1}^{m} f_{a_i L X_{i}}(\zeta) d\zeta 
    \\
    =\ & \int_{- \pi}^{\pi} \frac{1}{2 \pi} \sum_{ t = \lceil t_0 \cdot L \rceil}^{U} \cos(\zeta t) \cdot \EE \ex^{i \zeta X}.
\end{align*}
Define
\begin{align} \label{eqn:F}
    F(\zeta)
    =\ & \frac{1}{2 \pi} \cdot \left[ \sum_{ t = \lceil t_0 \cdot L \rceil}^{U} \cos(\zeta t) \right] \cdot \EE \ex^{i \zeta X}.
\end{align}
We conclude Step 2 with the following result.
\begin{proposition}
\label{prop:fourier}
With $F(\zeta)$ defined in \eqref{eqn:F}, we have
\begin{align*} 
    \PP_{X_{i}\sim\mathcal{N}_{\mathbb{Z}}(0,\sigma_i^2)} \left( U \geq \sum_{i=1}^{m} a_i L X_{i} > t_0 \cdot L \right) =\ & \int_{-\pi}^{\pi} F(\zeta) d \zeta.
\end{align*}
where $U$ is chosen according to Proposition \ref{prop:choice_U}.
\end{proposition}

\paragraph{Step 3: Trapezoidal rule.}
With the integral representation in Proposition \ref{prop:fourier}, we next compute $\int_{-\pi}^{\pi} F(\zeta) d\zeta$ numerically. Since $F(\zeta)$ is $2\pi$-periodic, we may equivalently write
\begin{align*}
    \int_{-\pi}^{\pi} F(\zeta) d \zeta = \int_{0}^{2 \pi} F(\zeta) d\zeta.
\end{align*}
We adopt the trapezoidal rule, together with the error bound in Lemma \ref{lemma:trapezoidal_error}. For a number of nodes $N \in \NN$, the trapezoidal approximation takes the form
\begin{equation}
    I_N = \frac{2 \pi}{N} \sum_{k=1}^{N} F(\zeta_{k}),
\end{equation}
where $\zeta_{k} = 2 \pi k /N$.
The following lemma, due to \cite{Trefethen2014exponentially}, shows that for periodic complex analytic functions the trapezoidal rule achieves exponential convergence in $N$.
\begin{lemma}[Theorem 3.2 in \cite{Trefethen2014exponentially}]
\label{lemma:trapezoidal_error}
Consider a $2\pi$-periodic complex analytic function $F(\zeta)$ satisfying $|F(\zeta)| \le M$ on the strip $-a < \mathrm{Im}\, \zeta < a$ for some $a>0$. Then, for any $N \ge 1$,
\begin{equation*}
    \left|\frac{2 \pi}{N} \sum_{k=1}^{N} F(\zeta_{k}) - \int_{0}^{2 \pi} F(\zeta) d\zeta\right| \leq \frac{4\pi M}{\ex^{aN} -1},
\end{equation*}
and the constant $4\pi$ is optimal.
\end{lemma}
The trapezoidal rule is simple to implement and converges dramatically faster. Empirically, the convergence is often even faster than exponential (see, e.g., \cite{TCS-065, Trefethen2014exponentially}). To our knowledge, this phenomenon has not been explicitly exploited in differential privacy, and it suggests a general acceleration strategy for previous numerical accountants \citep{Koskela2020computing, gopi2021numerical}.
Based on Lemma \ref{lemma:trapezoidal_error}, we obtain the following practical choice of $N$.
\begin{proposition}[Practical Choice of $N$]
\label{prop:choice_N}
For any $N \in \NN$ satisfying
\begin{align} \label{eqn:choice_N}
    N \geq 2 \cdot (U + 1) + dL \log(2 \cdot dL) \quad \text{with} \quad d = \sqrt{- \frac{ \sum_{i=1}^{m} a_i^2 \sigma_{i}^2}{2 \log (\Delta/8)}},  
\end{align}
we have
\begin{equation*}
    \left|\frac{2 \pi}{N} \sum_{k=1}^{N} F (\zeta_{k}) - \int_{0}^{2\pi} F(\zeta) d \zeta \right| \leq \Delta/4, \quad \text{with} \quad \zeta_{k} = \frac{2 \pi k}{N}.
\end{equation*}
\end{proposition}

\paragraph{Step 4: Sieve-type truncation.}
The remaining challenge is that evaluating $F(\zeta)$ at all $N$ nodes may be computationally expensive. With the stringent tolerance level $\Delta = 10^{-35}$ used in our experiments in Section \ref{sec:epsilondp}, the required $N$ is often $\gg 10^{6}$, depending on the choice of $k$ and $l$. To improve efficiency, we introduce an additional truncation of the trapezoidal sum by identifying a subset $\cC \subseteq [N]$ that satisfies the following desiderata:
First, $F(\zeta_k)$ is negligible for all $k \in [N]\setminus \cC$. Specifically, we want $$\frac{2 \pi}{N} \sum_{k \in [N] \setminus \cC} F (\zeta_{k}) \leq \Delta/4,$$ so that $\frac{2\pi}{N}\sum_{k \in \cC} F(\zeta_k)$ serves as an accurate approximation to $\frac{2\pi}{N}\sum_{k=1}^{N} F(\zeta_k)$.
Second, the set $\cC$ is as small as possible, so that evaluating $\frac{2\pi}{N}\sum_{k \in \cC} F(\zeta_k)$ is substantially faster than evaluating the full sum over $[N]$. 
Finally, the procedure for constructing $\cC$ is itself computationally efficient.

To this end, we propose Algorithm~\ref{alg:char_filter} to construct $\cC \subseteq [N]$, which is theoretically motivated by sieve methods in number theory. In our implementation, standard algorithmic tools (e.g., dynamic programming and binary search) are used to accelerate the selection of $\cC$. Let $X_{(i)} \sim \cN_{\ZZ}(0,\sigma_{(i)}^2)$ denote the random variable among $\{X_i \sim \cN_{\ZZ}(0,\sigma_i^2)\}_{i=1}^{m}$ with the $i$th largest noise parameter $\sigma_i^2$. Let $a_{(i)}$ be the corresponding weight in $\{a_i\}_{i=1}^{m}$, and write $f_{X_{(i)}}(a_{(i)}L\zeta)$ for the characteristic function of $a_{(i)} L X_{(i)}$.
Algorithm~\ref{alg:char_filter} proceeds as follows. Starting from $X_{(1)}$ with noise parameter $\sigma_{(1)}^2$, we discard all nodes $k \in [N]$ for which the characteristic function is sufficiently small, namely those satisfying $f_{X_{(1)}}(a_{(1)} L \zeta) < \Delta/(8 \cdot U)$. 
By the argument in Section \ref{sec:proof_method}, for any $\zeta$ satisfying $\bigl|f_{X_{(1)}}(a_{(1)}L\zeta)\bigr| < \Delta/(8 \cdot U)$, we have $2\pi \cdot |F(\zeta)| \le \Delta/4$. The algorithm then repeats the same filtering step sequentially for $X_{(2)},\ldots,X_{(m)}$, and the surviving indices constitute the final set $\cC$. 

\vspace{1mm}
\begin{algorithm}[!htp]
\SetAlgoLined
\KwIn{Number of nodes $N$; collection of all nodes $\{{2 \pi k}/{N}: k \in [N]\}$; discrete Gaussian distribution $X_{(i)}$ with $i$-th largest parameter $\sigma_{(i)}^2$; $f_{X_{(i)}}(a_{(i)} L \zeta)$ to be characteristic function of $a_{(i)} L X_{(i)}$.}
\vspace{1.5mm}

Candidate Nodes = [N] \hfill \texttt{//} initialization
\vspace{1.5mm}

\For{$i = 1$ \KwTo $m$}{
    Search for the smallest $a$ such that $\zeta^{(i)} = a \cdot \frac{\pi}{a_{(i)} L}$ satisfies $\left| f_{X_{(i)}}(a_{(i)} L \zeta^{(i)}) \right| < \frac{\Delta}{8 U}$\;
    \tcp{If $X_{(i)}$ has $r$ i.i.d.\ copies among $\{X_{j}\}_{j=1}^{m}$, then replace the test by $\left| f_{X_{(i)}}(a_{(i)} L \zeta^{(i)}) \right|^{r} < \frac{\Delta}{8 U}$.}
    \vspace{3mm}
    Update Candidate Nodes by \hfill \tcp{Keep ONLY the nodes that are close to peak}
    \vspace{-4mm}
    \begin{equation} \label{eqn:update_node}
    \begin{split}
        &\text{Candidate Nodes} = \text{Candidate Nodes}
        \\
        &\qquad \qquad \bigcap \left\{k \in [N]: \frac{2 \pi k}{N} \in \bigcup_{0 \leq j \leq a_{(i)} L} \left[ \frac{2 \pi}{a_{(i)} L} \cdot j - \zeta^{(i)}, \frac{2 \pi}{a_{(i)} L} \cdot j + \zeta^{(i)} \right] \right\}.
    \end{split}
    \end{equation}
}
\vspace{1.5mm}

\KwOut{$\cC =$ Candidate Nodes}
\caption{\textsc{Truncation}: computationally efficient to find $\cC$.}
\label{alg:char_filter}
\end{algorithm}


\vspace{1em}
Mathematically, $\cC$ is defined to be 
\begin{align*}
    \cC = \bigcap_{i=1}^{m} \bigcup_{0 \leq j \leq a_{(i)} L} \left[ \frac{2 \pi}{a_{(i)} L} \cdot j - \zeta^{(i)}, \frac{2 \pi}{a_{(i)} L} \cdot j + \zeta^{(i)} \right]
\end{align*}
The following proposition concludes Algorithm \ref{alg:char_filter} by showing that restricting the trapezoidal sum to the indices in $\cC$ incurs an error of at most $\Delta/4$. The proof of Proposition \ref{prop:C} is deferred to Appendix \ref{sec:proof_char_filter}.

\begin{proposition} \label{prop:C}
With $\cC$ output by Algorithm \ref{alg:char_filter}, we have
\begin{align*}
    \left| \frac{2 \pi}{N} \sum_{k=1}^{N} F (\zeta_{k}) - \frac{2 \pi}{N} \sum_{k \in \cC} F (\zeta_{k}) \right| \leq \Delta/4, \quad \text{with} \quad \zeta_{k} = \frac{2 \pi k}{N}.
\end{align*}
\end{proposition}

We now provide further intuition behind Algorithm \ref{alg:char_filter}. We begin with several basic properties of the characteristic function $f_{X_i}(\zeta)$ for all $i$.
\begin{lemma}
\label{lemma:char_property}
The characteristic function $f_{X_i}(\zeta)$ satisfies the following properties:
\begin{enumerate}
    \item $|f_{X_i}(\zeta)| \le 1$, and $f_{X_i}(\zeta)$ attains its maximum at $\zeta=0$ with $f_{X_i}(0)=1$.
    \item $f_{X_i}(\zeta)$ is $2\pi$-periodic; consequently, $f_{X_i}(a_i L \zeta)$ is periodic with period $2\pi/(a_i L)$.
    \item $f_{X_i}(\zeta)$ is strictly increasing on $(-\pi,0)$ and strictly decreasing on $(0,\pi)$.
\end{enumerate}
\end{lemma}

Lemma \ref{lemma:char_property} shows that $|f_{X_i}(a_i L \zeta)|$ exhibits periodic peaks at $\frac{2\pi}{a_{(i)} L} \cdot k$ for all integers $k$. The key idea of Algorithm \ref{alg:char_filter} is therefore to sieve (i.e., discard) those nodes at which at least one characteristic function $f_{X_{(i)}}(a_{(i)} L \zeta)$ is sufficiently small, exploiting this periodic structure.
In each iteration, we first determine a radius $\zeta^{(i)}$ such that $|f_{X_i}(a_i L \zeta)| < \Delta/(8 \cdot U)$ whenever $\zeta \notin [0,\zeta^{(i)}]$. By periodicity, any $\zeta_k$ whose distance to one of the points $\frac{2 \pi}{a_{(i)} L} \cdot j$ is greater than $\zeta^{(i)}$ satisfies $|f_{X_i}(a_i L \zeta)| < \Delta/(8 \cdot U)$, which leads to the update rule in \eqref{eqn:update_node}.
In Section \ref{sec:proof_method}, we show that removing all nodes $\zeta_k$ outside the union described in \eqref{eqn:update_node} incurs an error of at most $\Delta/4$. Hence, Algorithm \ref{alg:char_filter} enables the efficient yet rigorous elimination of nodes whose values of $F$ are negligible.

We now provide further intuition for starting the iteration with $X_{(1)}$, which has the largest noise parameter $\sigma_{(1)}^2$.
When $N$, or the number of Candidate Nodes, is large, it can be computationally expensive to identify all nodes whose distance to the set $\{\frac{2 \pi}{a_{(i)} L} \cdot j: 0 \leq j \leq a_{(i)} L \}$ is less than $\zeta^{(i)}$. We therefore begin with $X_{(1)}$, corresponding to the largest $\sigma_{(1)}^2$.
Observe that a larger noise parameter $\sigma_{(1)}^2$ corresponds to a smaller weight $a_{(1)}$. Consequently, the update rule in \eqref{eqn:update_node} involves the smallest number of periodic unions, each with the largest interval length. Empirically, this ordering provides the fastest reduction in the set of Candidate Nodes.

Algorithm \ref{alg:char_filter} substantially accelerates the computation. Empirically, the dominant computational cost lies in evaluating the function $F(\zeta)$. Under the allocation in Table \ref{tab:PLB_typical}, computing $\int_{-\pi}^{\pi} F(\zeta)d\zeta$ without Algorithm \ref{alg:char_filter} would require evaluating $F(\zeta_k)$ at all $k \in [N]$ with $N = 10^{6}$. Consequently, each subplot in Figure \ref{fig:eps_delta_paths_bypass} would involve more than $4 \times 10^{7}$ evaluations of $F(\zeta)$ at high precision, which is infeasible on a single CPU.
With Algorithm \ref{alg:char_filter}, our method evaluates $F(\zeta_k)$ only for $k \in \cC$, where $|\cC| = 203$. As a result, each subplot in Figure \ref{fig:eps_delta_paths_bypass} requires only about $8{,}120$ evaluations of $F(\zeta)$, and can be completed within approximately 0.5 hours.

We now summarize the entire procedure in the following theorem.
\begin{theorem}
\label{thm:method}
Equation \eqref{eqn:conv} can be numerically evaluated with error tolerance $\Delta$ using 
\begin{equation}
    \left| \mathbb{P}_{X_{i}\sim\mathcal{N}_{\mathbb{Z}}(0,\sigma_i^2)}\left(\sum_{i=1}^{m} a_i X_{i} > t_0 \right) - \frac{2 \pi}{N} \sum_{k \in \cC} F (\zeta_{k}) \right| \leq \Delta \quad \text{with} \quad \zeta_{k} = \frac{2 \pi k}{N},
\end{equation}
where
\begin{align*}
    F(\zeta)
    =\ & \frac{1}{2 \pi} \left[ \sum_{ t = \lceil t_0 \cdot L \rceil}^{U} \cos(\zeta t) \right] \cdot \EE \ex^{i \zeta X} = \frac{1}{2 \pi} \cdot \text{weight}(\zeta) \cdot \text{char}(\zeta),
\end{align*}
and $\text{weight}(\zeta)$ and $\text{char}(\zeta)$ have an explicit evaluable expression as the following when $\zeta \neq 0$ or $2\pi$.
\begin{align*}
    &\text{weight}(\zeta) = \sum_{ t = \lceil t_0 \cdot L \rceil}^{U} \cos(\zeta t) = \frac{1}{2} \left[\cos( \left\lceil t_0 L \right\rceil \cdot \zeta) + \cos(U \cdot \zeta) + \frac{\cos(\zeta/2)}{\sin(\zeta/2)} \left(\sin(U \cdot \zeta) - \sin(\left\lceil t_0 L \right\rceil \cdot \zeta) \right) \right],
    \\
    &\text{char}(\zeta) = \EE \ex^{i \zeta X} = \prod_{i=1}^{m} f_{X_{i}}(a_i L \zeta), \quad \text{with} \quad f_{X_{i}}(a_i L \zeta) = \left\{ \sum_{x_{i} = -\infty}^{\infty} \ex^{i a_i L \zeta x_i} \cdot \ex^{-x_{i}^2/2 \sigma_{i}^2} \right\} \cdot \left\{ \sum_{x_{i} = -\infty}^{\infty} \ex^{-x_{i}^2/2 \sigma_{i}^2} \right\}^{-1}.
\end{align*} 
Here $U$, $N$, and $\cC$ are chosen according to Propositions \ref{prop:choice_U}, \ref{prop:choice_N}, and \ref{prop:C}, respectively.
\end{theorem}

Theorem \ref{thm:method} has value beyond the specific application to the 2020 Census DHC File. It introduces a novel numerical method for privacy accounting that provably achieves exponential convergence, representing a substantial improvement over first-order methods \citep{koskela2021tight, gopi2021numerical} and existing higher-order approaches \citep{su20242020}, making it broadly applicable to a wide range of privacy accounting problems. 
Moreover, for privacy accounting tasks that require a large number of nodes, strategies analogous to Algorithm \ref{alg:char_filter} can be employed to dramatically reduce computational cost.

\section{Results}
\label{sec:results}

We extract the privacy budget allocations for all directed paths from the 2020 Census DHC File. Although there are millions of directed paths \citep{abowd2022Census}, the 2020 Census DHC File employs only 43 distinct privacy budget allocations across all paths.
Section \ref{sec:method} provides an efficient and practical procedure for computing $f_{kl}$ and $(\varepsilon_{kl}, \delta)$ for any pair of directed paths $k$ and $l$. Recall from Proposition \ref{prop:parallel_comp} that the overall privacy guarantee is given by $(\min_{k,l} f_{kl})^{**}$-DP and $(\max_{k,l} \varepsilon_{kl}, \delta)$-DP.
To compute these quantities, we leverage parallel computation on a CPU-based infrastructure and launch $43 \times 42/2 + 43 = 946$ jobs simultaneously, where each job evaluates the privacy level $f_{kl}$ or $(\varepsilon_{kl}, \delta)$ for a given pair $k,l \in [43]$. 
For the $f$-DP guarantee, we compute the lower convex envelope of $\min_{k,l} f_{kl}$ using the monotone chain algorithm, with results presented in Section \ref{sec:fdp}.
The overall $(\varepsilon, \delta)$-DP curve of the 2020 Census DHC File is then obtained by taking the pointwise maximum over $(\varepsilon_{kl}, \delta)$ across all pairs $(k,l)$ and the results are reported in Section \ref{sec:epsilondp}.
In practice, the U.S. Census Bureau tests multiple privacy budgets during the tuning stage and evaluates their performance on downstream tasks \citep{censusDemonstrationData, censusjsm, censusAnnouncing2030, censusDecennialCensus}. Motivated by our empirical findings, Section~\ref{sec:speedup} provides practical guidance for efficiently applying our method during privacy budget tuning, thereby avoiding repeated full-scale parallel computation of the overall privacy accounting.

\begin{figure}[!htp] 
\centering 
\begin{subfigure}[b]{0.4\textwidth}
\includegraphics[width=\textwidth]{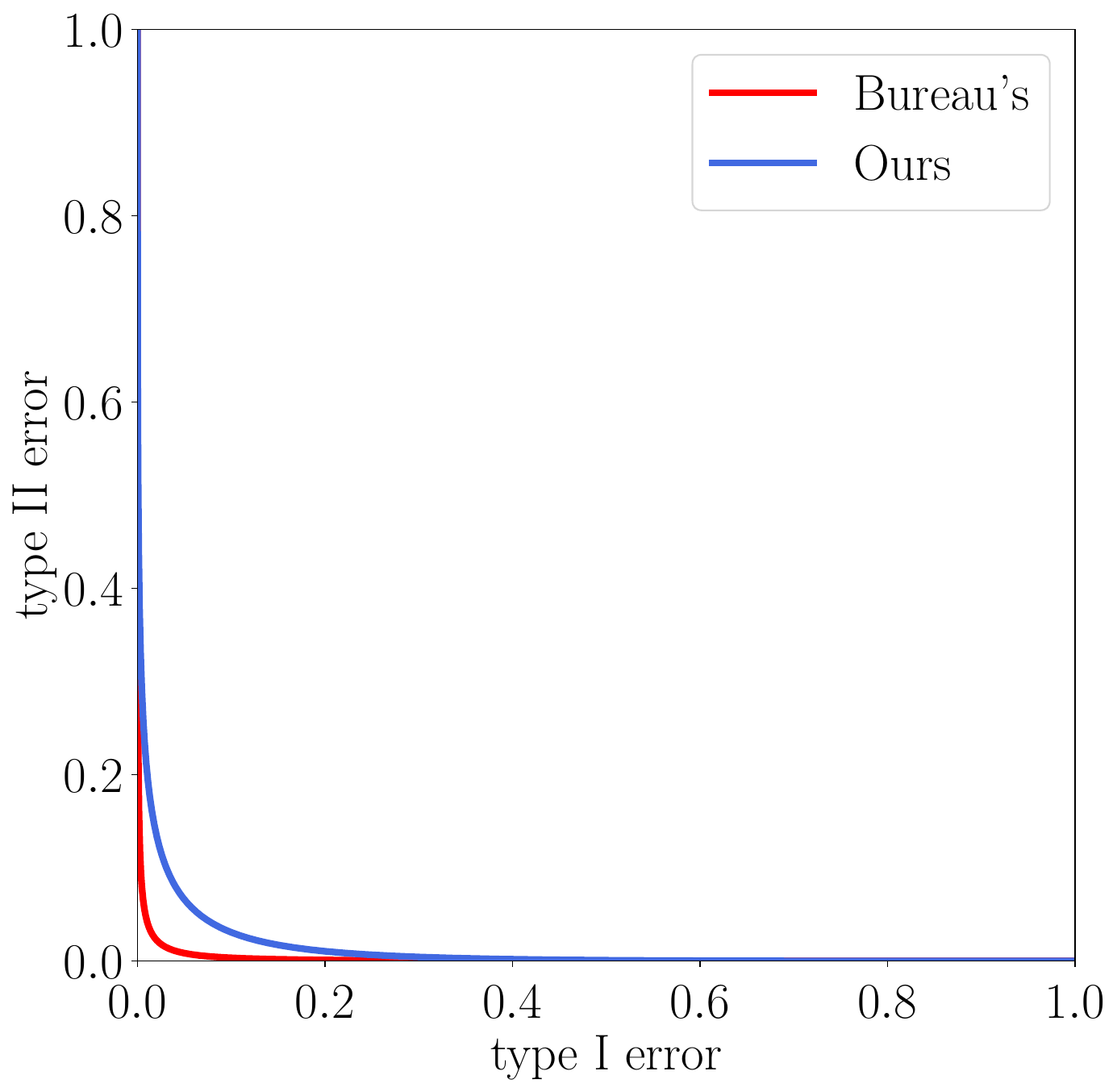} 
\end{subfigure}
\begin{subfigure}[b]{0.4\textwidth}
\includegraphics[width=\textwidth]{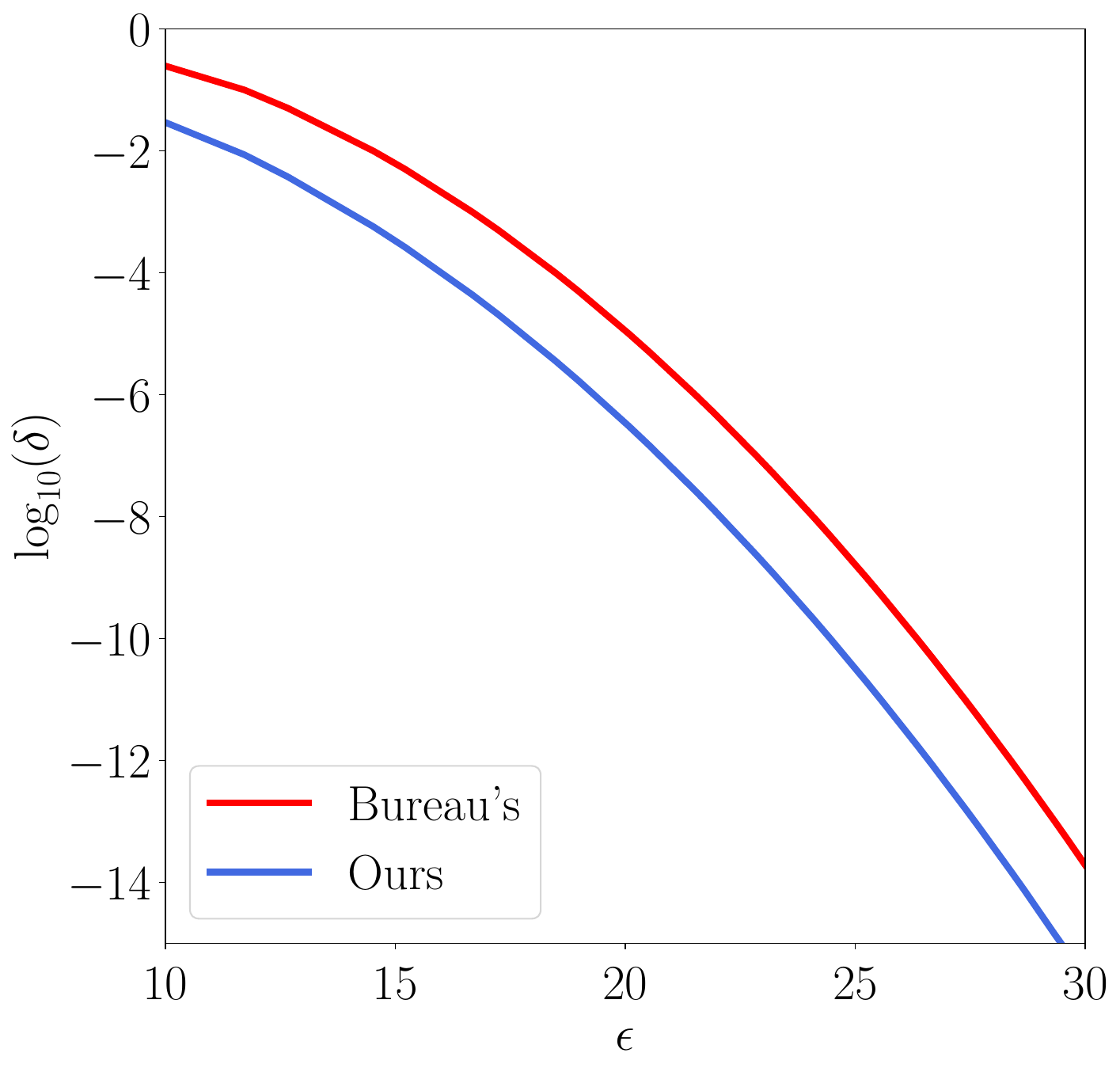}
\end{subfigure}
\caption{The $f$-DP curves and $(\varepsilon, \delta)$-DP curve for the 2020 Census DHC File computed using our method (blue) and the Census Bureau's accounting method (red).} 
\label{fig:overall-privacy} 
\end{figure}

\subsection{\texorpdfstring{$f$}{}-DP curves}
\label{sec:fdp}

The left panel of Figure \ref{fig:overall-privacy} characterizes the overall privacy guarantee of the 2020 Census DHC File in terms of $f$-DP curves. The blue curve represents the exact $f$-DP guarantee computed using our method, while the red curve is obtained by converting the $\rho$-zCDP guarantee to $f$-DP using AutoDP \citep{wang2019subsampled, zhu2019poission, zhu2020improving}. Under the same mechanism and injected noise, our method yields a strictly better privacy–utility trade-off.
Figure \ref{fig:trade_off_paths_bypass} presents 43 subplots illustrating the privacy guarantees of the composed mechanisms $[\widetilde{\mathbf{M}}_0, \widetilde{\mathbf{M}}_l]$, where $\widetilde{\mathbf{M}}_0$ corresponds to the privacy budget allocation in Table \ref{tab:PLB_typical}. The privacy guarantees of all composed mechanisms $[\widetilde{\mathbf{M}}_k, \widetilde{\mathbf{M}}_l]$ are available at \url{https://github.com/BuxinSu/Exact-Privacy-Accounting-for-2020-U.S.-Census.git}. Although there are 946 distinct composed mechanisms, their privacy profiles are different yet highly similar, as illustrated in Figure \ref{fig:trade_off_paths_bypass}.

\newlength{\cellh}
\setlength{\cellh}{2.5cm}

\newcommand{\tradeofflm}[1]{
  \parbox[c][\cellh][c]{0.16\textwidth}{\centering
    \includegraphics[width=\linewidth,height=\cellh,keepaspectratio]{Figures/v2/trade_off_curve_path_#1_to_13.pdf}
  }
}
\newcommand{\tradeoffsm}[1]{
  \parbox[c][\cellh][c]{0.16\textwidth}{\centering
    \includegraphics[width=\linewidth,height=\cellh,keepaspectratio]{Figures/v2/trade_off_curve_path_13_to_#1.pdf}
  }
}
\newcommand{\tradeoffbigfig}[1]{
  \parbox[t][2\cellh][t]{0.32\textwidth}{\centering\vspace{-0.85cm}
  \includegraphics[width=\linewidth,height=2\cellh,keepaspectratio]{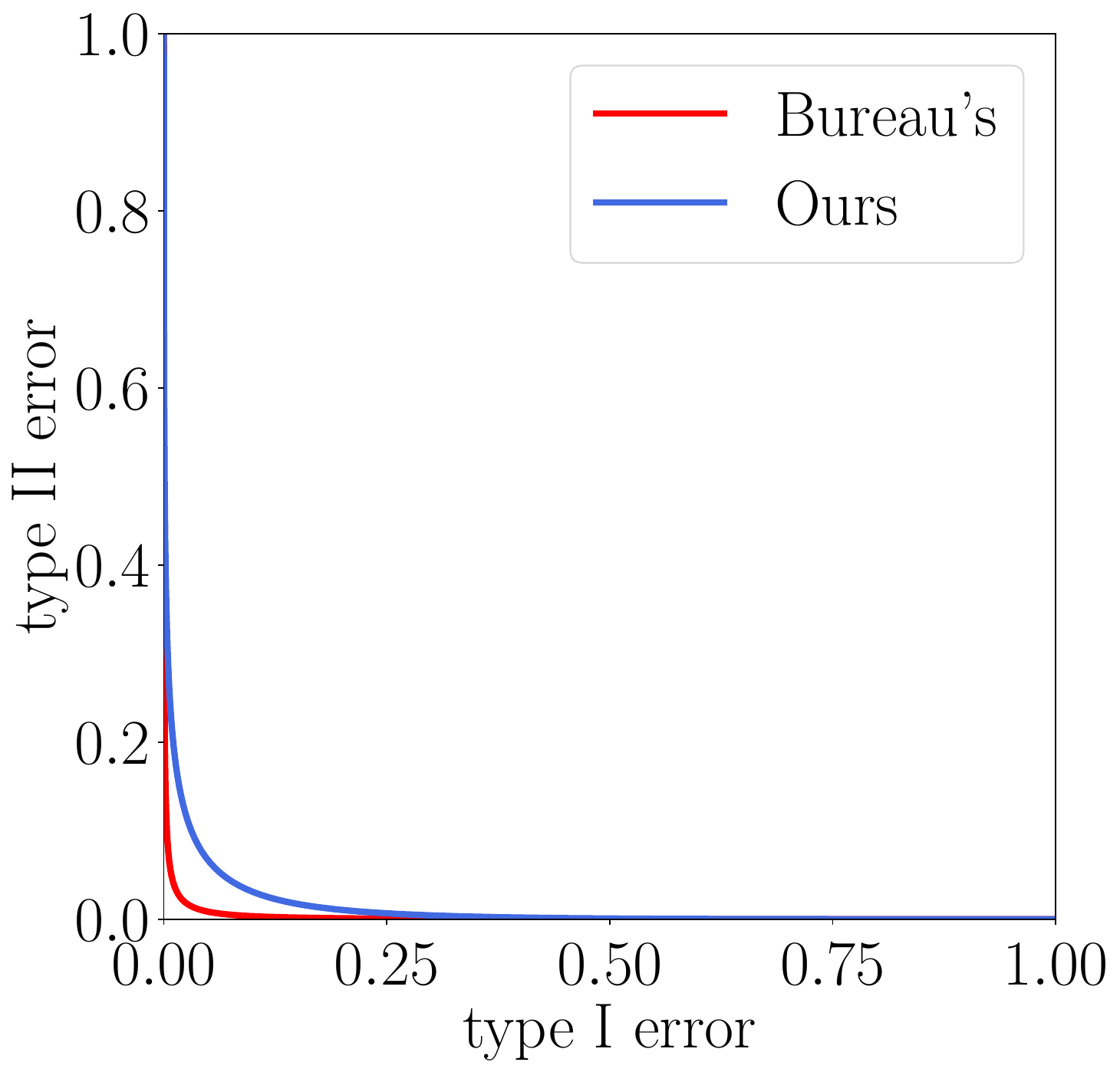}
  }
}

\begin{figure}[!htp]
    \centering
    \setlength{\tabcolsep}{2pt}
    \renewcommand{\arraystretch}{1.0}

    \begin{tabular}{@{}cccccc@{}}
        \tradeofflm{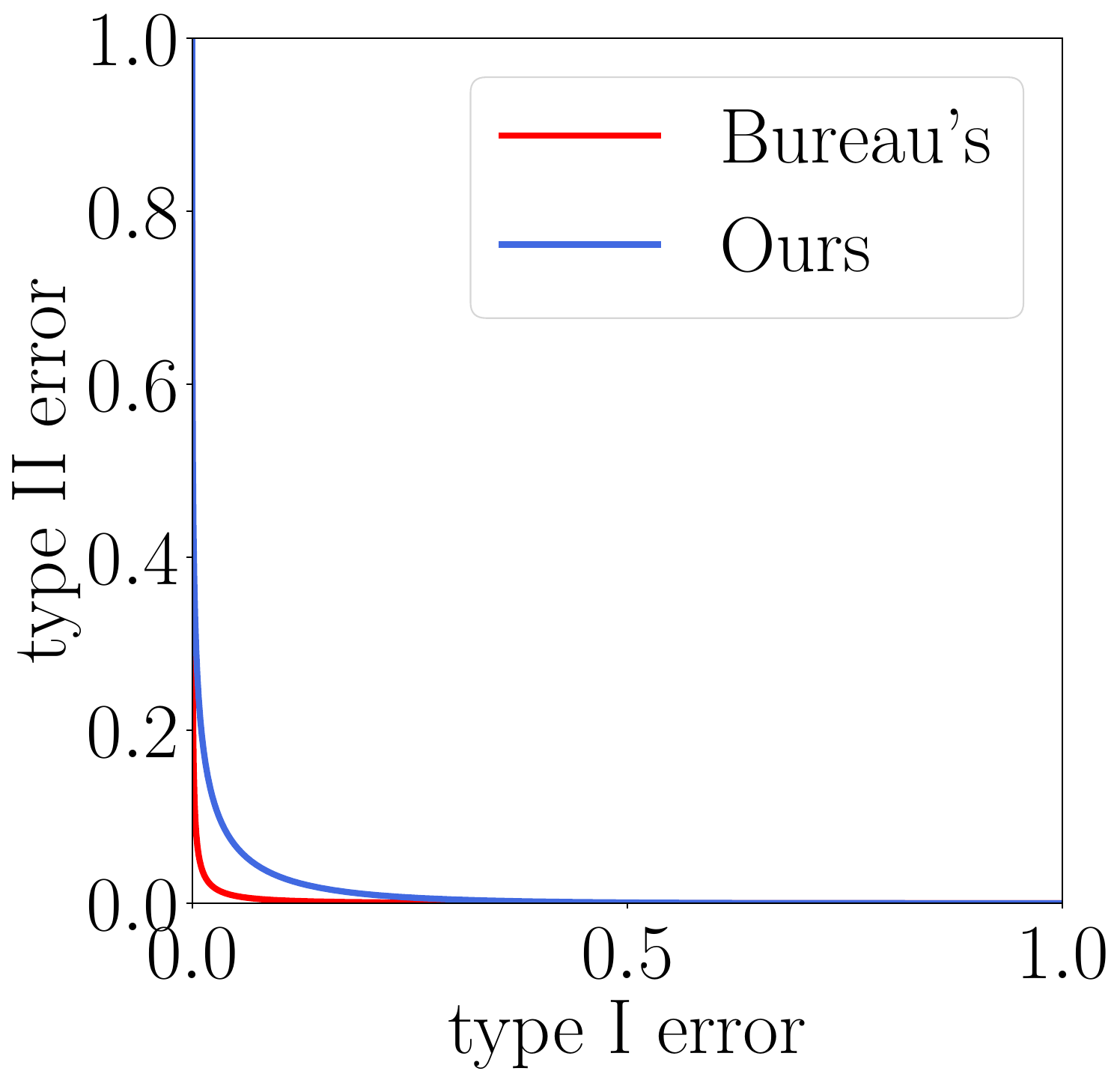}  & \tradeofflm{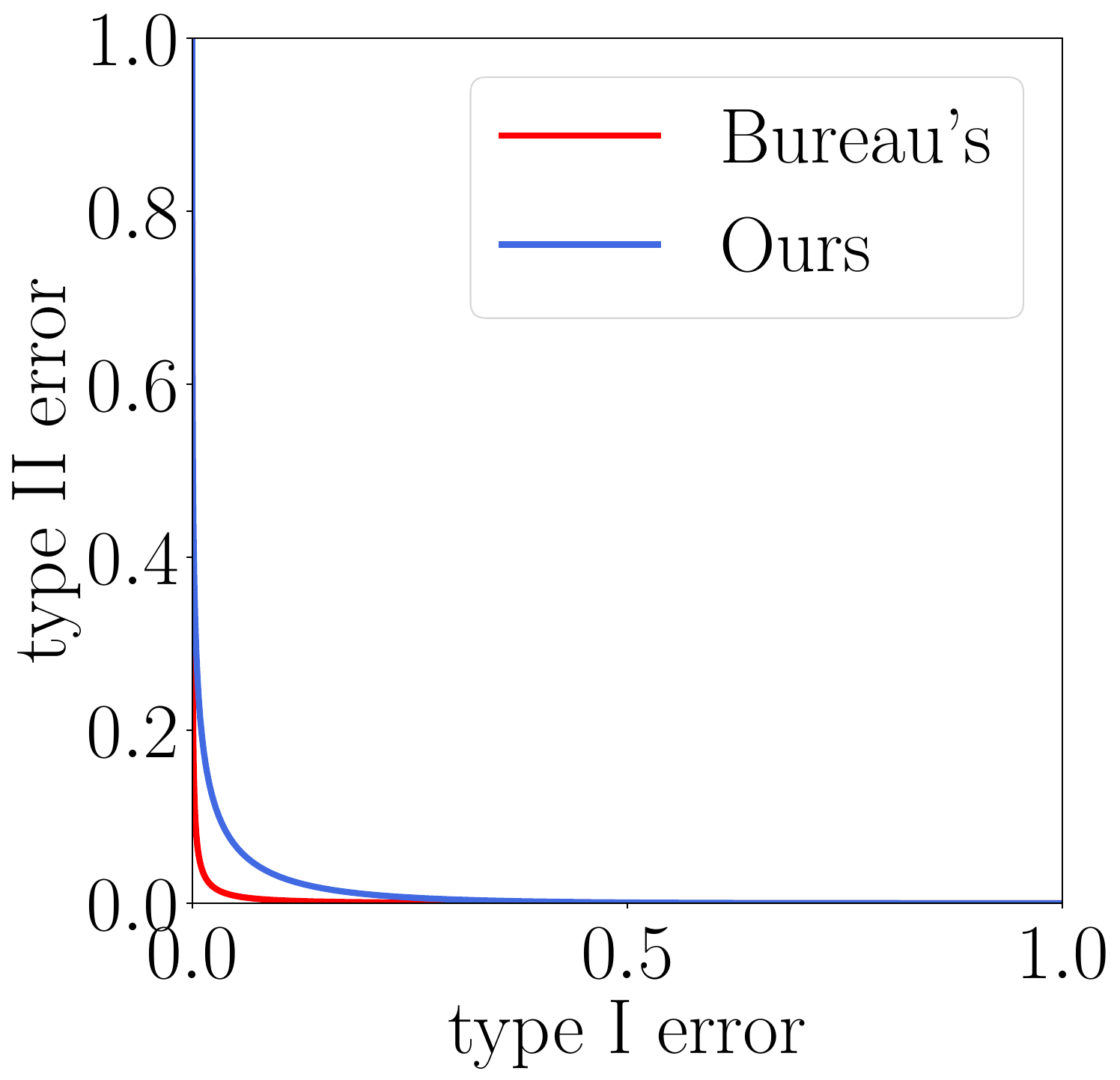}  & \tradeofflm{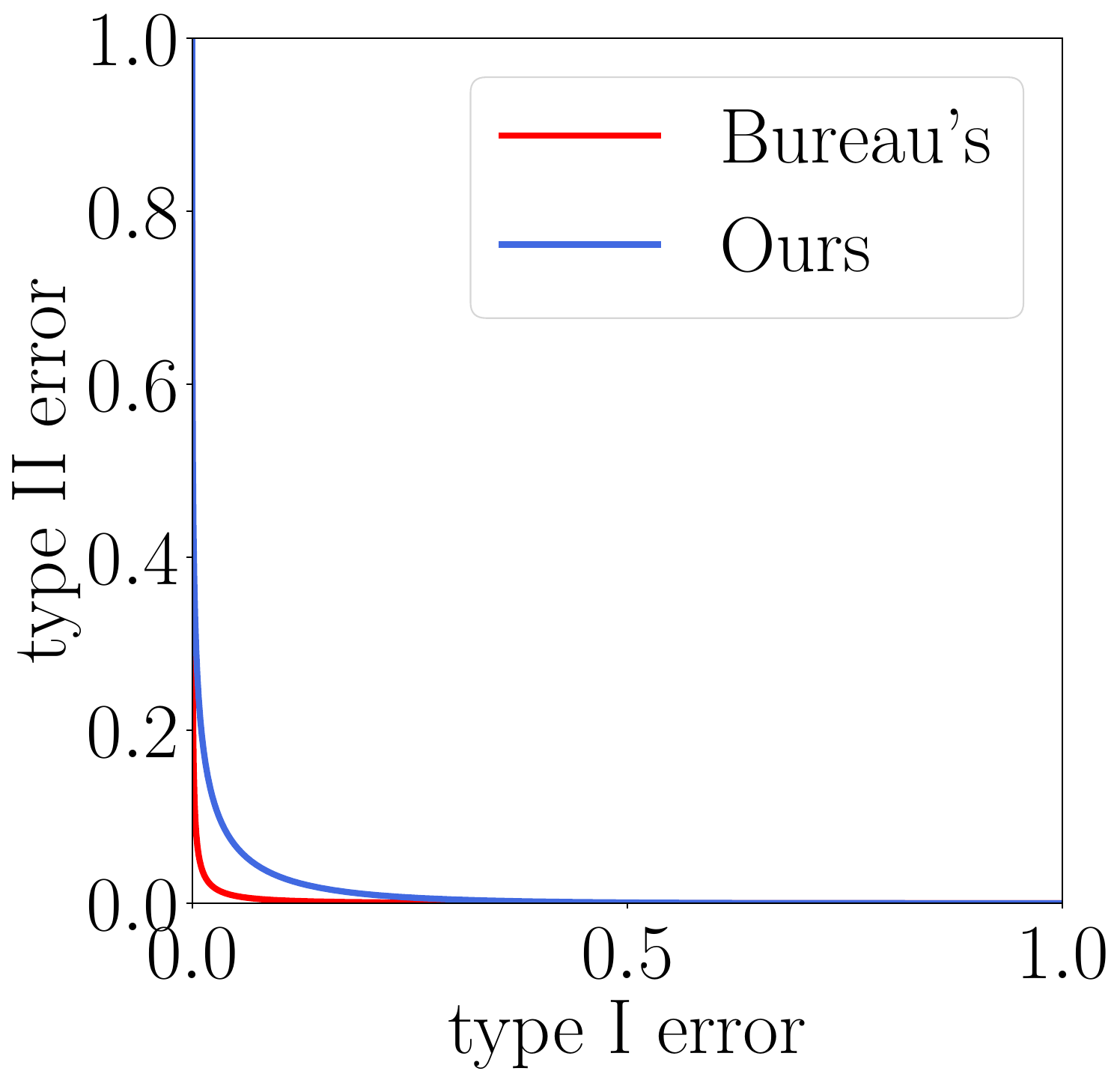}  & \tradeofflm{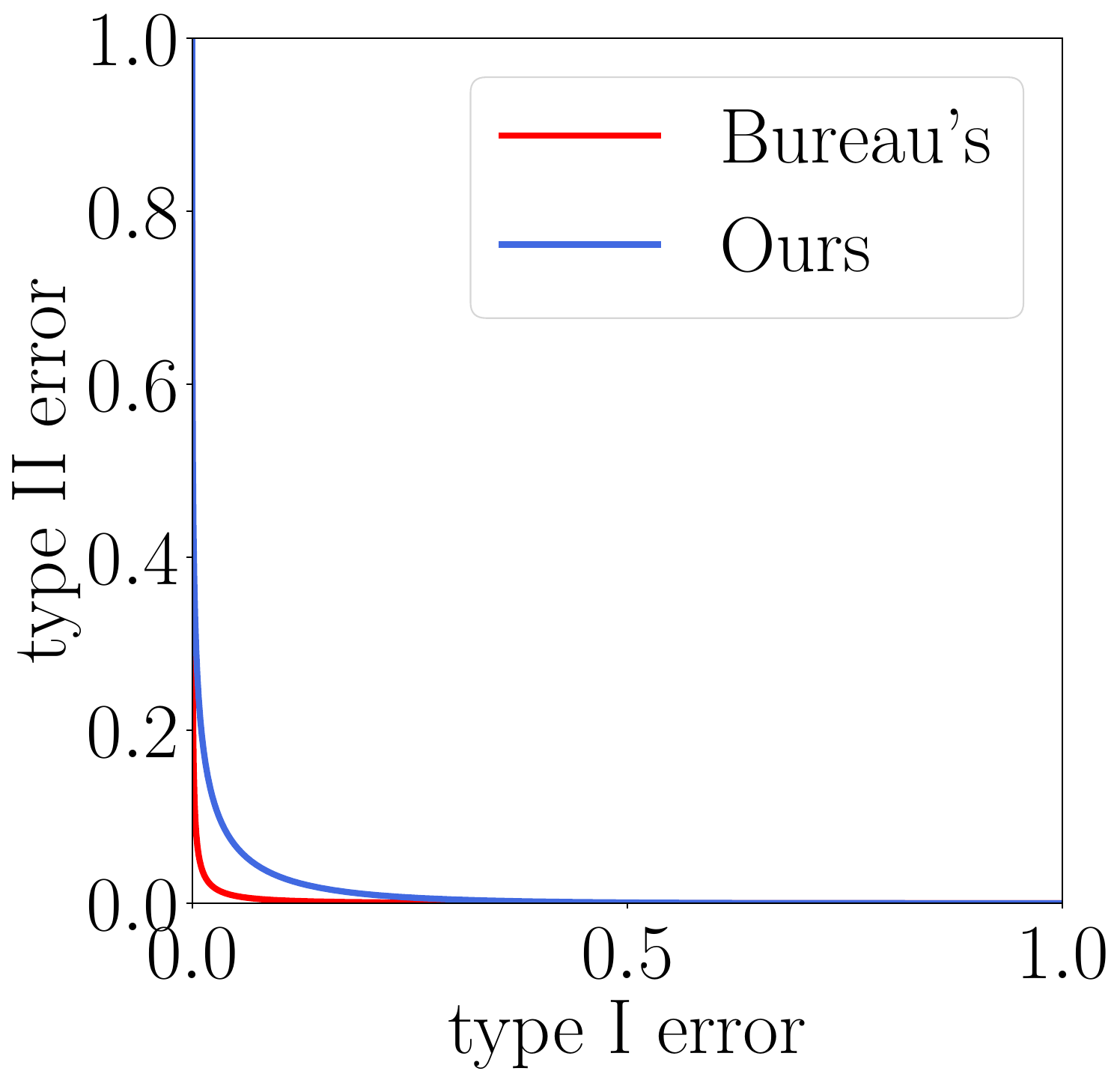}  & \tradeofflm{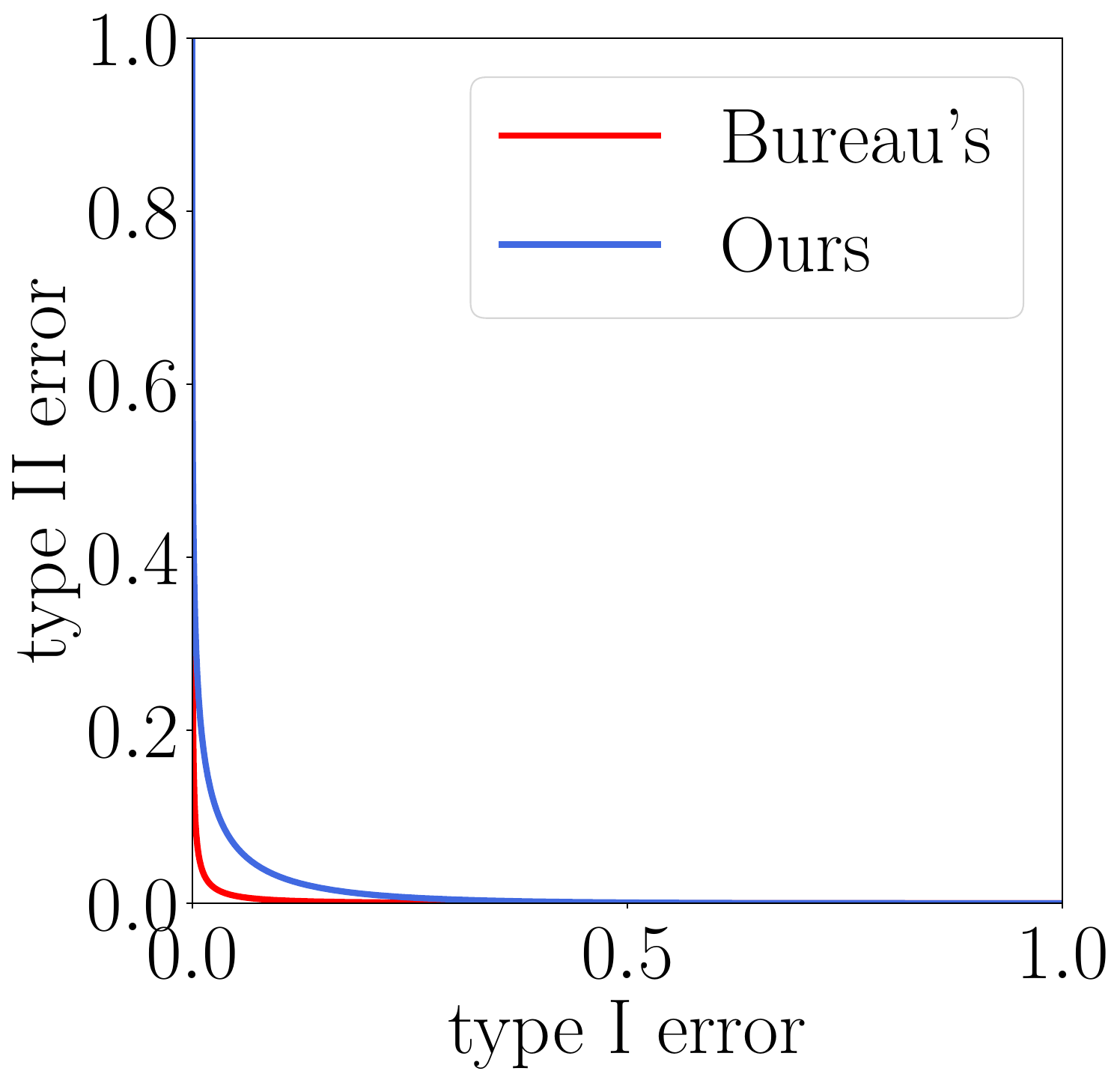}  & \tradeofflm{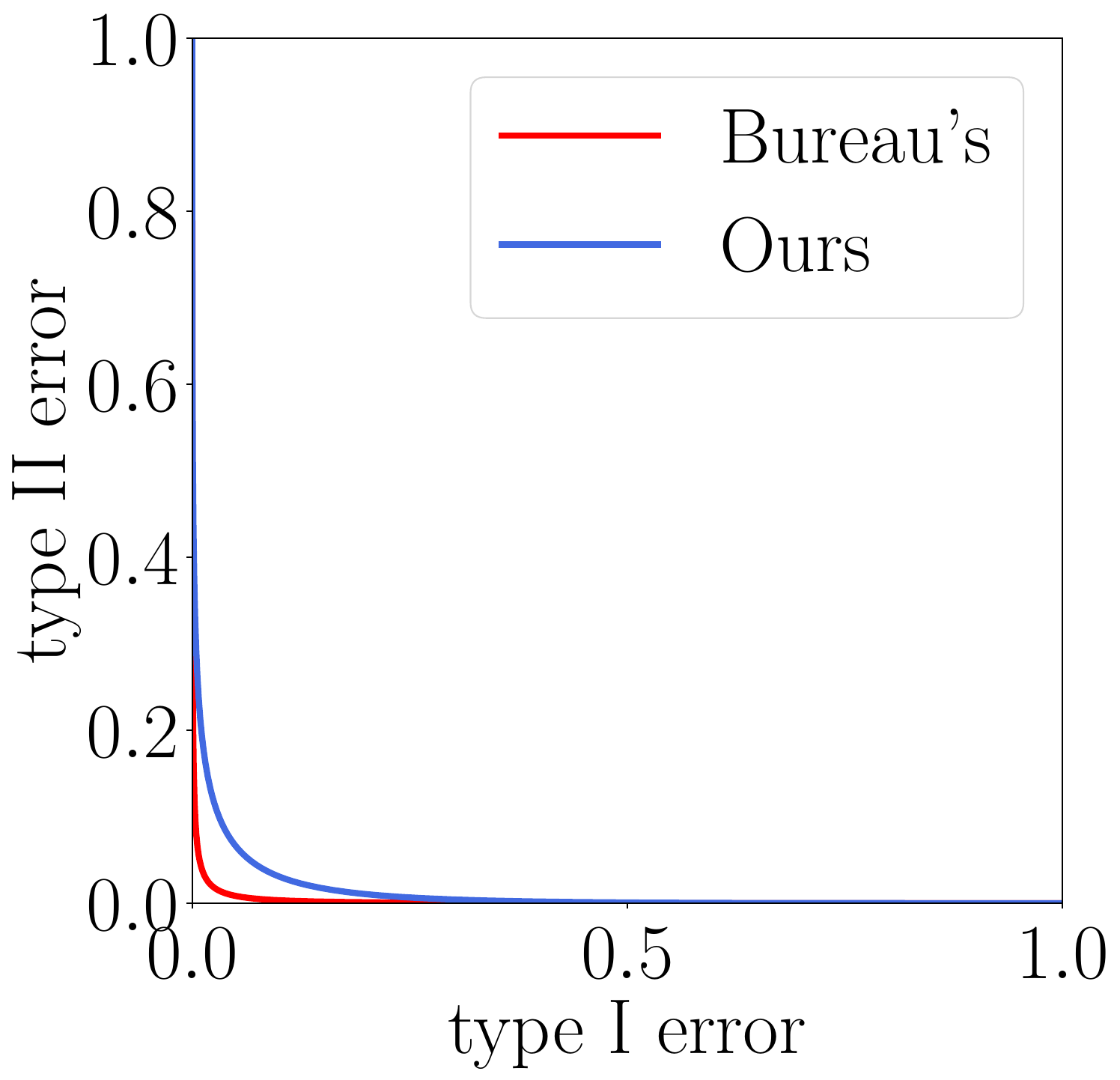}  \\
        \tradeofflm{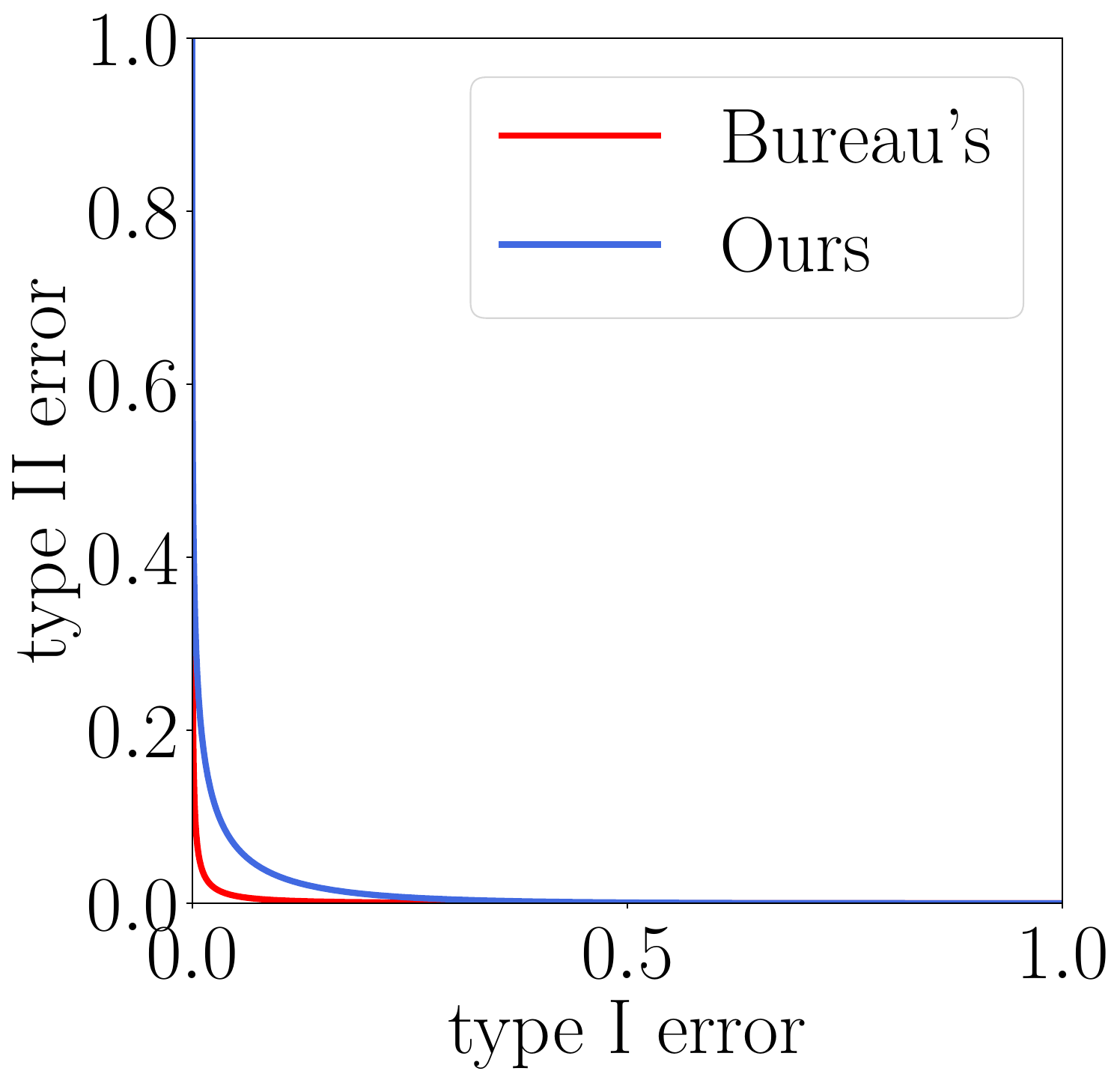}  & \tradeofflm{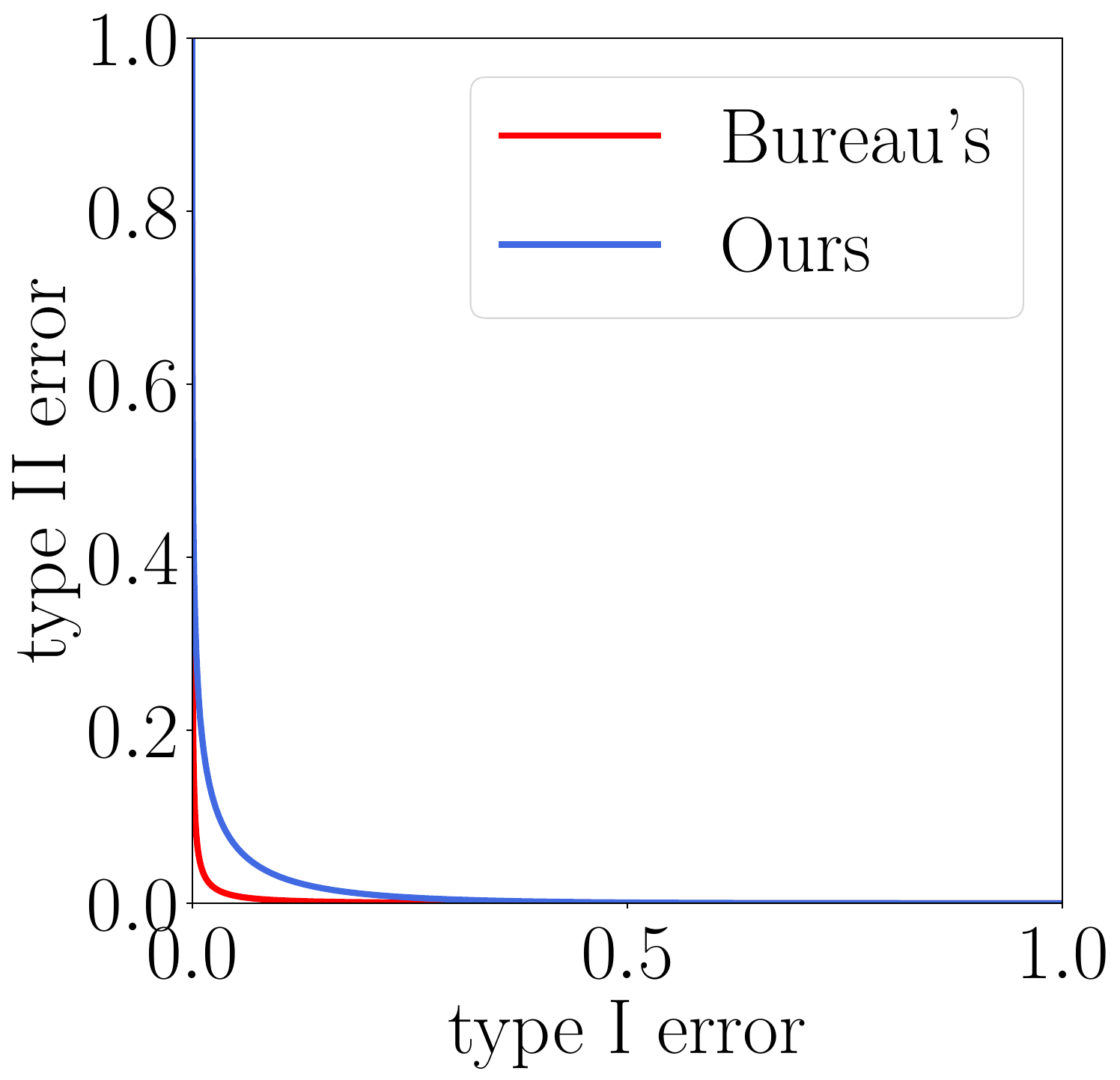}  & \tradeofflm{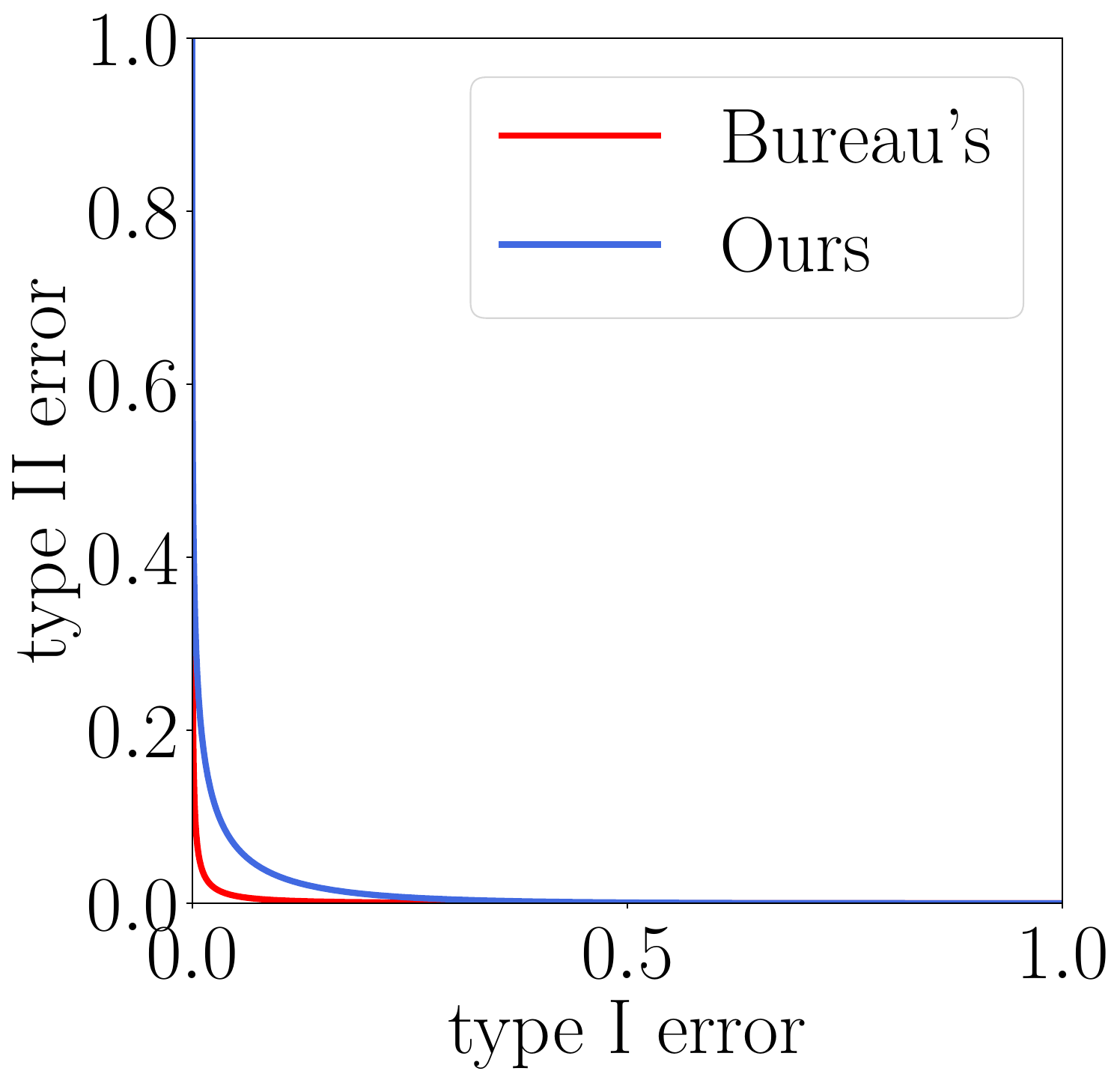}  & \tradeofflm{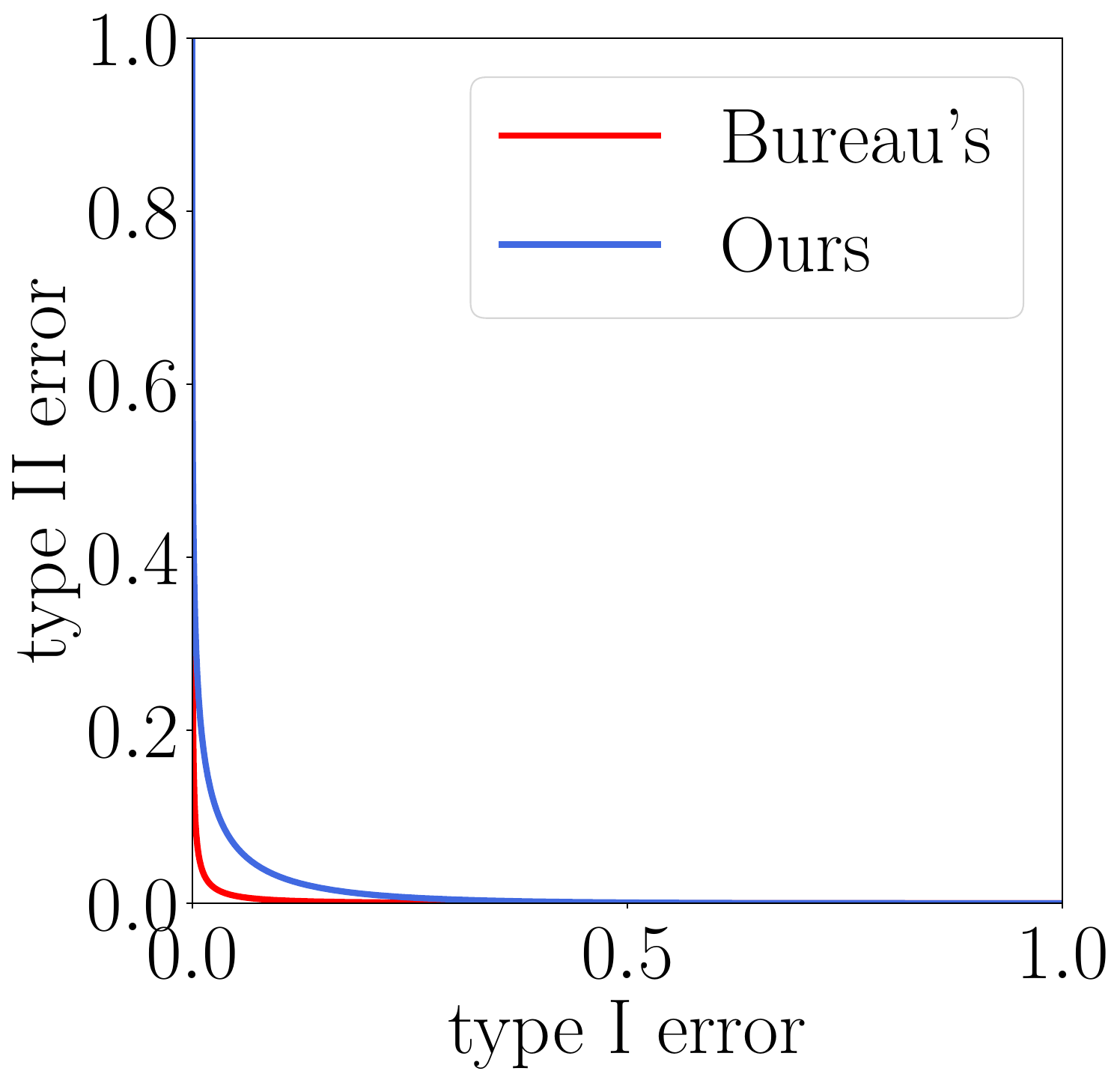} & \tradeofflm{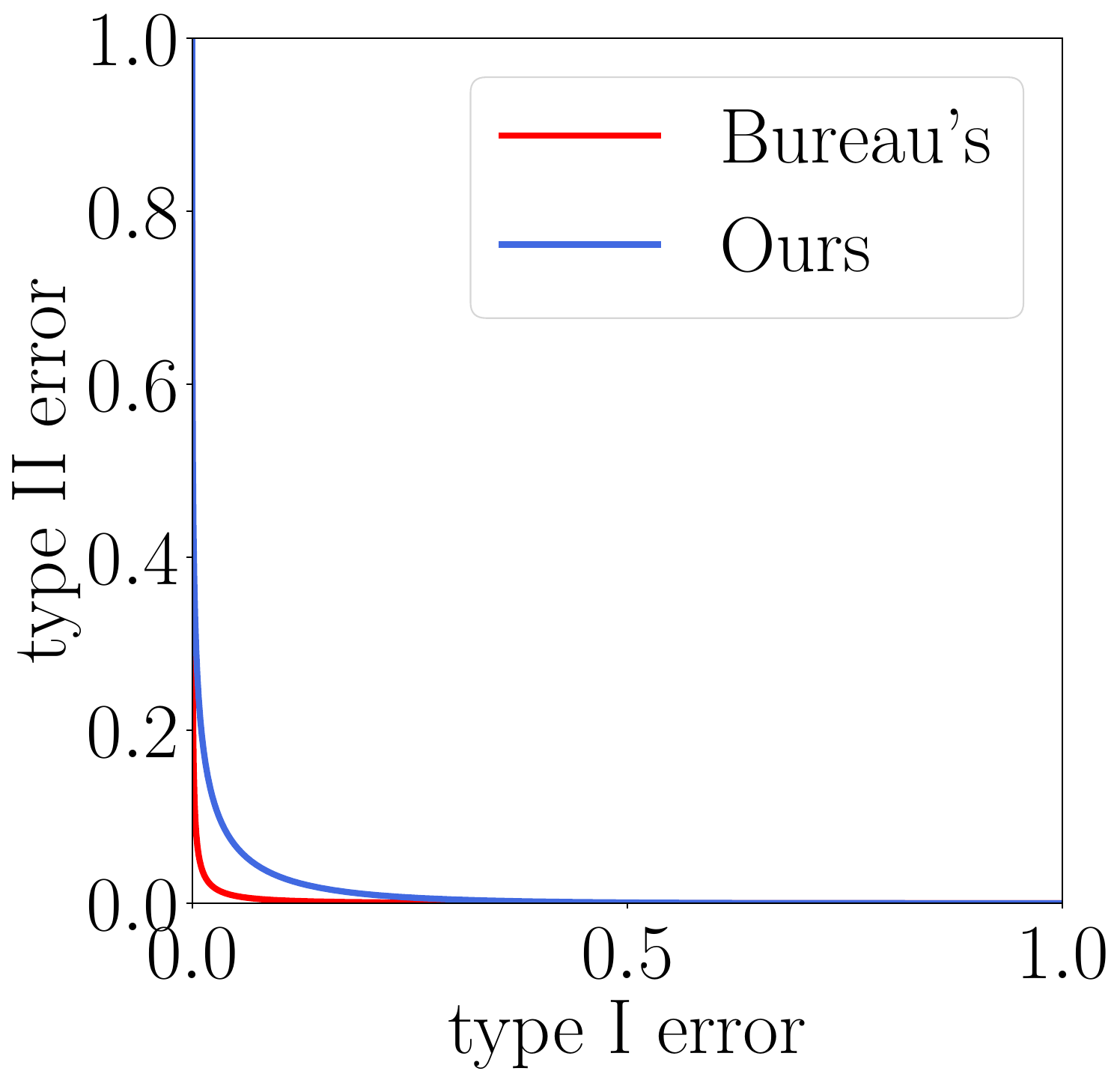} & \tradeofflm{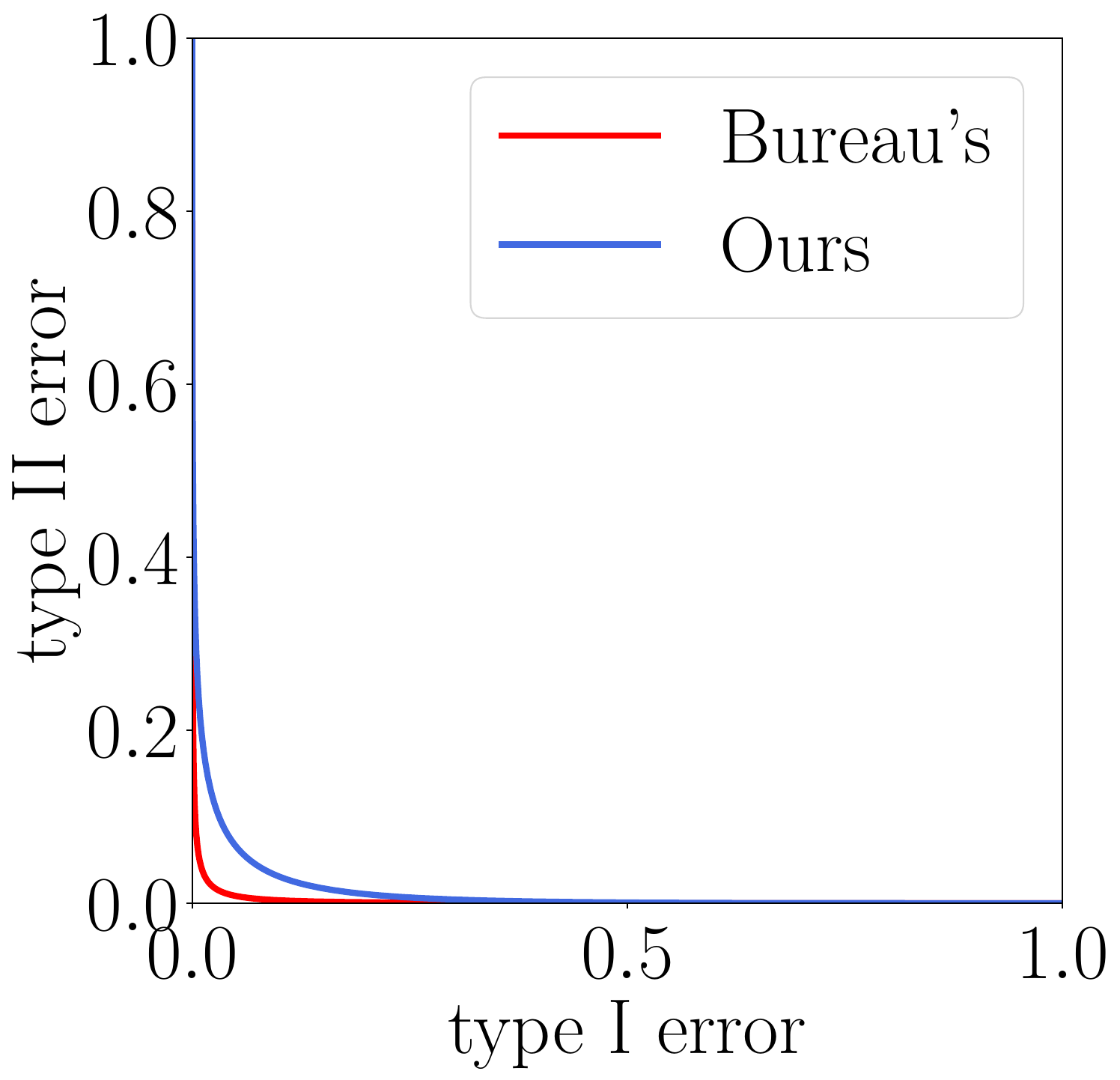} \\
        \tradeoffsm{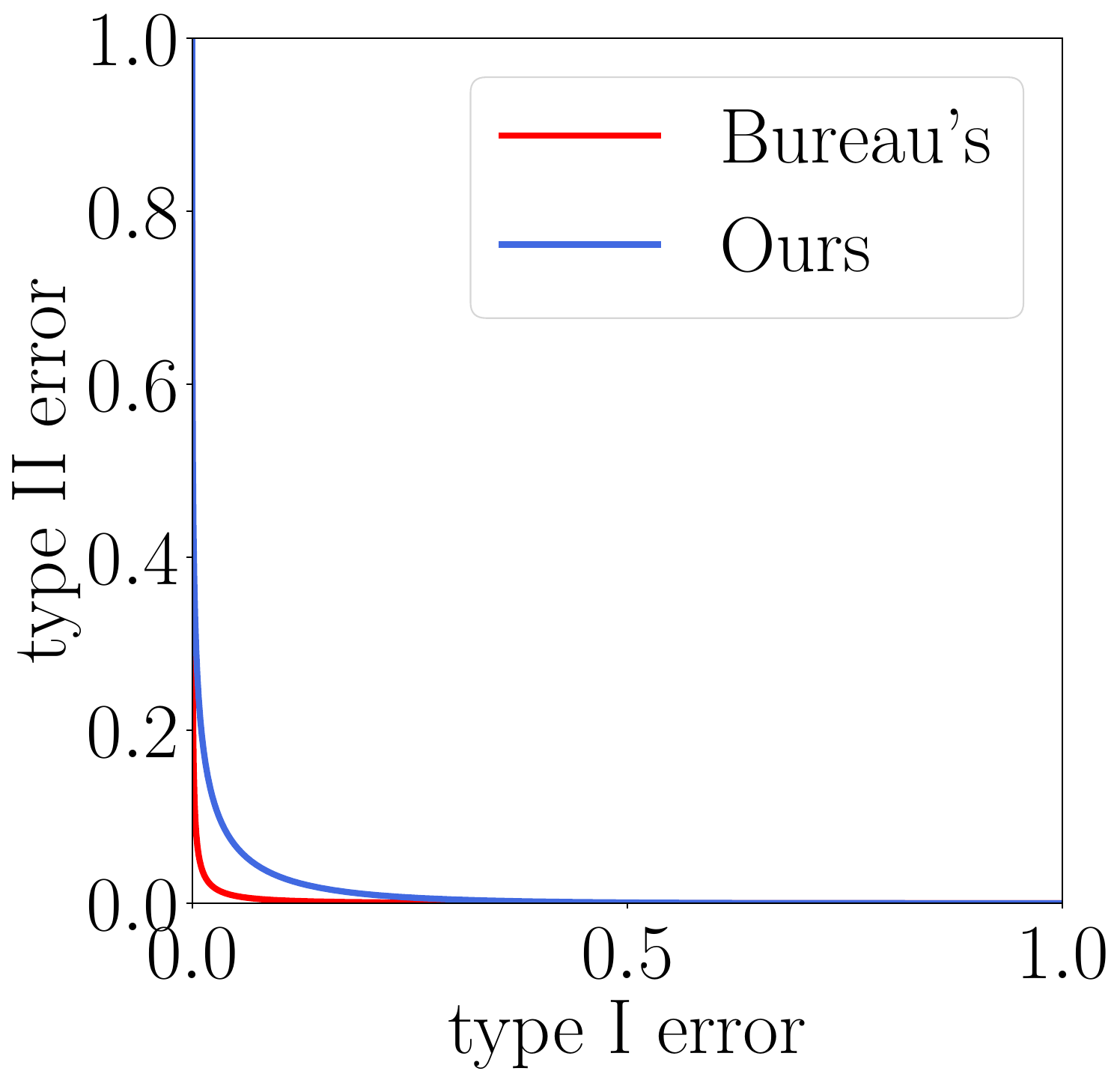} & \tradeoffsm{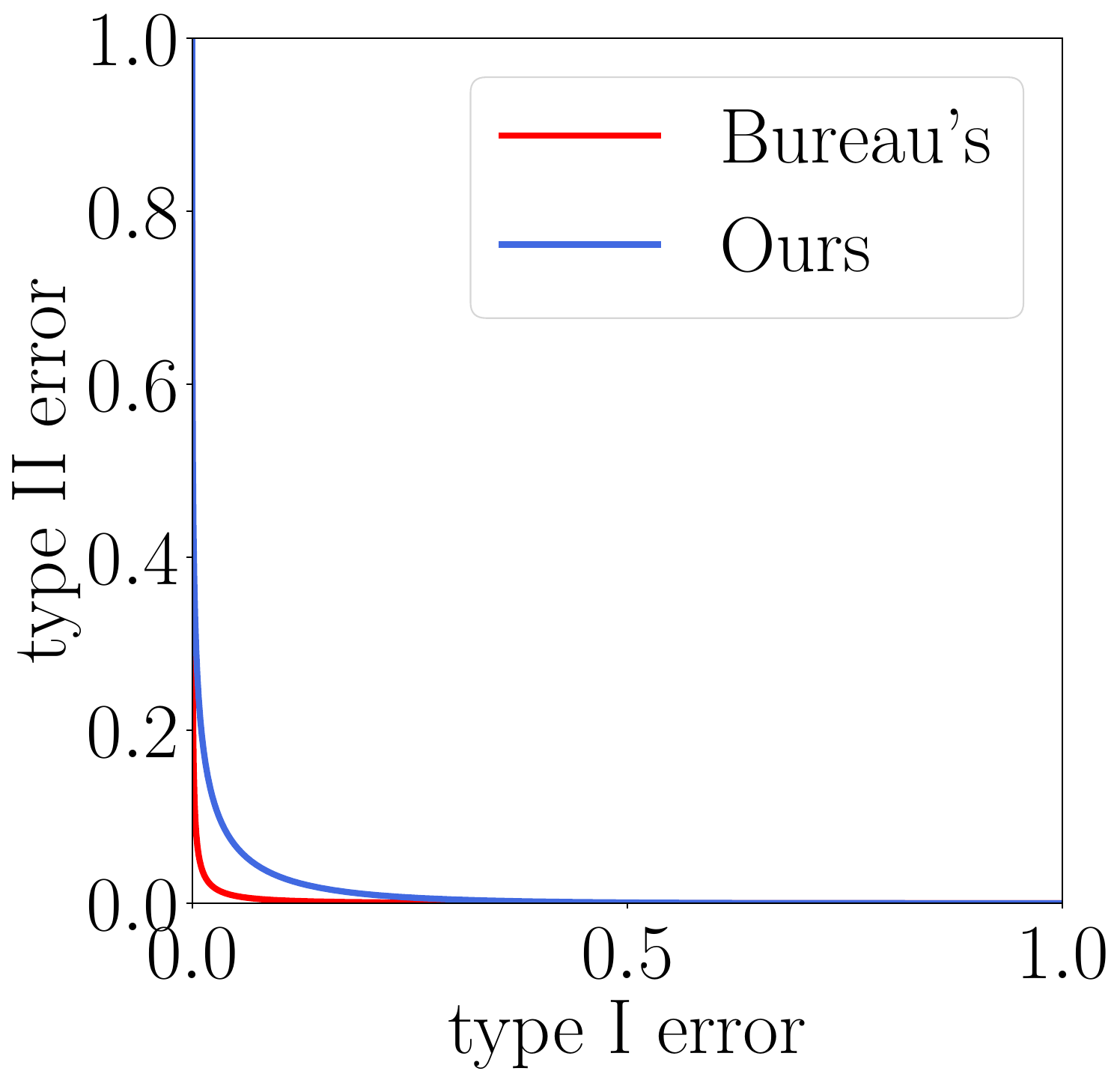} & \tradeoffsm{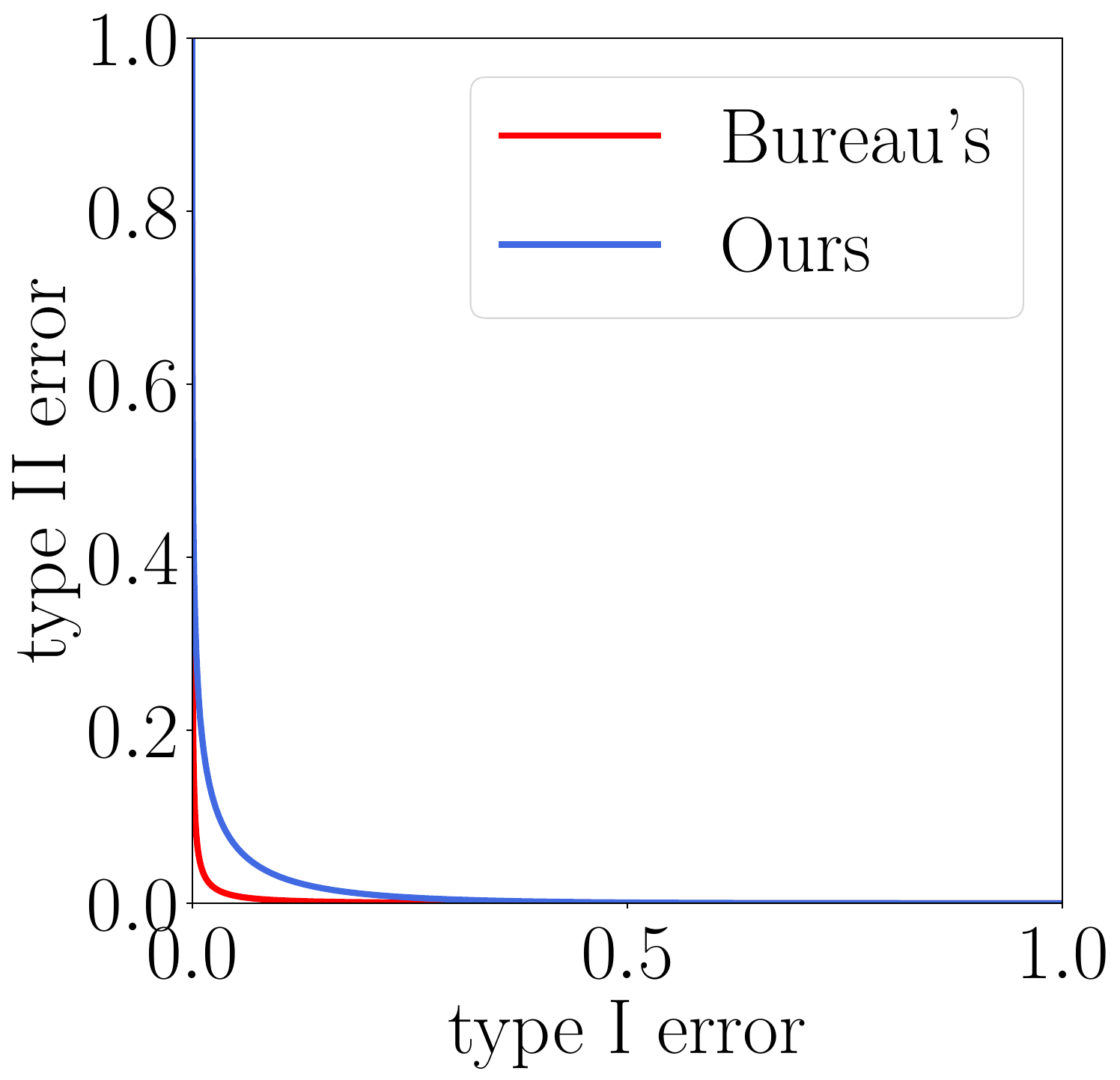} & \tradeoffsm{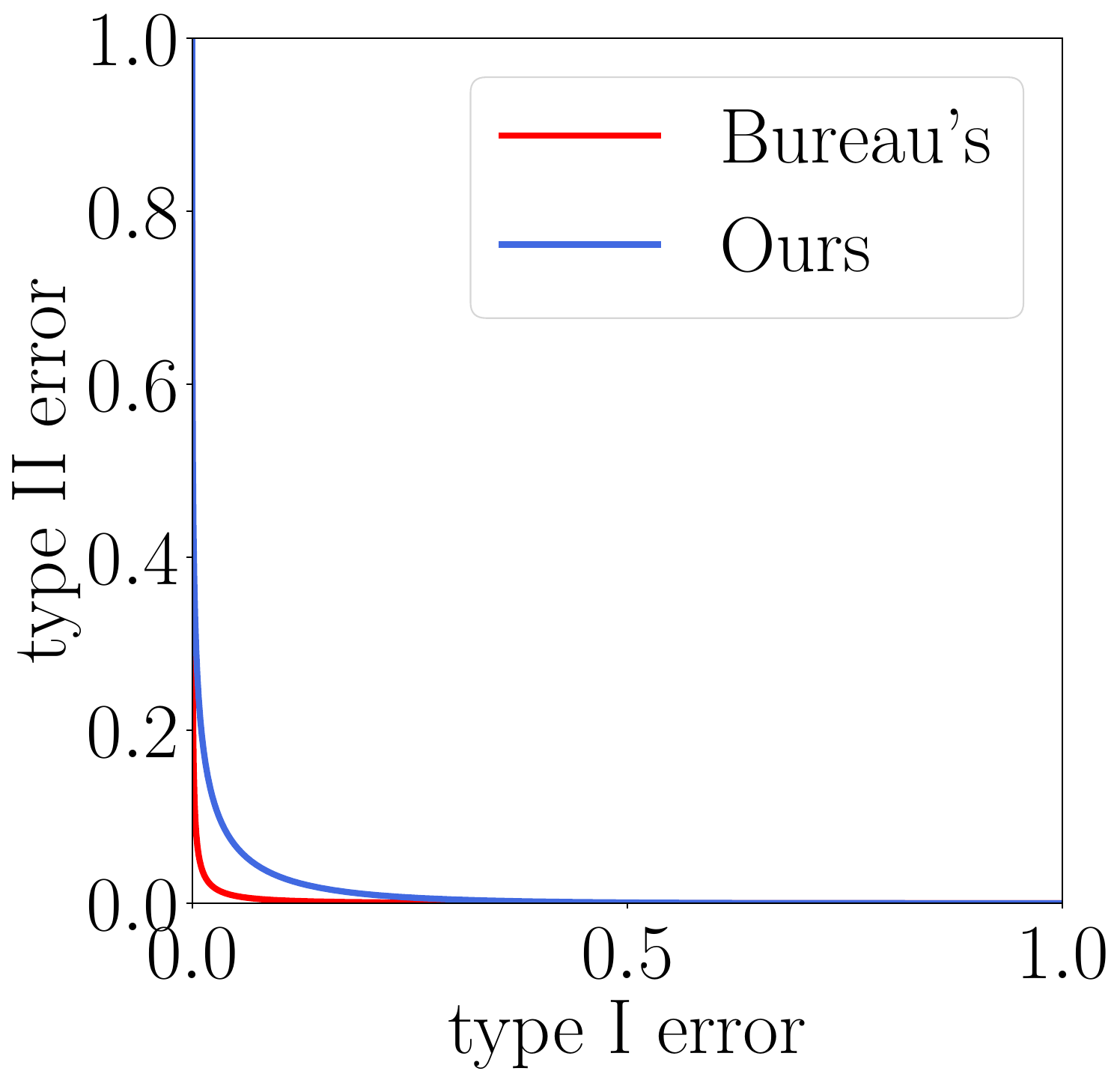} & \tradeoffsm{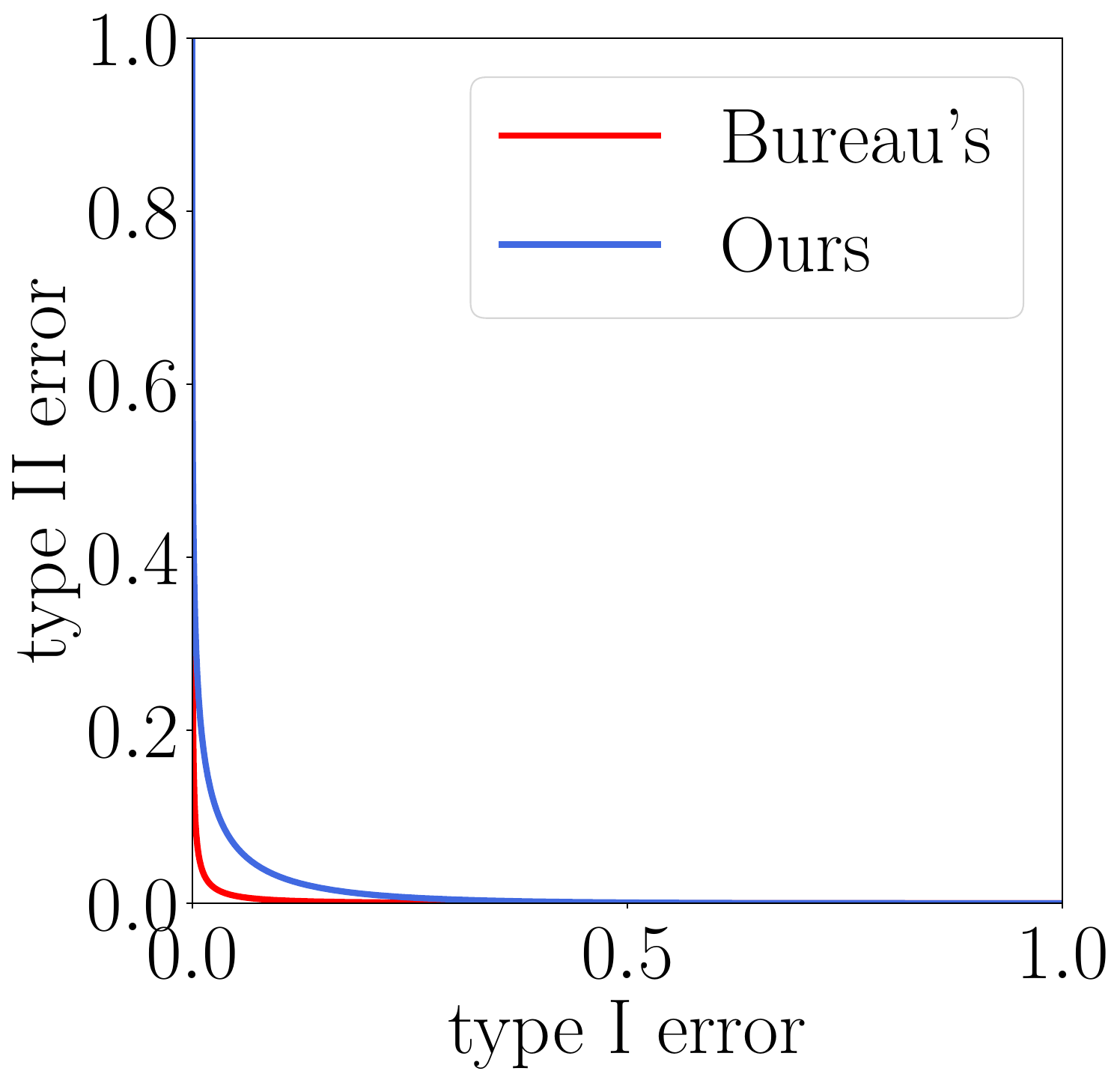} & \tradeoffsm{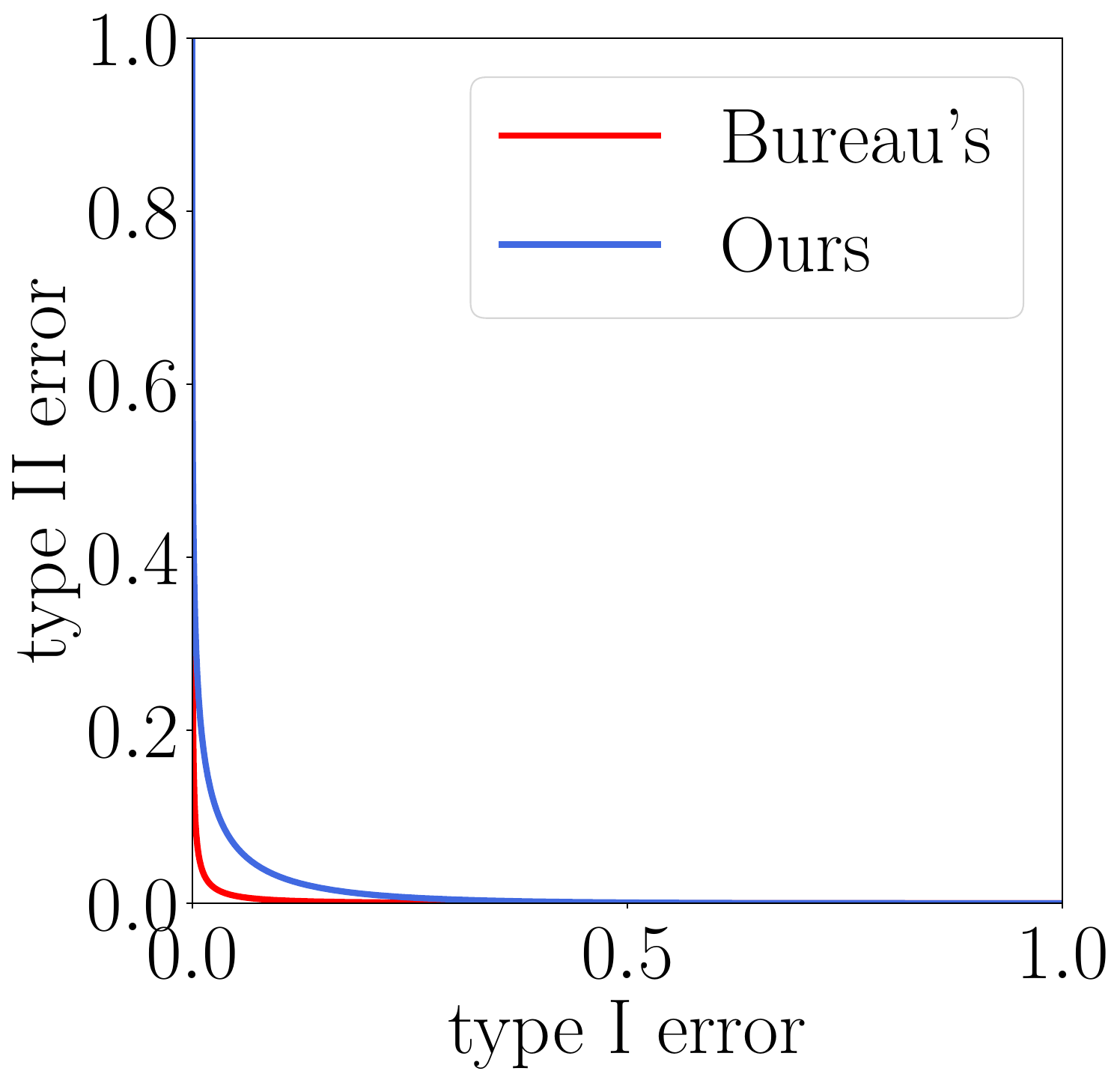} \\
        \tradeoffsm{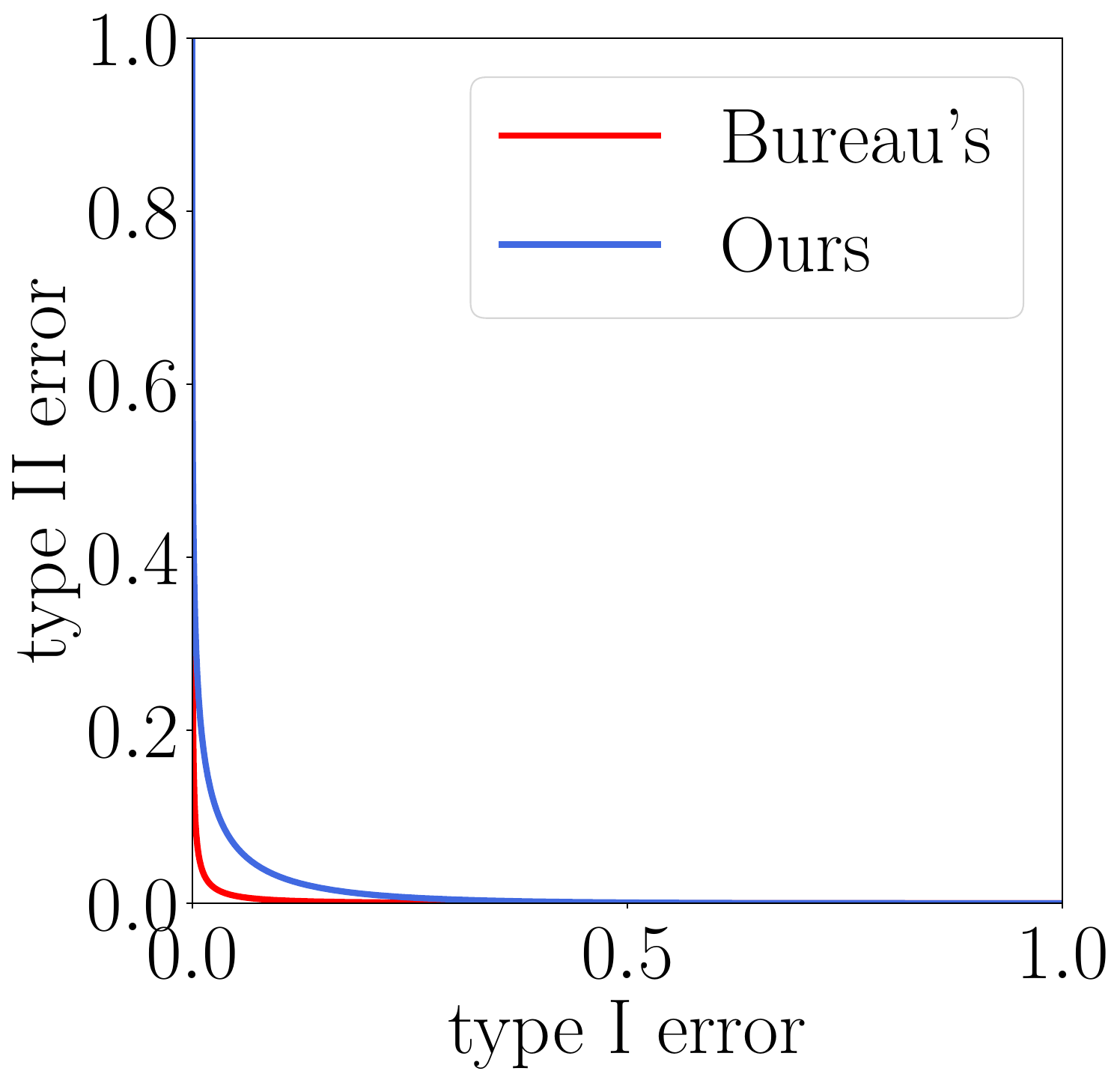} & \tradeoffsm{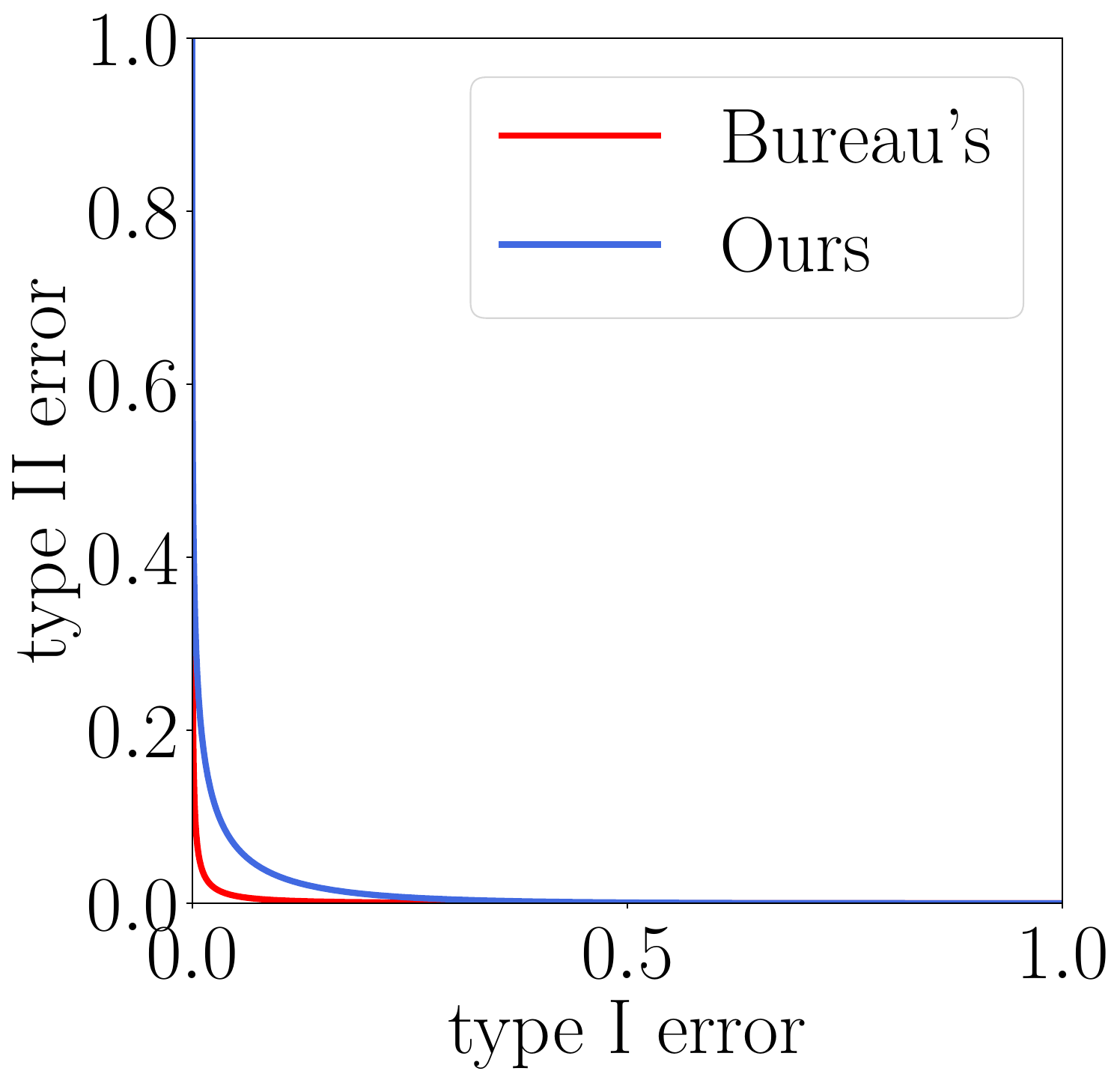} & \tradeoffsm{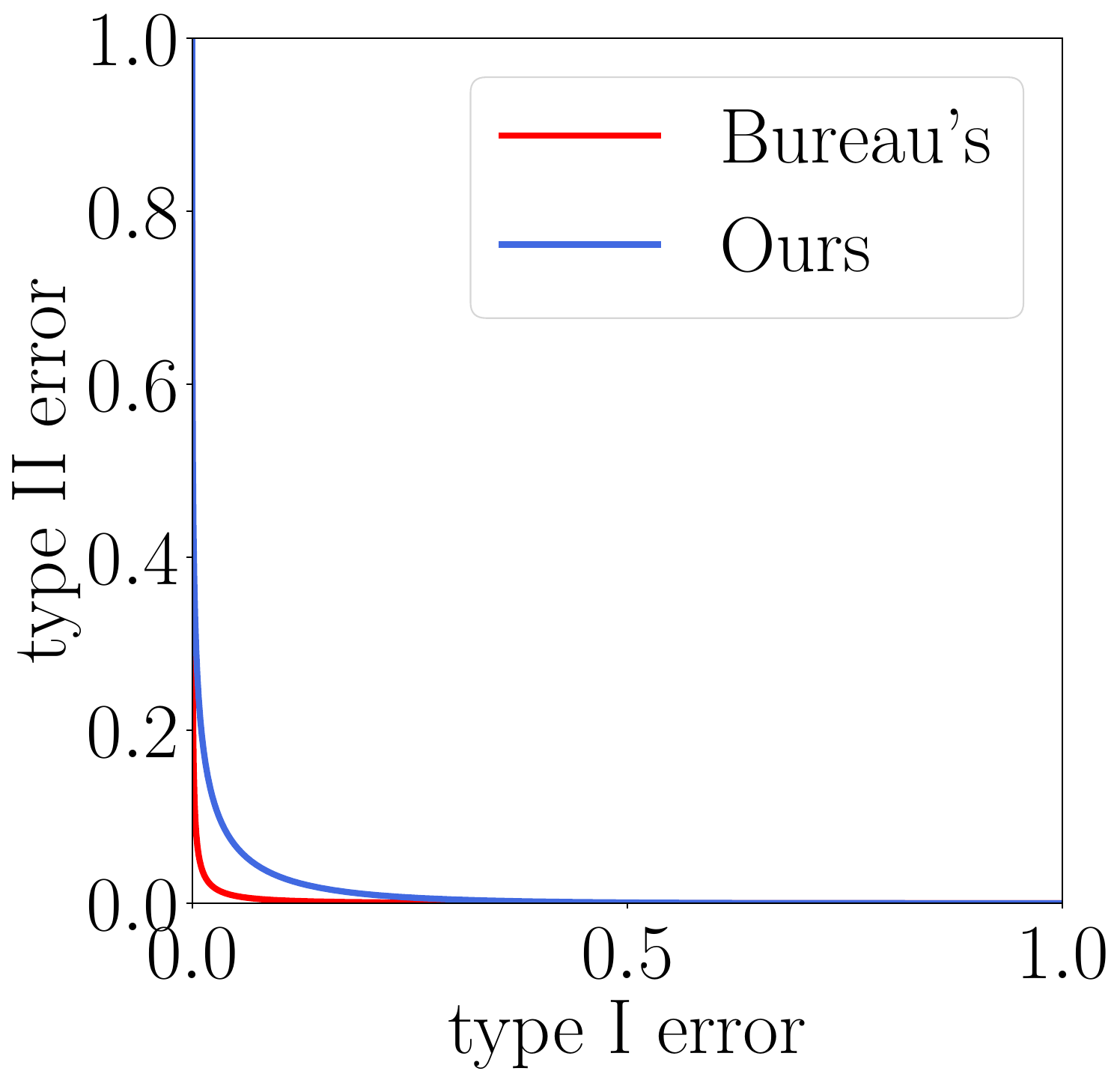} & \tradeoffsm{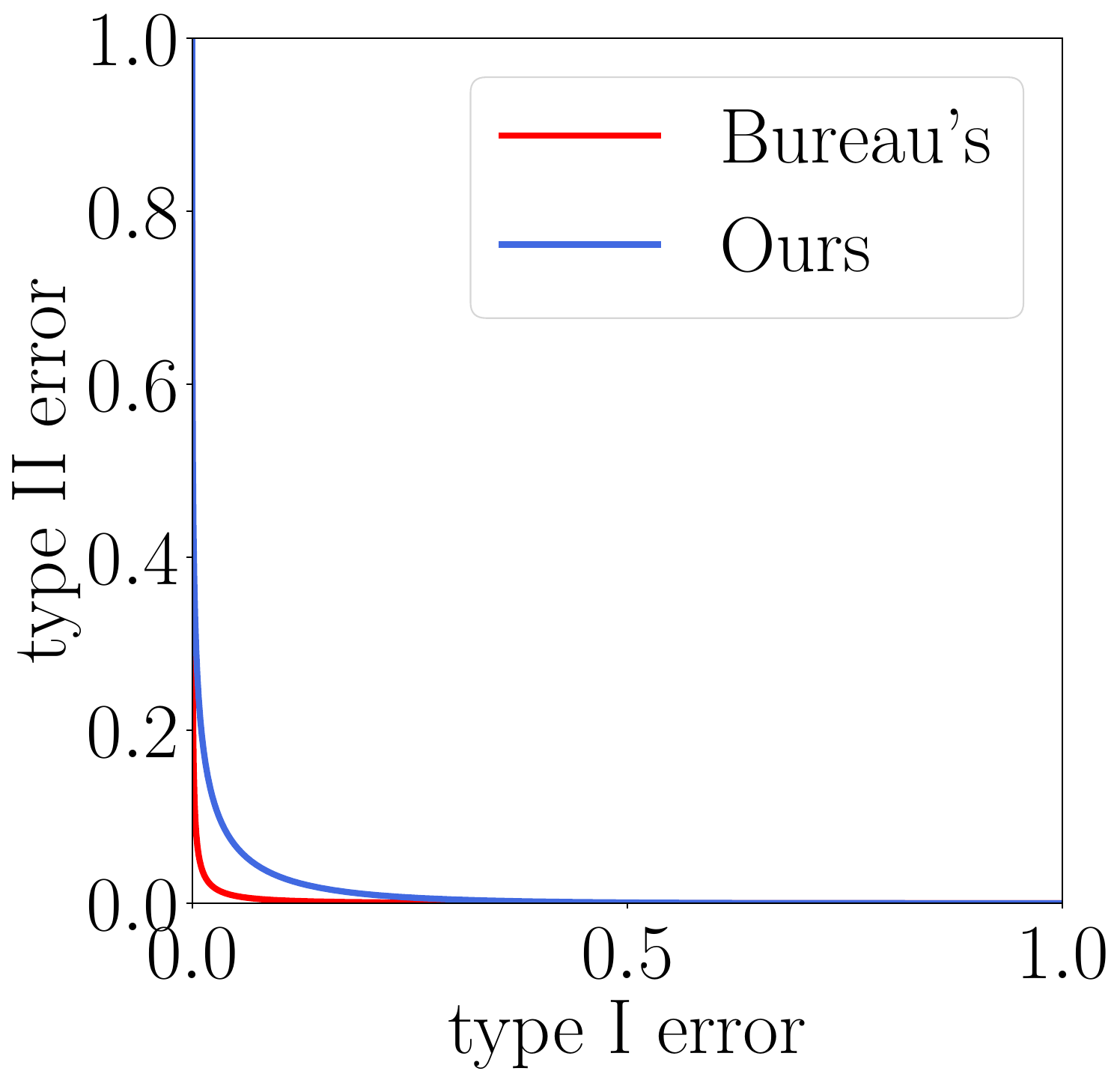} & \tradeoffsm{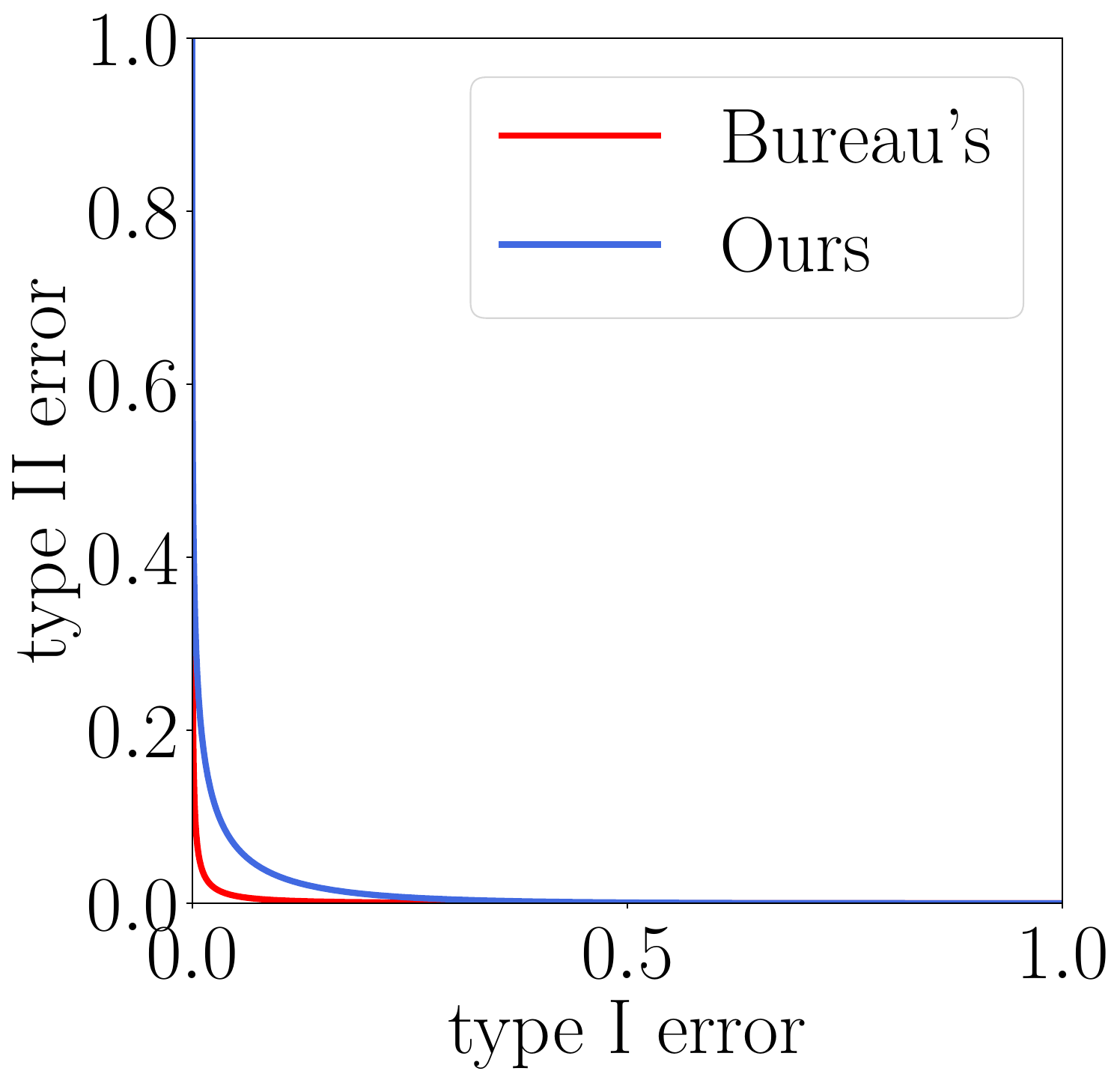} & \tradeoffsm{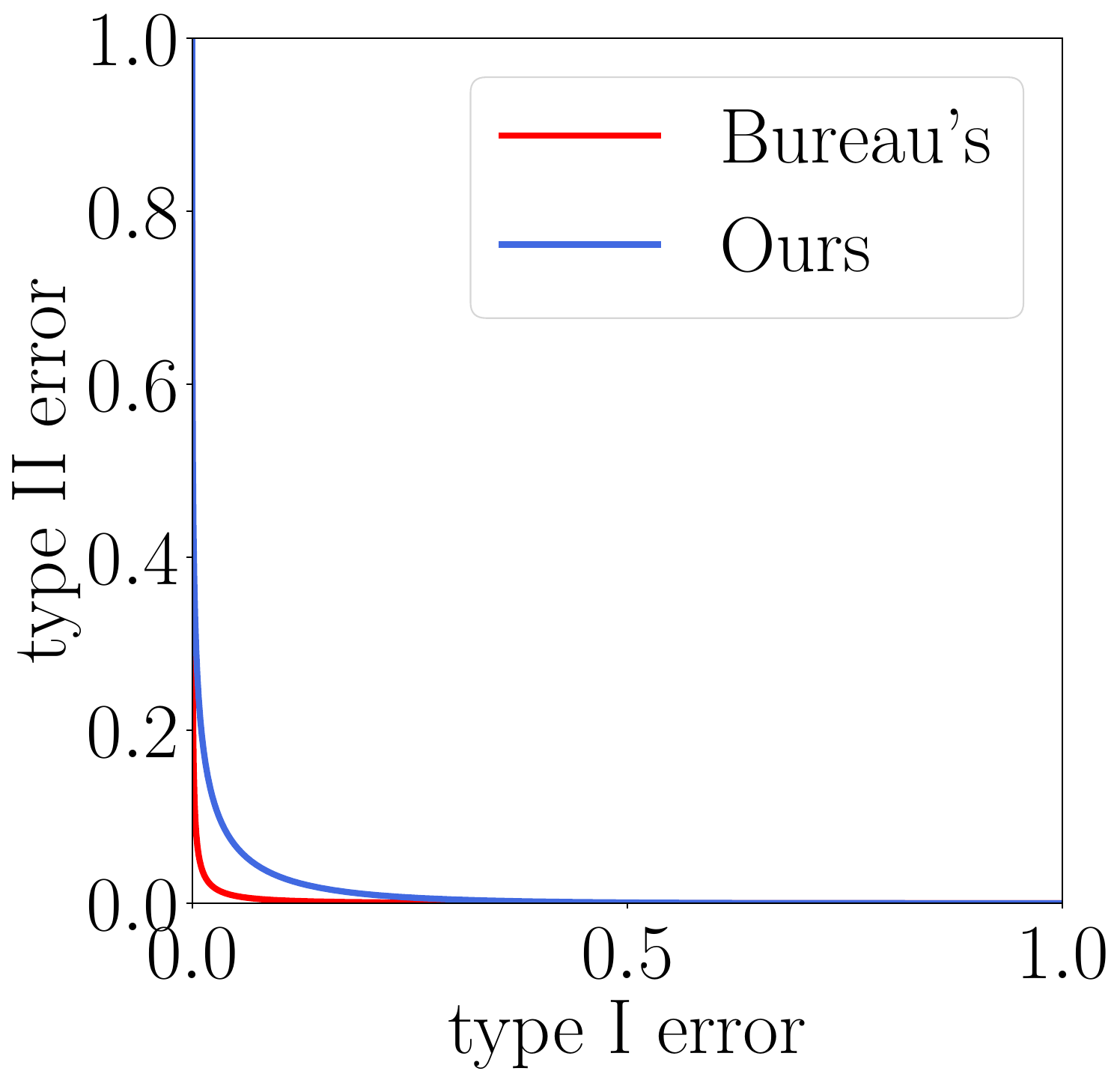} \\
        \tradeoffsm{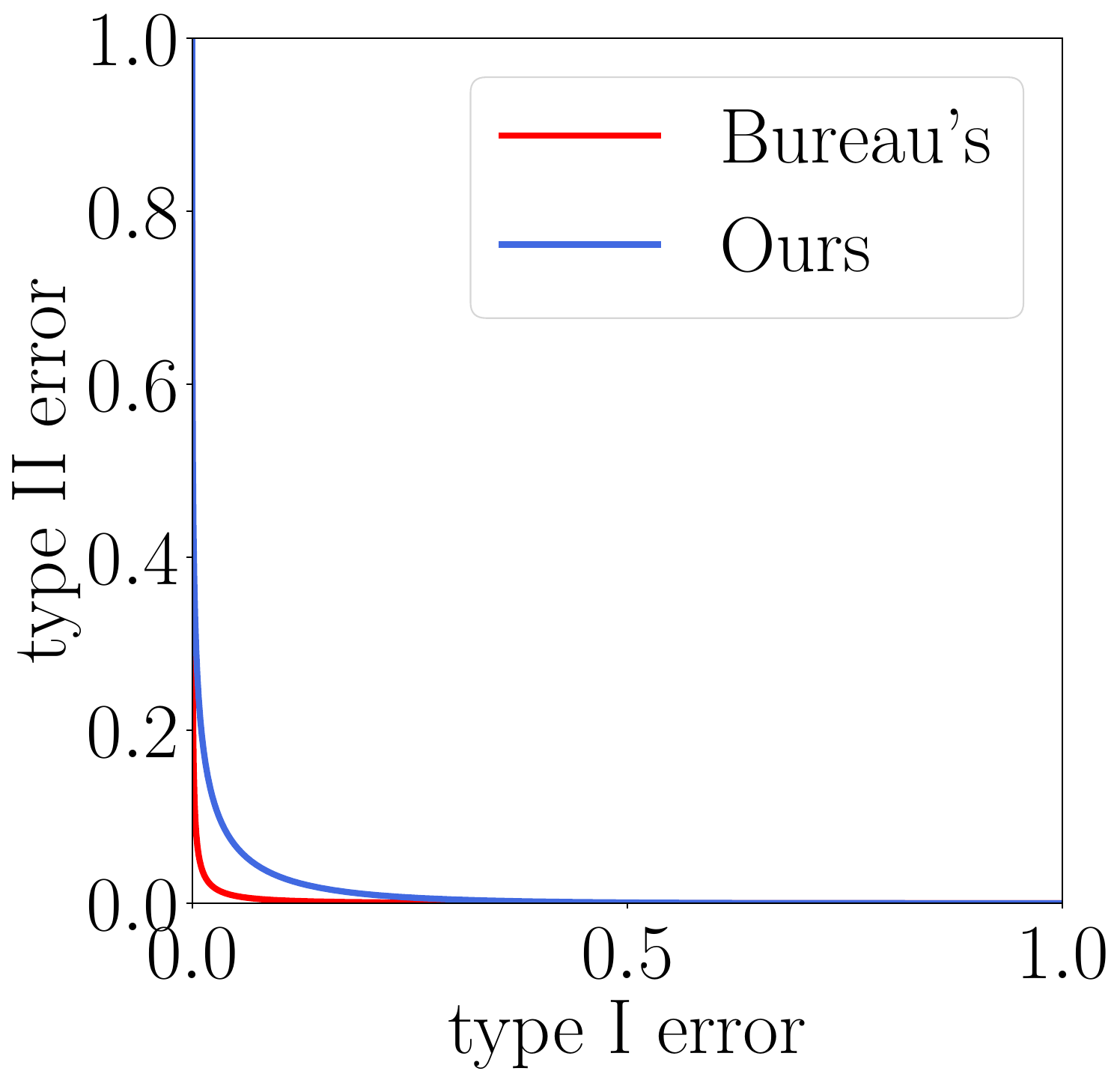} & \tradeoffsm{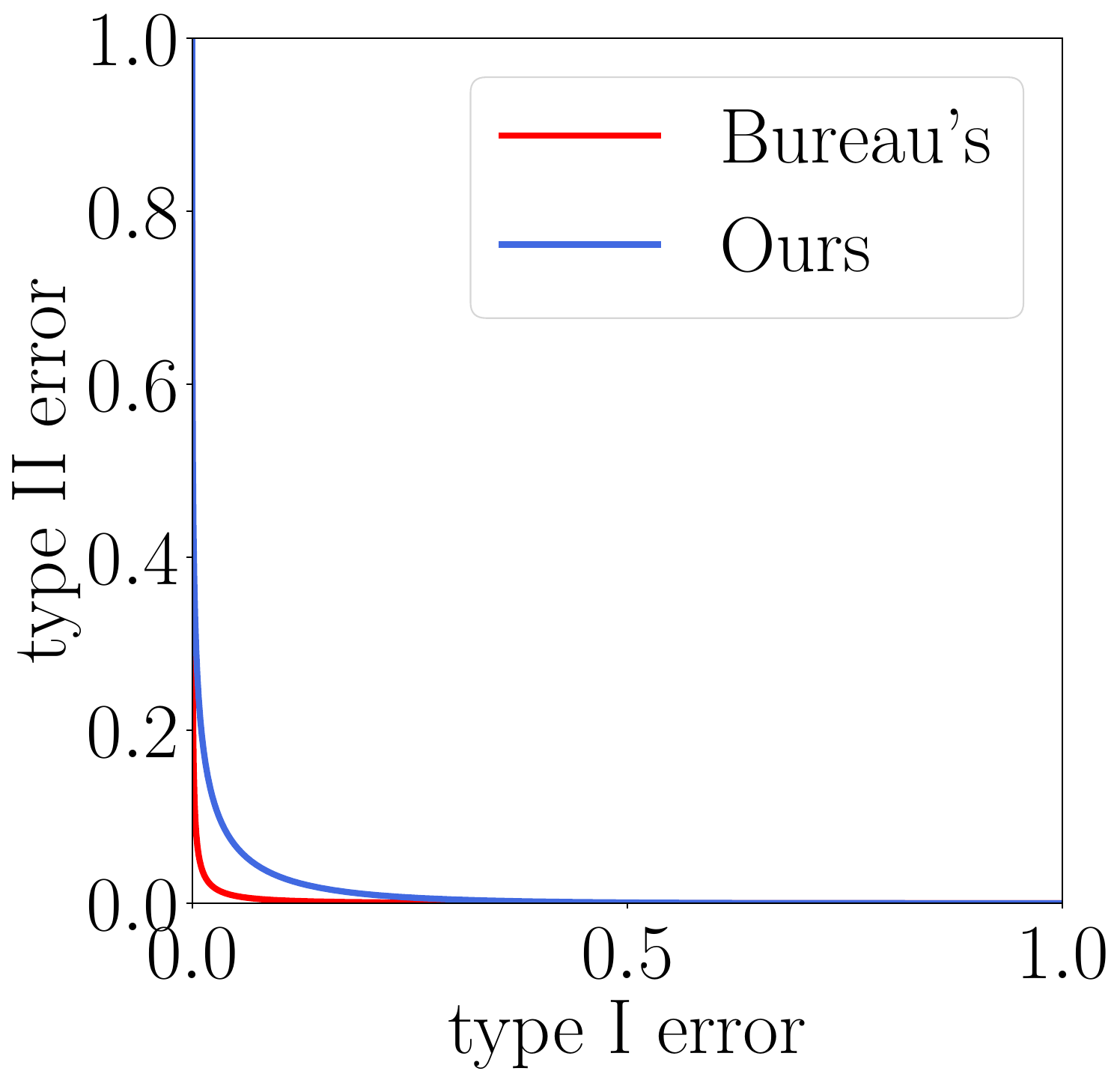} & \tradeoffsm{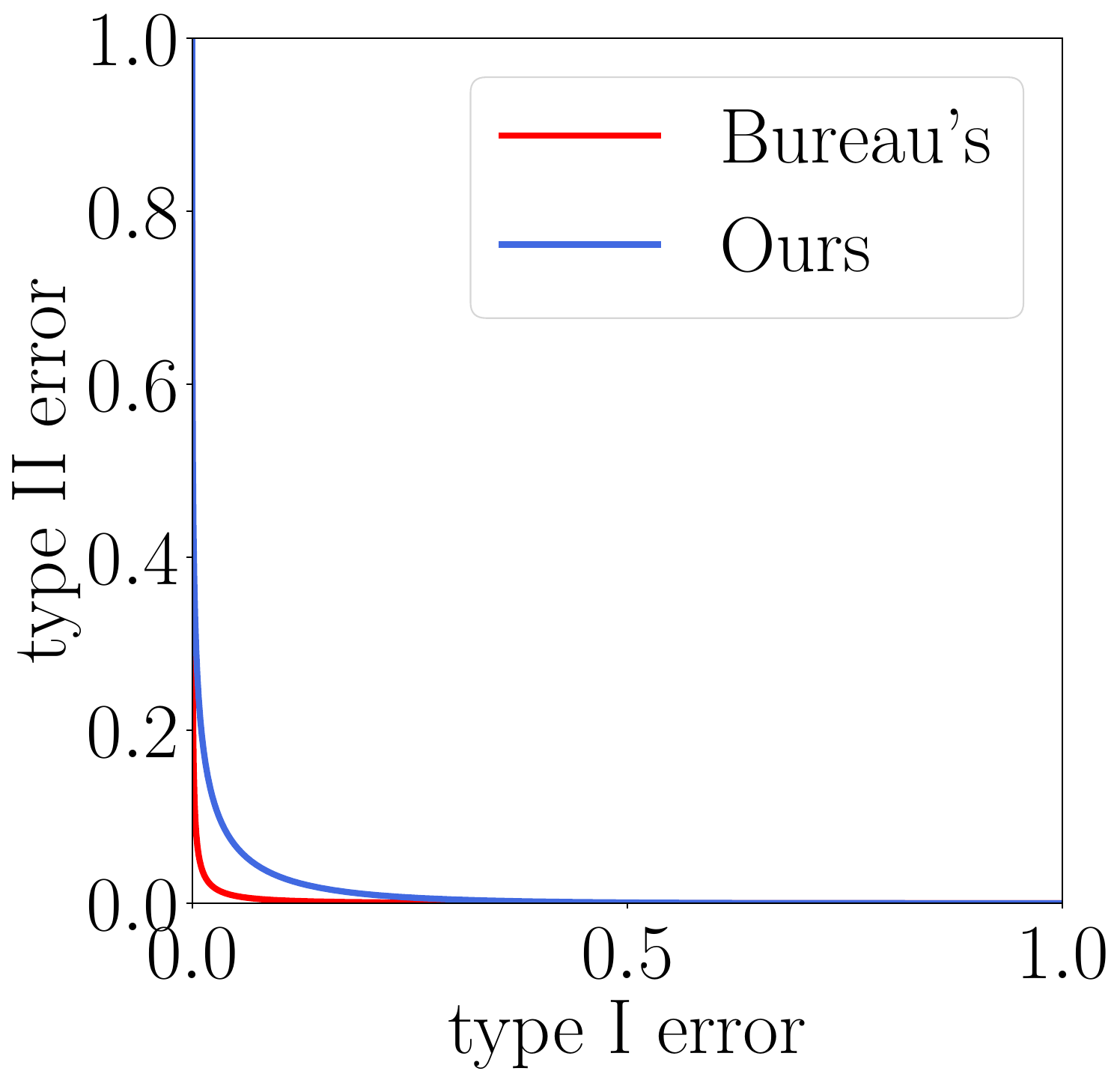} & \tradeoffsm{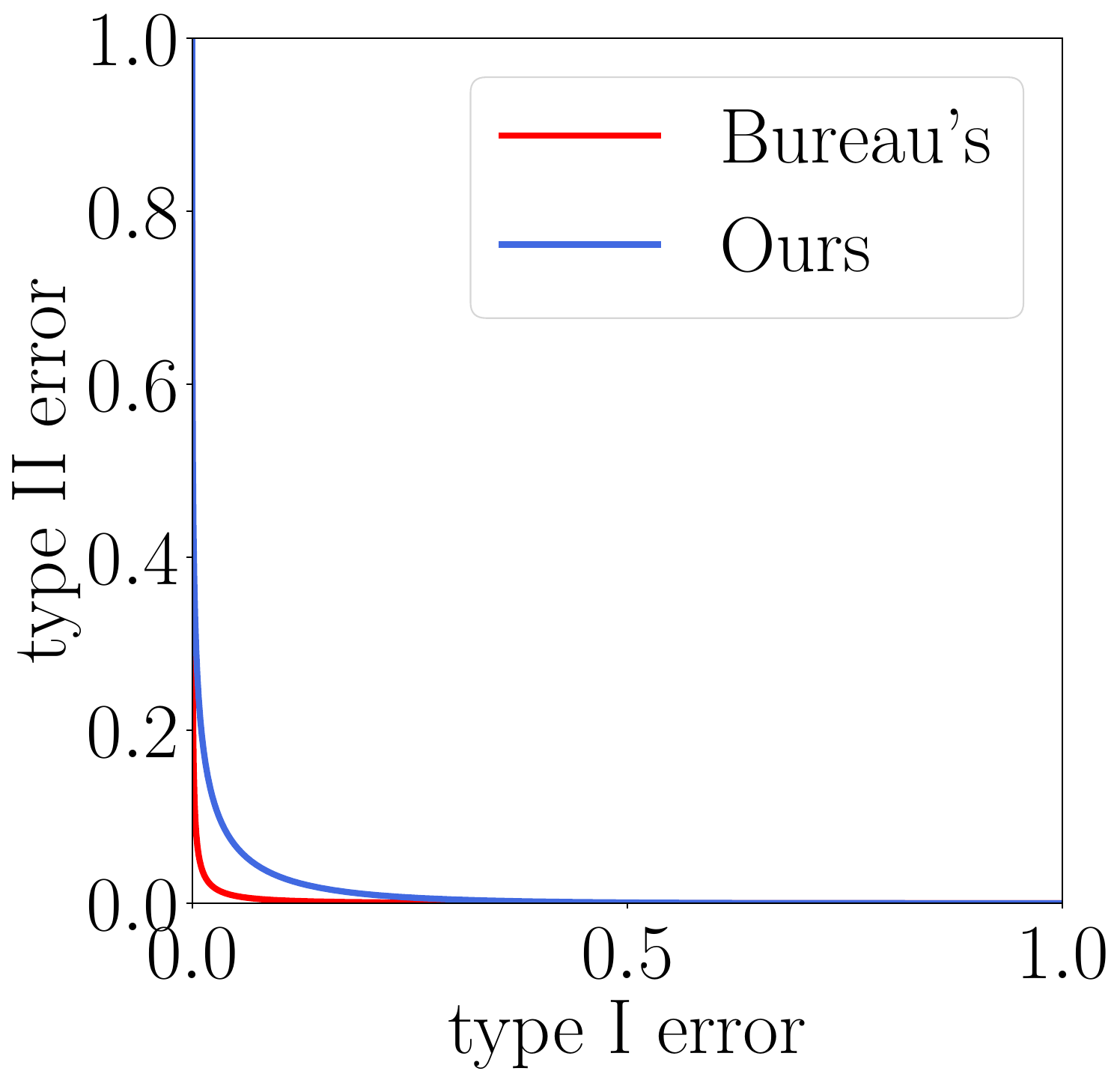} & \tradeoffsm{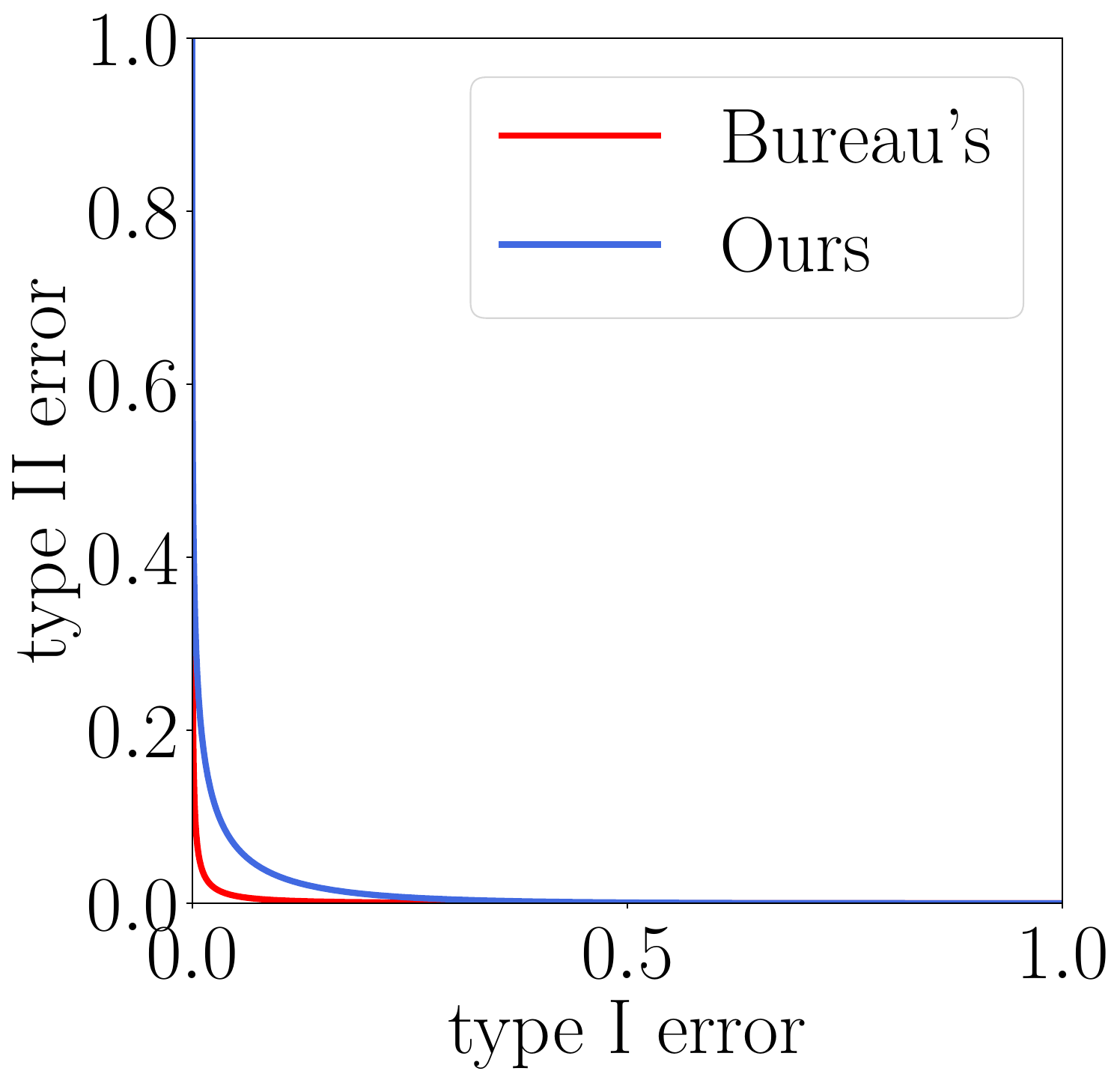} & \tradeoffsm{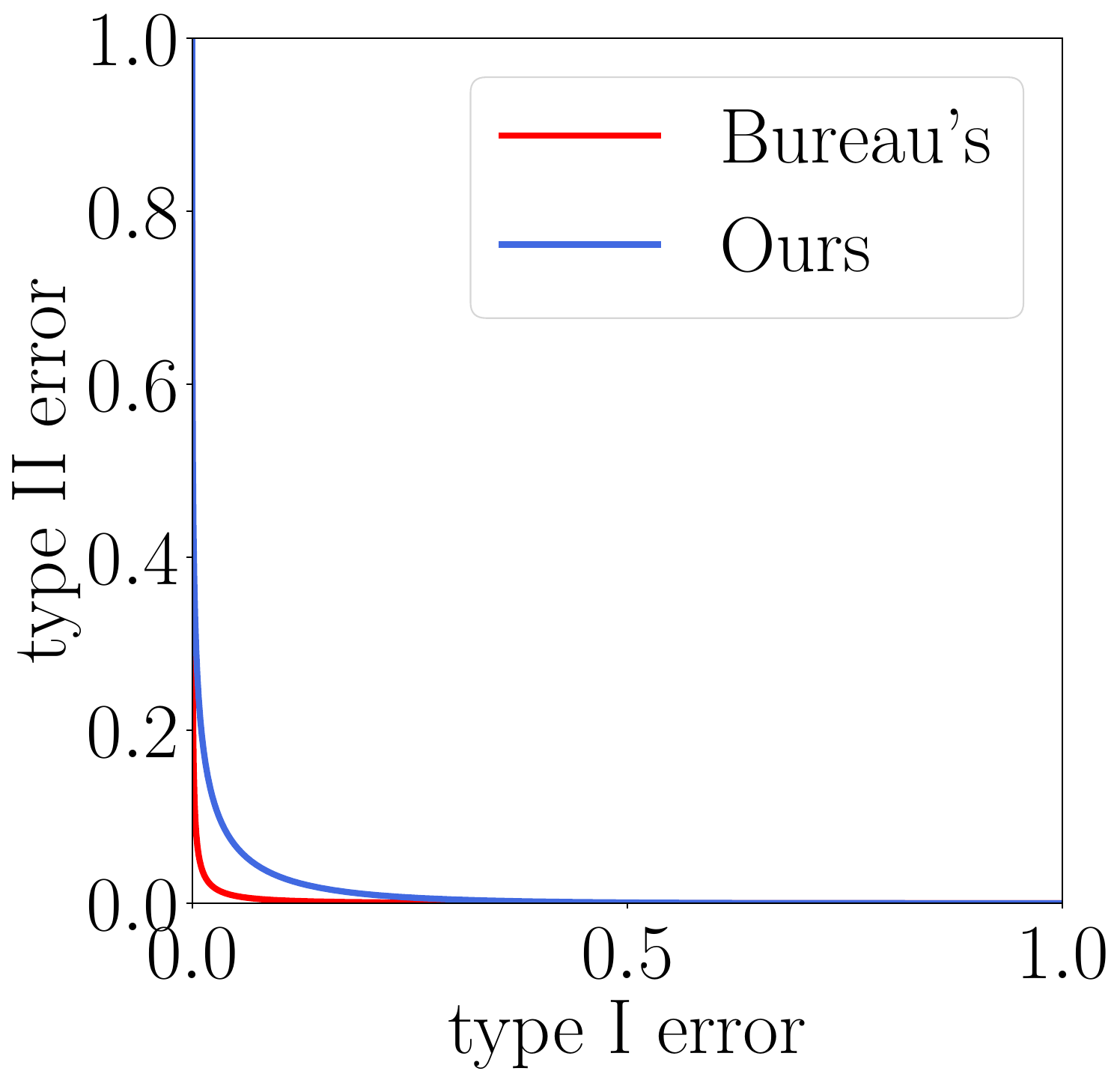} \\
        \tradeoffsm{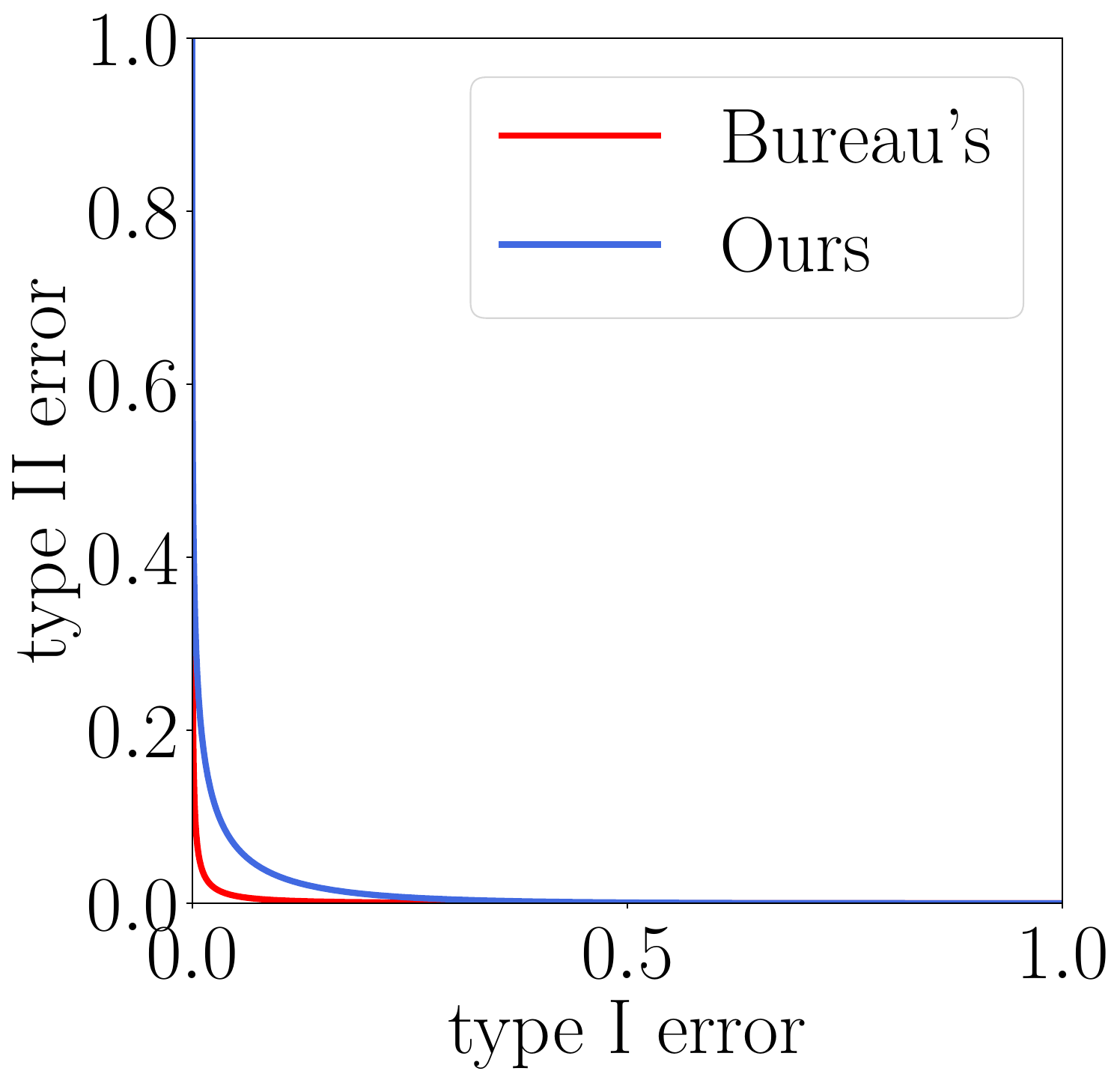} & \tradeoffsm{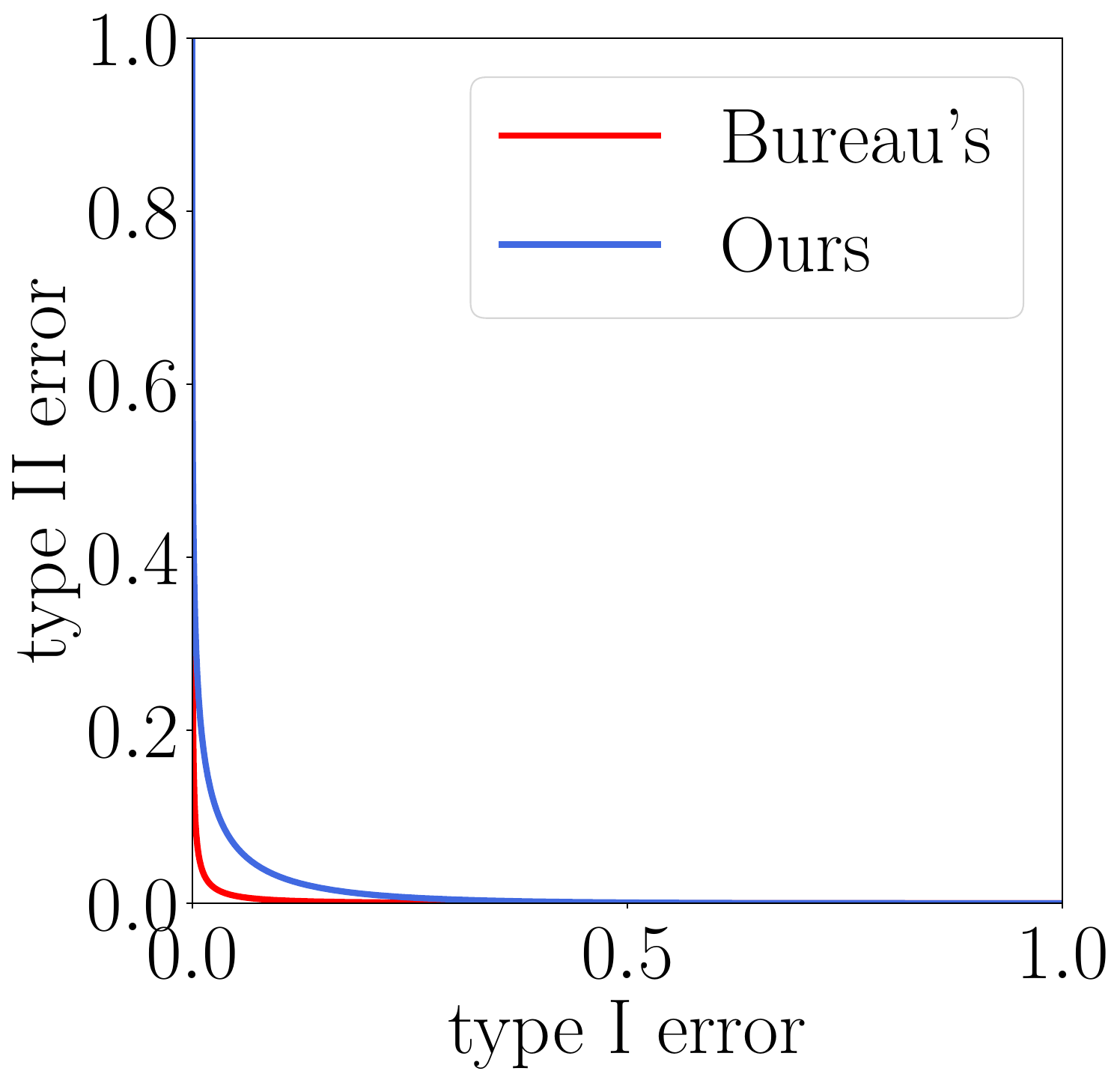} & \tradeoffsm{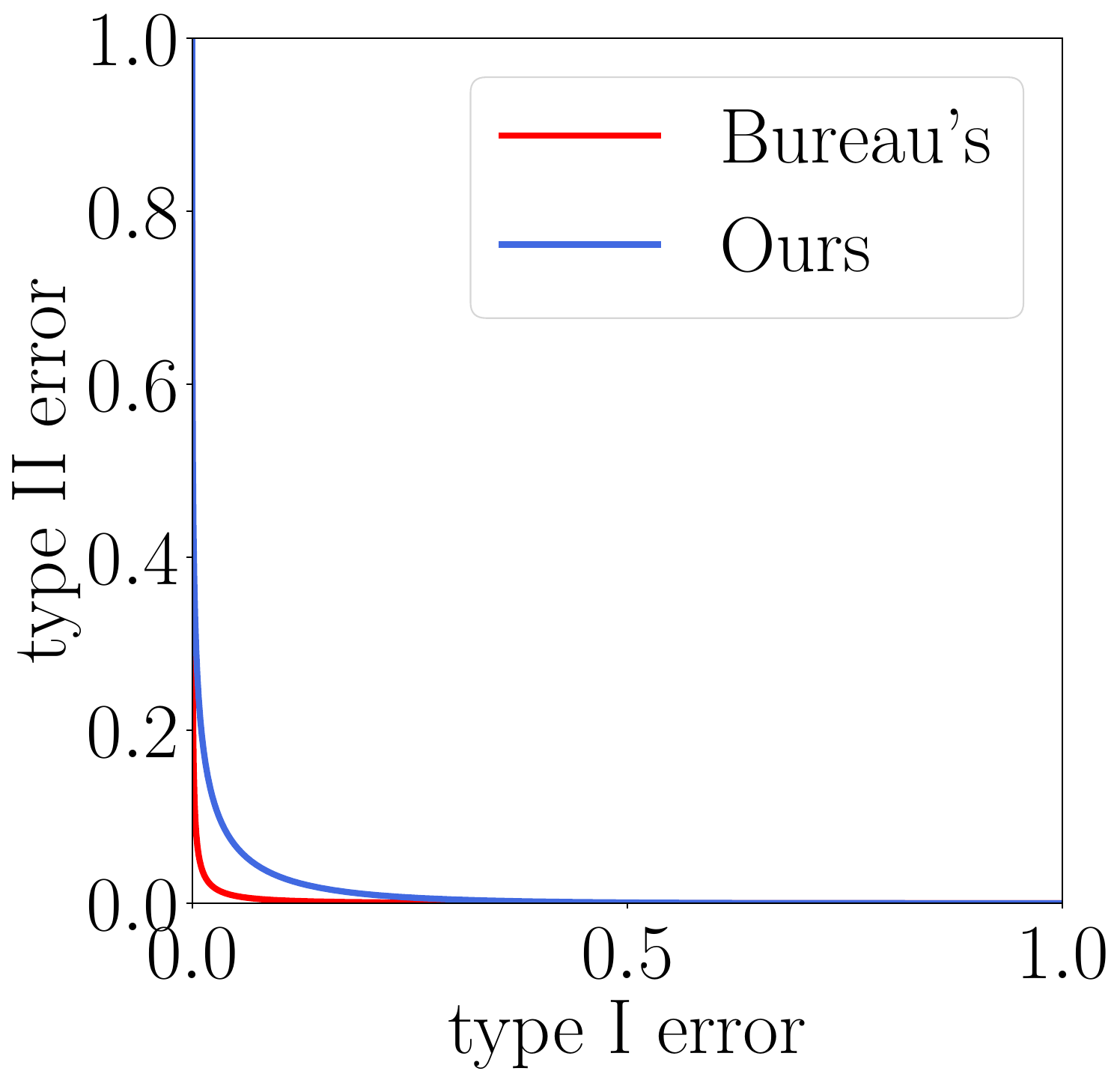} & \tradeoffsm{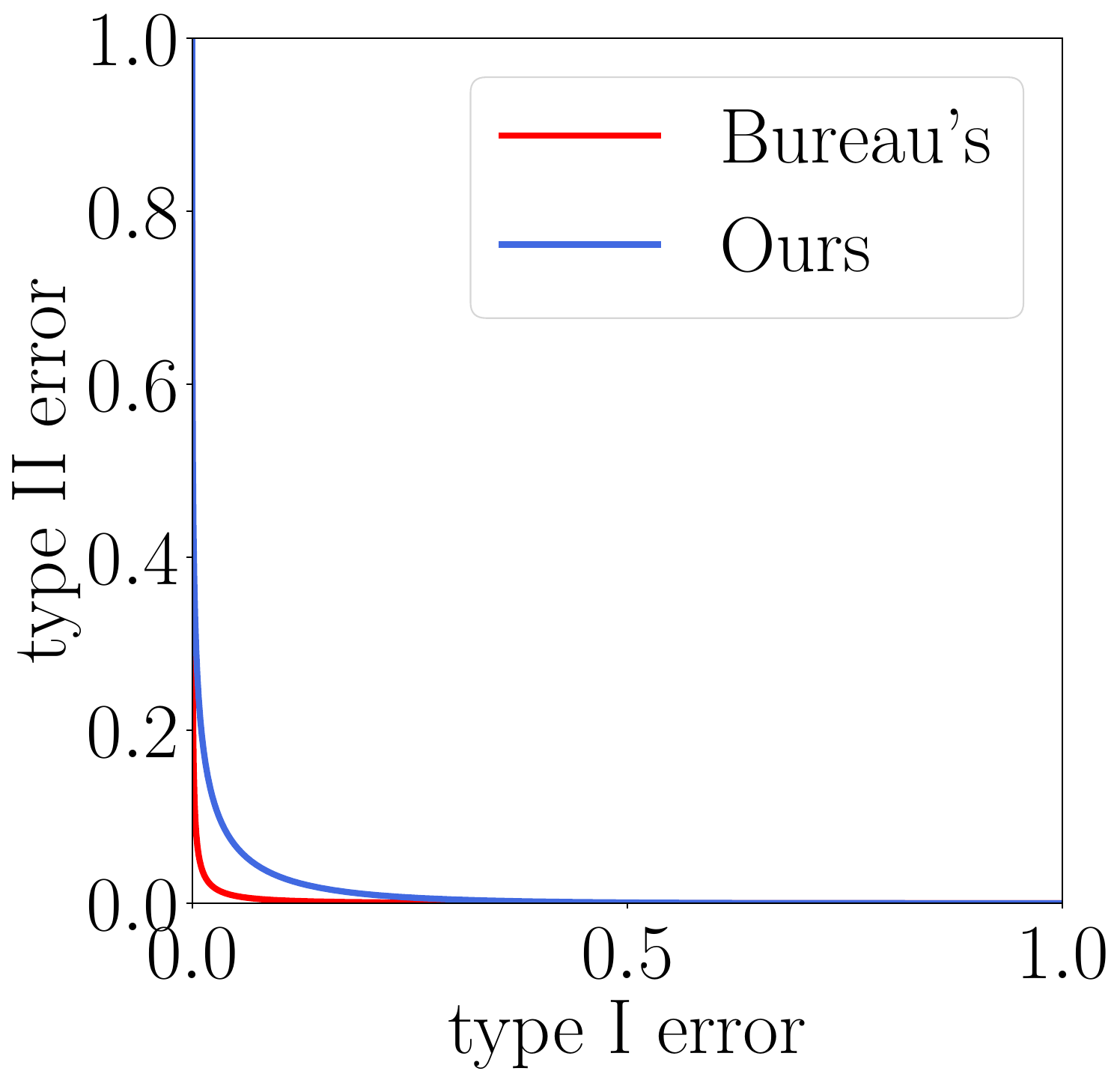} & \tradeoffsm{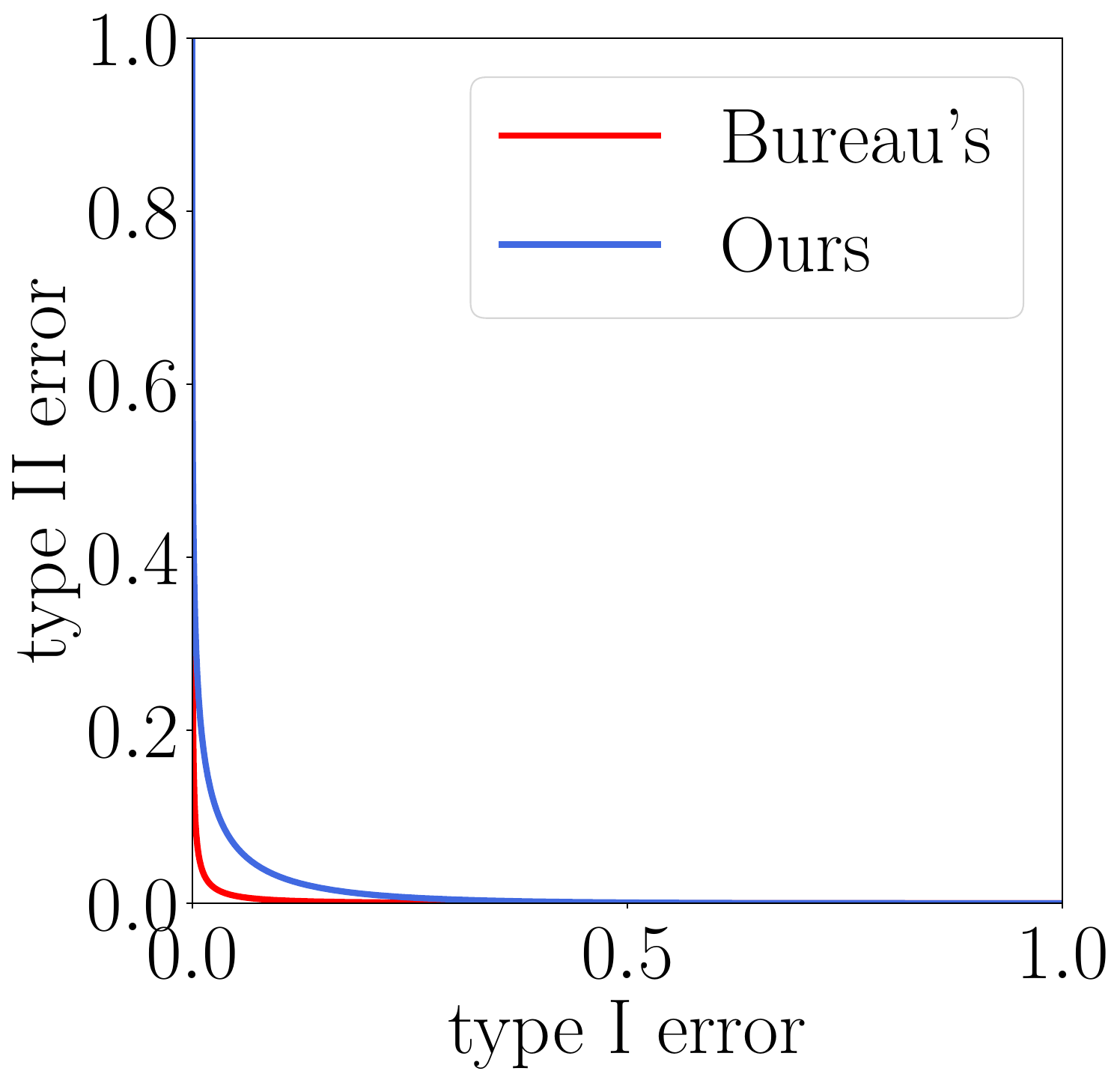} & \tradeoffsm{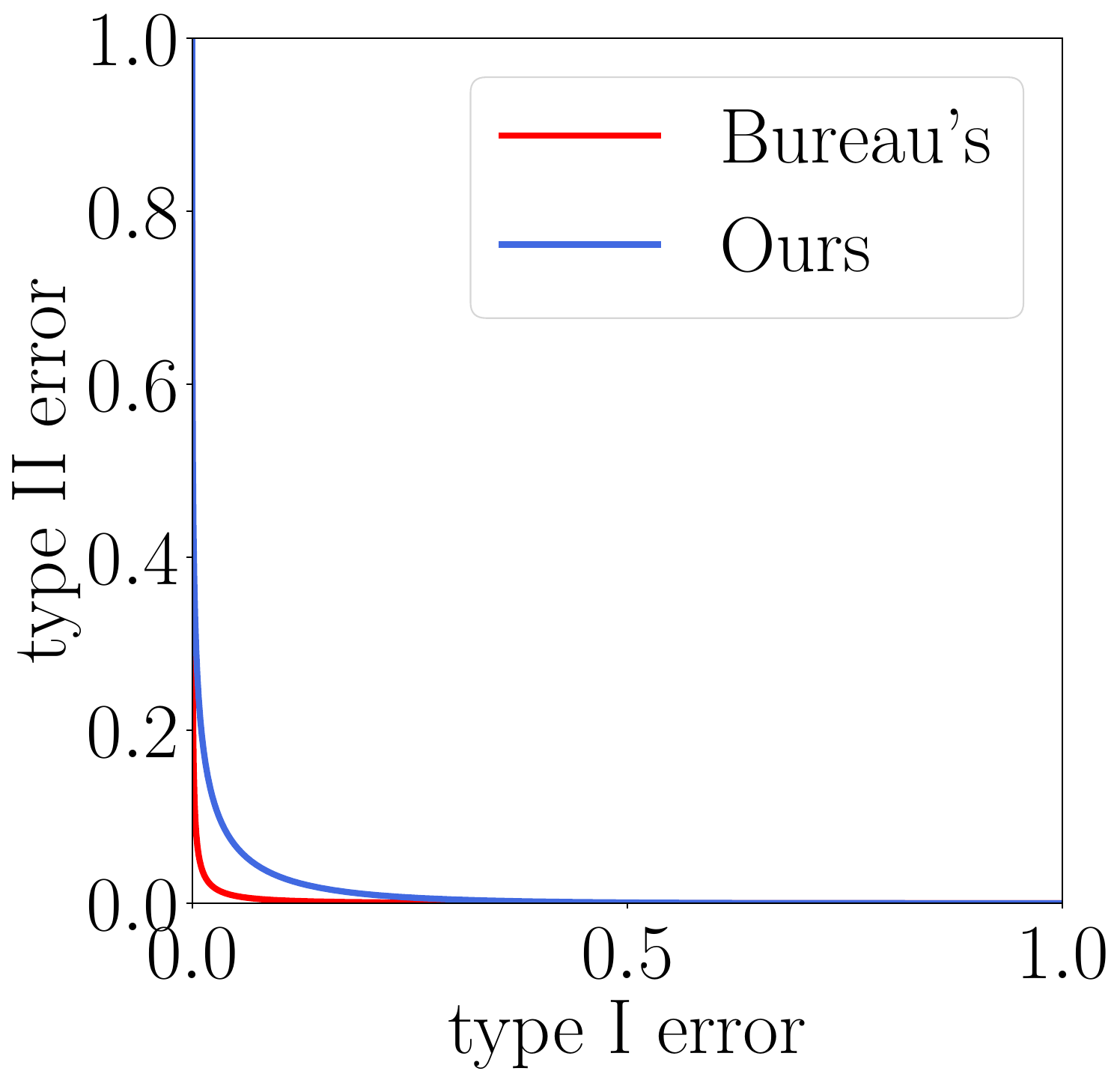} \\
        \tradeoffsm{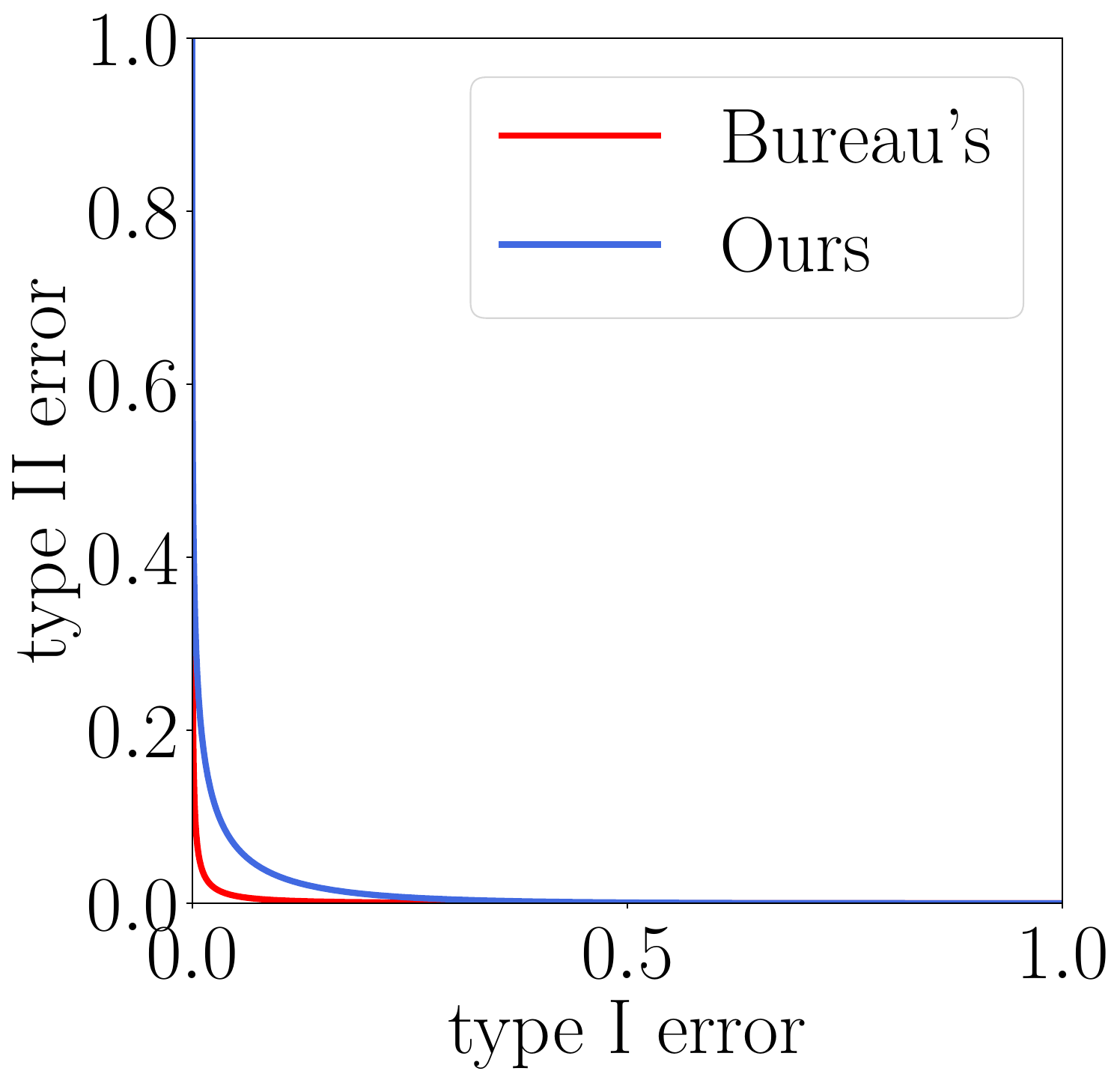} & \tradeoffsm{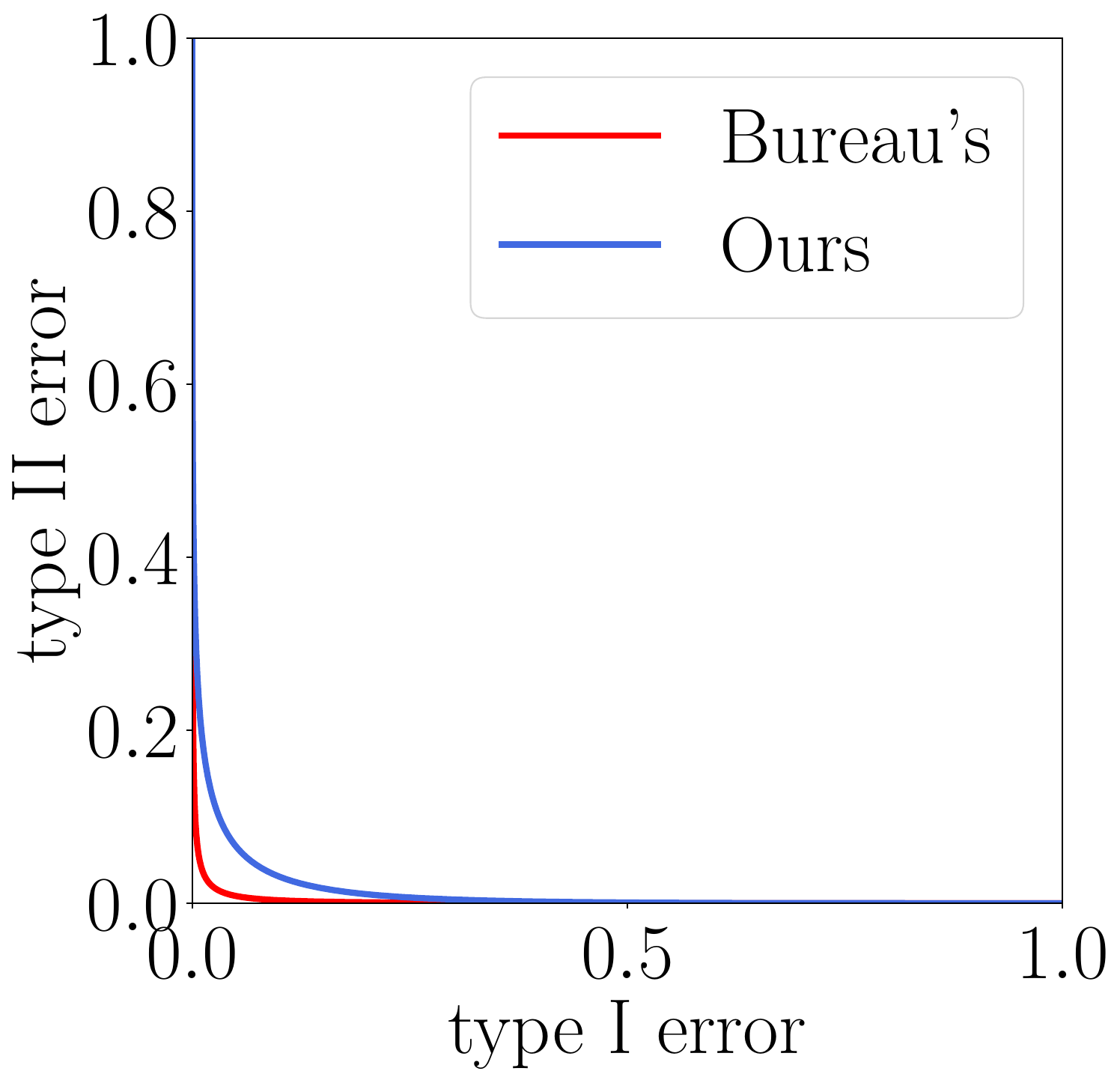} & \tradeoffsm{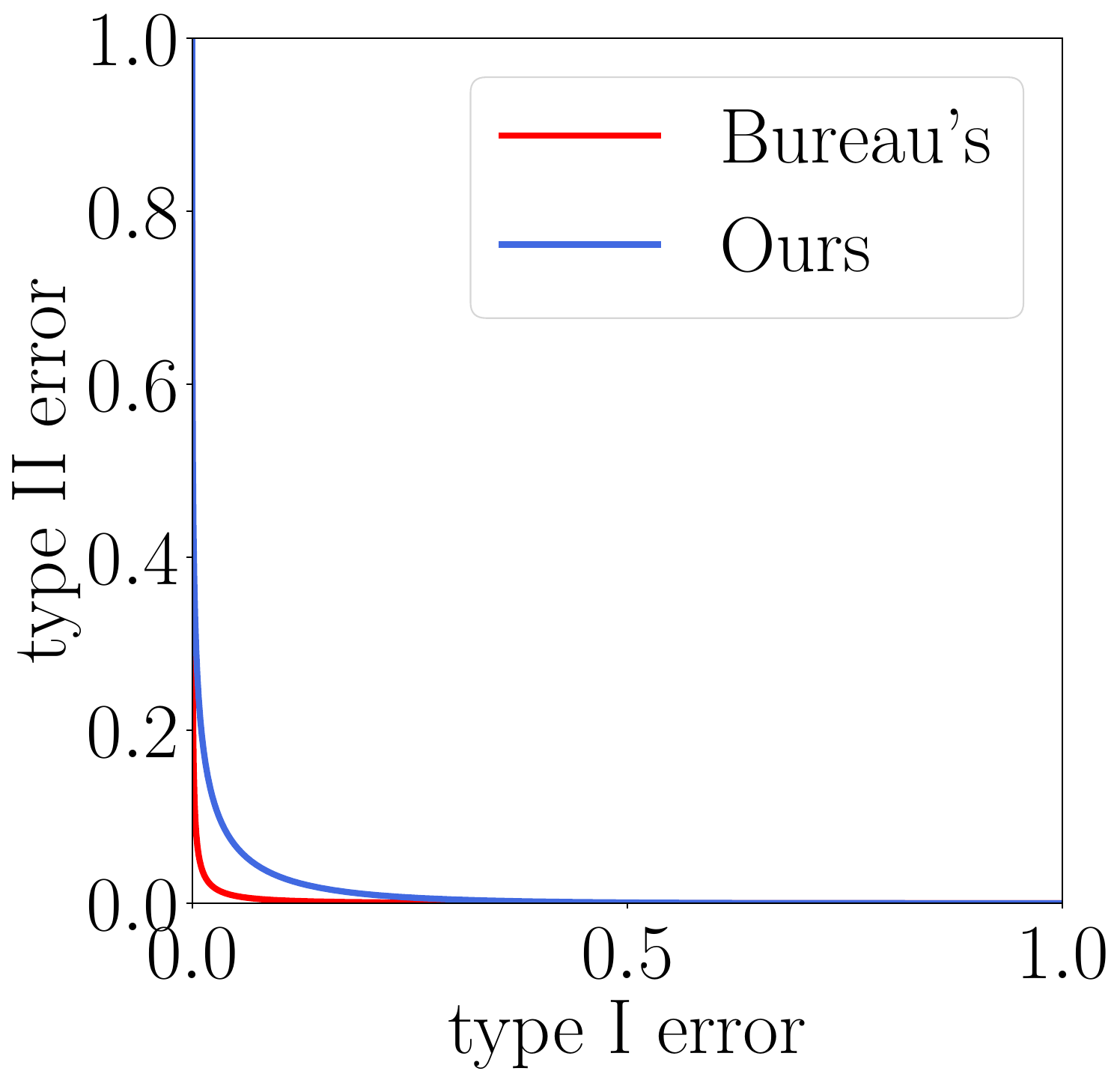} &  \multicolumn{3}{c}{\multirow[b]{2}{*}{\raisebox{0.5\cellh}{\tradeoffbigfig{13}}}} \\
        \tradeoffsm{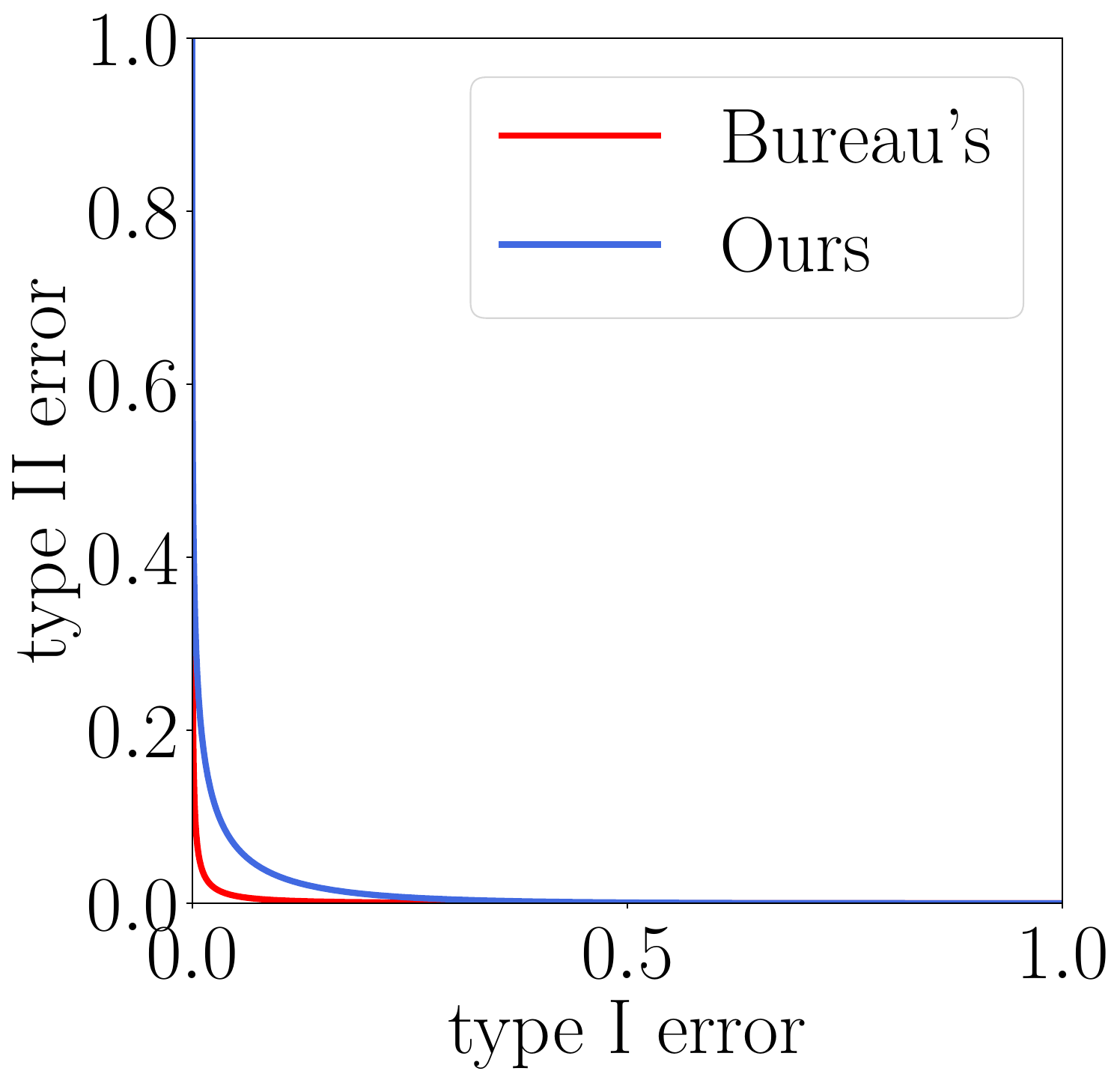} & \tradeoffsm{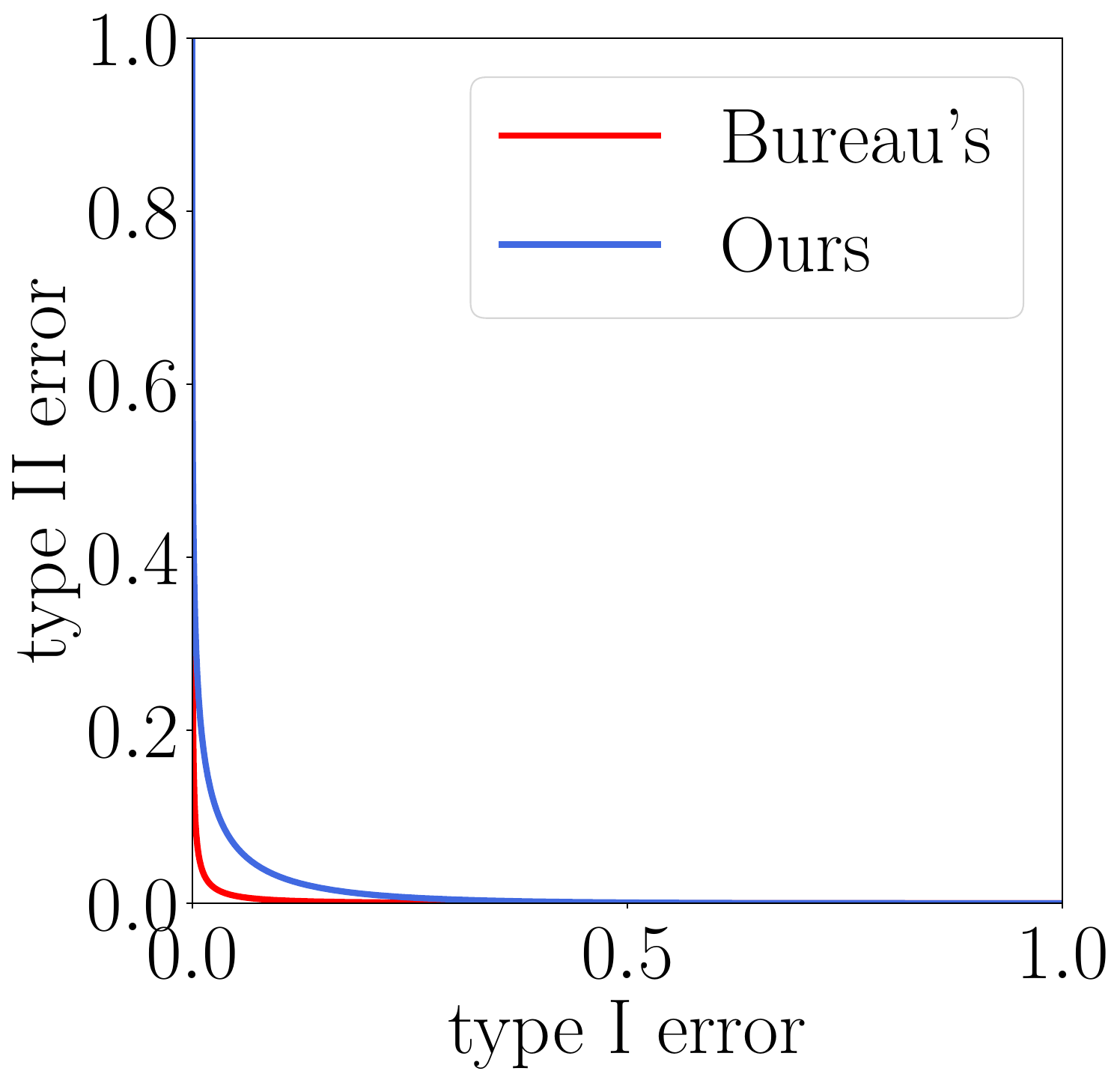} & \tradeoffsm{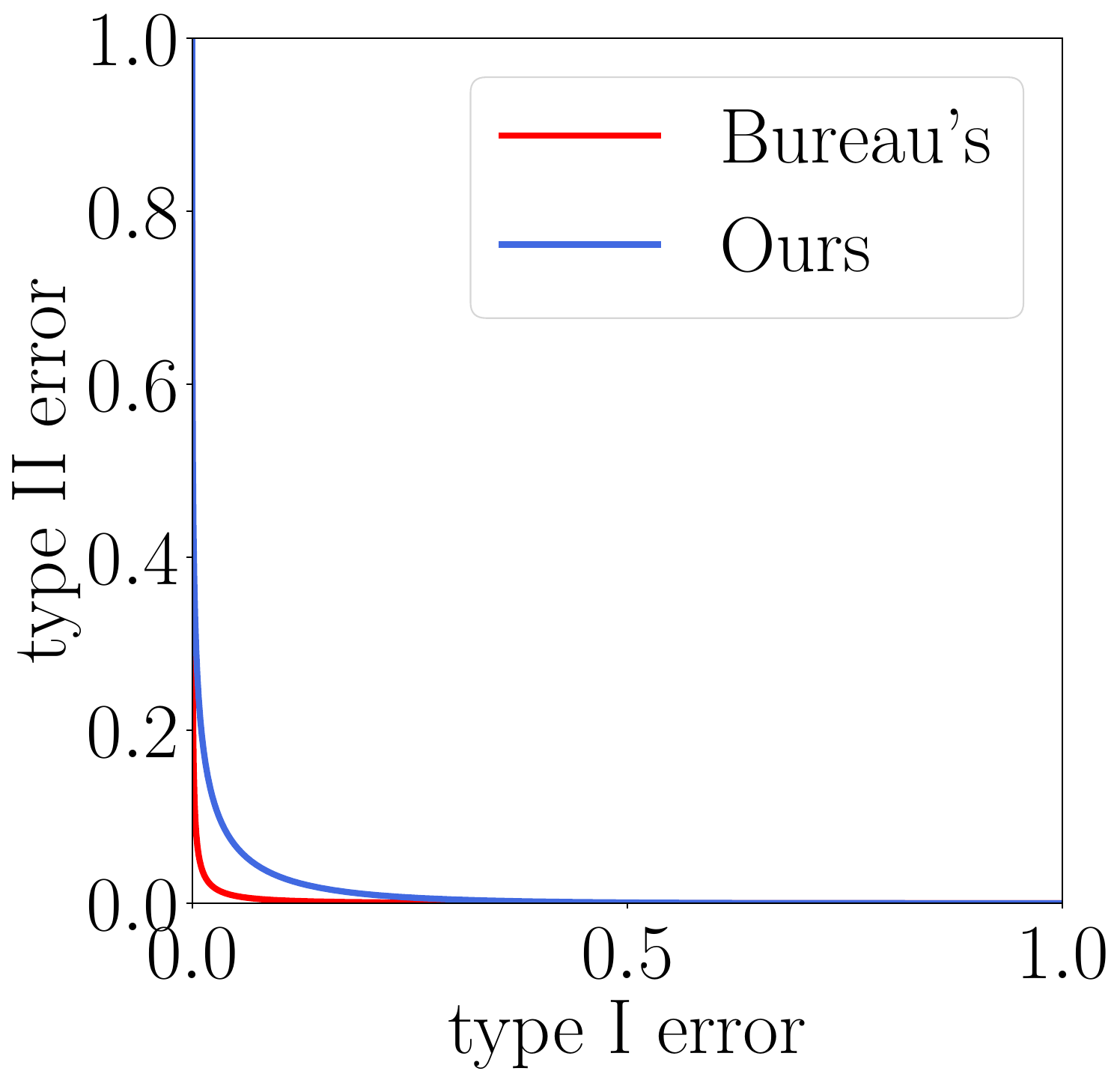} &  \multicolumn{3}{c}{}  \\
    \end{tabular}
    \vspace{-3mm}
    \caption{The $f$-DP curves for the composed mechanisms $[\widetilde{\mathbf{M}}_0, \widetilde{\mathbf{M}}_l]$ for all $l$, where $\widetilde{\mathbf{M}}_0$ corresponds to the privacy budget allocation in Table \ref{tab:PLB_typical}. Although the curves appear visually similar, they differ by amounts significantly larger than the numerical error introduced by floating-point arithmetic or the \texttt{mpmath} package, as quantified in Figure \ref{fig:sensitivity-privacy}. The enlarged panel in the lower-right corner displays the privacy profile of the composed mechanism $[\widetilde{\mathbf{M}}_0, \widetilde{\mathbf{M}}_0]$.}
    \label{fig:trade_off_paths_bypass}
\end{figure}

\subsection{\texorpdfstring{$(\varepsilon, \delta)$}{}-DP curves}
\label{sec:epsilondp}

Since the $(\varepsilon, \delta)$-DP curve is equivalent to the $f$-DP curve and has been suggested as an alternative way to report privacy guarantees \citep{gomez2025varepsilon}, we present the exact $(\varepsilon, \delta)$-DP curve for the 2020 Census DHC File in the right panel of Figure \ref{fig:overall-privacy}.
Similarly, Figure \ref{fig:eps_delta_paths_bypass} illustrates the privacy guarantees of the composed mechanisms $[\widetilde{\mathbf{M}}_0, \widetilde{\mathbf{M}}_l]$ for all $l$, expressed in terms of their $(\varepsilon, \delta)$-DP curves.

\newcommand{\epsilonlm}[1]{%
  \parbox[c][\cellh][c]{0.16\textwidth}{\centering
    \includegraphics[width=\linewidth,height=\cellh,keepaspectratio]{Figures/v2/epsilon_delta_curve_path_#1_to_13.pdf}
  }
}
\newcommand{\epsilonsm}[1]{%
  \parbox[c][\cellh][c]{0.16\textwidth}{\centering
    \includegraphics[width=\linewidth,height=\cellh,keepaspectratio]{Figures/v2/epsilon_delta_curve_path_13_to_#1.pdf}
  }
}
\newcommand{\epsilonbigfig}[1]{
  \parbox[t][2\cellh][t]{0.32\textwidth}{\centering\vspace{-0.85cm}
  \includegraphics[width=\linewidth,height=2\cellh,keepaspectratio]{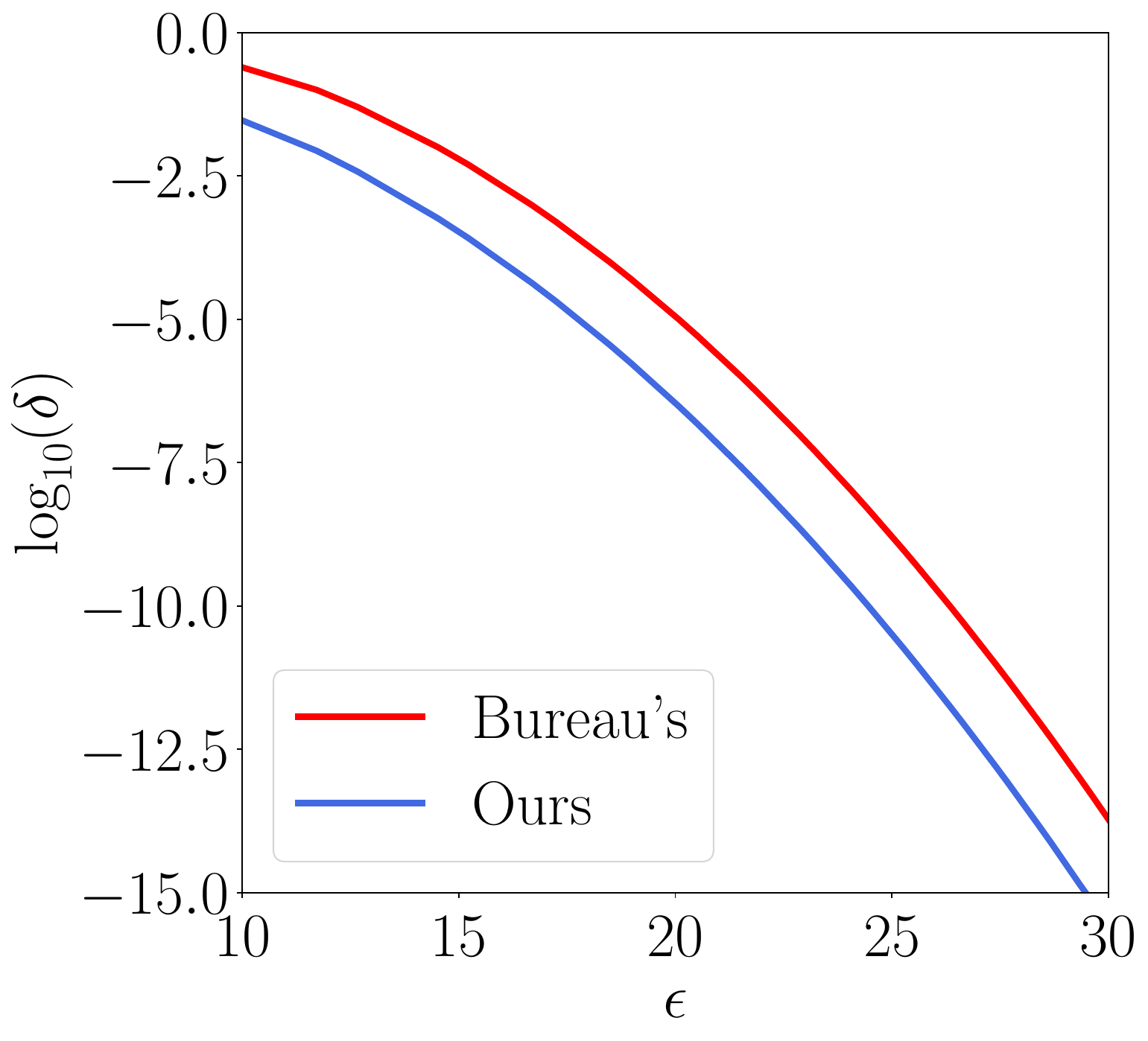}
  }
}

\begin{figure}[!htp]
    \centering
    \setlength{\tabcolsep}{2pt}
    \renewcommand{\arraystretch}{1.0}

    \begin{tabular}{@{}cccccc@{}}
        \epsilonlm{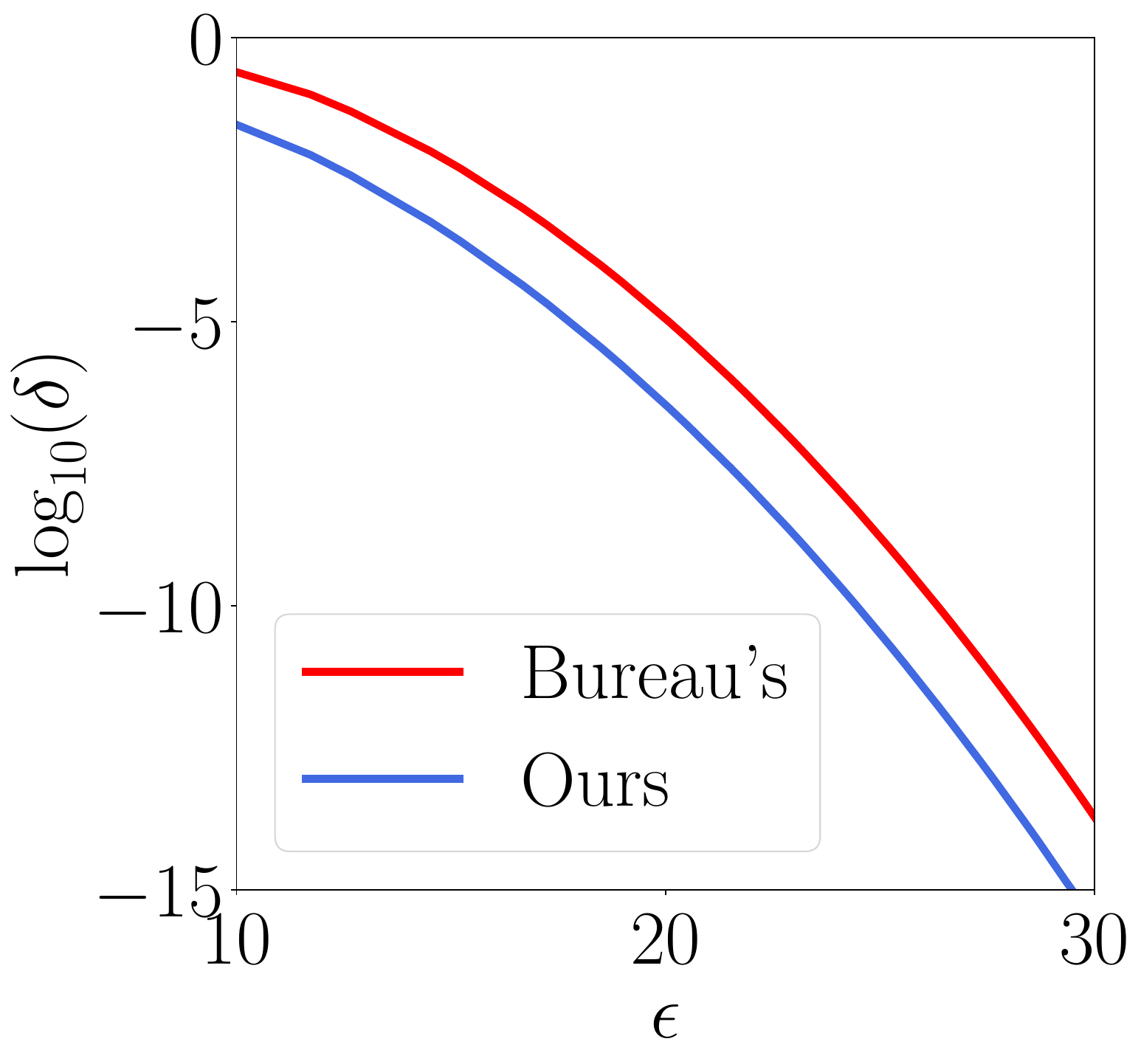}  & \epsilonlm{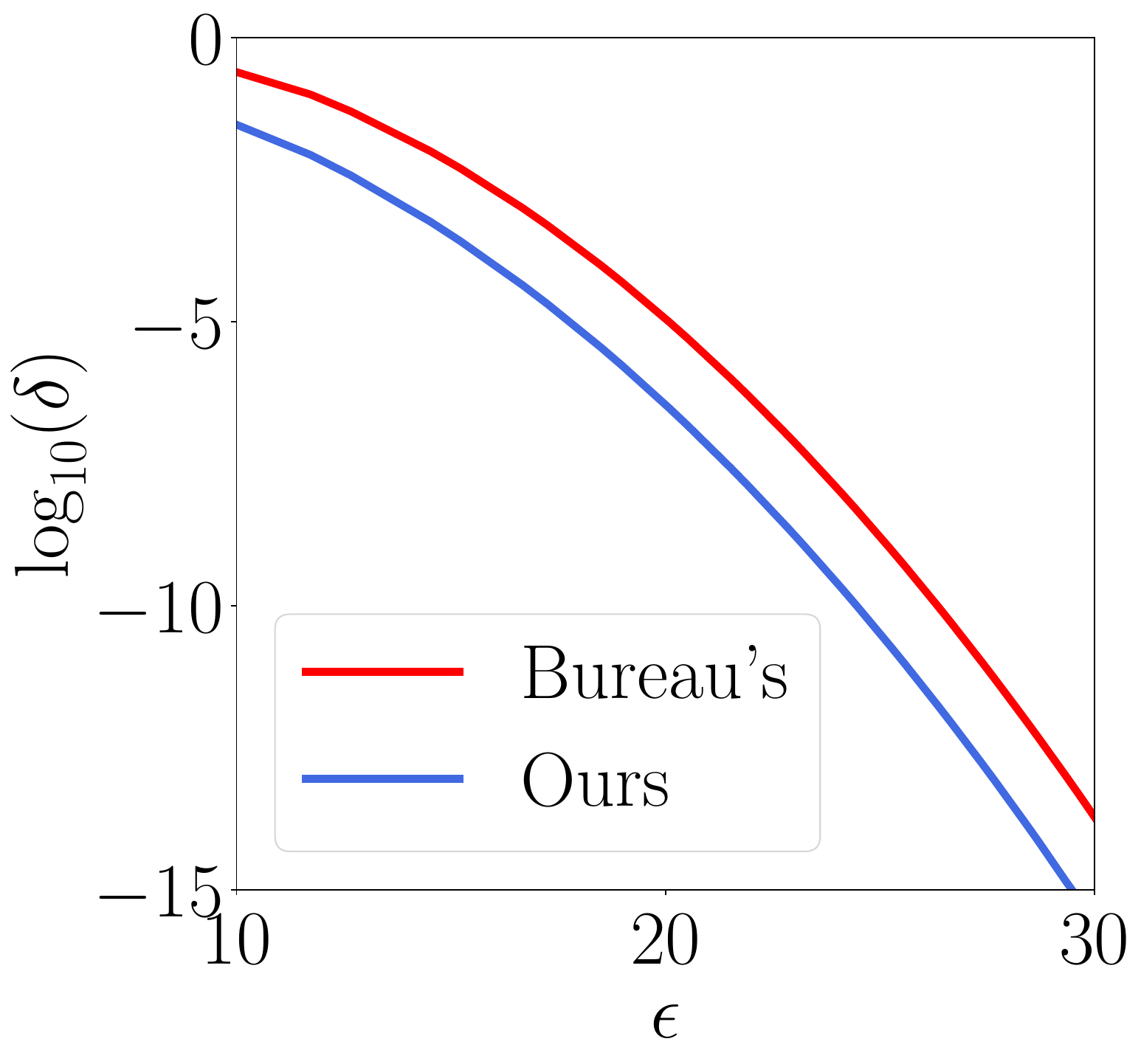}  & \epsilonlm{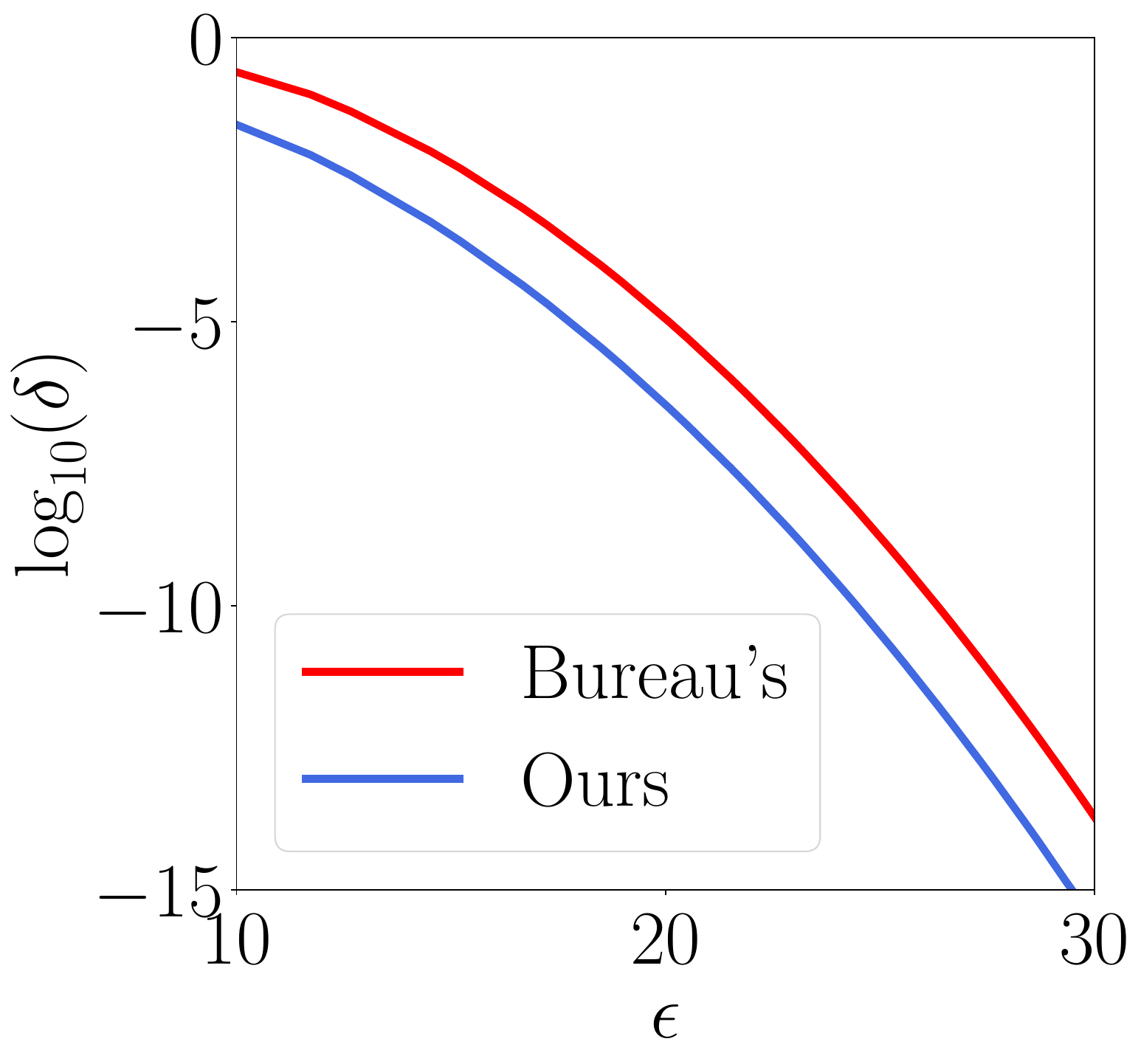}  & \epsilonlm{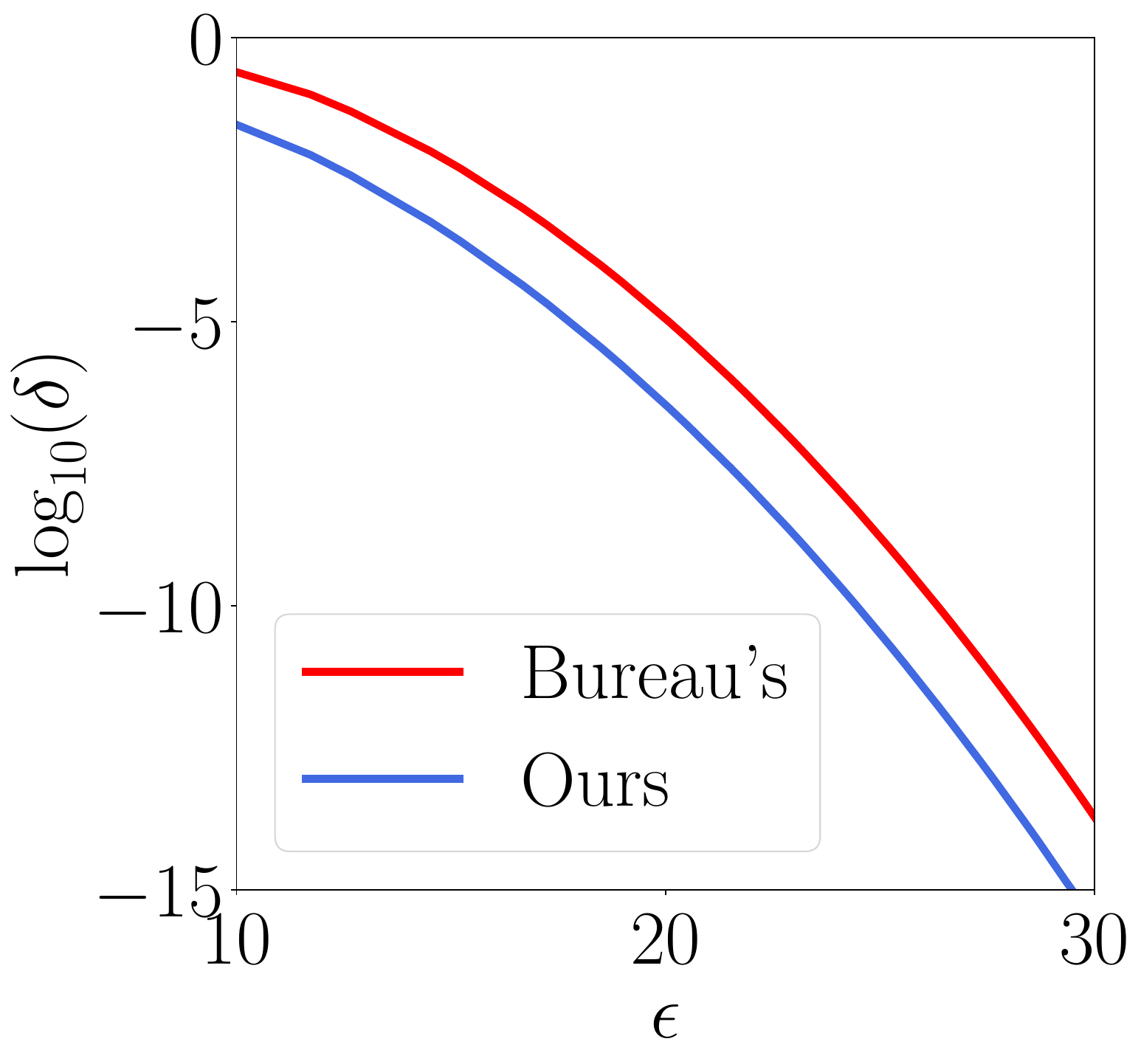}  & \epsilonlm{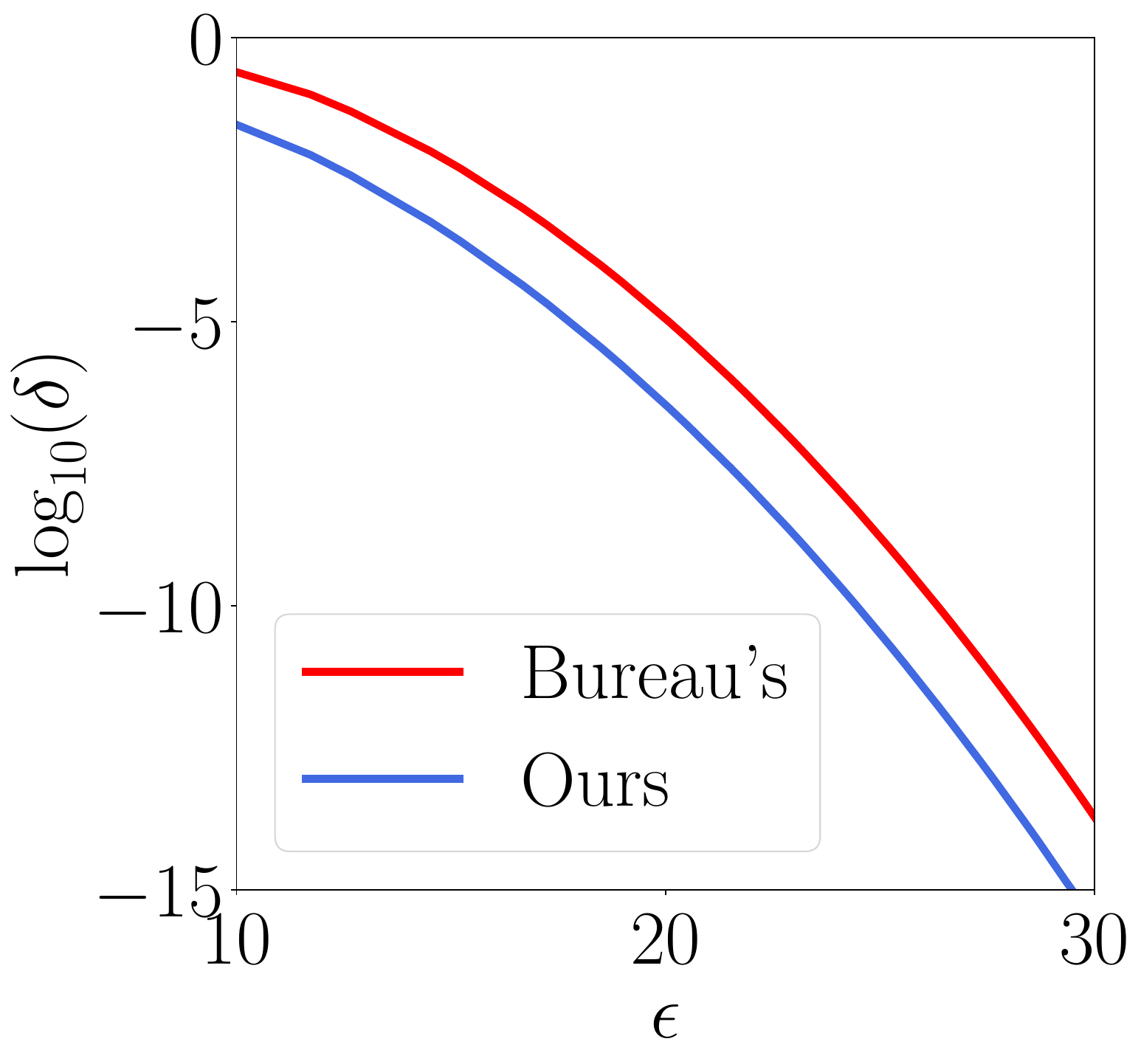}  & \epsilonlm{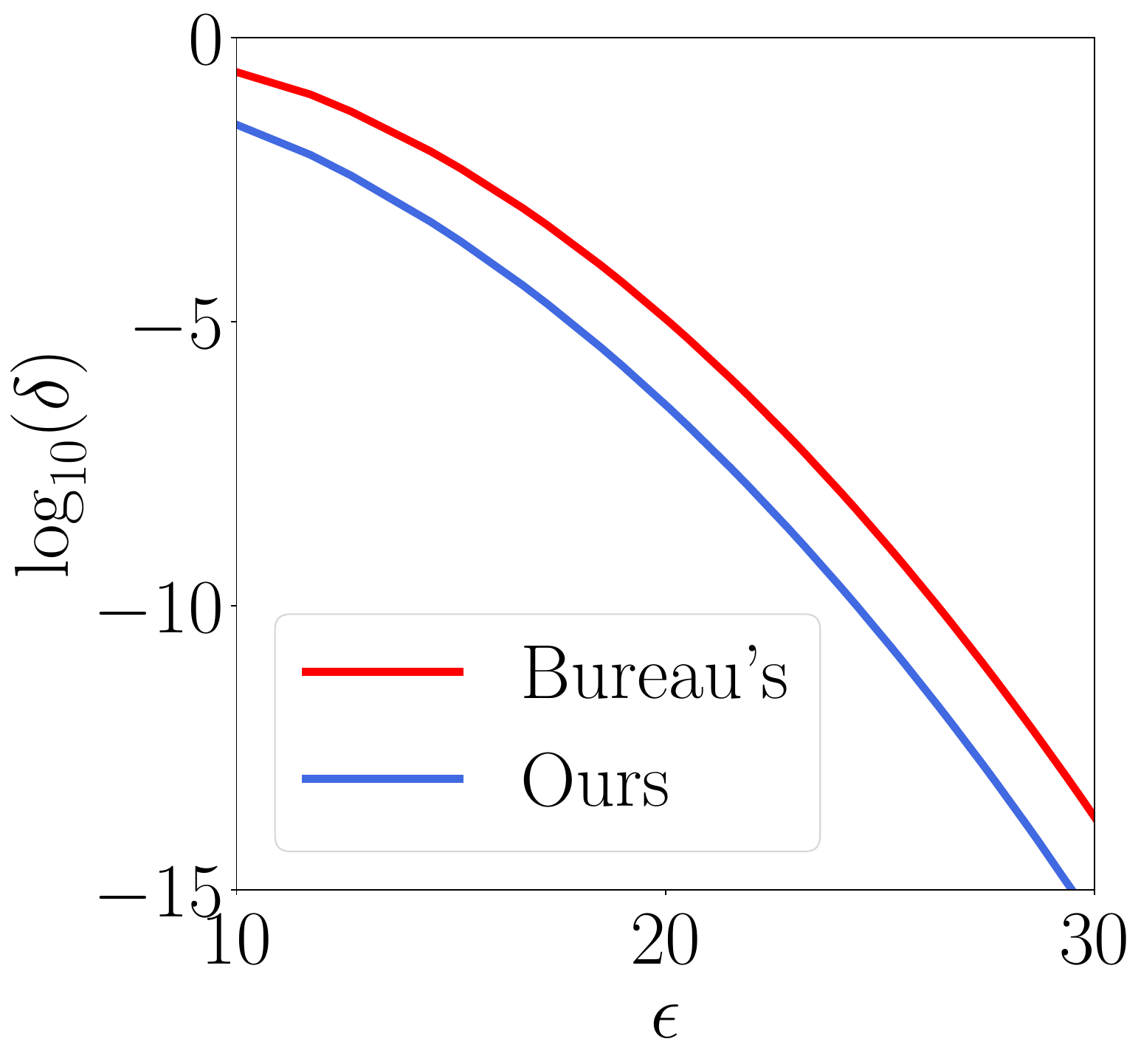}  \\
        \epsilonlm{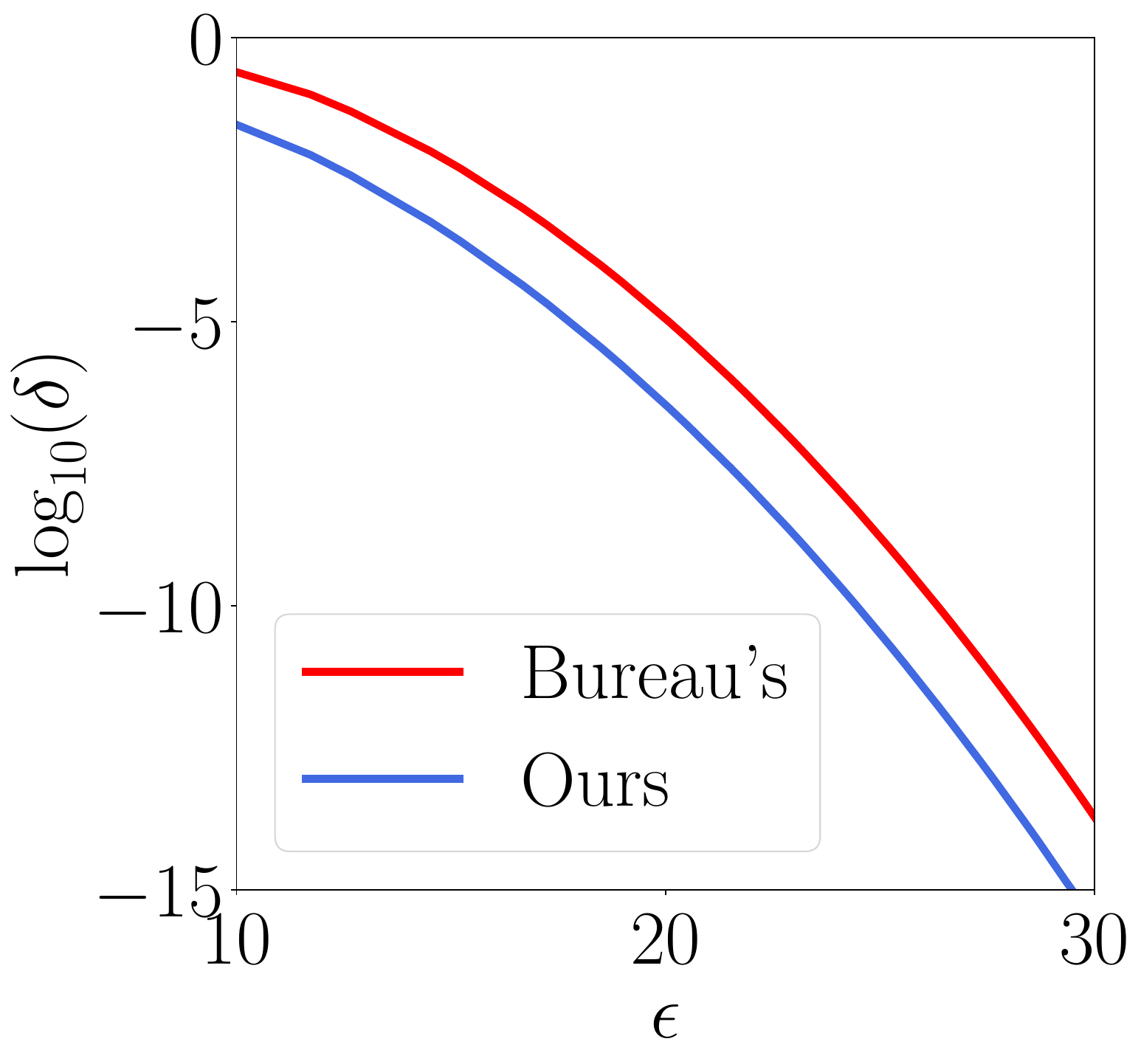}  & \epsilonlm{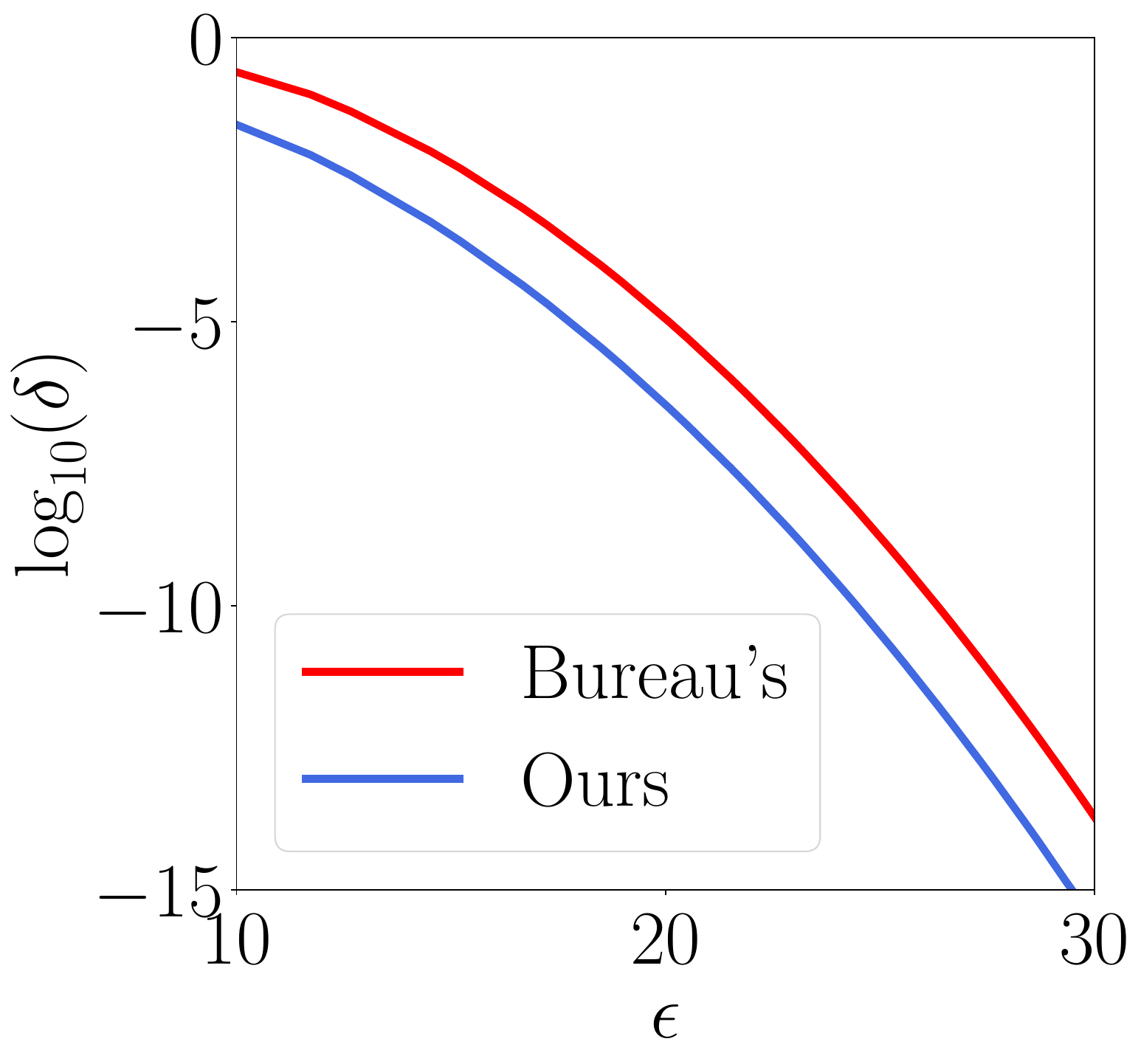}  & \epsilonlm{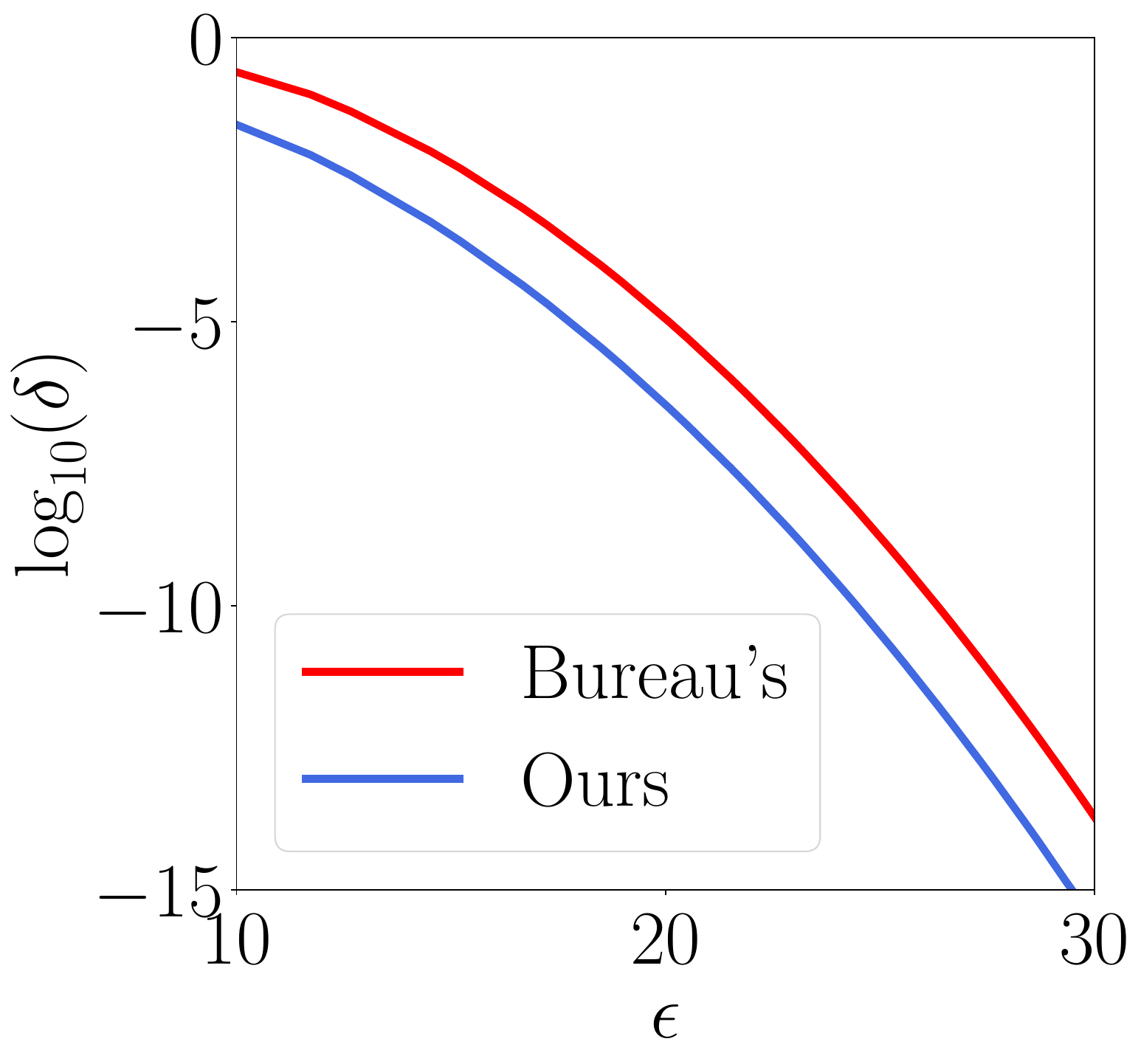}  & \epsilonlm{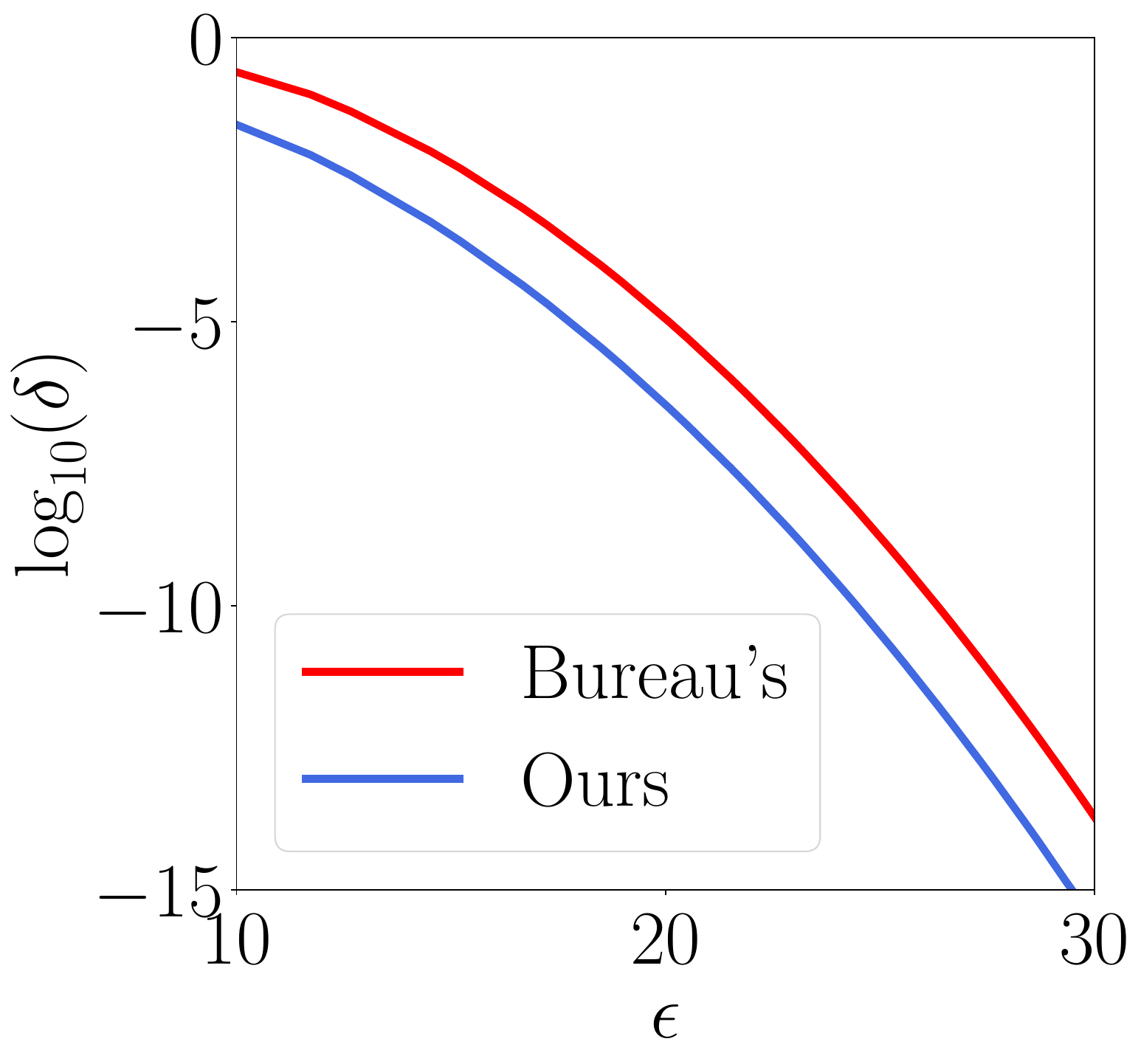} & \epsilonlm{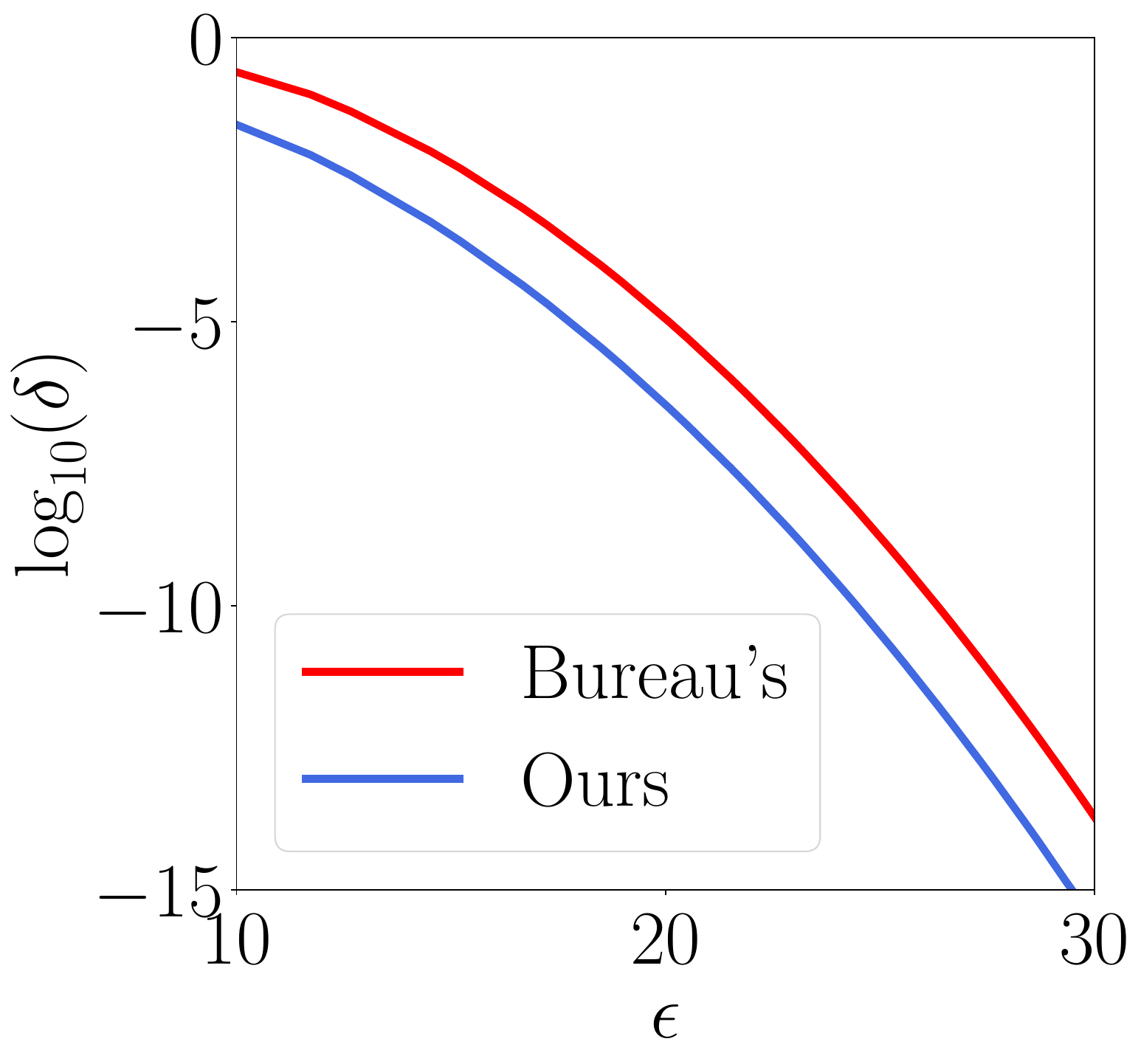} & \epsilonlm{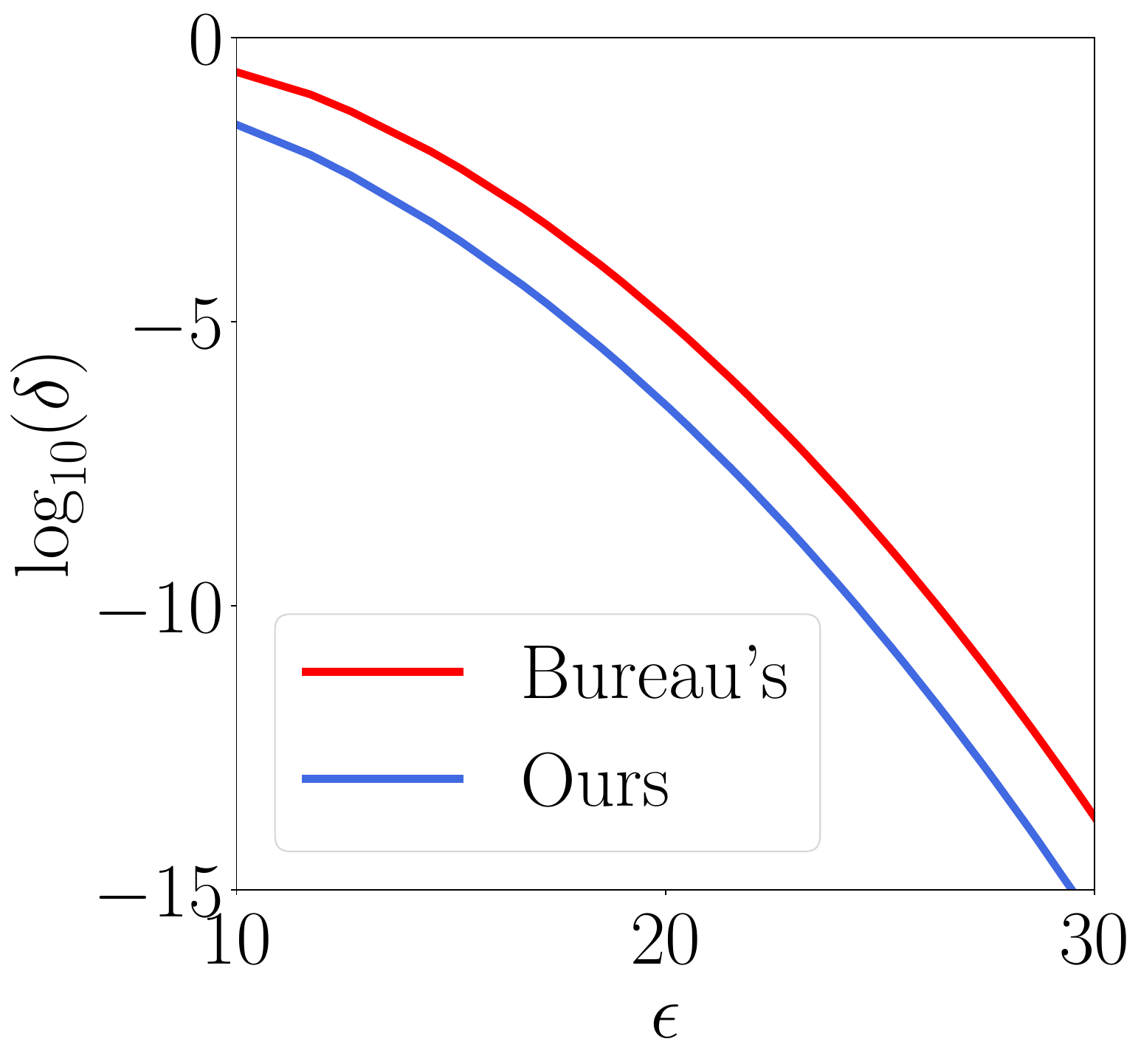} \\
        \epsilonsm{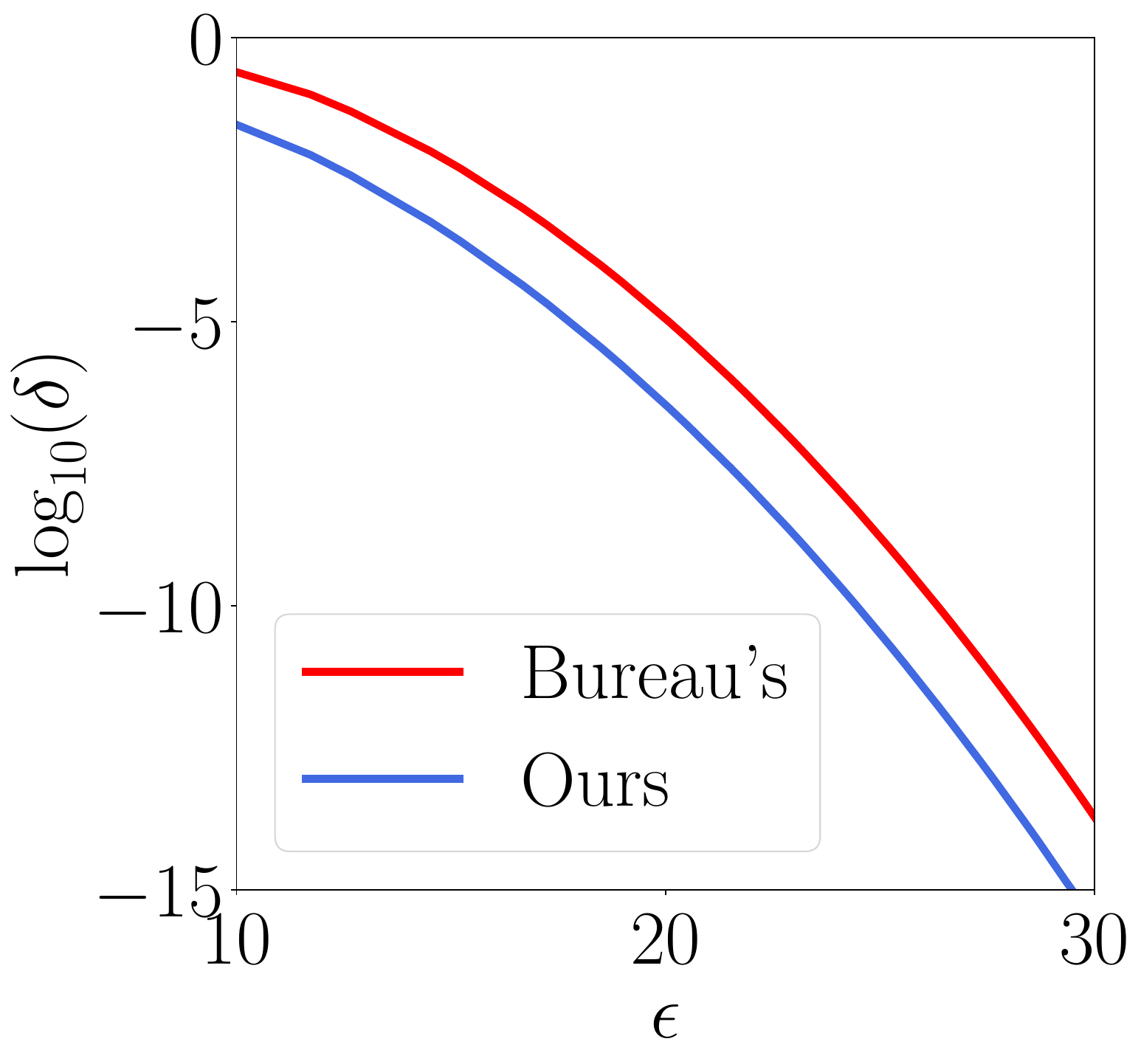} & \epsilonsm{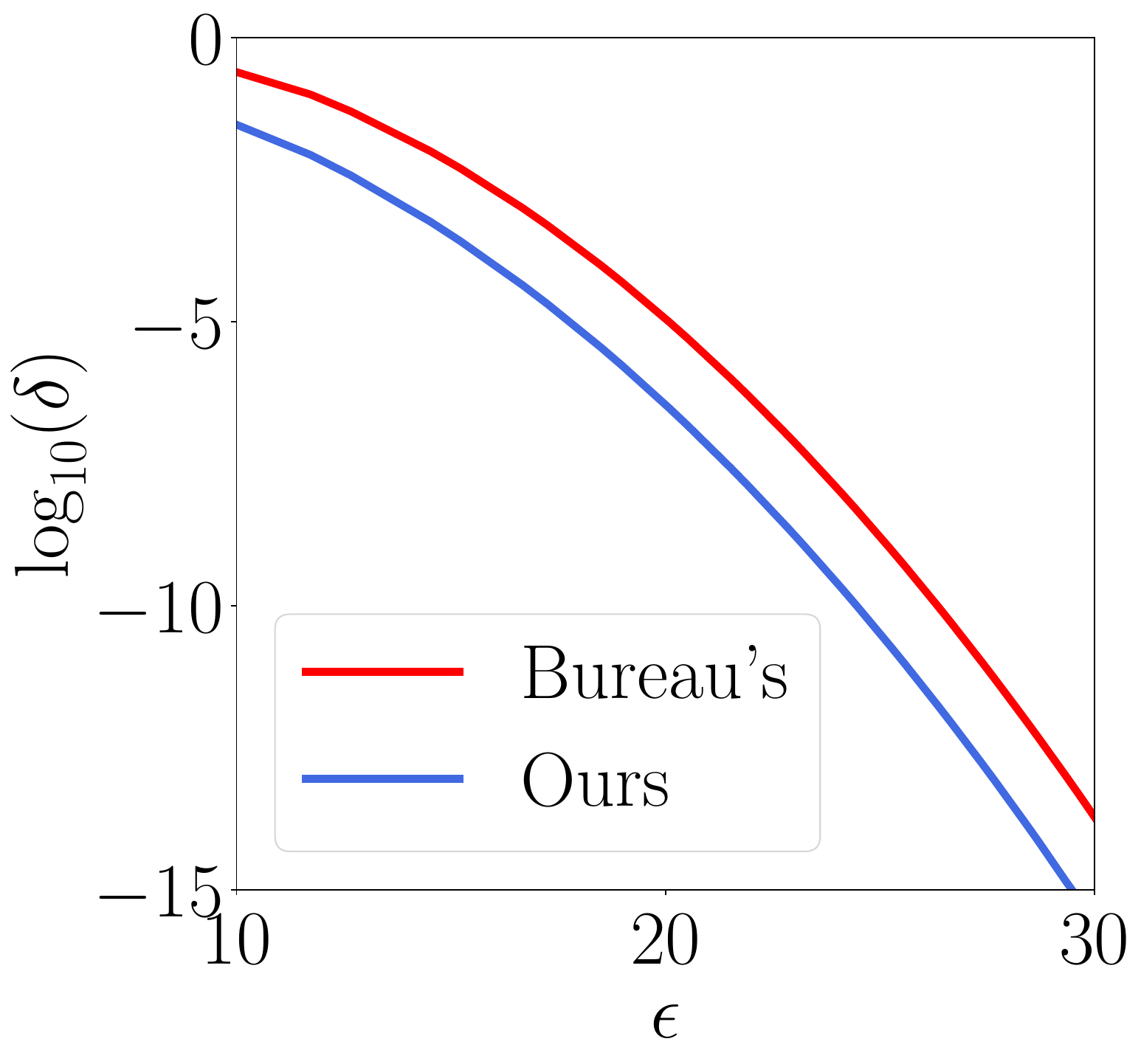} & \epsilonsm{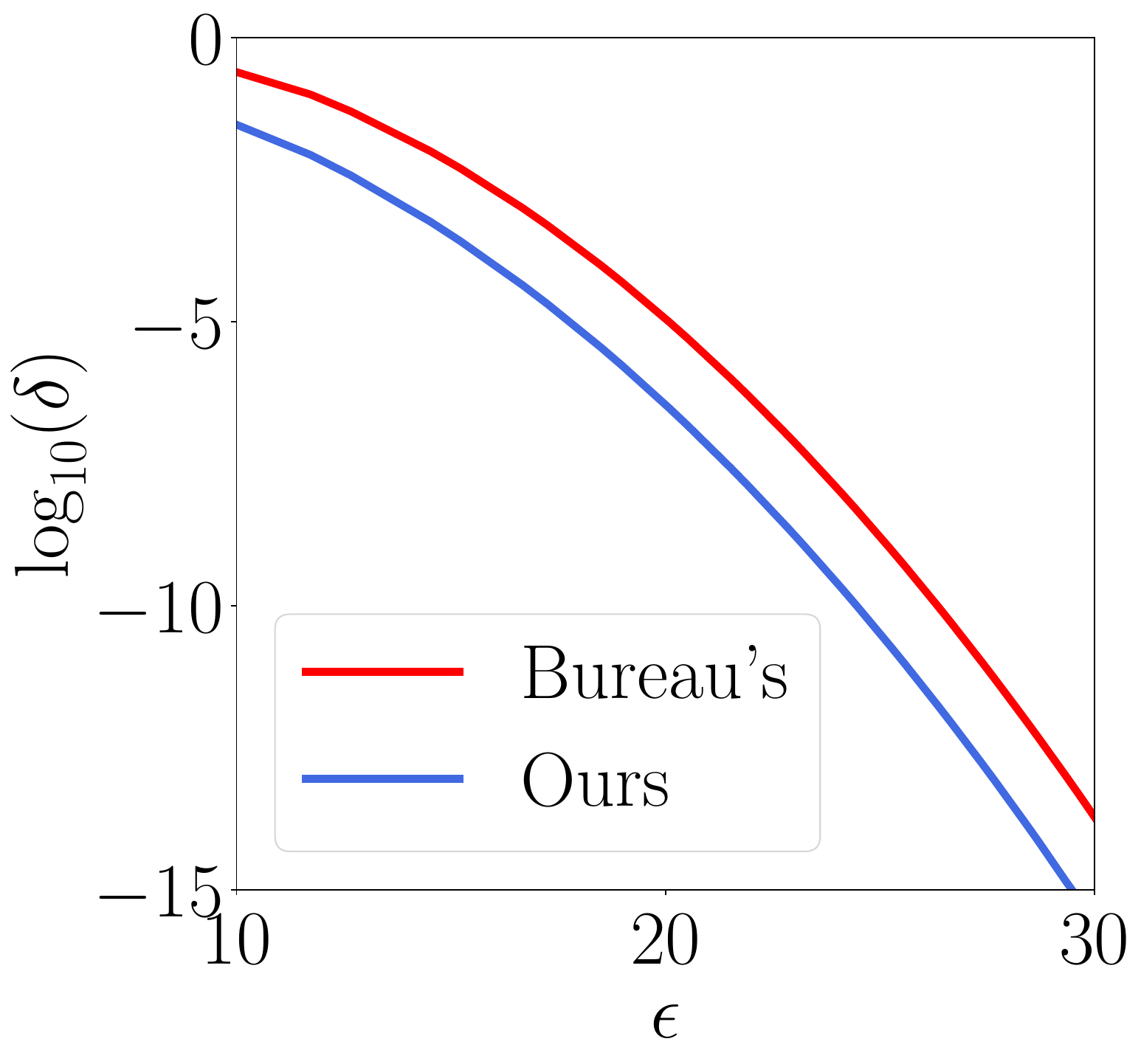} & \epsilonsm{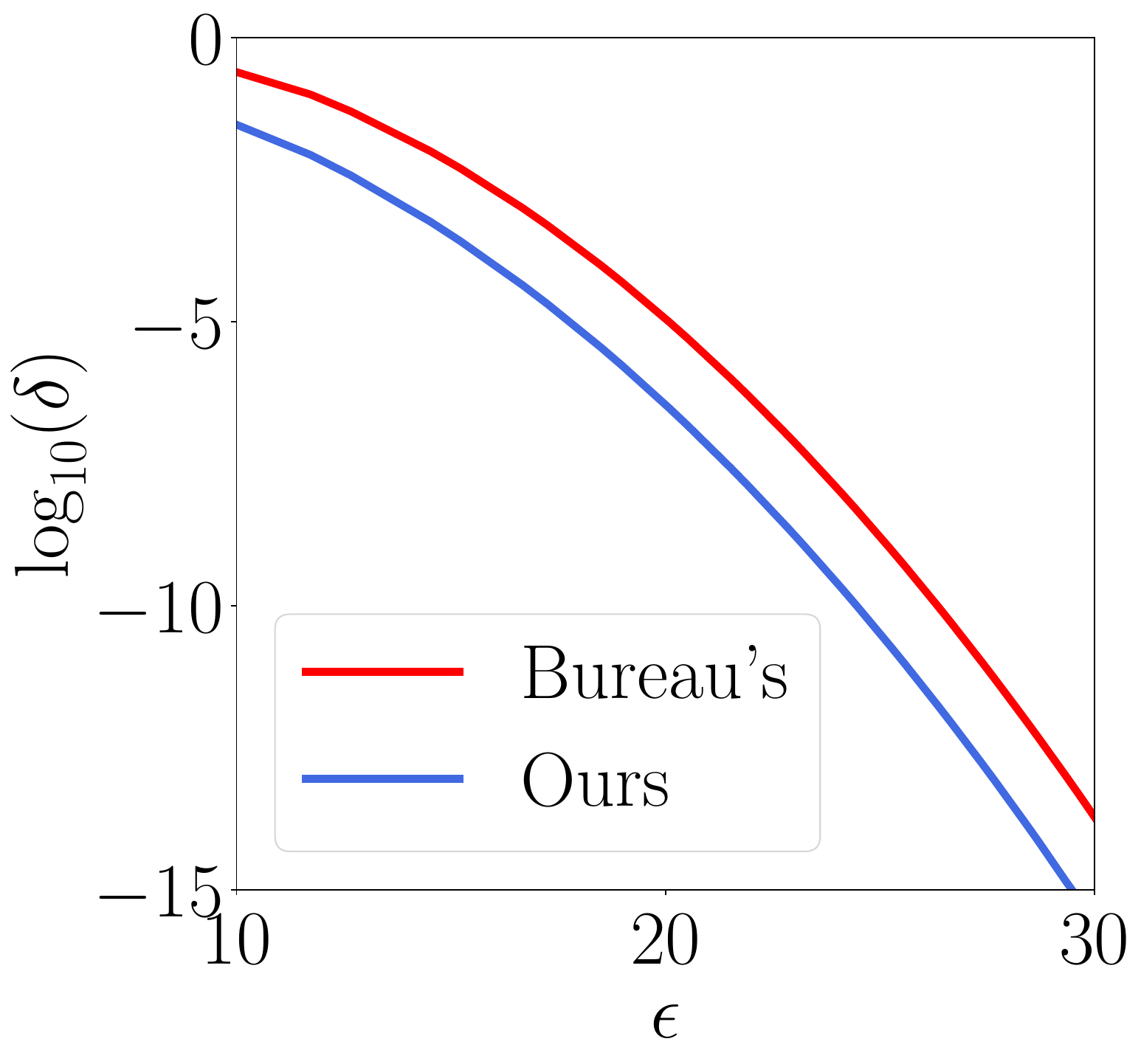} & \epsilonsm{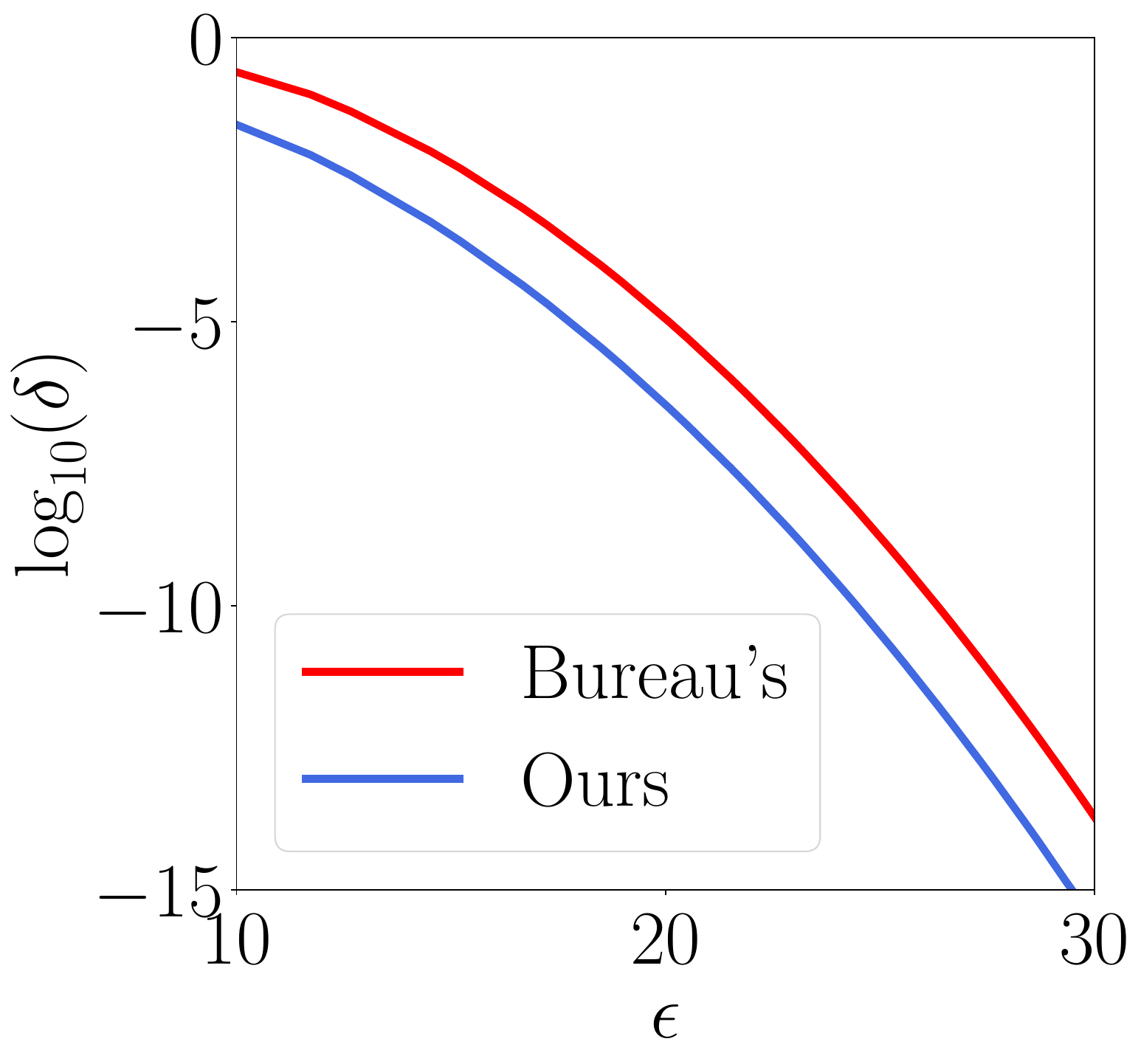} & \epsilonsm{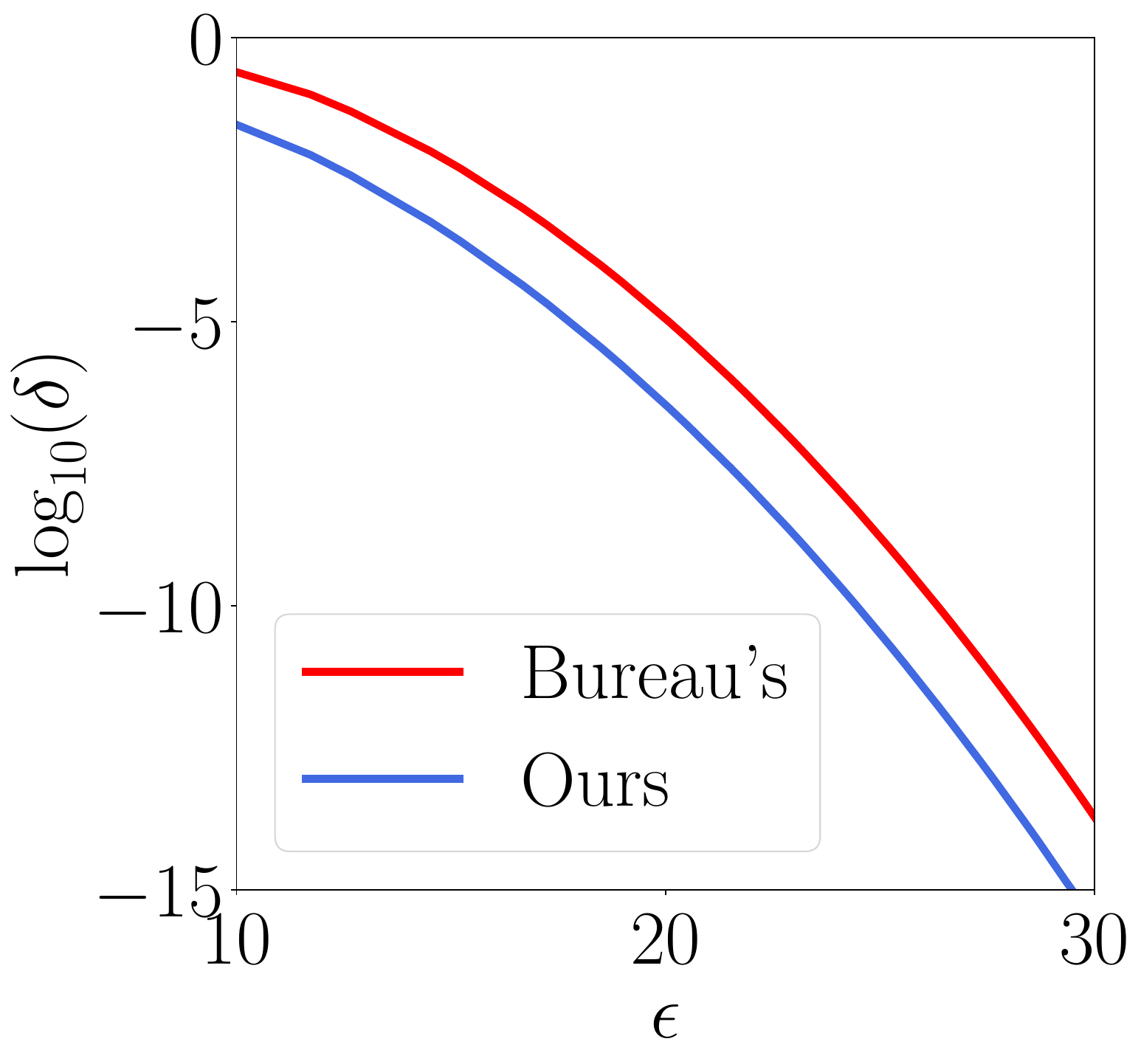} \\
        \epsilonsm{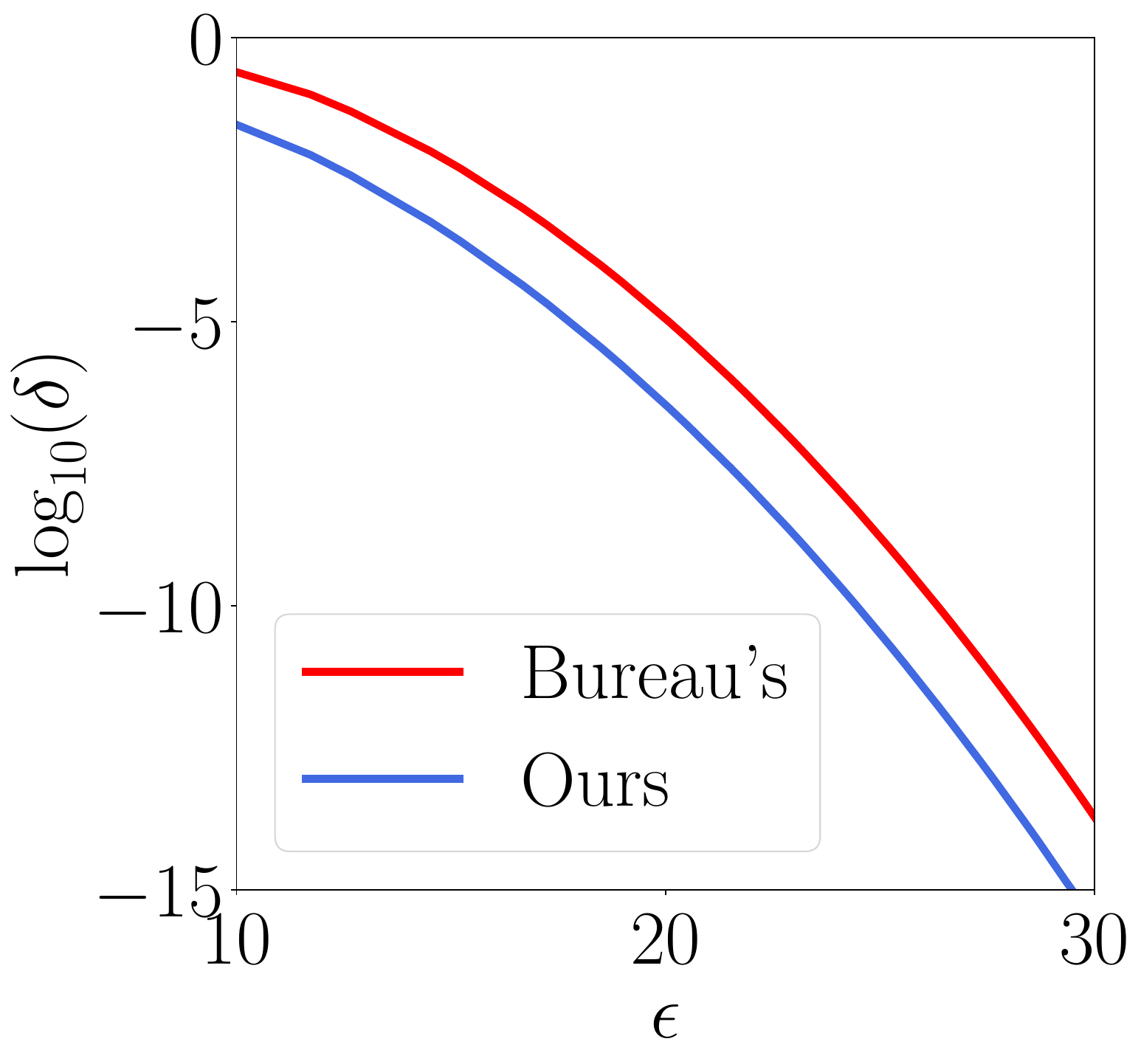} & \epsilonsm{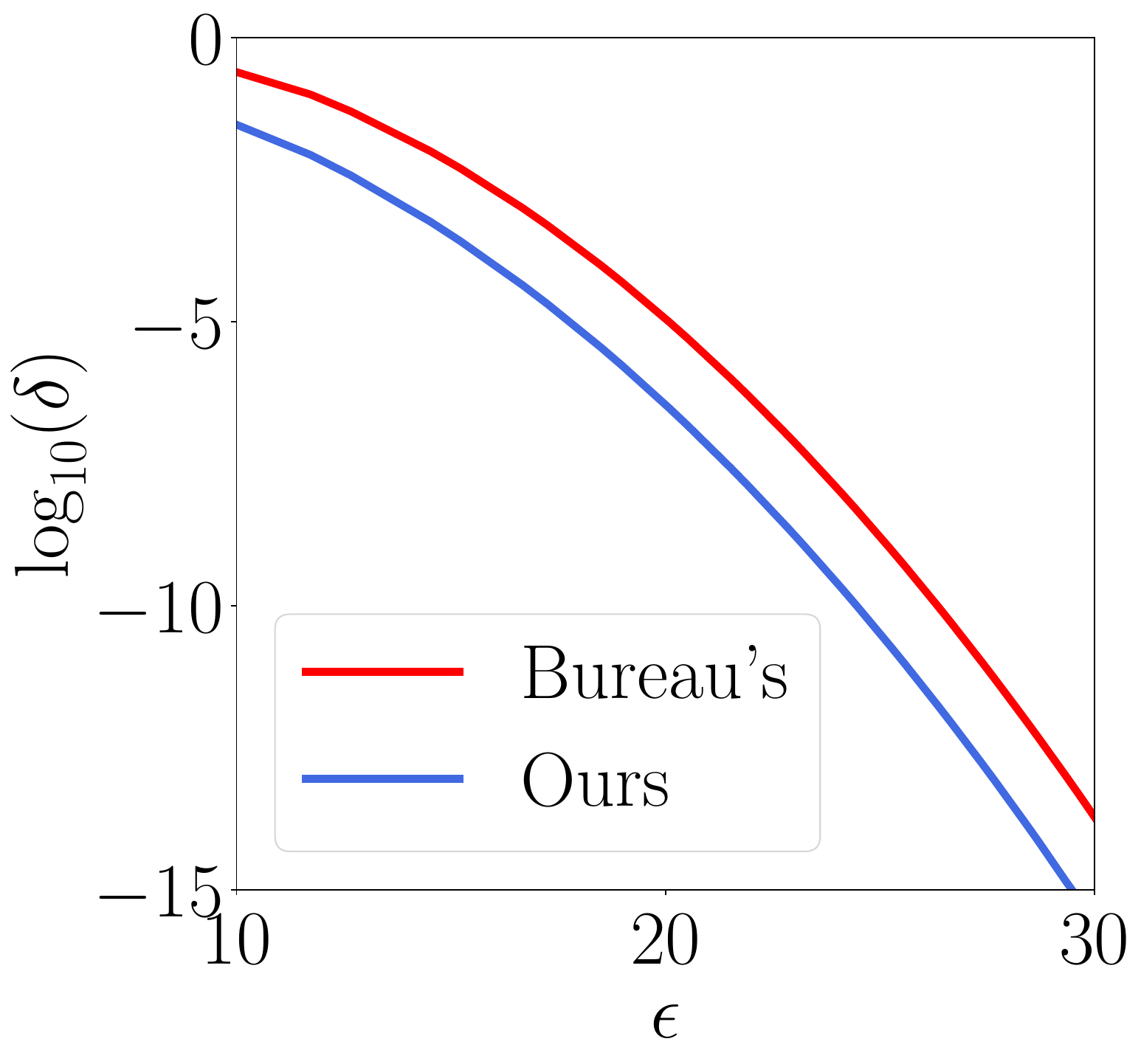} & \epsilonsm{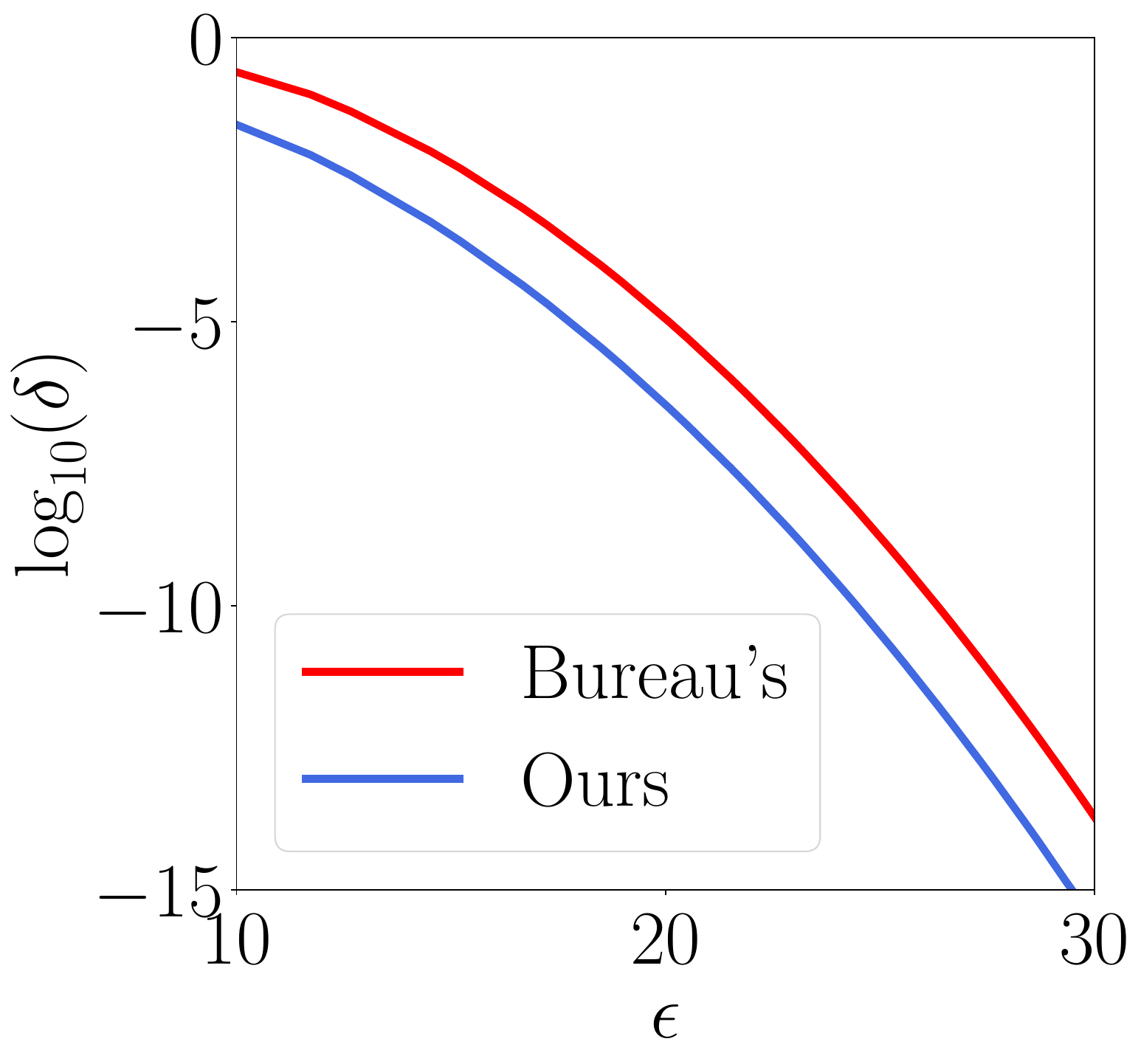} & \epsilonsm{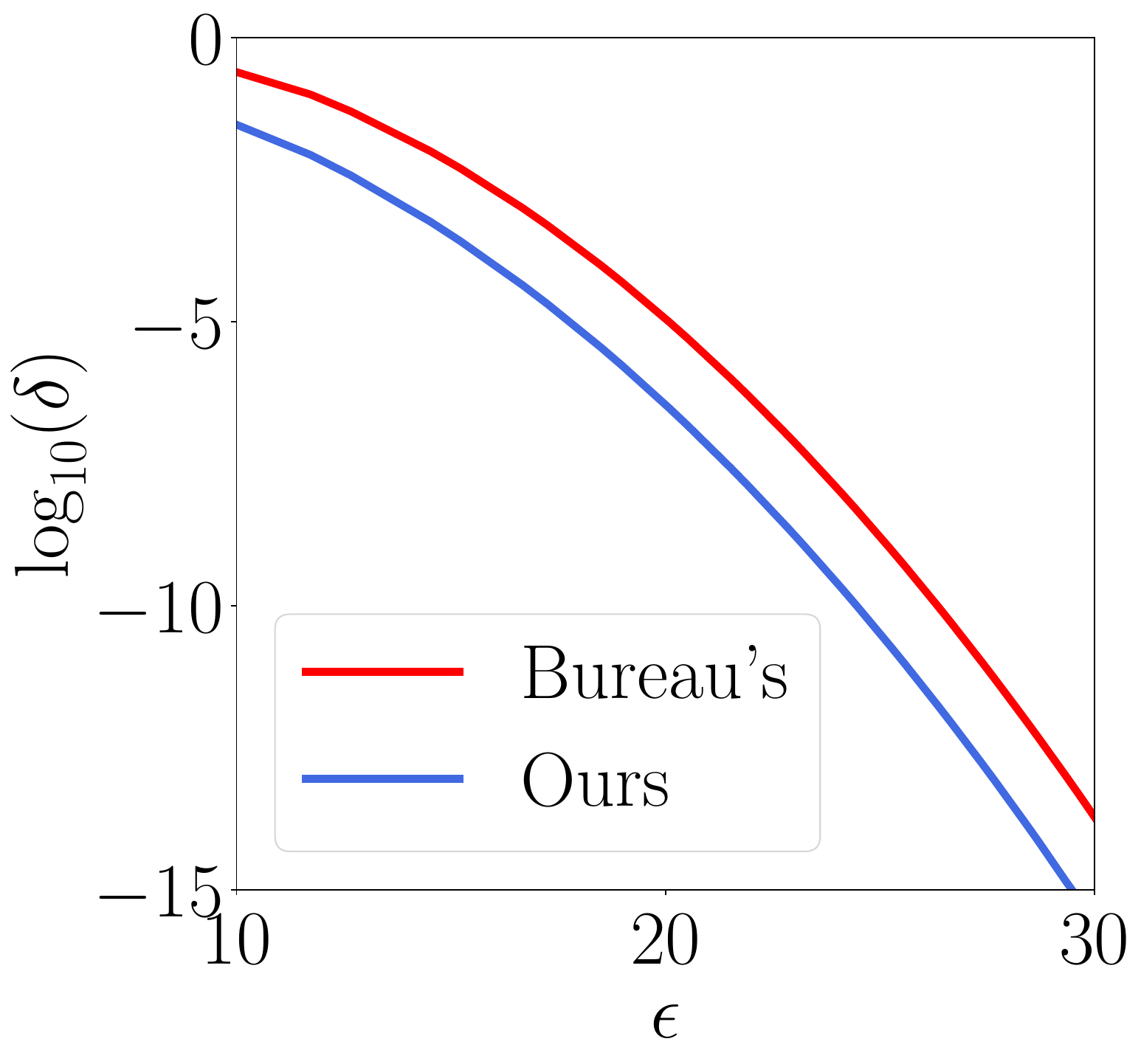} & \epsilonsm{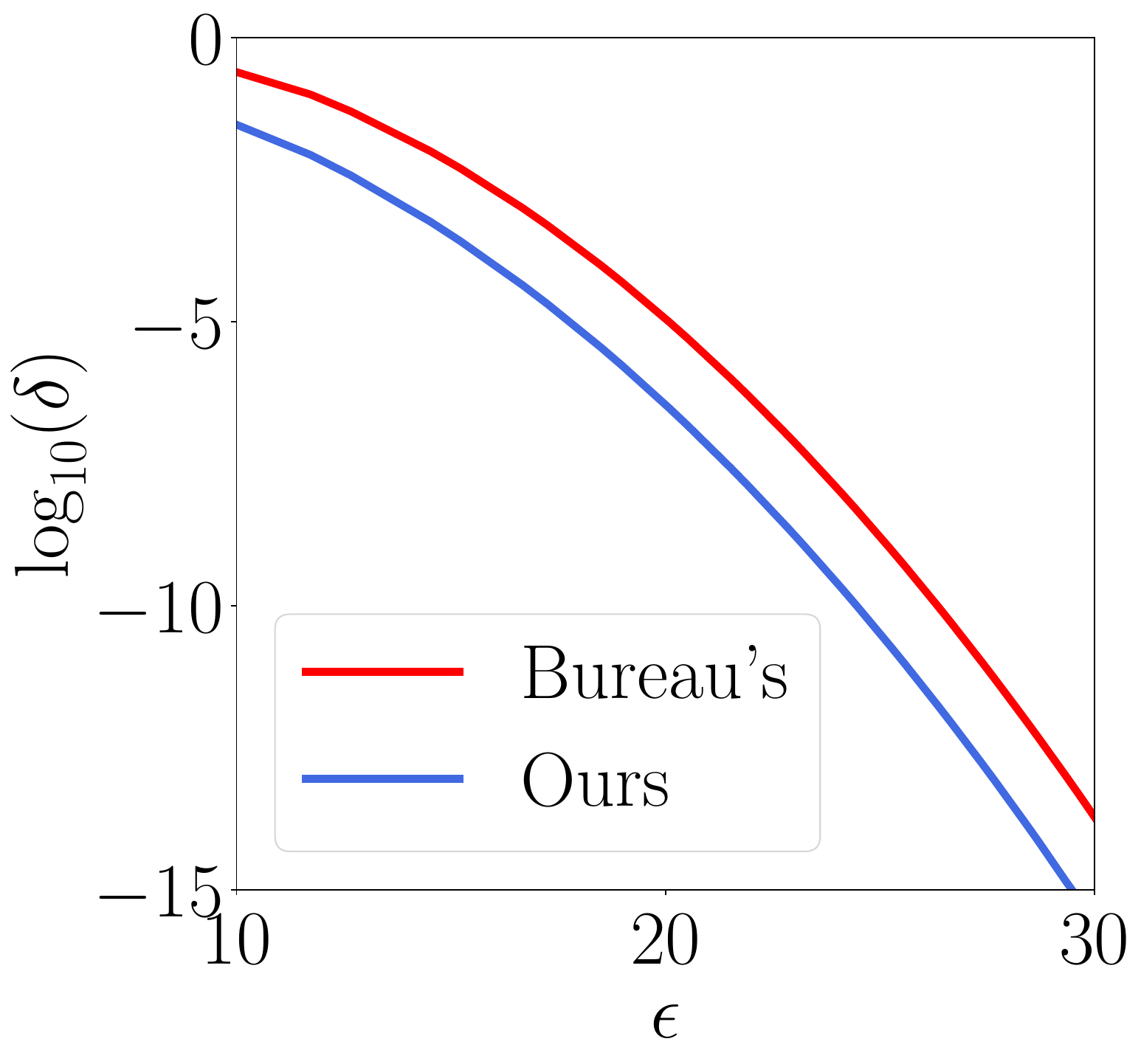} & \epsilonsm{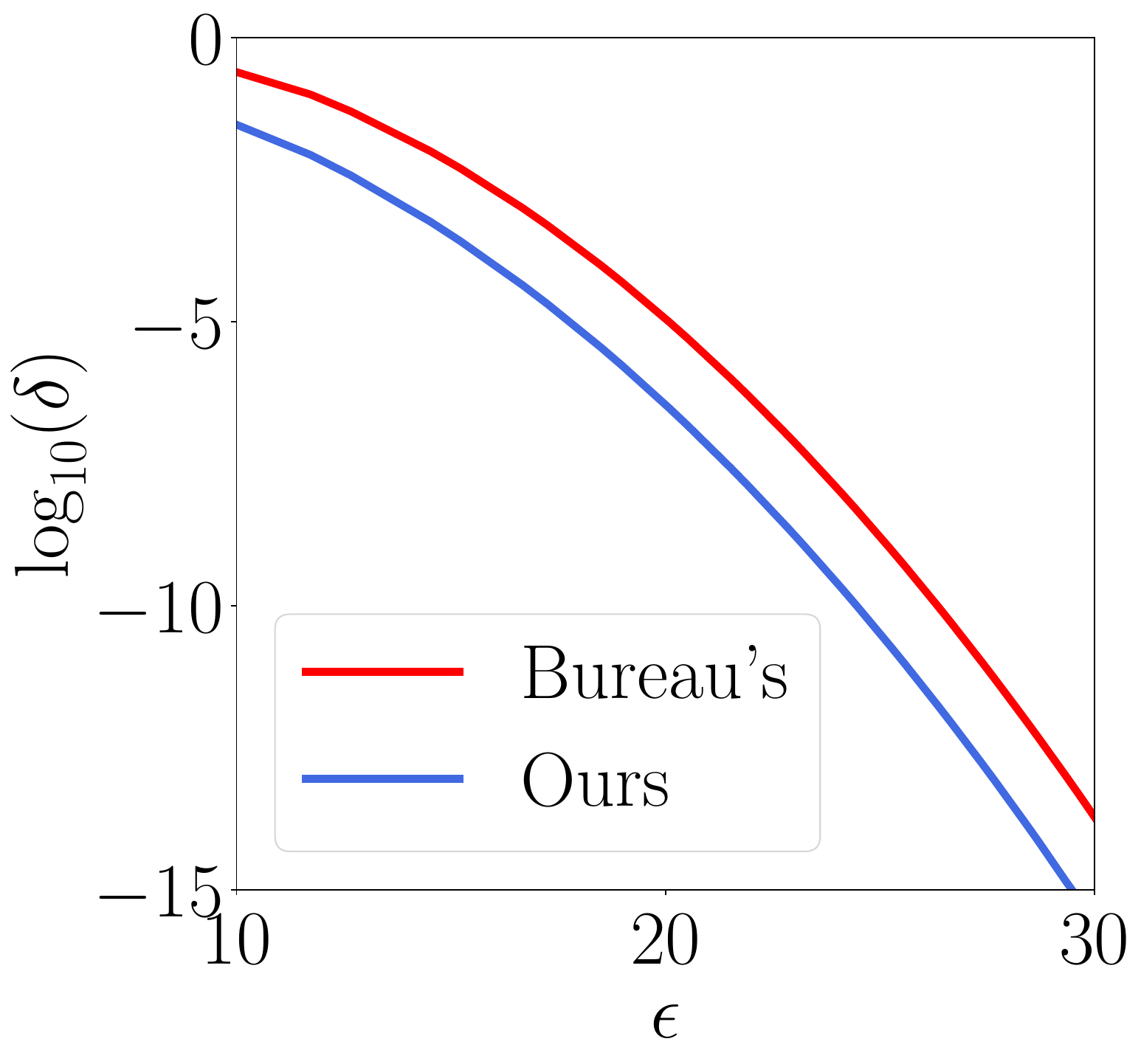} \\
        \epsilonsm{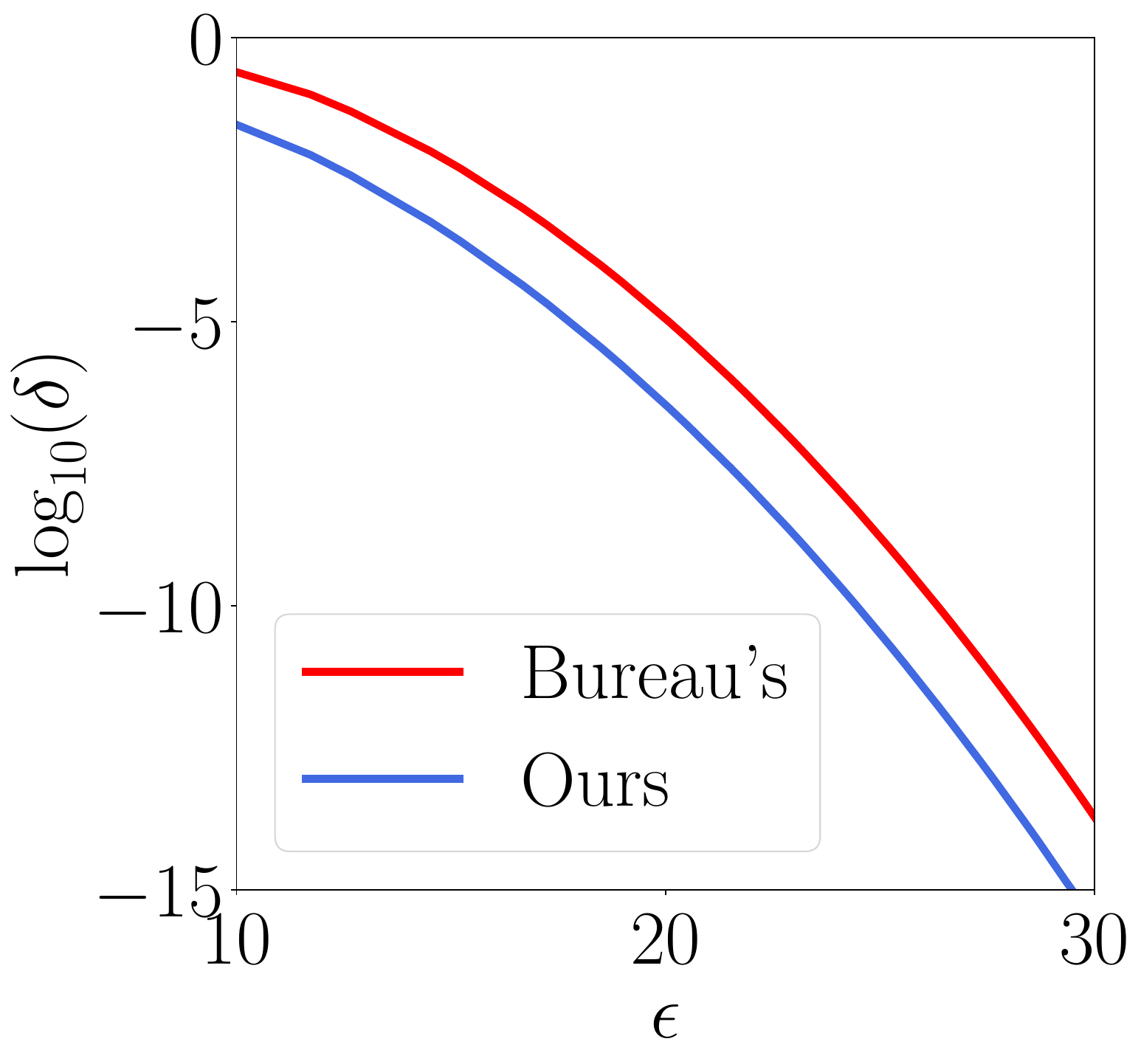} & \epsilonsm{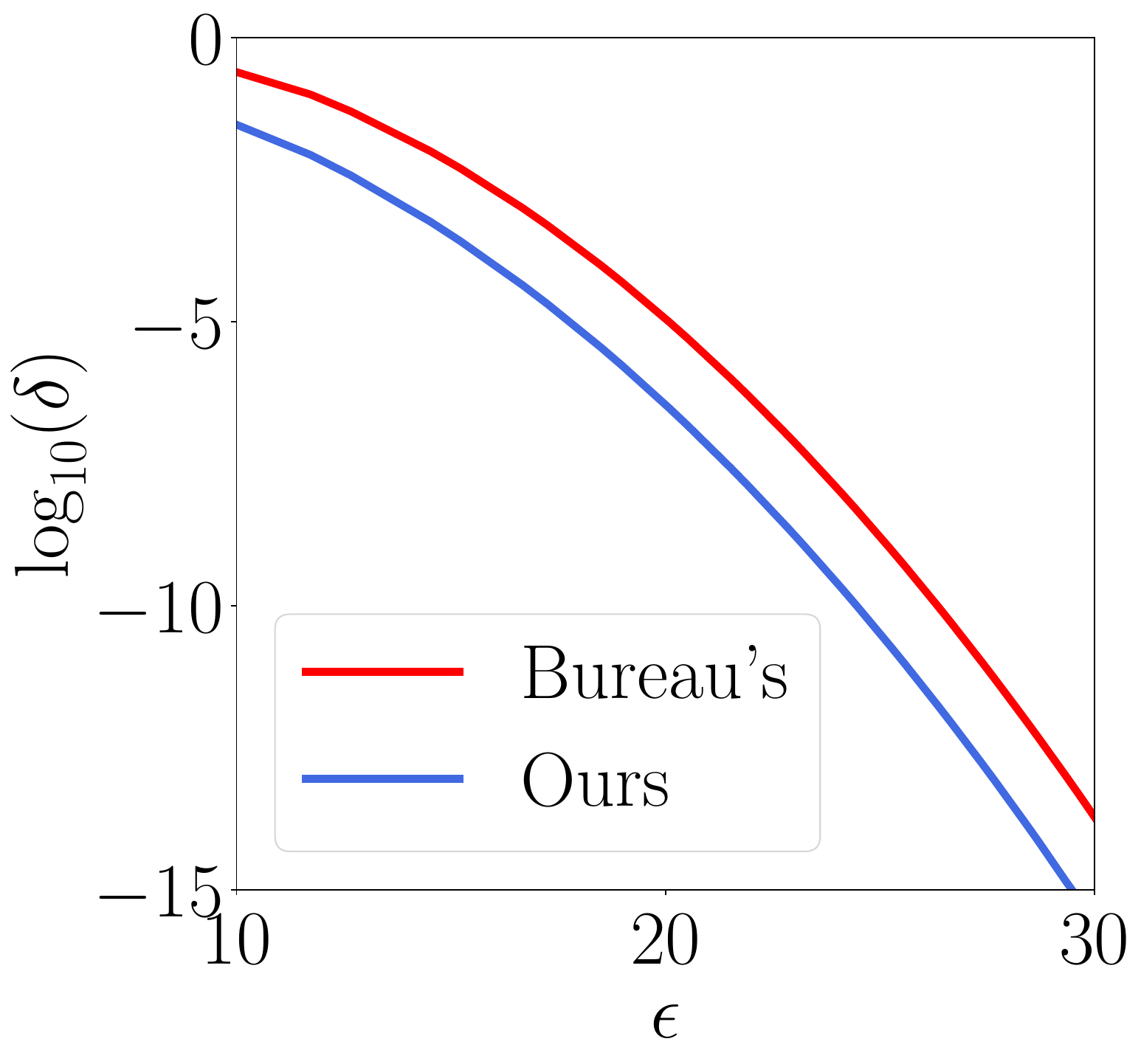} & \epsilonsm{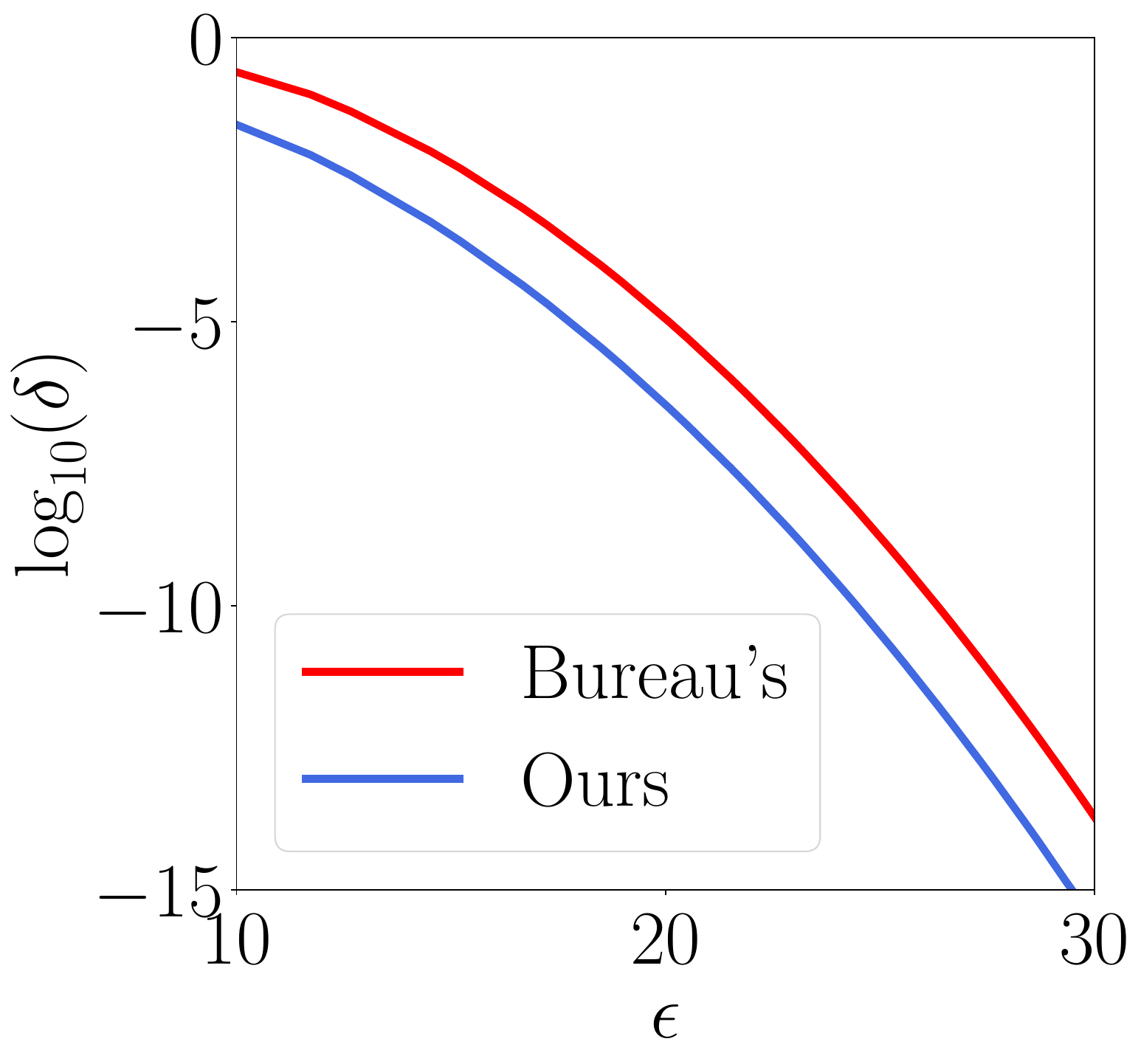} & \epsilonsm{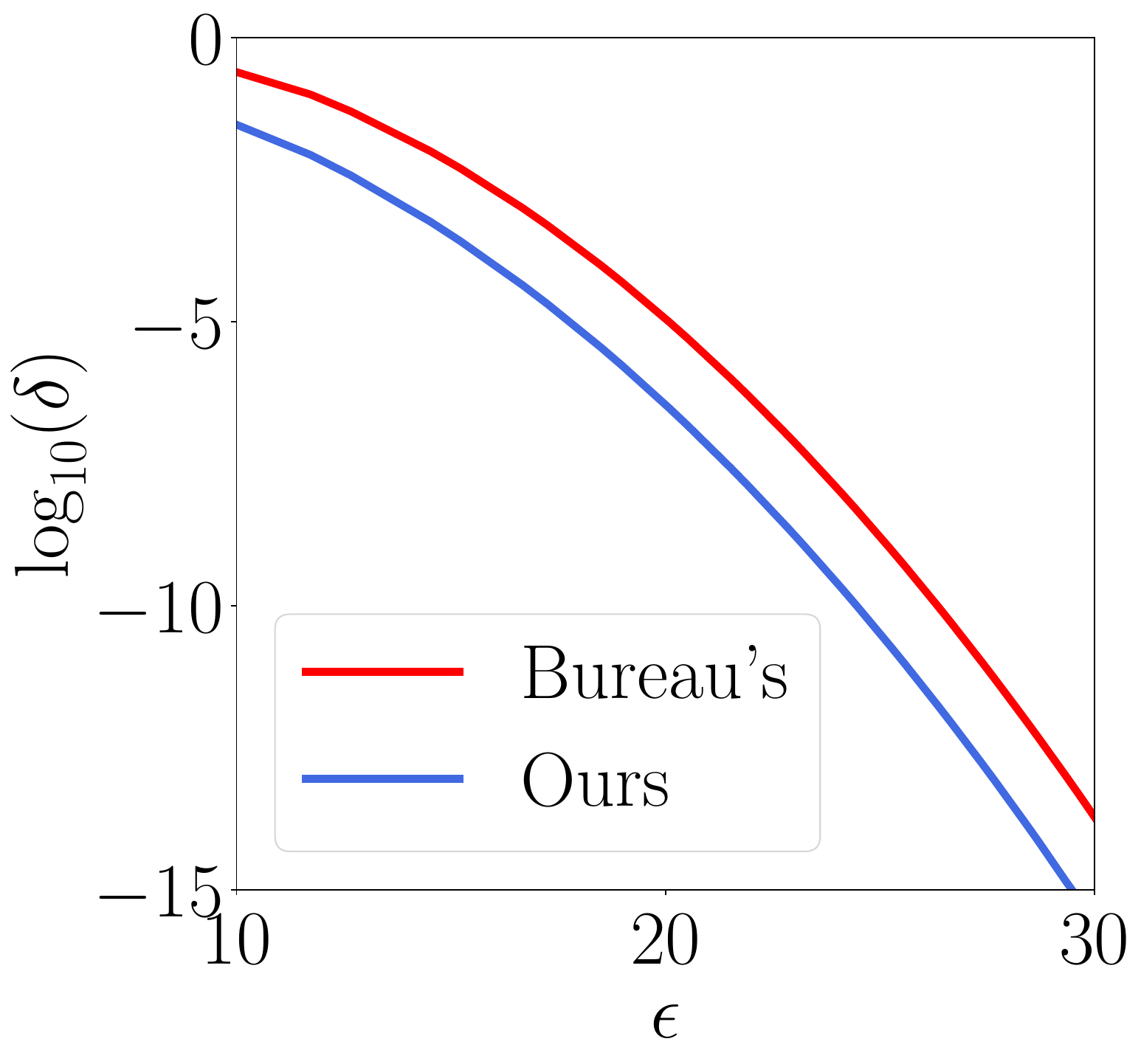} & \epsilonsm{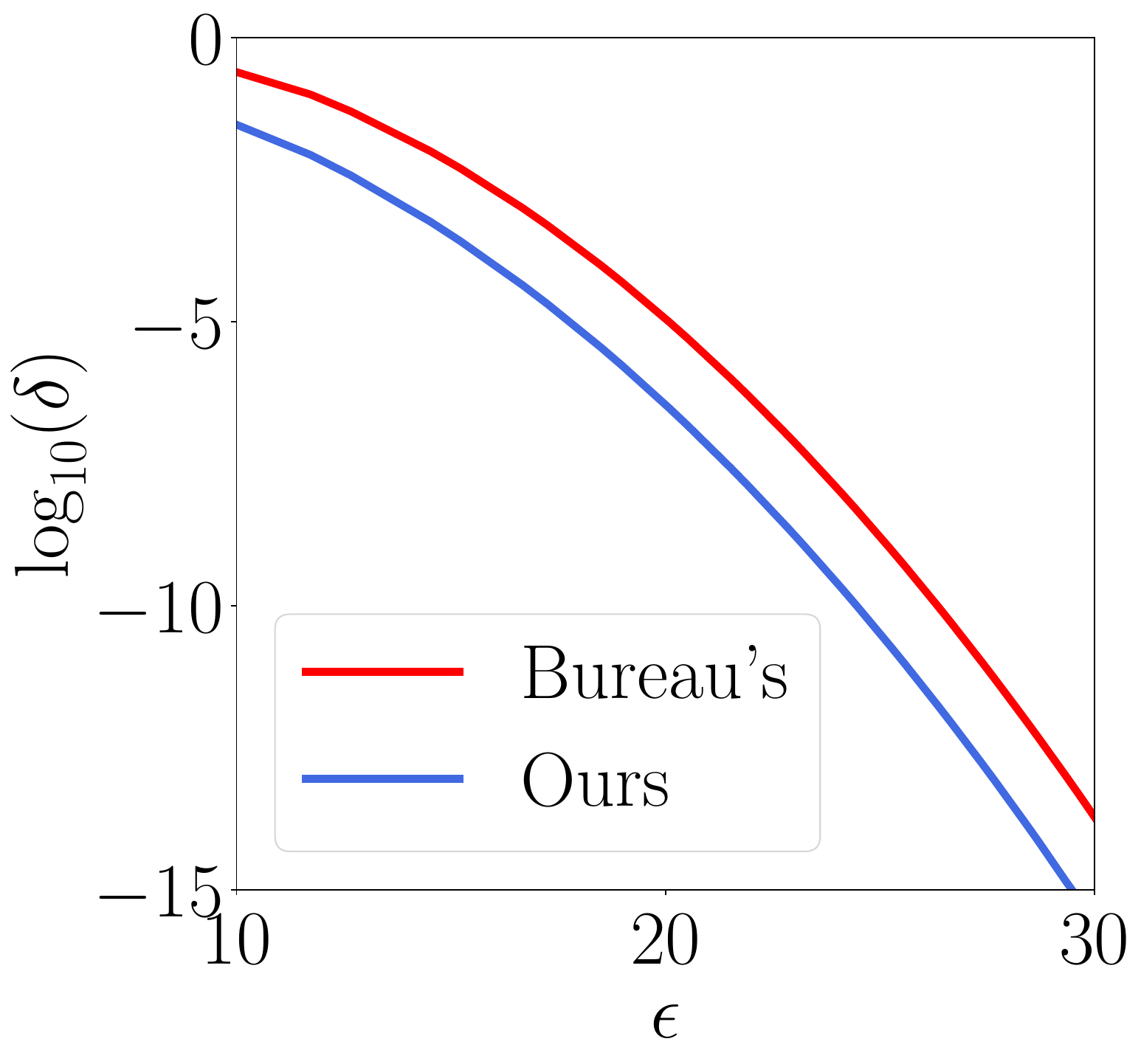} & \epsilonsm{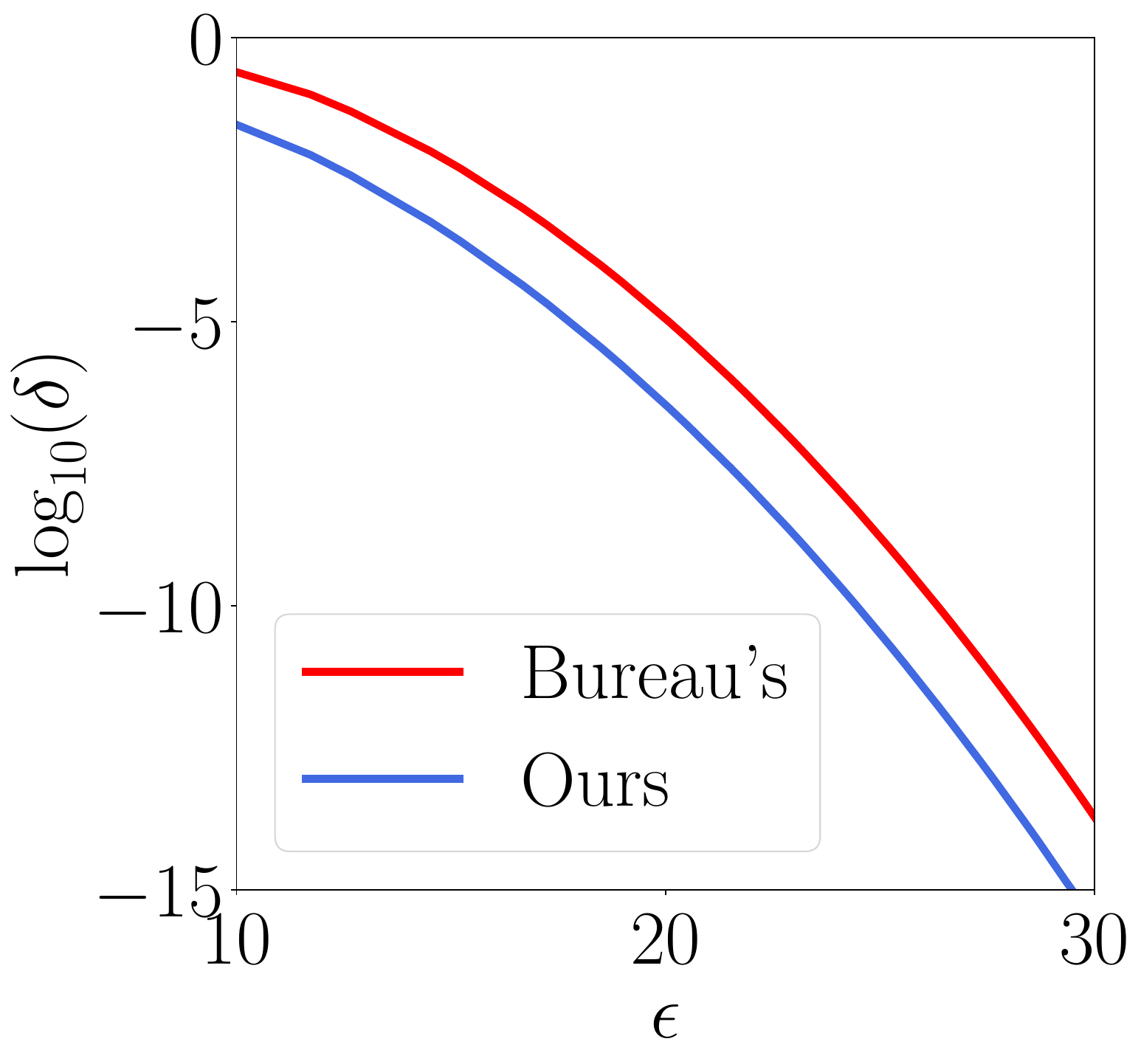} \\
        \epsilonsm{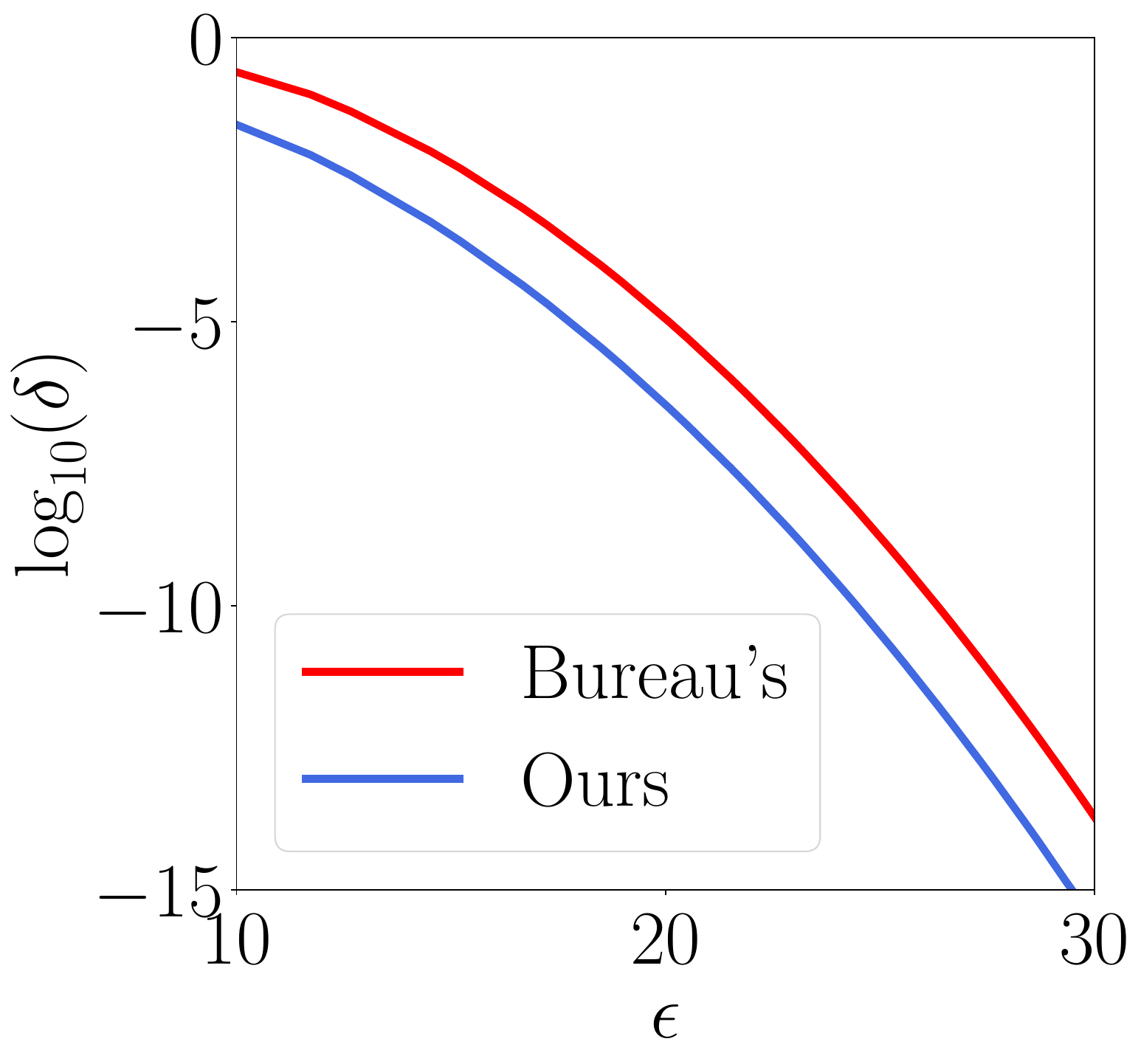} & \epsilonsm{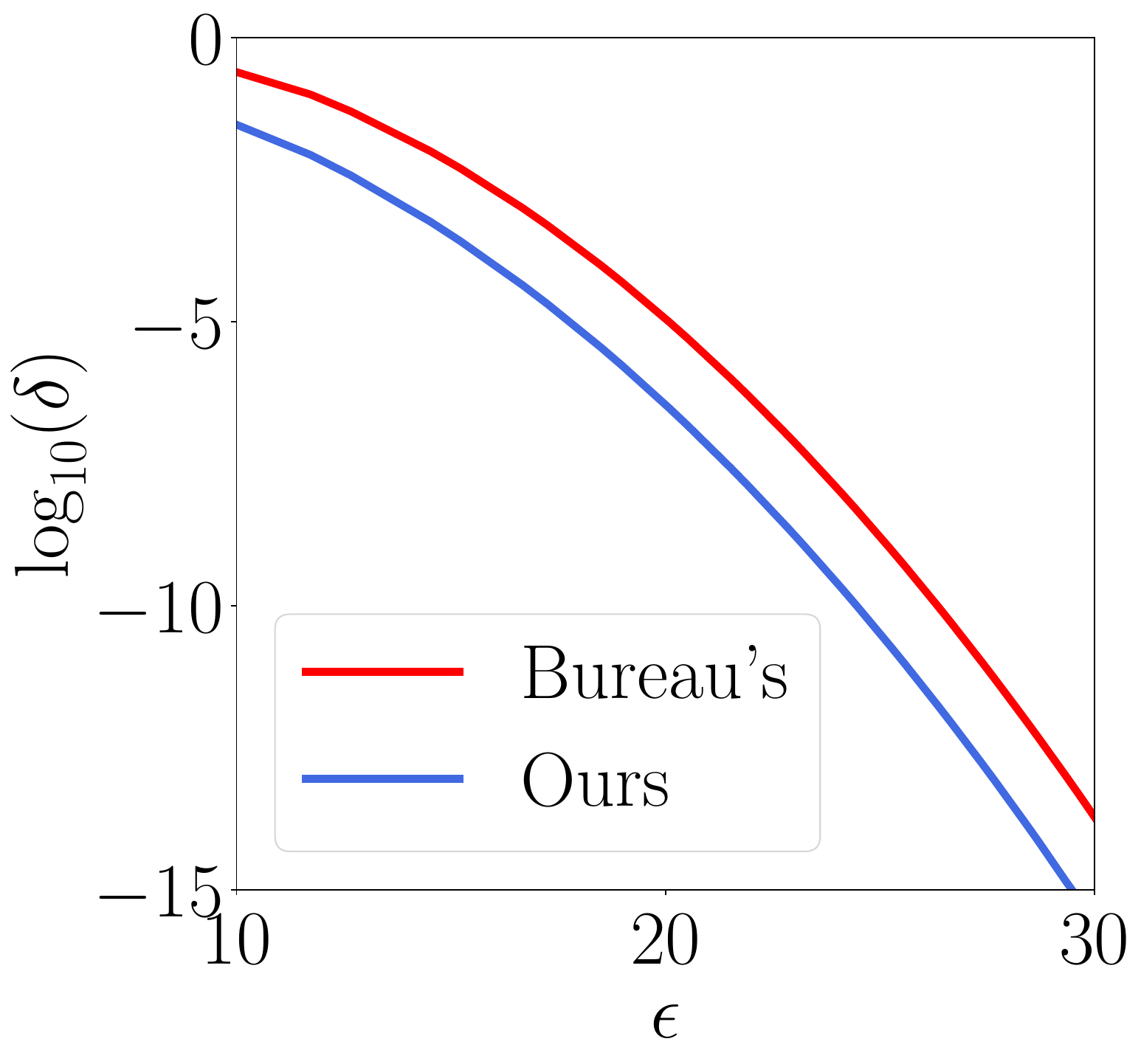} & \epsilonsm{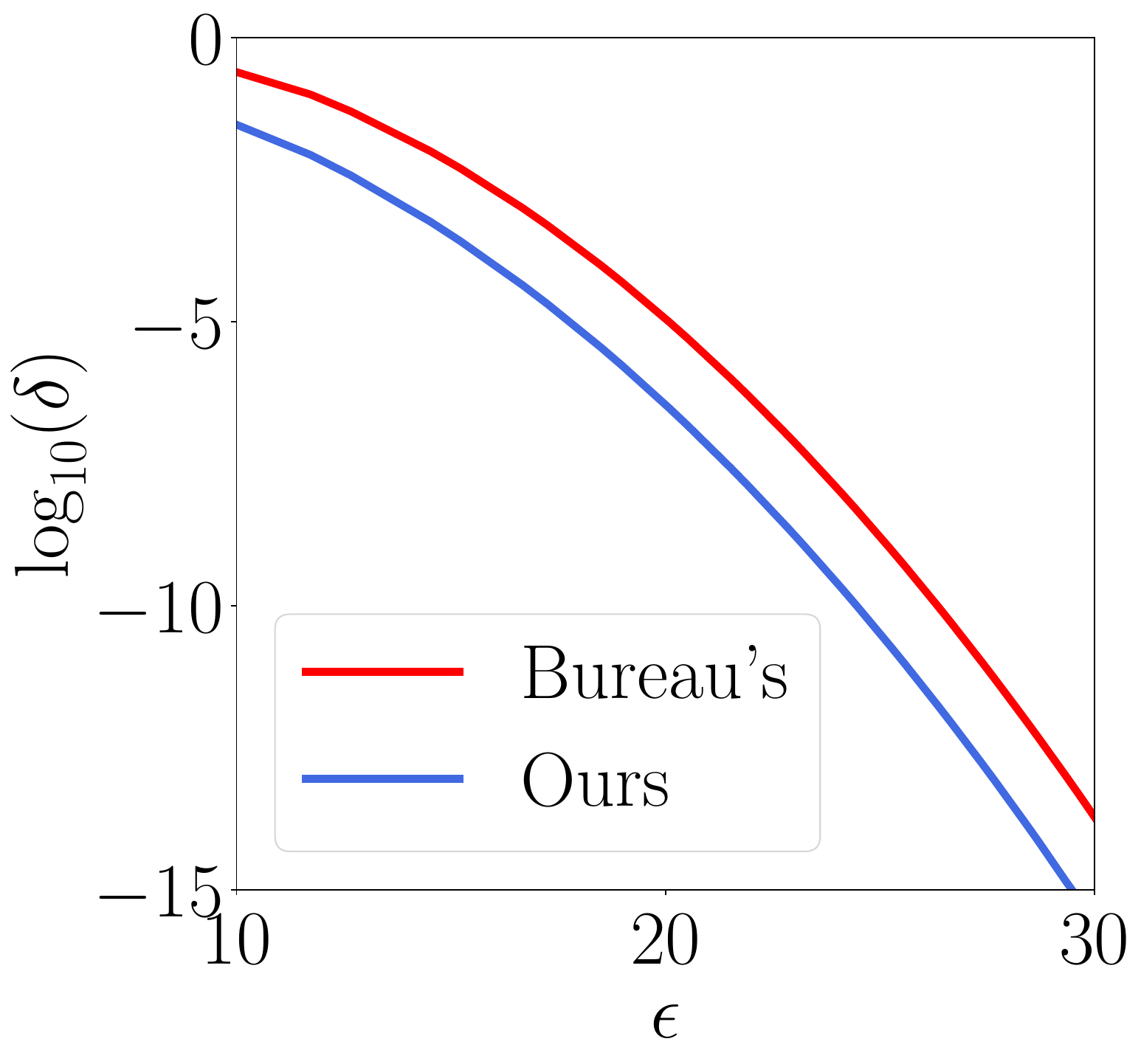} & \epsilonsm{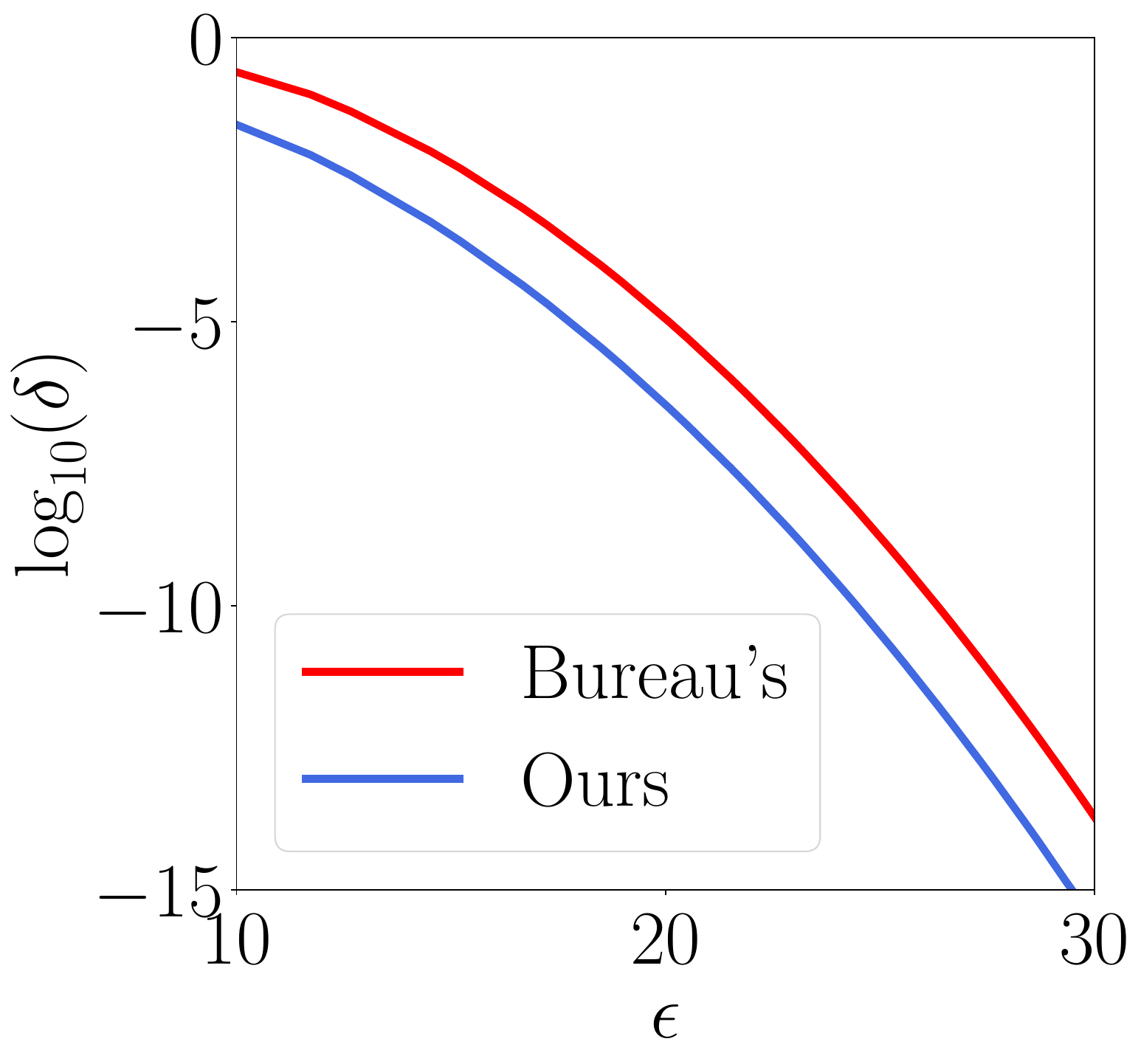} & \epsilonsm{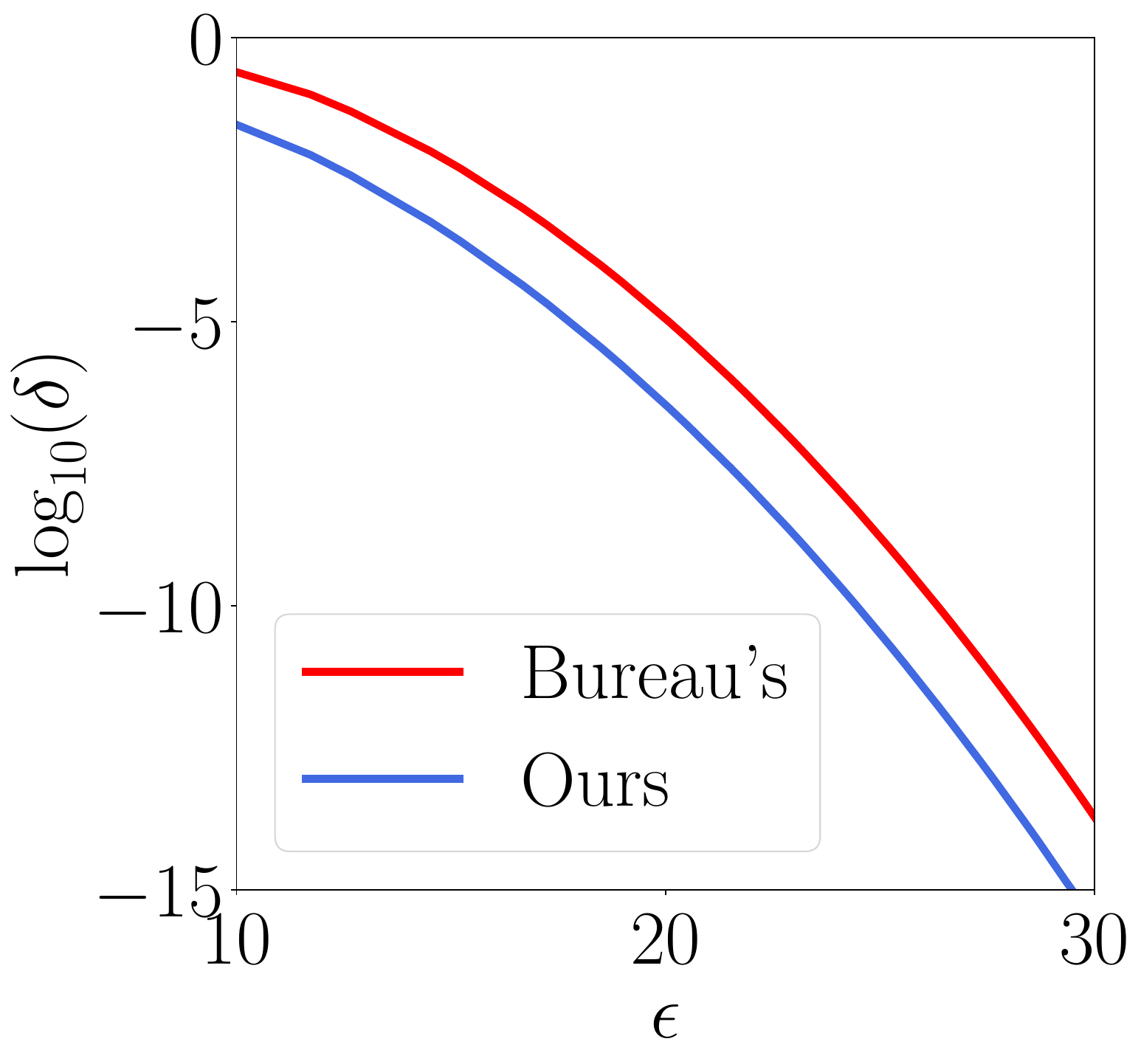} & \epsilonsm{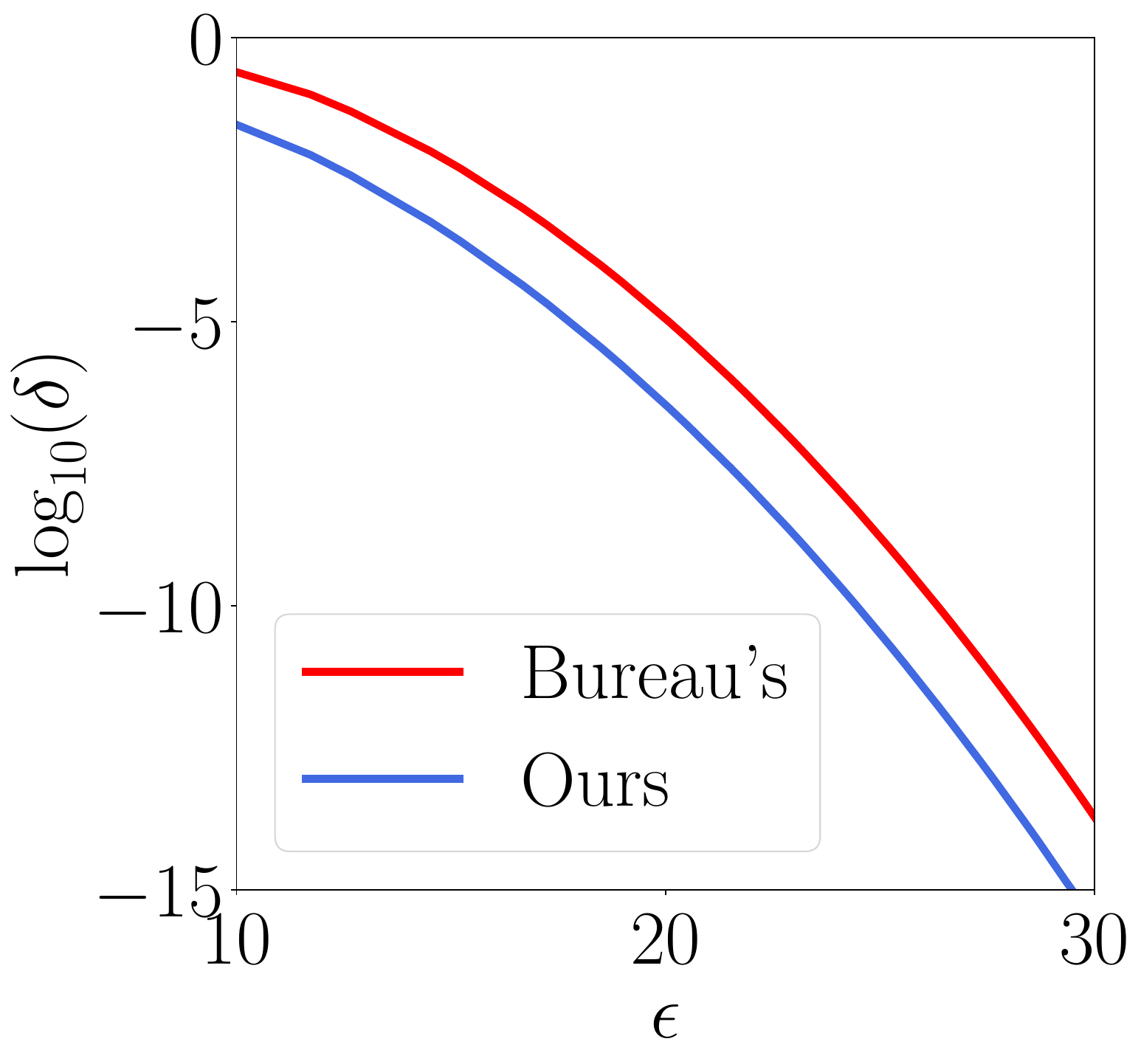} \\
        \epsilonsm{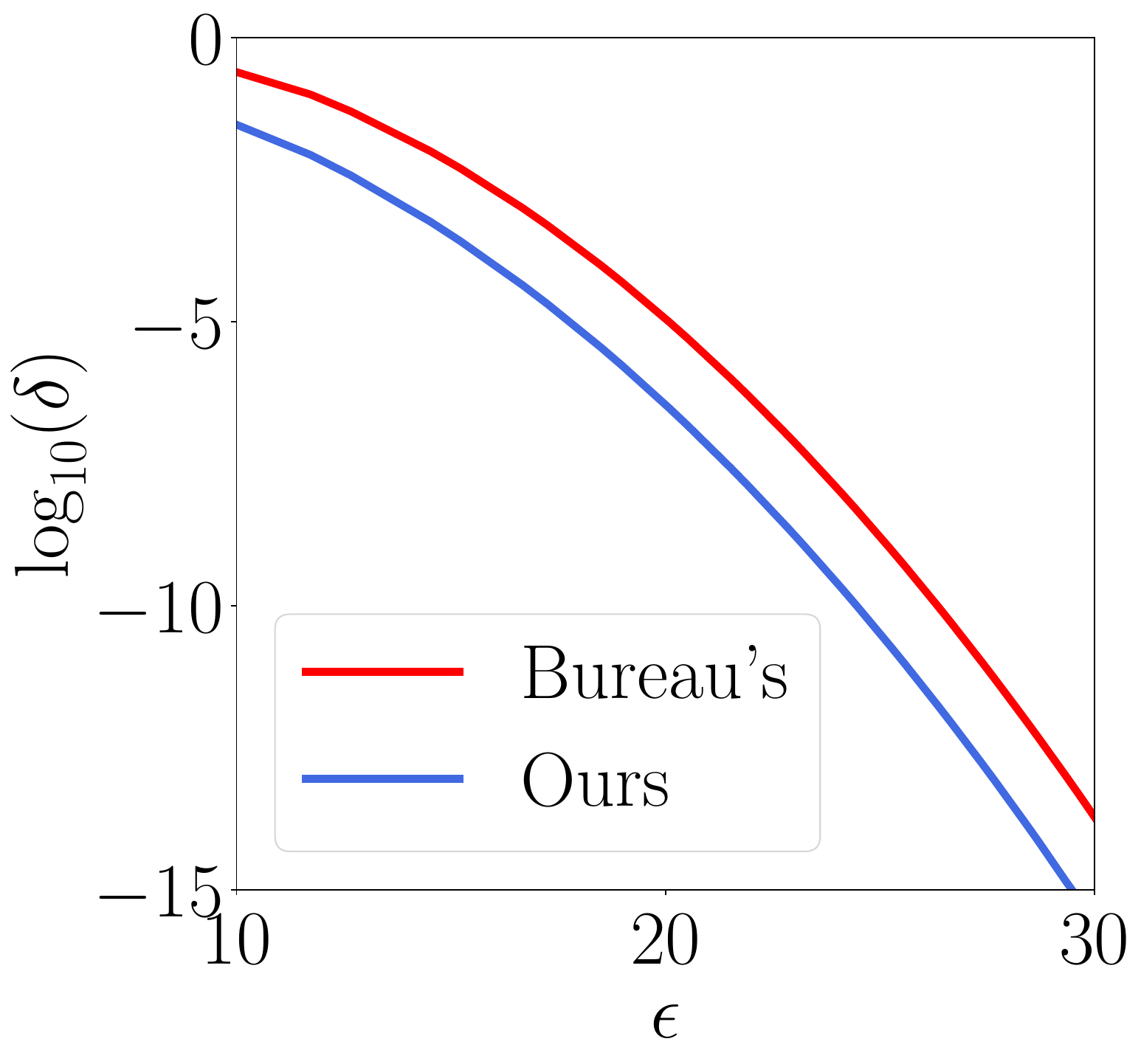} & \epsilonsm{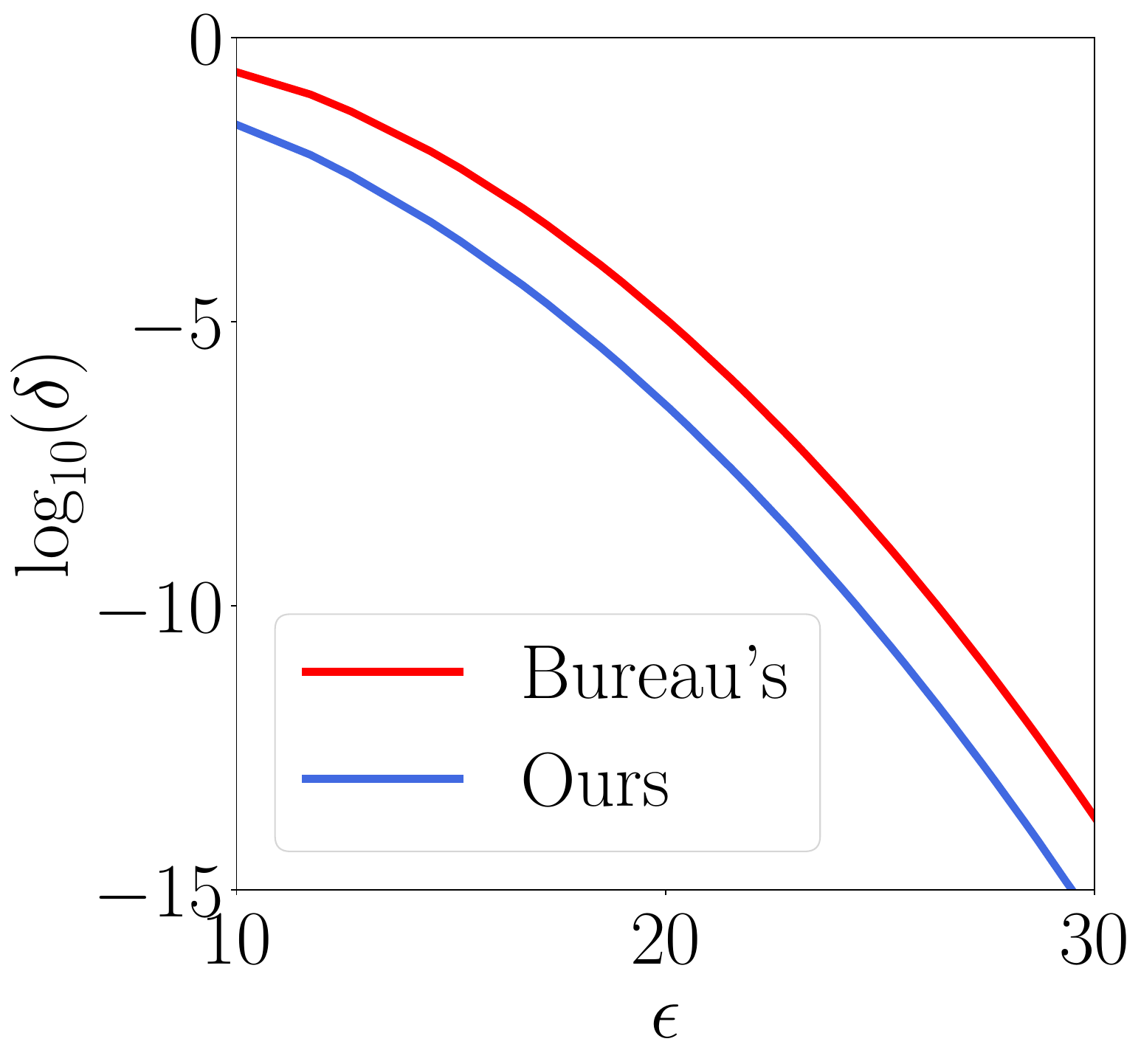} & \epsilonsm{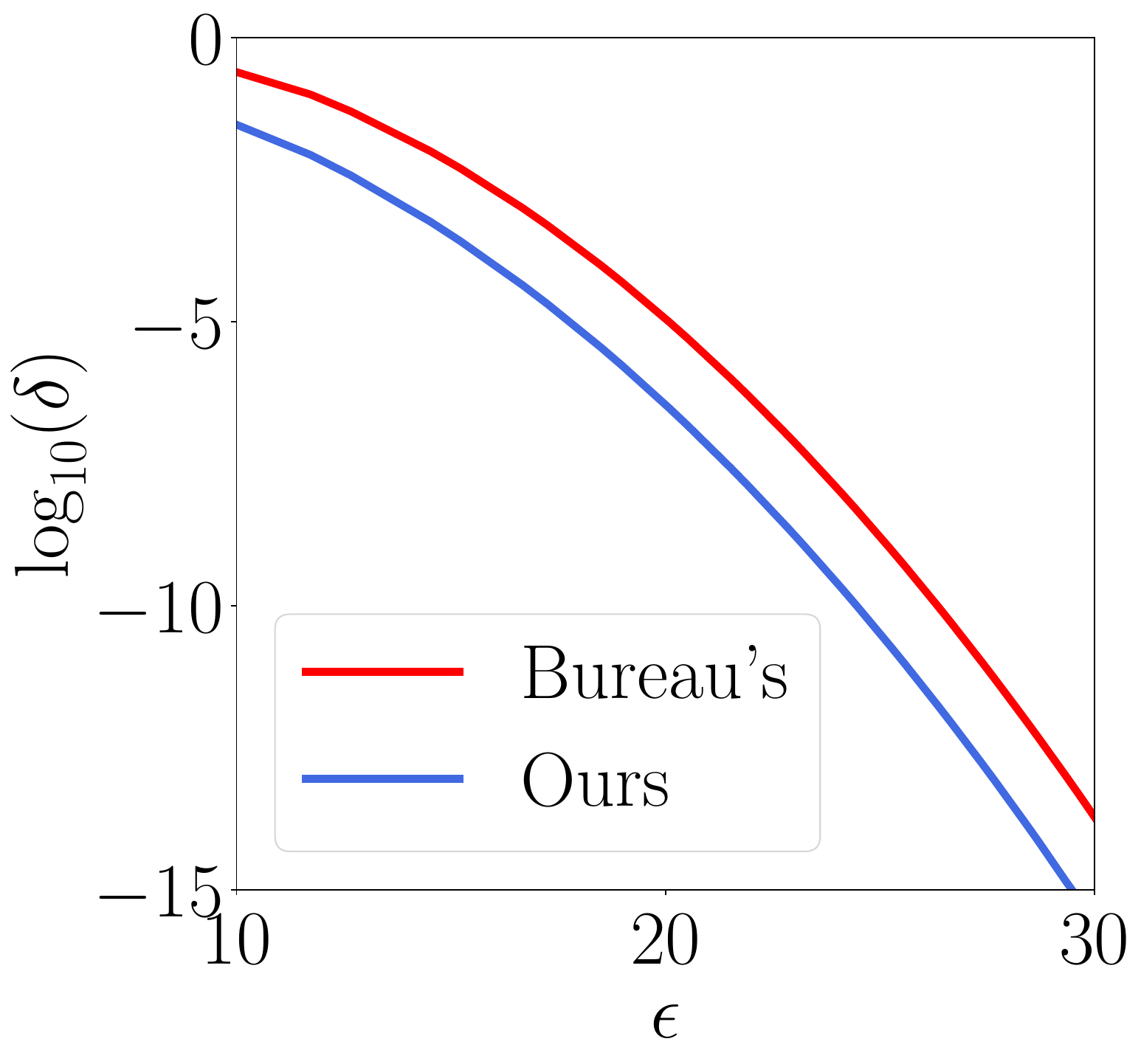} &  \multicolumn{3}{c}{\multirow[b]{2}{*}{\raisebox{0.5\cellh}{\epsilonbigfig{13}}}} \\
        \epsilonsm{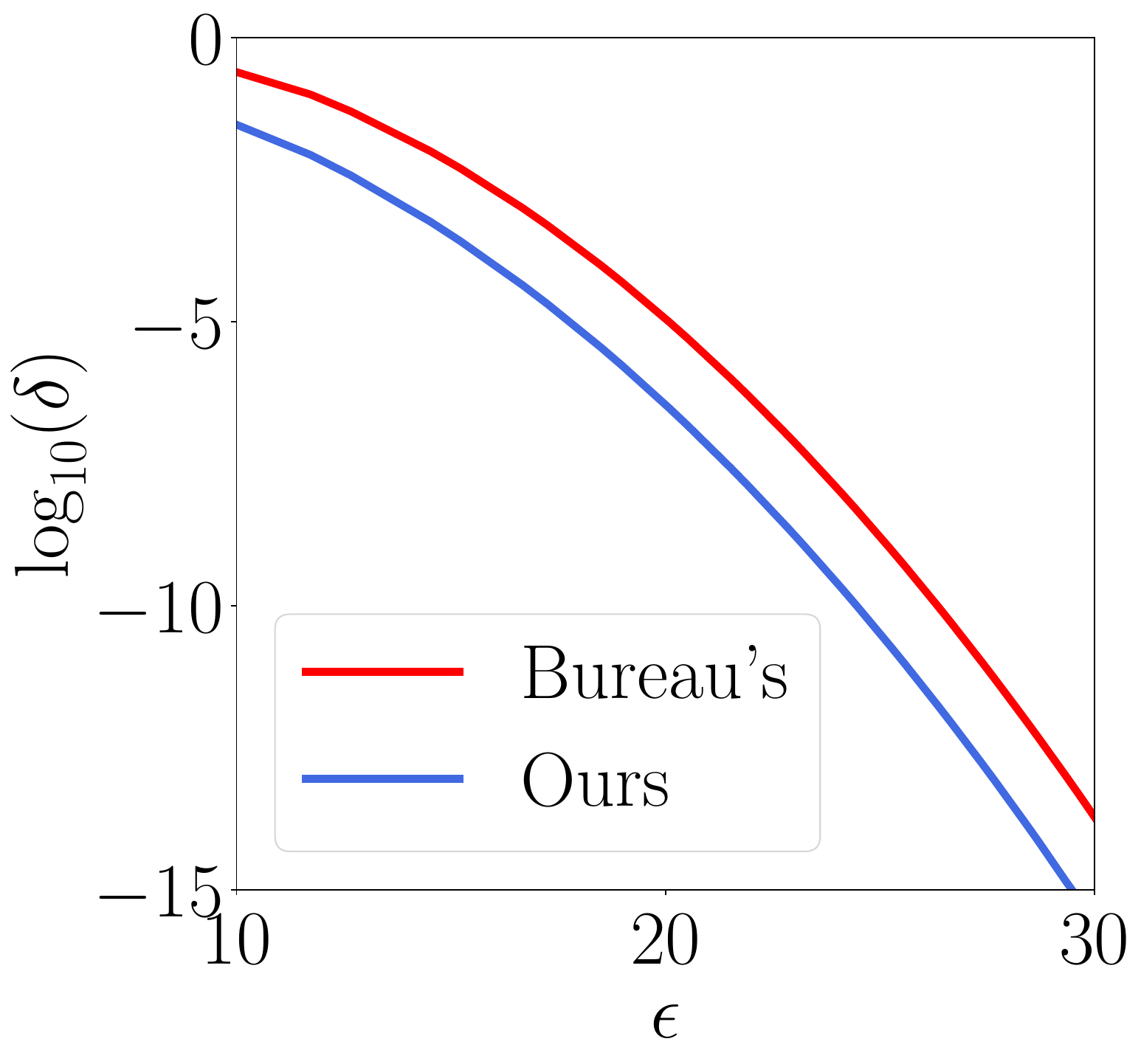} & \epsilonsm{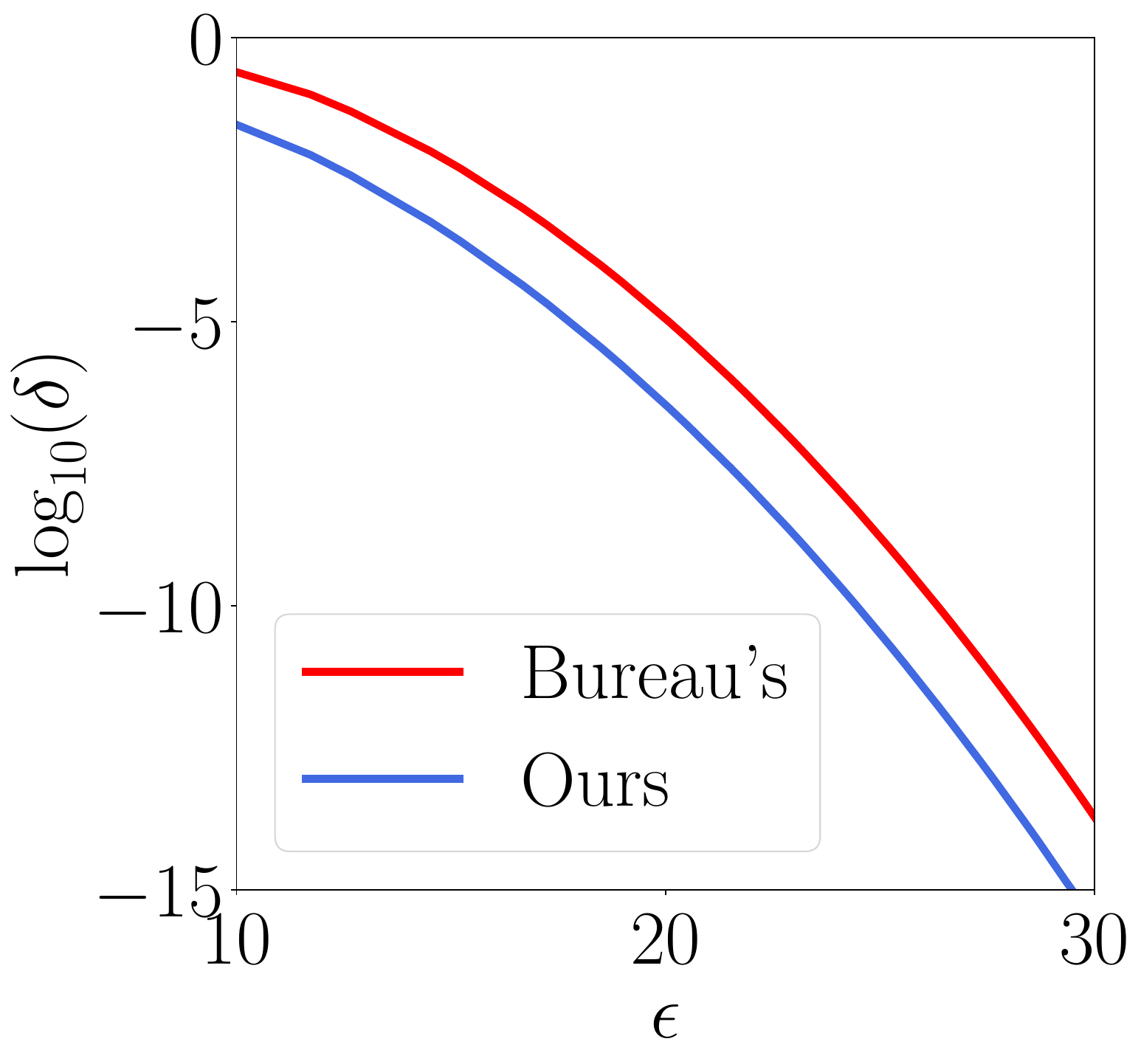} & \epsilonsm{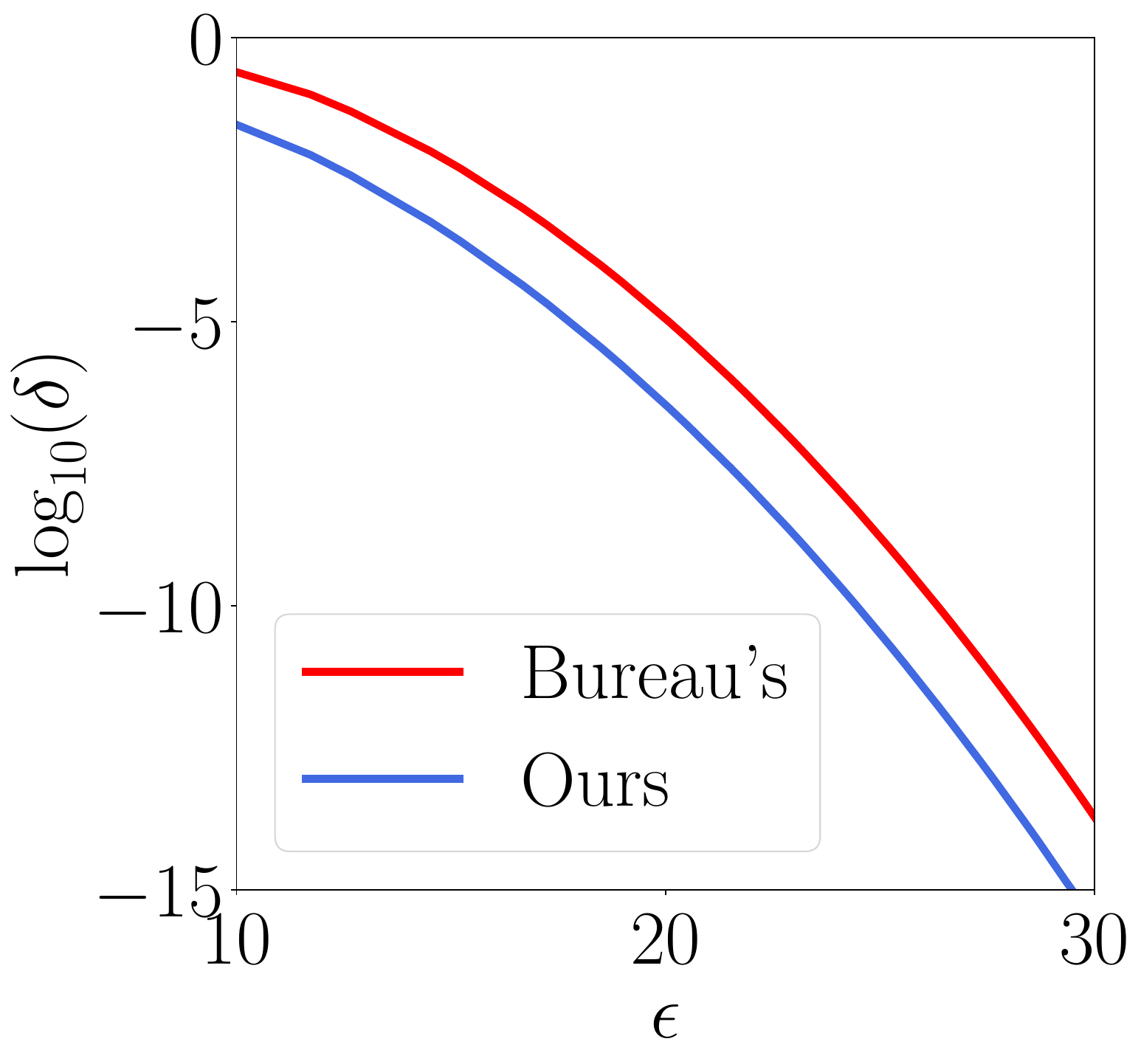} &  \multicolumn{3}{c}{}  \\
    \end{tabular}
    \vspace{-3mm}
    \caption{The $(\varepsilon, \delta)$-DP curves for the composed mechanisms $[\widetilde{\mathbf{M}}_0, \widetilde{\mathbf{M}}_l]$ for all $l$, where $\widetilde{\mathbf{M}}_0$ corresponds to the privacy budget allocation in Table \ref{tab:PLB_typical}. Although the curves appear highly similar, they differ by amounts significantly larger than the numerical error introduced by floating-point arithmetic or the \texttt{mpmath} package, as detailed in Figure \ref{fig:sensitivity-privacy}. The enlarged panel in the lower-right corner displays the privacy guarantee of the composed mechanism $[\widetilde{\mathbf{M}}_0, \widetilde{\mathbf{M}}_0]$.}
    \label{fig:eps_delta_paths_bypass}
\end{figure}

\subsection{Practical advice on privacy budget tuning}
\label{sec:speedup}

In practice, the U.S. Census Bureau considers multiple privacy budget allocations prior to deployment and evaluates their performance on downstream tasks \citep{censusDemonstrationData, censusjsm, censusAnnouncing2030, censusDecennialCensus}. However, determining the privacy budget based on the overall privacy guarantee can be computationally expensive, as it requires large-scale parallel computation. We therefore provide practical guidance to avoid repeatedly recomputing the full privacy accounting.
As shown in Figures \ref{fig:trade_off_paths_bypass} and \ref{fig:eps_delta_paths_bypass}, the privacy guarantees of different composed mechanisms $[\widetilde{\mathbf{M}}_k, \widetilde{\mathbf{M}}_l]$ are not identical but are highly similar. Motivated by this observation, we recommend using the composed mechanism $[\widetilde{\mathbf{M}}_0, \widetilde{\mathbf{M}}_0]$, with the privacy budget allocation given in Table \ref{tab:PLB_typical}, as a proxy during the privacy budget tuning stage.\footnote{More precisely, we recommend using a privacy budget allocation in which no geographic level is bypassed.} After the tuning process is completed, the exact overall privacy guarantee can then be computed to ensure full mathematical rigor.

To make this recommendation more precise, we first quantify how the privacy guarantee of the composed mechanism $[\widetilde{\mathbf{M}}_0, \widetilde{\mathbf{M}}_0]$ differs from the overall privacy guarantee. We measure the relative difference between the two as
\begin{align*}
    \frac{f_{00} - (\min_{k,l} f_{kl})^{**}}{f_{00}}, \quad \text{and} \quad \frac{\varepsilon_{00} - \max_{k,l} \varepsilon_{kl}}{\varepsilon_{00}}. 
\end{align*} 
Figure \ref{fig:sensitivity-privacy} shows that, although the overall privacy level differs slightly from the value computed using the allocation in Table \ref{tab:PLB_typical}, the discrepancy is empirically confined to a very small range. Specifically, the overall privacy guarantee is observed to uniformly lie within
\begin{equation}
\begin{split}
\label{eqn:sensitivity_bound}
    &\max_{k,l} \varepsilon_{kl} \in \left[\varepsilon_{00}, (1 + 0.005) \varepsilon_{00} \right]\\
    & (\min_{k,l} f_{kl})^{**} \in \left[(1 - 1.0 \times 10^{-4})f_{00}, f_{00} \right].
\end{split}
\end{equation}
Although \eqref{eqn:sensitivity_bound} is not a formal theoretical guarantee, it provides strong empirical evidence that the overall privacy level is relatively insensitive to the choice among different privacy budget allocations. This observation motivates the following practical strategy for privacy budget tuning. In practice, the Census Bureau may use the privacy budget allocation in which no geographic level is bypassed, the composed mechanism $[\widetilde{\mathbf{M}}_0, \widetilde{\mathbf{M}}_0]$, as a proxy during the tuning phase. After tuning is complete, a final evaluation of the full overall privacy guarantee can be performed, which is expected to yield very similar performance.

\begin{figure}[!htp] 
\centering 
\begin{subfigure}[b]{0.42\textwidth}
\includegraphics[height=\textwidth]{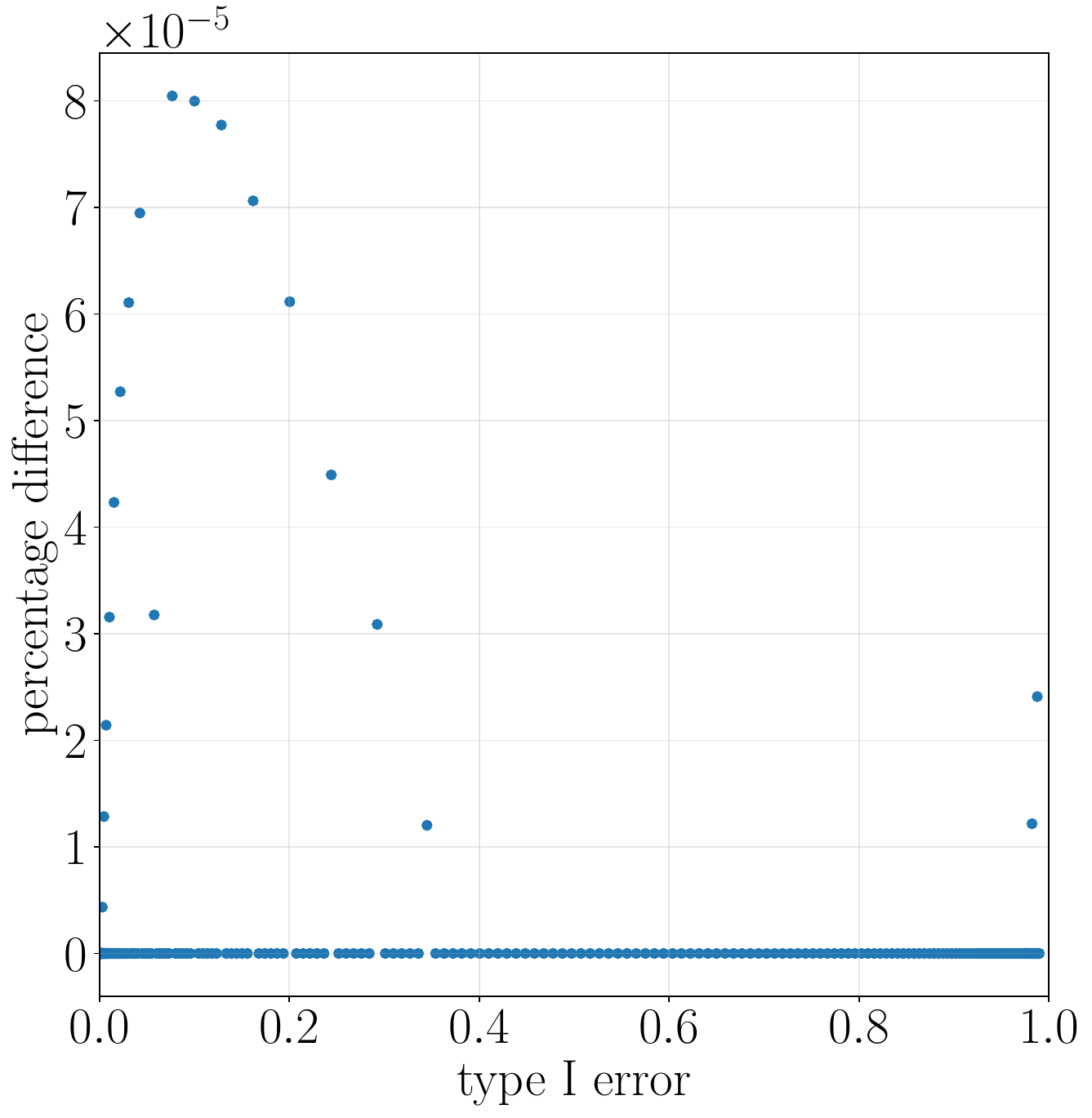} 
\end{subfigure}
\hspace{0.1\textwidth}
\begin{subfigure}[b]{0.42\textwidth}
\includegraphics[height=\textwidth]{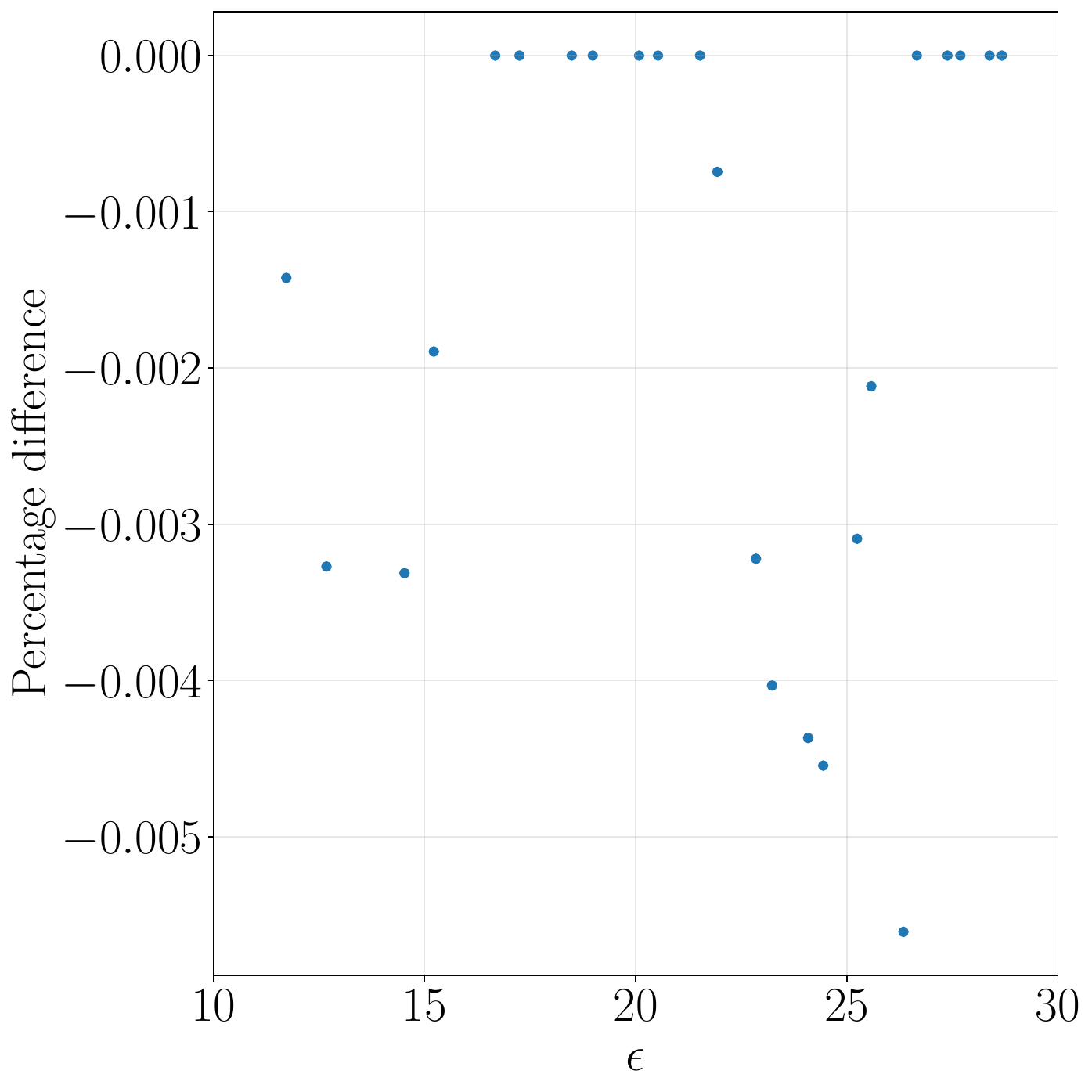}
\end{subfigure}
\caption{Percentage difference between the privacy level of the composed mechanism $[ \widetilde{\mathbf{M}}_0, \widetilde{\mathbf{M}}_0 ]$ and the overall privacy level in terms of $f$-DP (left) and $(\varepsilon, \delta)$-DP (right).} 
\label{fig:sensitivity-privacy} 
\end{figure}

\section{Approximation Gap from the Continuous Gaussian}
\label{sec:theory}

In this section, we discuss why exact privacy accounting for compositions of the discrete Gaussian mechanism, both in this paper and in prior work \citep{su20242020}, differs substantially from that of the (continuous) Gaussian mechanism \citep{dwork2014algorithmic, balle2018improving, Dong2022Gaussian}.
Propositions \ref{prop:eps_delta_curve},  \ref{prop:f_dp_curve} and Theorem 8 in \cite{balle2018improving} show that privacy accounting for compositions of discrete and continuous Gaussian mechanisms is mathematically equivalent to evaluating the tail probability of weighted convolutions of discrete and continuous Gaussian distributions, respectively. Empirically, Figure \ref{fig:discrete-continuous} illustrates the probability mass function of the convolution of two discrete Gaussian distributions and its continuous Gaussian counterpart, revealing a pronounced discrepancy between the discrete convolution and the smooth Gaussian density curve.

To gain further intuition, we examine the difference from the perspective of the central limit theorem. Even under i.i.d.\ compositions and equal weight, the convolution of discrete Gaussian variables produces a distribution that differs from its continuous Gaussian counterpart. For i.i.d.\ continuous Gaussian variables, the normalized sum has density
\begin{align*}
    \text{d}\, \PP_{X_i \sim \mathcal{N}(0,\sigma^2)}\left( \frac{1}{\sqrt{\var \left(\sum_{i=1}^m X_i \right)}} \sum_{i=1}^m X_i = y \right) = \frac{\ex^{-y^2/2}}{\sqrt{2 \pi}}.
\end{align*}
In contrast, a second-order CLT refinement—known as the Edgeworth expansion \citep{hall2013bootstrap, buhler2018explicit, derumigny2024explicit}—yields the following approximation for the $m$-fold convolution of discrete Gaussian variables:
\begin{align*}
    \PP_{X_i \sim \mathcal{N}_{\ZZ}(0,\sigma^2)} \left( \frac{1}{\sqrt{\var \left(\sum_{i=1}^m X_i \right)}} \sum_{i=1}^m X_i = y \right) \approx \frac{1}{\sqrt{\var \left(\sum_{i=1}^m X_i \right)}} \cdot \frac{\ex^{-y^2/2}}{\sqrt{2 \pi}}  \bigg(1 +  \frac{\lambda_{4}}{24 m} (3 - 6 y^2 + y^4) \bigg), 
\end{align*}
where the additional term $\frac{\lambda_4}{24m}(3 - 6y^2 + y^4)$ captures the discrepancy between the discrete and continuous Gaussian mechanisms.
For the noise parameters used in the 2020 Census DHC File, the magnitude of $\lambda_4/(24m)$ is not negligible. The illustrative result in Fact \ref{fact:lambda_4} shows that $\lambda_4/(24m) \gg \Delta$ under some of the adopted noise parameters. Theorem \ref{thm:edgeworth} formalizes this observation by establishing a rigorous uniform error bound of order $m^{-3/2}$.

Let $X_i \overset{\text{i.i.d.}}{\sim} \mathcal{N}_{\ZZ}(0,\sigma^2)$ and define $S_m = \frac{1}{B_m}\sum_{i=1}^m X_i,$ where $B_m = \sqrt{\var(\sum_{i=1}^m X_i)}$. 
As illustrated in Fact A.3 of \cite{su20242020}, $\var(X_i)$ is often close to, but not identical to, the noise parameter $\sigma^2$.
Let $\kappa_r=\kappa_r(X_i)$ denote the $r$-th cumulant of $X_i$, and define
\begin{align*}
    \lambda_{r} =\ & \kappa_r/\kappa_2^{r/2},
    \quad K_{r} = \EE |X_i|^r/\kappa_2^{r/2}.
\end{align*}
Using these quantities, we can derive an explicit uniform error bound for the Edgeworth expansion.
\begin{theorem}[Uniform error bound for the Edgeworth expansion]
\label{thm:edgeworth}
For any $y$ in the support of $S_m$, define the local error
    \begin{align*} 
        R_{m}(y) := \left| \PP \left( S_{m} = y \right) - \frac{1}{B_m} \cdot \frac{\ex^{-y^2/2}}{\sqrt{2 \pi}}  \bigg(1 +  \frac{\lambda_{4}}{24 m} (3 - 6 y^2 + y^4) \bigg)\right|.
    \end{align*}
Then, there exists a constant $C>0$ such that
    \begin{align} \label{eqn:edgeworth_order}
        \sup_{y\in\mathbb{Z}/B_m} R_m(y) \leq \frac{C}{m^{3/2}}.
    \end{align}
More precisely, let $U^{m}(c)$ and $U_m(\zeta)$ be defined in \eqref{eqn: U^m bound} and \eqref{eqn:U_m_bound}. For any $c$ satisfying $U^{m}(c)<1$, we have, 
    \begin{align}
    \begin{split} \label{eqn:edgeworth_bound}
        \sup_{y\in\mathbb{Z}/B_m} R_{m} \left( y \right) 
        \leq\ & m^{-1/2} \cdot \frac{1}{\pi \kappa_2^{1/2}} \int_{c \cdot \sqrt{\log(m)}}^{\infty} \ex^{-\zeta^2/2} \left(1 + \frac{\lambda_{4} \zeta^4}{24 m} \right) d\zeta 
        \\
        &+ m^{-1/2} \cdot \frac{1}{\pi \kappa_2^{1/2}} \int_{0}^{c \cdot \sqrt{\log(m)}} \exp \left( - \frac{\zeta^2}{2} \right) \times \bigg\{ \Lambda'_5 + \Lambda'_6 + \Lambda'_7 \bigg\}  d\zeta\\
        &+ \exp\left(- \frac{c^2}{2}  \log(m)  + \frac{\cos(\theta_0) - 1 + \theta_0^2/2}{2 \theta_0^3} \cdot 8 c^3 K_3 \cdot \left(\frac{\log^3 (m)}{m} \right)^{1/2} \right) ,
    \end{split}
    \end{align}
    where $\Lambda'_5, \Lambda'_6, \Lambda'_7$ are given by 
    \begin{align*}
        \Lambda'_5 =\ & \frac{\zeta^4}{8m} \cdot \left(\frac{1}{(1 - U^{m}(c))^2} + 1 \right) + \frac{K_{6}}{6!} \frac{\zeta^6}{m^2} + \frac{1}{(1 - U^{m}(c))^2} \cdot \bigg\{\frac{K_{6} \zeta ^6}{48 m^2} + \frac{K_{8} \zeta ^8}{640 m^3} + \frac{K_{10} \zeta ^{10}}{17280 m^4} + \frac{K_{12} \zeta ^{12}}{1036800 m^{5}} \bigg\}\\
        \Lambda'_6 =\ & \frac{1}{2} \left(\frac{K_{4}}{4!} \frac{\zeta^4}{m} +  \frac{K_{6}}{6!} \frac{\zeta^6}{m^2} + \frac{U_m(\zeta) }{2(1 - U^{m}(c))^2} \right)^2\\
        \Lambda'_7 =\ & \frac{1}{6} \exp \left(\frac{c^4 K_{4}}{4!} \frac{\log(m)^2}{m} + \frac{c^6 K_{6}}{6!} \frac{\log(m)^3}{m^2} + \frac{U_m (c \sqrt{\log(m)}{})}{2(1 - U^{m}(c))^2} \right) \times \left(\frac{K_{4}}{4!} \frac{\zeta^4}{m} + \frac{K_{6}}{6!} \frac{\zeta^6}{m^2} + \frac{U_m(\zeta)}{2(1 - U^{m}(c))^2} \right)^3,
    \end{align*}
    and $\theta_0$ is the unique root of the equation
    \begin{align} \label{eqn: theta_0}
        \theta^2 + 2 \theta \sin(\theta) + 6 (\cos(\theta) - 1) = 0, 
    \end{align}
    within the interval $(0, 2\pi)$.
\end{theorem}

The proof of Theorem \ref{thm:edgeworth} is postponed to Section \ref{sec:proof_theory}. 
We now briefly explain why the right-hand side of Equation~\eqref{eqn:edgeworth_bound} is of order $O(m^{-3/2})$, as stated in Equation~\eqref{eqn:edgeworth_order}.
Recall that Equation~\eqref{eqn:edgeworth_bound} holds for any $c$ satisfying
\begin{align} 
     U^{m}(c) =\ & c^2 \cdot \frac{\log(m)}{2m} + c^4 \cdot \frac{K_4}{4!} \left(\frac{\log(m)}{m} \right)^2 + c^6 \cdot \frac{K_6}{6!} \left(\frac{\log(m)}{m}\right)^3 < 1
\end{align}
When $m$ is sufficiently large, we may choose $c = 2$. With this choice, we have
\begin{align*}
    &\frac{1}{\pi \kappa_2^{1/2}} \int_{c \cdot \sqrt{\log(m)}}^{\infty} \ex^{-\zeta^2/2} \left(1 + \frac{\lambda_{4} \zeta^4}{24 m} \right) d\zeta = o(m^{-3/2})\\
    &\exp\left(- \frac{c^2}{2} \log(m)  + \frac{\cos(\theta_0) - 1 + \theta_0^2/2}{2 \theta_0^3} \cdot 8 c^3 K_3 \cdot \left(\frac{\log^3 (m)}{m} \right)^{1/2} \right) = o(m^{-3/2}) \\
    & \Lambda'_5 = O\left(\frac{1}{m}\right), \quad 
    \Lambda'_6 = O\left(\frac{1}{m^2}\right), \quad \Lambda'_7 = O\left(\frac{1}{m^3}\right).
\end{align*}
Combining these bounds, we conclude that $R_{m}(y)$ is bounded above by $O(m^{-3/2})$.

\section{Discussion}
\label{sec:discussion}




In this paper, we provide, for the first time, an exact characterization of the overall privacy guarantee of the 2020 Census DHC File in terms of $f$-DP and its corresponding $(\varepsilon,\delta)$-DP profile.
Compared with the U.S.\ Census Bureau’s current upper-bounding approach, these curves offer a more direct and policy-relevant interpretation of re-identification risk and are easier to communicate in empirical research and policymaking across the social sciences \citep{Sullivan2020Coming}, political science \citep{ansolabehere2008end,cohen2021census}, and economics \citep{autor2003rise,bureau2023guidance}, while achieving the optimal privacy--utility trade-off with the exact computation of the composition bounds.
Achieving exact privacy accounting for the discrete Gaussian mechanism presents two key challenges: analytic composition methods developed for the continuous Gaussian setting \citep{kairouz2021distributed,zhu2022optimal} do not extend to the discrete case, and exact $(\varepsilon,\delta)$-DP and $f$-DP curves require extremely high numerical precision together with large-scale parallel composition (e.g., 946 mechanisms for the 2020 Census DHC File), making existing numerical approaches \citep{Koskela2020computing,gopi2021numerical,su20242020} computationally impractical. 

To address this, we develop a new framework that integrates numerical analysis and number-theoretic ideas.
Our acceleration rests on three numerical innovations tailored to the heterogeneity and extreme-precision regime of the discrete Gaussian mechanism. We recast exact privacy accounting into a discrete Fourier–based numerical integration problem, and then explicitly exploit the exponential convergence of the trapezoidal rule for complex-analytic periodic integrands. To make this exponential regime usable in the fully heterogeneous setting, we map rational weights onto a common integer lattice, enabling sharp sub-Gaussian truncation bounds. Finally, inspired by sieve methods in number theory, we prune quadrature nodes by rigorously discarding regions where the oscillatory integrand must be negligible, reducing the required evaluations from $\sim 10^6$ nodes down to as few as $203$ in representative DHC settings; combined with practical tuning guidance, this makes exact privacy characterization feasible at scale.


More broadly, our contribution is a fast and numerically reliable framework for exact privacy accounting in large-composition settings. This makes accurate privacy characterization practically feasible and enables policymakers to treat privacy cost as a tunable resource when allocating budget across releases. Beyond the discrete Gaussian mechanism and the Census setting, several directions merit further study. For instance, developing a direct analytical construction of the sieve (e.g., via lattice reduction) could eliminate the current sequential search and further accelerate the algorithm. In addition, although the Census Bureau’s current allocation weights are rational, extending the method to handle irrational weights and almost-periodic integrands could enable more flexible privacy budget allocations. Relatedly, integrating the framework with other privacy accounting paradigms may provide accurate privacy budgets with tighter privacy--utility trade-offs at scale. Finally, it may be fruitful to explore whether these techniques can inform the design of improved integer-valued mechanisms for official statistics.

{\small
\section*{Acknowledgments}

This work was supported in part by a Meta Faculty Research Award, and Wharton AI for Business.
Chendi Wang was supported by the Fundamental Research Funds for the Central Universities (No.\ 20720251050) and NSFC Basic Science Center Project for Econometric Modeling and Economic Policy Studies (No.\ 71988101).

\bibliographystyle{abbrvnat}
\bibliography{ref}
}

\clearpage
\appendix

\begin{center}
  {\LARGE Supplementary Materials for\par}
  \vspace{0.6em}
  \begin{minipage}{\textwidth}
    \centering\LARGE
    \scititle
  \end{minipage}
  \vspace{0.6em}
\end{center}

\section{Preliminaries}

\subsection{Useful facts for discrete Gaussian distributions}
\label{sec:useful_facts}
We introduce some useful facts for discrete Gaussian distributions in this section, including the sub-Gaussian tail bound for discrete Gaussian distributions and properties of the characteristic functions.
The proof of this section can be found in \cite{Canonne2020discrete, su20242020}.

\paragraph{Properties of discrete Gaussian distributions.} 
The following proposition on the moment generating function and tail bound shows that the discrete Gaussian is sub-Gaussian.
\begin{proposition}[Lemma 16 \& Corollary 17 in \cite{Canonne2020discrete}]
\label{prop:subgaussian}
    A discrete Gaussian random variable $X \sim \mathcal{N}_{\ZZ}(0,\sigma^2)$ is sub-Gaussian with variance proxy $\sigma^2.$ For any $t \in \RR$, we have
    \begin{align*}
        \EE \ex^{t X} \leq \ex^{t^2 \sigma^2/2}.
    \end{align*}
    Moreover, we have the tail bound
    \begin{align*}
        \PP(X > t) \leq \ex^{-t^2/2 \sigma^2}. 
    \end{align*}
\end{proposition}


The following proposition records the numerical values of $\lambda_4$ defined in Section \ref{sec:theory} for some examples of noise parameters $\sigma^2$.  
\begin{fact} \label{fact:lambda_4}
    The $\lambda_4$ defined in Section \ref{sec:theory} takes the following values.  
    For $\sigma^2 = 0.1$, $\lambda_4 \approx 72.21.$
    For $\sigma^2 = 0.2$, $\lambda_4 \approx 4.12.$
    For $\sigma^2 = 0.4$, $\lambda_4 \approx 0.19.$
    For $\sigma^2 = 1.0$, $\lambda_4 \approx 8.34 \times 10^{-6}.$
    For $\sigma^2 = 2.0$, $\lambda_4 \approx 8.92 \times 10^{-14}.$
    For $\sigma^2 = 4.0$, $\lambda_4 \approx 2.55 \times 10^{-30}.$
    For $\sigma^2 = 10.0$, $\lambda_4 < 10^{-50}.$
\end{fact}

\paragraph{Characteristic functions of discrete Gaussian distribution.}
The characteristic function of $S_m$ defined in Section \ref{sec:theory} can be represented as follows.
\begin{align} \label{eqn:char_func}
\begin{split}
    f_{S_m}(t) = \EE \ex^{it S_m} =\ & \left( \frac{\sum_{u=- \infty}^{\infty} \ex^{- u^2/2 \sigma^2} \ex^{i \cdot t/B_m \cdot u}}{\sum_{u=- \infty}^{\infty} \ex^{- u^2/2 \sigma^2}} \right)^m\\
    \overset{(a)}{=}\ & \left( \frac{\sum_{u=- \infty}^{\infty} \ex^{- \sigma^2 (t/B_m - 2 \pi u)^2/2}}{\sum_{u=- \infty}^{\infty} \ex^{- 2 \sigma^2 \pi^2 u^2}} \right)^m,
\end{split}
\end{align}
where Equality (a) holds due to Poisson summation formula.
For $X \sim \mathcal{N}_{\ZZ}(0, \sigma^2)$, characteristic functions $f_{X}(\zeta)$ have the following properties.
\begin{lemma}[Proposition A.6 in \cite{su20242020}]
    We have the following observation on $f_{X}(\zeta)$:
    \begin{enumerate}
    \item $|f_{X}(\zeta)| \leq 1$, $f_{X}(\zeta)$ achieves its maximum at $\zeta=0$ with a maximum value of $f_{X}(\zeta) = 1$. 
    \item $f_{X}(\zeta)$ is periodic with period $2 \pi$, therefore $f_{X}(a_i L \zeta)$ is periodic with period $2 \pi/a_i L$.
    \item $f_{X}(\zeta)$ is strictly increasing on $(- \pi, 0)$ and is strictly decreasing on $(0, \pi)$.
    \end{enumerate}
\end{lemma}

\section{Proof of Section \ref{sec:foundation}}
\label{sec:proof_foundation}

\subsection{Proof of Proposition \ref{prop:eps_delta_curve}}

The probability mass function of $P_i$ and $Q_i$ is given by the following.
\begin{align*}
    \prod_{i=1}^{m} \text{d} P_i(x_1, \cdots, x_m) = \prod_{i=1}^{m} \frac{\ex^{-x_i^2/2\sigma_i^2}}{\sum_{x \in \ZZ} \ex^{-x^2/2\sigma_i^2}}, \quad \prod_{i=1}^{m} \text{d} Q_i(x_1, \cdots, x_m) = \prod_{i=1}^{m} \frac{\ex^{-(x_i - \mu)^2/2\sigma_i^2}}{\sum_{x \in \ZZ} \ex^{-x^2/2\sigma_i^2}}.
\end{align*}
Given that $\mu = 1$ is an integer, we have
\begin{align*}
    \sum_{i=1}^m \log \frac{\text{d} Q_i(Y_i)}{\text{d} P_i(Y_i)} =\ & \log\left( \prod_{i=1}^{m} \frac{\ex^{-(Y_i - \mu)^2/2\sigma_i^2}}{\sum_{x \in \ZZ} \ex^{-x^2/2\sigma_i^2}} \bigg/ \prod_{i=1}^{m} \frac{\ex^{-Y_i^2/2\sigma_i^2}}{\sum_{x \in \ZZ} \ex^{-x^2/2\sigma_i^2}} \right) 
    \\
    =\ & \sum_{i=1}^{m} \log\left( \frac{\ex^{-(Y_i - \mu)^2/2\sigma_i^2}}{\sum_{x \in \ZZ} \ex^{-x^2/2\sigma_i^2}} \bigg/ \frac{\ex^{-Y_i^2/2\sigma_i^2}}{\sum_{x \in \ZZ} \ex^{-x^2/2\sigma_i^2}} \right) 
    \\
    =\ & \sum_{i=1}^{m} \frac{Y_{i}}{\sigma_{i}^2} - \frac{\mu^2}{2 \sigma_{i}^2}
    \\
    \overset{d}{=}\ & \sum_{i=1}^{m} \frac{X_{i}}{\sigma_{i}^2} + \frac{\mu^2}{2 \sigma_{i}^2},
\end{align*}
and
\begin{align*}
    \sum_{i=1}^m \log \frac{\text{d} Q_i(X_i)}{\text{d} P_i(X_i)} =\ & \sum_{i=1}^{m} \frac{X_{i}}{\sigma_{i}^2} - \frac{\mu^2}{2 \sigma_{i}^2}.
\end{align*}
By Lemma \ref{lemma:privacy-profile}, we have
\begin{align*}
    \delta(\varepsilon_{kl}) =\ & \PP \left( \sum_{i=1}^{m} \frac{X_{i}}{\sigma_i^2} > \varepsilon_{kl} - \sum_{i=1}^{m} \frac{\mu^2}{2 \sigma_{i}^2} \right) 
    - \ex^{\varepsilon_{kl}} \cdot \PP\left( \sum_{i=1}^{m} \frac{X_{i}}{\sigma_i^2} > \varepsilon_{kl} + \sum_{i=1}^{m} \frac{\mu^2}{2 \sigma_{i}^2} \right).
\end{align*}
Let $a_i = \rho_i/\rho$. Recall that $\mu=1$, and Equation \eqref{eqn:noise}, we have
\begin{align*}
    \delta(\varepsilon_{kl}) =\ & \PP \left( \sum_{i=1}^{m} a_{i} X_{i} > \frac{\varepsilon_{kl}}{\rho} - \sum_{i=1}^{m} \frac{\mu^2}{2 \rho \sigma_{i}^2} \right) 
    - \ex^{\varepsilon_{kl}} \cdot \PP\left( \sum_{i=1}^{m} a_{i} X_{i} > \frac{\varepsilon_{kl}}{\rho} + \sum_{i=1}^{m} \frac{\mu^2}{2 \rho \sigma_{i}^2} \right)
    \\
    =\ & \PP \left( \sum_{i=1}^{m} a_{i} X_{i} > \frac{\varepsilon_{kl}}{\rho} - \sum_{i=1}^{m} \frac{a_i}{2} \right) 
    - \ex^{\varepsilon_{kl}} \cdot \PP\left( \sum_{i=1}^{m} a_{i} X_{i} > \frac{\varepsilon_{kl}}{\rho} + \sum_{i=1}^{m} \frac{a_i}{2} \right)
    \\
    =\ & \PP \left( \sum_{i=1}^{m} a_{i} X_{i} > \frac{\varepsilon_{kl}}{\rho} - 1 \right) 
    - \ex^{\varepsilon_{kl}} \cdot \PP\left( \sum_{i=1}^{m} a_{i} X_{i} > \frac{\varepsilon_{kl}}{\rho} + 1 \right).
\end{align*}
The last equation follows from the fact that we are considering the composed composition $[ \widetilde{\mathbf{M}}_k, \widetilde{\mathbf{M}}_l ]$ and $(\sigma_1^2, \sigma_2^2, \cdots, \sigma_m^2)$ denotes parameters of noises injected in $\widetilde{\mathbf{M}}_k$ and $\widetilde{\mathbf{M}}_l$ determined by Equation \eqref{eqn:noise}. Therefore, $\sum_{i} \rho_i = 2\rho$ and $\sum_{i} a_i = 2$.

\subsection{Proof of Proposition \ref{prop:f_dp_curve}}

By Neyman-Pearson Lemma, we have
\begin{align*}
    \alpha_{kl}(\zeta) =\ & \PP_{X_i \sim \cN_{\ZZ}(0, \sigma_{i}^2)} \left( \sum_{i=1}^m \log \frac{\text{d} Q_i(X_i)}{\text{d} P_i(X_i)} > \zeta \right) + c \cdot \PP_{X_{i} \sim \cN_{\ZZ}(0, \sigma_{i}^2)} \left( \sum_{i=1}^m \log \frac{\text{d} Q_i(X_i)}{\text{d} P_i(X_i)} = \zeta \right), 
    \\
    \beta_{kl}(\zeta) =\ & \PP_{Y_i \sim \cN_{\ZZ}(\mu, \sigma_{i}^2)} \left( \sum_{i=1}^m \log \frac{\text{d} Q_i(Y_i)}{\text{d} P_i(Y_i)} \leq \zeta \right) - c \cdot \PP_{Y_i \sim \cN_{\ZZ}(\mu, \sigma_{i}^2)} \left( \sum_{i=1}^m \log \frac{\text{d} Q_i(Y_i)}{\text{d} P_i(Y_i)} = \zeta \right)
\end{align*}
for some $c$. 
Equivalently, after change of variable, as $\mu = 1$ we have the parametric equation
\begin{align*}
    \alpha_{kl}(\zeta) =\ & \PP_{X_{i} \sim \cN_{\ZZ}(0, \sigma_{i}^2)} \left( \sum_{i=1}^m a_{i} X_{i} > \zeta + \sum_{i=1}^m \frac{a_{i}}{2} \right) + c \cdot \PP _{X_{i} \sim \cN_{\ZZ}(0, \sigma_{i}^2)} \left( \sum_{i=1}^m a_{i} X_{i} = \zeta + \sum_{i=1}^m \frac{a_{i}}{2} \right),
    \\
    \beta_{kl}(\zeta) =\ & \PP_{X_{i} \sim \cN_{\ZZ}(0, \sigma_{i}^2)} \left( \sum_{i=1}^m a_{i} X_{i} \leq \zeta - \sum_{i=1}^m \frac{a_{i}}{2} \right) - c \cdot \PP _{X_{i} \sim \cN_{\ZZ}(0, \sigma_{i}^2)} \left( \sum_{i=1}^m a_{i} X_{i} = \zeta - \sum_{i=1}^m \frac{a_{i}}{2} \right).
\end{align*}
As $\sum_i a_i = 2$, the trade-off function is determined by the following parametric equation:
\begin{align*}
    \alpha_{kl}(\zeta) =\ & \PP_{X_{i} \sim \cN_{\ZZ}(0, \sigma_{i}^2)} \left( \sum_{i=1}^m a_{i} X_{i} > \zeta + 1 \right) + c \cdot \PP _{X_{i} \sim \cN_{\ZZ}(0, \sigma_{i}^2)} \left( \sum_{i=1}^m a_{i} X_{i} = \zeta + 1 \right),
    \\
    \beta_{kl}(\zeta) =\ & \PP_{X_{i} \sim \cN_{\ZZ}(0, \sigma_{i}^2)} \left( \sum_{i=1}^m a_{i} X_{i} \leq \zeta - 1 \right) - c \cdot \PP _{X_{i} \sim \cN_{\ZZ}(0, \sigma_{i}^2)} \left( \sum_{i=1}^m a_{i} X_{i} = \zeta - 1 \right).
\end{align*}

\section{Proof of Section \ref{sec:method}}
\label{sec:proof_method}

\subsection{Proof of Proposition \ref{prop:choice_U}}
\label{sec:proof_choice_U}
\begin{proof}
    By Lemma 16 in \cite{Canonne2020discrete}, we have
    \begin{align*}
        \EE_{X_{i} \sim \mathcal{N}_{\mathbb{Z}}(0, \sigma_{i}^2)} \left[\ex^{t X_{i}}\right] \leq \ex^{t^2 \sigma_{i}^2 / 2}.
    \end{align*}
    Therefore, 
    \begin{align*}
        \EE_{X_{i} \sim \mathcal{N}_{\mathbb{Z}}(0, \sigma_{i}^2)} \left[\ex^{t \sum_{i=1}^{m} a_i L X_{i} }\right] = \prod_{i=1}^{m} \EE_{X_{i} \sim \mathcal{N}_{\mathbb{Z}}(0, \sigma_{i}^2)} \left[\ex^{t a_i L X_{i} }\right] \leq \ex^{t^2 \cdot \sum_{i=1}^{m} a_i^2 L^2 \sigma_{i}^2 / 2},
    \end{align*}
    which further implies that 
    \begin{align*}
        \PP_{X_{i}\sim\mathcal{N}_{\mathbb{Z}}(0,\sigma_i^2)} \left( \left| \sum_{i=1}^{m} a_i L X_{i} \right| > U \right) \leq 2 \cdot \ex^{- U^2/ 2 \sum_{i=1}^{m} a_i^2 L^2 \sigma_{i}^2}.
    \end{align*}
    For $U$ satisfying Equation \eqref{eqn:choice_U}, by Corollary 17 in \cite{Canonne2020discrete}, we have
    \begin{align*}
        \PP_{X_{i}\sim\mathcal{N}_{\mathbb{Z}}(0,\sigma_i^2)} \left( \left| \sum_{i=1}^{m} a_i L X_{i} \right| > U \right) \leq 2 \cdot \ex^{- U^2/ 2 \sum_{i=1}^{m} a_i^2 L^2 \sigma_{i}^2} \leq \Delta/4.
    \end{align*}
    This completes the proof of Proposition \ref{prop:choice_U} that 
    \begin{align*}
        \left| \mathbb{P}_{X_{i}\sim\mathcal{N}_{\mathbb{Z}}(0,\sigma_i^2)}\left[\sum_{i=1}^{m} a_i X_{i} > t_0 \right] - \PP_{X_{i}\sim\mathcal{N}_{\mathbb{Z}}(0,\sigma_i^2)} \left( U \geq \sum_{i=1}^{m} a_i L X_{i} > t_0 \cdot L \right) \right| \leq \Delta/4
    \end{align*}
    for any $U$ satisfying Equation \eqref{eqn:choice_U}.
\end{proof}

\subsection{Proof of Proposition \ref{prop:choice_N}}
\label{sec:proof_choice_N}

\begin{proof}
Consider Lemma \ref{lemma:trapezoidal_error} with the complex strip boundary $a = 1/(dL)$ for an optimized parameter $d > 0$. 
Let $\zeta = x + iy$ for $x, y \in \mathbb{R}$ restricted to the complex strip $-1/(dL) \le y \le 1/(dL)$. 
We aim to rigorously establish an upper bound $M$ such that $|F(\zeta)| \le M$ across this strip.
Recall the definition of $F(\zeta)$ from Equation \eqref{eqn:F}:
$$ 
F(\zeta) = \frac{1}{2\pi} \left[ \sum_{t=\lceil t_0 \cdot L \rceil}^U \cos(\zeta t) \right] \cdot \mathbb{E}[e^{i\zeta X}].
$$
By the triangle inequality, we can bound the magnitude of the product by the product of the magnitudes:
\begin{equation}\label{eqn:f_upper}
    |F(\zeta)| \le \frac{1}{2\pi} \left( \sum_{t=\lceil t_0 \cdot L \rceil}^U |\cos(\zeta t)| \right) \cdot |\mathbb{E}[e^{i\zeta X}]|
\end{equation}

\noindent
\textbf{Step 1: Bounding the Expectation.}
We must bound the magnitude of the complex expectation without improperly factoring dependent variables. Using Jensen's inequality, we have
$$ 
|\mathbb{E}[e^{i(x+iy)X}]| = |\mathbb{E}[e^{ixX} e^{-yX}]| \le \mathbb{E}[|e^{ixX} e^{-yX}|] = \mathbb{E}[|e^{ixX}| e^{-yX}] = \mathbb{E}[e^{-yX}].
$$
Because the discrete Gaussian convolution $X = \sum_{i=1}^m a_i L X_i$ is symmetrically distributed around $0$, we have $\mathbb{E}[e^{-yX}] = \mathbb{E}[e^{yX}]$. 
The moment generating function is an even, convex function of $y$, so its maximum on the interval $y \in [-1/(dL), 1/(dL)]$ is achieved at the boundary $y = 1/(dL)$.
Thus, uniformly across the strip:
$$ 
|\mathbb{E}[e^{i\zeta X}]| \le \mathbb{E}[e^{X/(dL)}].
$$
By definition of $X = \sum_{i=1}^{m} a_i L X_{i}$, we have
\begin{align*}
    \EE \ex^{X/dL} = \prod_{i=1}^{m} \ex^{(a_i/d) \cdot X_{i}} =\ & \prod_{i=1}^{m} \left\{ \sum_{x_{i} = -\infty}^{\infty} \ex^{(a_i/d) \cdot  x_i} \cdot \ex^{-x_{i}^2/2 \sigma_{i}^2} \right\} \cdot \prod_{i=1}^{m} \left\{ \sum_{x_{i} = -\infty}^{\infty} \ex^{-x_{i}^2/2 \sigma_{i}^2} \right\}^{-1}
    \\
    =\ & \prod_{i=1}^{m} \ex^{a_i^2 \sigma_{i}^2/2 d^2} \cdot \left\{ \sum_{x_{i} = -\infty}^{\infty} \ex^{- \left( \frac{x_i}{\sqrt{2 \sigma_i^2}} - \frac{a_i \sigma_{i}}{d \sqrt{2}} \right)^2} \right\} \cdot \prod_{i=1}^{m} \left\{ \sum_{x_{i} = -\infty}^{\infty} \ex^{-x_{i}^2/2 \sigma_{i}^2} \right\}^{-1}
    \\
    \overset{(c)}{<}\ & \prod_{i=1}^{m} \ex^{a_i^2 \sigma_{i}^2/2 d^2} \cdot \left\{ \sum_{x_{i} = -\infty}^{\infty} \ex^{- \left( \frac{x_i}{\sqrt{2 \sigma_i^2}} \right)^2} \right\} \cdot \prod_{i=1}^{m} \left\{ \sum_{x_{i} = -\infty}^{\infty} \ex^{-x_{i}^2/2 \sigma_{i}^2} \right\}^{-1}
    \\
    <\ & \prod_{i=1}^{m} \ex^{a_i^2 \sigma_{i}^2/2 d^2} = \ex^{\sum_{i=1}^{m} a_i^2 \sigma_{i}^2/2 d^2},
\end{align*}
where Equation (c) follows from Fact A.4 in \cite{su20242020}.

\noindent
\textbf{Step 2: Bounding the Finite Geometric Sum.}
Next, we bound the sum of the cosine terms. For $\zeta = x + iy$:
$$ 
|\cos(\zeta t)| = \left| \frac{e^{i(x+iy)t} + e^{-i(x+iy)t}}{2} \right| \le \frac{e^{-yt} + e^{yt}}{2} = \cosh(yt) \le e^{|y||t|} \le e^{|t|/(dL)}.
$$
This forms a finite geometric series with common ratio $r = e^{1/(dL)} > 1$. Its exact value is:
$$ 
\sum_{t=\lceil t_0 \cdot L \rceil}^U |\cos(\zeta t)| \leq 2 \sum_{t=0}^U \left(e^{1/(dL)}\right)^t = 2 \cdot \frac{e^{(U+1)/(dL)} - 1}{e^{1/(dL)} - 1} < 2 \cdot \frac{e^{(U+1)/(dL)}}{e^{1/(dL)} - 1}.
$$
Using the fundamental exponential inequality $e^z - 1 > z$ for all $z > 0$, we substitute $z = 1/(dL)$ to get $\frac{1}{e^{1/(dL)} - 1} < dL$. 
Thus, the summation is strictly bounded by:
$$ 
\sum_{t=\lceil t_0 \cdot L \rceil}^U |\cos(\zeta t)| < 2 \cdot dL \cdot e^{(U+1)/(dL)}. 
$$

\noindent
\textbf{Step 3: Establishing $M$ and the Trapezoidal Error Bound.}
Substituting the rigorous bounds from Step 1 and 2 back into Equation \eqref{eqn:f_upper}, the absolute upper bound $M$ for $|F(\zeta)|$ on the complex strip is:
$$ M = \frac{1}{2\pi} \left( 2 \cdot dL \cdot e^{(U+1)/(dL)} \right) \exp\left( \frac{\sum_{i=1}^m a_i^2 \sigma_i^2}{2d^2} \right) = \frac{dL}{\pi} \exp\left( \frac{U+1}{dL} + \frac{\sum_{i=1}^m a_i^2 \sigma_i^2}{2d^2} \right). $$
By Lemma \ref{lemma:trapezoidal_error}, the approximation error of the trapezoidal rule is bounded by $\frac{4\pi M}{e^{N/(dL)} - 1}$. To guarantee this error is at most $\Delta / 4$, we require:
$$ 
\frac{4\pi M}{e^{N/(dL)} - 1} \le \frac{\Delta}{4} \iff e^{N/(dL)} \ge 1 + \frac{16\pi M}{\Delta},
$$
yielding the sufficient condition:
$$ 
e^{N/(dL)} \ge 1 + \frac{16\pi M}{\Delta} = 1 + \frac{16 dL}{\Delta} \exp\left( \frac{U+1}{dL} + \frac{\sum_{i=1}^m a_i^2 \sigma_i^2}{2d^2} \right).
$$
Taking the natural logarithm of both sides and multiplying by $dL$, we obtain the explicit requirement for $N$:
$$ 
N \ge U + 2 + \frac{L \sum_{i=1}^m a_i^2 \sigma_i^2}{2d} + dL \log\left( \frac{16 dL}{\Delta} \right).
$$
We expand the logarithm:
\begin{equation} \label{eqn:f_bound_2}
    N \ge U + 2 + \frac{L \sum_{i=1}^m a_i^2 \sigma_i^2}{2d} - dL \log\left( \frac{\Delta}{8} \right) + dL \log(2 \cdot dL).
\end{equation}

\noindent
\textbf{Step 4: Optimizing $d$ via AM-GM.}
To minimize the required number of integration nodes $N$, we find the optimal parameter $d > 0$. 
The threshold in Equation \eqref{eqn:f_bound_2} is dominated by the two terms containing $d$. 
By the AM-GM inequality, their sum is globally minimized when the terms are perfectly equal:
$$ 
\frac{L \sum_{i=1}^m a_i^2 \sigma_i^2}{2d} = -dL \log\left( \frac{\Delta}{8} \right) \iff d = \frac{\sqrt{\sum_{i=1}^m a_i^2 \sigma_i^2}}{\sqrt{-2 \log(\Delta/8)}}.
$$
Substituting this optimal value of $d$ back into the dominant terms evaluates to exactly twice their balanced value:
$$ 
2 \left( -dL \log\left( \frac{\Delta}{8} \right) \right) = L \sqrt{ -2 \left( \sum_{i=1}^m a_i^2 \sigma_i^2 \right) \log\left( \frac{\Delta}{8} \right) } \leq U.
$$
Substituting this combined minimum back into Equation \eqref{eqn:f_bound_2} yields the final, formally rigorous condition for $N$:
$$ 
N \geq 2 \cdot (U + 1) + dL \log(2 \cdot dL) \ge U + 2 + L \sqrt{ -2 \left( \sum_{i=1}^m a_i^2 \sigma_i^2 \right) \log\left( \frac{\Delta}{8} \right) } + dL \log(2 \cdot dL). 
$$
This explicitly corrects the required parameters to guarantee a numerical tolerance of $\Delta/4$, completing the proof.
\end{proof}

\subsection{Proof of Proposition \ref{prop:C}}
\label{sec:proof_char_filter}

\begin{proof}
    By triangle inequality, it suffices to prove for each $k \in [N] - \cC$, we have $2 \pi \cdot |F(\zeta_{k})| \leq \Delta/4$. 
    Recall that Equation \eqref{eqn:F} that 
    \begin{align*}
        F(\zeta)
        =\ & \frac{1}{2 \pi} \cdot \left[ \sum_{ t = \lceil t_0 \cdot L \rceil}^{U} \cos(\zeta t) \right] \cdot \prod_{i=1}^{m} f_{a_i L X_i}(\zeta).
    \end{align*}
    According to Equation \eqref{eqn:update_node}, we have for each $k \in [N] - \cC$, we have
    $$
        |f_{a_i L X_i}(\zeta_k)| < \Delta/(8 \cdot U) 
    $$
    for some $i$.
    This implies that for each $k \in [N] - \cC$, we have
    \begin{align*}
        2\pi \cdot |F(\zeta_k)|
        =\ & \left| \left[ \sum_{ t = \lceil t_0 \cdot L \rceil}^{U} \cos(\zeta t) \right] \cdot \prod_{i=1}^{m} f_{a_i L X_i}(\zeta) \right| \\
        \leq\ & \left| \left[ \sum_{ t = \lceil t_0 \cdot L \rceil}^{U} \cos(\zeta t) \right] \cdot \min\left\{|f_{X_1}(a_i L \zeta)|, \cdots, |f_{X_m}(a_i L \zeta)| \right\} \right| \\
        \leq\ & \left| \left[ \sum_{ t = \lceil t_0 \cdot L \rceil}^{U} \cos(\zeta t) \right] \cdot \Delta/(8 \cdot U) \right| \leq \Delta/4. 
    \end{align*}
    By triangle inequality, this completes the proof of Proposition \ref{prop:C}.
\end{proof}

\section{Proof of Theorem \ref{thm:edgeworth}}
\label{sec:proof_theory}

Recall the notation $X_i \overset{I.I.D.}{\sim} \mathcal{N}_{\ZZ}(0,\sigma^2)$ and $S_m = \frac{1}{B_m}\sum_{i=1}^m X_i,$ where $B_m = \sqrt{\var(\sum_{i=1}^m X_i)}$. 
$\kappa_r = \kappa_r(X_i)$ is the abbreviation for the $r^{th}$ cumulant of $X_{i}$ and $\lambda_{r}$ and $K_{r}$ as:
\begin{align*}
    \lambda_{r} =\ & \kappa_r/\var(X_i)^{r/2}, \quad K_{r} = \EE |X_i|^r/\var(X_i)^{r/2}.
\end{align*}
By Exercise 3.3.2 (iii) in \cite{durrett2019probability}, the Fourier inversion for discrete random variables is given by 
\begin{align} \label{eqn:discrete_fourier_general}
    \PP(S_m = y) = \frac{1}{2 \pi B_m} \int_{-\pi B_m}^{\pi B_m} \ex^{-i \zeta y} f_{S_m}(\zeta) d\zeta.
\end{align}
By Fourier Inversion, the edgeworth expansion can be written as
\begin{align} \label{eqn:edgeworth_fourier_general}
\begin{split}
    &\frac{1}{B_{m}} \cdot \frac{\ex^{-y^2/2} }{\sqrt{2 \pi}} \bigg(1 +  \frac{\lambda_{4}}{24 m} (3 - 6 y^2 + y^4) \bigg) = \frac{1}{2\pi B_{m}} \int_{-\infty}^{\infty} \ex^{-i \zeta y} \ex^{-\zeta^2/2} \left(1 + \frac{\lambda_{4} \zeta^4}{24 m} \right) d\zeta
\end{split}
\end{align}
Combining \eqref{eqn:discrete_fourier_general} and \eqref{eqn:edgeworth_fourier_general}, we have
\begin{align} \label{omega1 + omega2}
\begin{split}
    R_{m} \leq\ & \frac{1}{2\pi B_{m}}  \bigg|\int_{-\infty}^{\infty} \ex^{-i \zeta y} \ex^{-\zeta^2/2} \left(1 + \frac{\lambda_{4} \zeta^4}{24 m}\right) d\zeta - \int_{-\pi B_m}^{\pi B_m} \ex^{-i \zeta y} f_{S_m}(\zeta) d\zeta \bigg|\\
    \leq\ & \frac{1}{2\pi B_{m}}  \bigg|\int_{-\pi B_m}^{\pi B_m} \ex^{-i \zeta y} \ex^{-\zeta^2/2} \left(1 + \frac{\lambda_{4} \zeta^4}{24 m} \right) d\zeta - \int_{-\pi B_m}^{\pi B_m} \ex^{-i \zeta y} f_{S_m}(\zeta) d\zeta \bigg|\\
    &+ \frac{1}{\pi B_{m}} \left|\int_{\pi B_m}^{\infty} \ex^{-i \zeta y} \ex^{-\zeta^2/2} \left(1 + \frac{\lambda_{4} \zeta^4}{24 m} \right) d\zeta \right|\\
    =:\ & \Lambda_1 + \Lambda_2
\end{split}
\end{align}
\begin{lemma}[Upper bound on $\Lambda_2$]
\label{lemma: omega2}
By triangle inequality, we have
    \begin{align*}
        \Lambda_{2} \leq \frac{1}{\pi B_{m}} \int_{\pi B_m}^{\infty} \ex^{-\zeta^2/2} \left(1 + \frac{\lambda_{4} \zeta^4}{24 m} \right) d\zeta
    \end{align*}
\end{lemma} 
It remains to provide an upper bound on $\Lambda_{1}$. 
\begin{align*}
    \Lambda_{1} \leq\ & \frac{1}{2\pi B_{m}}  \bigg|\int_{-\pi B_m}^{\pi B_m} \ex^{-i \zeta y} \ex^{-\zeta^2/2} \left(1 + \frac{\lambda_{4} \zeta^4}{24 m} \right) d\zeta - \int_{-\pi B_m}^{\pi B_m} \ex^{-i \zeta y} f_{S_m}(\zeta) d\zeta \bigg|\\
    \leq\ & \frac{1}{2\pi B_{m}} \int_{-\pi B_m}^{\pi B_m} \left| \ex^{-\zeta^2/2} \left(1 + \frac{\lambda_{4} \zeta^4}{24 m} \right) - f_{S_m}(\zeta) \right| d\zeta\\
    \leq\ & \frac{1}{\pi B_{m}} \int_{0}^{\pi B_m} \left| \ex^{-\zeta^2/2} \left(1 + \frac{\lambda_{4} \zeta^4}{24 m} \right) - f_{S_m}(\zeta) \right| d\zeta
\end{align*}
where the last line follows from the fact that $f_{S_{m}}$ is even function. 
This further implies that 
\begin{align} \label{omega3 + omega4}
\begin{split}
    \Lambda_{1} \leq\ & \frac{1}{\pi B_{m}} \int_{0}^{c \cdot \sqrt{\log(m)}} \left| \ex^{-\zeta^2/2} \left(1 + \frac{\lambda_{4} \zeta^4}{24 m} \right) - f_{S_m}(\zeta) \right| d\zeta\\
    &+ \frac{1}{\pi B_{m}} \int_{c \cdot \sqrt{\log(m)}}^{\pi B_m} \left| \ex^{-\zeta^2/2} \left(1 + \frac{\lambda_{4} \zeta^4}{24 m} \right) - f_{S_m}(\zeta) \right| d\zeta\\
    =\ & \Lambda_{3} + \Lambda_{4}
\end{split}
\end{align}
for $c > 0$ satisfying the condition in Lemma \ref{lemma: omega3}. 
Obtaining a sharp upper bound of $\Lambda_{3}$ and $\Lambda_{4}$ is technically involved. We summarize them into the following two lemmas and the proof is deferred to Section \ref{sec:proof_omega3} and \ref{sec:proof_omega4}. 
\begin{lemma}[Upper bound on $\Lambda_3$]
\label{lemma: omega3}
For any $c$ that $U^{m}(c) < 1$, we have
    \begin{align*}
        \Lambda_{3} \leq\ & \frac{1}{\pi B_{m}} \int_{0}^{c \cdot \sqrt{\log(m)}} \exp \left( - \frac{\zeta^2}{2} \right) \times \bigg\{ \Lambda'_5 + \Lambda'_6 + \Lambda'_7 \bigg\}  d\zeta
    \end{align*}
    where 
\begin{align*}
    \Lambda'_5 =\ & \frac{\zeta^4}{8m} \cdot \left(\frac{1}{(1 - U^{m}(c))^2} + 1 \right) + \frac{K_{6}}{6!} \frac{\zeta^6}{m^2} + \frac{1}{(1 - U^{m}(c))^2} \cdot \bigg\{\frac{K_{6} \zeta ^6}{48 m^2} + \frac{K_{8} \zeta ^8}{640 m^3} + \frac{K_{10} \zeta ^{10}}{17280 m^4} + \frac{K_{12} \zeta ^{12}}{1036800 m^{5}} \bigg\}\\
    \Lambda'_6 =\ & \frac{1}{2} \left(\frac{K_{4}}{4!} \frac{\zeta^4}{m} +  \frac{K_{6}}{6!} \frac{\zeta^6}{m^2} + \frac{U_m(\zeta) }{2(1 - U^{m}(c))^2} \right)^2\\
    \Lambda'_7 =\ & \frac{1}{6} \exp \left(\frac{c^4 K_{4}}{4!} \frac{\log(m)^2}{m} + \frac{c^6 K_{6}}{6!} \frac{\log(m)^3}{m^2} + \frac{U_m (c \sqrt{\log(m)}{})}{2(1 - U^{m}(c))^2} \right) \times \left(\frac{K_{4}}{4!} \frac{\zeta^4}{m} + \frac{K_{6}}{6!} \frac{\zeta^6}{m^2} + \frac{U_m(\zeta)}{2(1 - U^{m}(c))^2} \right)^3
\end{align*}
and $U_m(\zeta), U^{m}(c)$ are given in equation \eqref{eqn:U_m_bound} and \eqref{eqn: U^m bound}. 
\end{lemma}
Now, we remain to give an upper bound on $\Lambda_{4}$:
\begin{align*}
        \Lambda_{4} =\ & \frac{1}{\pi B_{m}} \int_{c \cdot \sqrt{\log(m)}}^{\pi B_m} \left| \ex^{-\zeta^2/2} \left(1 + \frac{\lambda_{4} \zeta^4}{24 m} \right) - f_{S_m}(\zeta) \right| d\zeta\\
        \leq\ & \frac{1}{\pi B_{m}} \int_{c \cdot \sqrt{\log(m)}}^{\pi B_m} \ex^{-\zeta^2/2} \left(1 + \frac{\lambda_{4} \zeta^4}{24 m} \right) + \frac{1}{\pi B_{m}} \int_{c \cdot \sqrt{\log(m)}}^{\pi B_m} |f_{S_m} (\zeta)|  d\zeta
    \end{align*}
\begin{lemma}[Upper bound on $\Lambda_{4}$]
\label{lemma: omega4}
Let $\theta_0$ be unique root of the Equation \eqref{eqn: theta_0} within the interval $(0, 2\pi)$. We have
    \begin{align*}
        \Lambda_{4}
        \leq\ & \frac{1}{\pi B_{m}} \int_{c \cdot \sqrt{\log(m)}}^{\pi B_m} \ex^{-\zeta^2/2} \left(1 + \frac{\lambda_{4} \zeta^4}{24 m} \right) \\
        &+ \exp\left(- \frac{c^2}{2}  \log(m)  + \frac{\cos(\theta_0) - 1 + \theta_0^2/2}{2 \theta_0^3} \cdot 8 c^3 K_3 \cdot \left(\frac{\log^3 (m)}{m} \right)^{1/2} \right) 
    \end{align*}
\end{lemma}
Combining Lemma \ref{lemma: omega2}, \ref{lemma: omega3}, \ref{lemma: omega4} completes the proof of Theorem \ref{thm:edgeworth}. 


\subsection{Proof of Lemma \ref{lemma: omega3}}
\label{sec:proof_omega3}
\begin{proof}
    The characteristic function $f_{X_j}$ is given by
\begin{align*}
    f_{X_j}(\zeta/B_m) =\ & 1 - \frac{\EE X_j^2}{2!} \frac{\zeta^2}{B_m^2} + \frac{\EE X_j^4}{4!} \frac{\zeta^4}{B_m^4} - \frac{\EE X_j^6}{6!} \frac{\zeta^6}{B_m^6} + \cdots   
\end{align*}
Applying a Taylor-Lagrange expansion, there exists a complex number $\theta_{1,j,m}$ satisfying $|\theta_{1,j,m}| < 1$ such that 
\begin{align*}
    U_{j,m}(\zeta) :=\ & f_{X_j}(\zeta/B_m) - 1\\
    =\ & - \frac{\EE X_j^2}{2!} \frac{\zeta^2}{B_m^2}+ \frac{\EE X_j^4}{4!} \frac{\zeta^4}{B_m^4}- \frac{\theta_{1,j,m} \EE X_j^6}{6!} \frac{\zeta^6}{B_m^6}
\end{align*}
The triangle inequality gives the following upper bound on $|U_{j,m}(\zeta)|$ as follows: 
\begin{align}\label{eqn:abs_U_m}
\begin{split}
    |U_{j,m}(\zeta)| \leq\ & \frac{\EE X_j^2}{2!} \frac{\zeta^2}{B_m^2}+ \frac{\EE X_j^4}{4!} \frac{\zeta^4}{B_m^4}+ \frac{\EE X_j^6}{6!} \frac{\zeta^6}{B_m^6}
\end{split}
\end{align}
Denote $U^{m}(c)$ to be the upper bound of $|U_{j,m}(\zeta)|$ for all $\zeta \in \left[ 0, c \sqrt{\log(m)} \right]$.
\begin{align} \label{eqn: U^m bound}
\begin{split}
U^{m}(c) :=\ & c^2 \cdot \frac{\log(m)}{2m} + c^4 \cdot \frac{K_4}{4!} \left(\frac{\log(m)}{m} \right)^2 + c^6 \cdot \frac{K_6}{6!} \left(\frac{\log(m)}{m}\right)^3 > |U_{j,m}(\zeta)|
\end{split}
\end{align}
We choose $c$ so that the right hand side of Equation \eqref{eqn:abs_U_m} is smaller than 1. 
Specifically, we define $c$ satisfying $U^{m}(c) < 1$.
This ensures existence of a complex number $\theta_{2,j,m}$ such that $|\theta_{2,j,m}| < 1$ and 
\begin{align*}
    \log(f_{X_j}(\zeta/B_m)) =\ & \log(1 + U_{j,m}(\zeta))\\
    =\ & U_{j,m}(\zeta) - \frac{U_{j,m}(\zeta)^2}{2 (1 + \theta_{2,j,m}(\zeta) U_{j,m}(\zeta))^2}
\end{align*}
Summing of all $j$, we obtain
\begin{align*}
    f_{S_m}(\zeta) =\ & \exp \left(\sum_{j=1}^{m} U_{j,m}(\zeta) - \frac{U_{j,m}(\zeta)^2}{2 (1 + \theta_{2,j,m}(\zeta) U_{j,m}(\zeta))^2} \right)\\
    =\ & \exp \left( \sum_{j=1}^{m} - \frac{\EE X_j^2}{2!} \frac{\zeta^2}{B_m^2} + \frac{\EE X_j^4}{4!} \frac{\zeta^4}{B_m^4} - \frac{\theta_{1,j,m} \EE X_j^6}{6!} \frac{\zeta^6}{B_m^6} - \frac{U_{j,m}(\zeta)^2}{2 (1 + \theta_{2,j,m}(\zeta) U_{j,m}(\zeta))^2} \right)\\
    =\ & \exp \left( - \frac{\zeta^2}{2} \right) \times \exp \left( \sum_{j=1}^{m} \frac{\EE X_j^4}{4!} \frac{\zeta^4}{B_m^4} - \frac{\theta_{1,j,m} \EE X_j^6}{6!} \frac{\zeta^6}{B_m^6} - \frac{U_{j,m}(\zeta)^2}{2 (1 + \theta_{2,j,m}(\zeta) U_{j,m}(\zeta))^2} \right)
\end{align*}
Applying Taylor expansion to $\exp(x)$, there exists a complex number $\theta_{3,m}$ with 
\begin{align} \label{eqn:theta_3_m}
\begin{split}
    |\theta_{3,m}| \leq\ & \sup_{\zeta \in [0, c \sqrt{\log(m)}]} \exp \left(\sum_{j=1}^{m} \frac{\EE X_j^4}{4!} \frac{\zeta^4}{B_m^4} + \frac{\EE X_j^6}{6!} \frac{\zeta^6}{B_m^6} + \frac{|U_{j,m}(\zeta)|^2}{2 |1 + \theta_{2,j,m}(\zeta) U_{j,m}(\zeta)|^2} \right)\\
    =\ & \sup_{\zeta \in [0, c \sqrt{\log(m)}]} \exp \left(\frac{K_{4}}{4!} \frac{\zeta^4}{m} + \frac{K_{6}}{6!} \frac{\zeta^6}{m^2} + \sum_{j=1}^{m} \frac{|U_{j,m}(\zeta)|^2}{2 |1 + \theta_{2,j,m}(\zeta) U_{j,m}(\zeta)|^2} \right) 
\end{split}
\end{align}
such that 
\begin{align*}
    f_{S_m}(\zeta) 
    =\ & \exp \left( - \frac{\zeta^2}{2} \right) \times \Bigg\{1 + \sum_{j=1}^{m} \frac{\EE X_j^4}{4!} \frac{\zeta^4}{B_m^4} - \frac{\theta_{1,j,m} \EE X_j^6}{6!} \frac{\zeta^6}{B_m^6} - \frac{U_{j,m}(\zeta)^2}{2 (1 + \theta_{2,j,m}(\zeta) U_{j,m}(\zeta))^2} \\
    &+ \frac{1}{2}\left(\sum_{j=1}^{m} \frac{\EE X_j^4}{4!} \frac{\zeta^4}{B_m^4} - \frac{\theta_{1,j,m} \EE X_j^6}{6!} \frac{\zeta^6}{B_m^6} - \frac{U_{j,m}(\zeta)^2}{2 (1 + \theta_{2,j,m}(\zeta) U_{j,m}(\zeta))^2}  \right)^2\\
    &+ \frac{\theta_{3,m}}{6} \left(\sum_{j=1}^{m} \frac{\EE X_j^4}{4!} \frac{\zeta^4}{B_m^4} - \frac{\theta_{1,j,m} \EE X_j^6}{6!} \frac{\zeta^6}{B_m^6} - \frac{U_{j,m}(\zeta)^2}{2 (1 + \theta_{2,j,m}(\zeta) U_{j,m}(\zeta))^2} \right)^3\Bigg\} \\
    =\ & \exp \left( - \frac{\zeta^2}{2} \right) \times \Bigg\{1 + \frac{\lambda_{4}}{4! m} \zeta^4 + \sum_{j=1}^{m} \frac{3 \kappa_2{}^2}{4!} \frac{\zeta^4}{B_m^4} - \frac{\theta_{1,j,m} \EE X_j^6}{6! } \frac{\zeta^6}{B_m^6} - \frac{U_{j,m}(\zeta)^2}{2 (1 + \theta_{2,j,m}(\zeta) U_{j,m}(\zeta))^2}\\
    &+ \frac{1}{2}\left(\sum_{j=1}^{m} \frac{\EE X_j^4}{4!} \frac{\zeta^4}{B_m^4} - \frac{\theta_{1,j,m} \EE X_j^6}{6! } \frac{\zeta^6}{B_m^6} - \frac{U_{j,m}(\zeta)^2}{2 (1 + \theta_{2,j,m}(\zeta) U_{j,m}(\zeta))^2} \right)^2\\
    &+ \frac{\theta_{3,m}}{6} \left(\sum_{j=1}^{m} \frac{\EE X_j^4}{4!} \frac{\zeta^4}{B_m^4} - \frac{\theta_{1,j,m} \EE X_j^6}{6! } \frac{\zeta^6}{B_m^6} - \frac{U_{j,m}(\zeta)^2}{2 (1 + \theta_{2,j,m}(\zeta) U_{j,m}(\zeta))^2} \right)^3\Bigg\}
\end{align*}
Using the triangle inequality, we have
\begin{align} 
\begin{split}
    &\bigg|f_{S_m} - \ex^{-\zeta^2/2} \left(1 + \frac{\lambda_{4} \zeta^4}{24 m}\right)  \bigg|\\
    \leq\ & \exp \left( - \frac{\zeta^2}{2} \right) \times \Bigg\{\Bigg| \sum_{j=1}^{m} \frac{3 \kappa_2{}^2}{4!} \frac{\zeta^4}{B_m^4} - \frac{\theta_{1,j,m} \EE X_j^6}{6! } \frac{\zeta^6}{B_m^6} - \frac{U_{j,m}(\zeta)^2}{2 (1 + \theta_{2,j,m}(\zeta) U_{j,m}(\zeta))^2} \Bigg|\\
    &+ \frac{1}{2}\left(\sum_{j=1}^{m} \frac{\EE X_j^4}{4!} \frac{\zeta^4}{B_m^4} + \frac{\EE X_j^6}{6! } \frac{\zeta^6}{B_m^6} + \frac{|U_{j,m}(\zeta)|^2}{2 |1 + \theta_{2,j,m}(\zeta) U_{j,m}(\zeta)|^2} \right)^2\\
    &+ \frac{|\theta_{3,m}|}{6} \left(\sum_{j=1}^{m} \frac{\EE X_j^4}{4!} \frac{\zeta^4}{B_m^4} + \frac{\EE X_j^6}{6! } \frac{\zeta^6}{B_m^6} + \frac{|U_{j,m}(\zeta)|^2}{2 |1 + \theta_{2,j,m}(\zeta) U_{j,m}(\zeta)|^2} \right)^3 \Bigg\}\\
    =:\ & \exp \left( - \frac{\zeta^2}{2} \right) \times \bigg\{ \Lambda_5 + \Lambda_6 + \Lambda_7 \bigg\}
\end{split}
\end{align}
In order to give explicit upper bound on $\Lambda_{5}, \Lambda_{6}$ and $\Lambda_{7}$, \eqref{eqn:abs_U_m} implies that 
\begin{align*}
    \sum_{j=1}^{m} |U_{j,m}(\zeta)|^2 \leq\ & \sum_{j=1}^{m} \left( \frac{\EE X_j^2}{2!} \frac{\zeta^2}{B_m^2}+ \frac{\EE X_j^4}{4!} \frac{\zeta^4}{B_m^4}+ \frac{\EE X_j^6}{6!} \frac{\zeta^6}{B_m^6} \right)^2\\
    \leq\ & \sum_{j=1}^{m} \frac{(\EE |X_j|^2)^2 \zeta ^4}{4 B_m^4} +\frac{(\EE |X_j|^4)^2 \zeta ^8}{576 B_m^8} +\frac{(\EE |X_j|^6)^2 \zeta ^{12}}{518400 B_m^{12}} +\frac{(\EE |X_j|^2) (\EE |X_j|^4) \zeta ^6}{24 B_m^6}\\
    &+\frac{(\EE |X_j|^2) (\EE |X_j|^6) \zeta ^8}{720 B_m^8} +\frac{(\EE |X_j|^4) (\EE |X_j|^6) \zeta ^{10}}{8640 B_m^{10}}
\end{align*}
By Jensen's inequality, for any $k,l \leq 6$, 
\begin{align} \label{eqn:jensen_ineq}
\begin{split}
   \sum_{j=1}^{m} \frac{1}{B_m^{k+l}} (\EE |X_j|^k) (\EE |X_j|^l) \leq\ & \sum_{j=1}^{m} \frac{1}{B_m^{k+l}} (\EE |X_j|^{k+l})^{k/k+l} (\EE |X_j|^{k+l})^{l/k+l}\\
   =\ & \frac{K_{k+l,m}}{m^{(k+l-2)/2}} 
\end{split}
\end{align}
Therefore, Equation \eqref{eqn:jensen_ineq} allows us to conclude
\begin{equation} \label{eqn:U_m_bound}
\begin{split}
    \sum_{j=1}^{m} |U_{j,m}(\zeta)|^2
    \leq\ & \frac{K_{4} \zeta ^4}{4 m}
    +\frac{K_{6} \zeta ^6}{24 m^2} 
    +\frac{K_{8} \zeta ^8}{720 m^3}+\frac{K_{8} \zeta ^8}{576 m^3} +\frac{K_{10} \zeta ^{10}}{8640 m^4}
    +\frac{K_{12} \zeta ^{12}}{518400 m^{5}}\\
    =\ & \frac{K_{4} \zeta ^4}{4 m}
    +\frac{K_{6} \zeta ^6}{24 m^2} 
    +\frac{K_{8} \zeta ^8}{320 m^3} 
    +\frac{K_{10} \zeta ^{10}}{8640 m^4}
    +\frac{K_{12} \zeta ^{12}}{518400 m^{5}}\\
    :=\ & U_m(\zeta)
\end{split}
\end{equation}
Equation \eqref{eqn:U_m_bound} allows us to bound $\Lambda_{5}, \Lambda_{6}$ and $\Lambda_{7}$ as follows: 
\begin{align*}
    \Lambda_5 =\ & \Bigg| \sum_{j=1}^{m} \frac{3 \kappa_2{}^2}{4!} \frac{\zeta^4}{B_m^4} - \frac{\theta_{1,j,m} \EE X_j^6}{6! } \frac{\zeta^6}{B_m^6} - \frac{U_{j,m}(\zeta)^2}{2 (1 + \theta_{2,j,m}(\zeta) U_{j,m}(\zeta))^2} \Bigg|\\
    \leq\ & \frac{\zeta^4}{8m} + \sum_{j=1}^{m} \frac{U_{j,m}(\zeta)^2}{2 (1 + \theta_{2,j,m}(\zeta) U_{j,m}(\zeta))^2}  + \frac{K_{6}}{6!} \frac{\zeta^6}{m^2}\\
    \leq\ & \frac{\zeta^4}{8m} + \sum_{j=1}^{m} \frac{1}{2 (1 + \theta_{2,j,m}(\zeta) U_{j,m}(\zeta))^2} \times \bigg(  \frac{(\EE |X_j|^2)^2 \zeta ^4}{4 B_m^4} +\frac{(\EE |X_j|^4)^2 \zeta ^8}{576 B_m^8} +\frac{(\EE |X_j|^6)^2 \zeta ^{12}}{518400 B_m^{12}}\\
    &+\frac{(\EE |X_j|^2) (\EE |X_j|^4) \zeta ^6}{24 B_m^6} +\frac{(\EE |X_j|^2) (\EE |X_j|^6) \zeta ^8}{720 B_m^8} +\frac{(\EE |X_j|^4) (\EE |X_j|^6) \zeta ^{10}}{8640 B_m^{10}}\bigg) + \frac{K_{6}}{6!} \frac{\zeta^6}{m^2}\\
    \leq\ & \frac{\zeta^4}{8m} + \sum_{j=1}^{m} \frac{1}{2 (1 + \theta_{2,j,m}(\zeta) U_{j,m}(\zeta))^2} \times \bigg(  \frac{(\EE |X_j|^2)^2 \zeta ^4}{4 B_m^4} +\frac{(\EE |X_j|^4)^2 \zeta ^8}{576 B_m^8} +\frac{(\EE |X_j|^6)^2 \zeta ^{12}}{518400 B_m^{12}}\\
    &+\frac{(\EE |X_j|^2) (\EE |X_j|^4) \zeta ^6}{24 B_m^6} +\frac{(\EE |X_j|^2) (\EE |X_j|^6) \zeta ^8}{720 B_m^8} +\frac{(\EE |X_j|^4) (\EE |X_j|^6) \zeta ^{10}}{8640 B_m^{10}}\bigg) + \frac{K_{6}}{6!} \frac{\zeta^6}{m^2}\\
    \leq\ & \left(1 + \sum_{j=1}^{m} \frac{1}{(1 + \theta_{2,j,m}(\zeta) U_{j,m}(\zeta))^2} \frac{1}{m} \right) \frac{\zeta^4}{8m}\\
    &+ \sum_{j=1}^{m} \frac{1}{2 (1 + \theta_{2,j,m}(\zeta) U_{j,m}(\zeta))^2} \cdot \Bigg( \frac{(\EE |X_j|^4)^2 \zeta ^8}{576 B_m^8} +\frac{(\EE |X_j|^6)^2 \zeta ^{12}}{518400 B_m^{12}}\\
    &+\frac{(\EE |X_j|^2) (\EE |X_j|^4) \zeta ^6}{24 B_m^6} +\frac{(\EE |X_j|^2) (\EE |X_j|^6) \zeta ^8}{720 B_m^8} +\frac{(\EE |X_j|^4) (\EE |X_j|^6) \zeta ^{10}}{8640 B_m^{10}} \Bigg) + \frac{K_{6}}{6!} \frac{\zeta^6}{m^2}
\end{align*}
Recall Equation \eqref{eqn: U^m bound} and the fact that $|\theta_{2,j,m}| < 1$, we conclude
\begin{align} \label{omega5}
\begin{split}
    \Lambda_{5}(\zeta) \leq\ & \left(\frac{1}{(1 - U^{m}(c))^2} + 1  \right) \frac{\zeta^4}{8m}\\
    &+ \frac{1}{2} \sum_{j=1}^{m} \frac{1}{(1 - U^{m}(c))^2} \cdot \Bigg( \frac{(\EE |X_j|^4)^2 \zeta ^8}{576 B_m^8} +\frac{(\EE |X_j|^6)^2 \zeta ^{12}}{518400 B_m^{12}}\\
    &+\frac{(\EE |X_j|^2) (\EE |X_j|^4) \zeta ^6}{24 B_m^6} +\frac{(\EE |X_j|^2) (\EE |X_j|^6) \zeta ^8}{720 B_m^8} +\frac{(\EE |X_j|^4) (\EE |X_j|^6) \zeta ^{10}}{8640 B_m^{10}} \Bigg) \Bigg| + \frac{K_{6}}{6!} \frac{\zeta^6}{m^2}\\
    \leq\ & \frac{\zeta^4}{8m} \cdot \left(\frac{1}{(1 - U^{m}(c))^2} + 1 \right) + \frac{K_{6}}{6!} \frac{\zeta^6}{m^2}\\
    &+ \frac{1}{(1 - U^{m}(c))^2} \cdot \bigg\{\frac{K_{6} \zeta ^6}{48 m^2} + \frac{K_{8} \zeta ^8}{640 m^3} + \frac{K_{10} \zeta ^{10}}{17280 m^4} + \frac{K_{12} \zeta ^{12}}{1036800 m^{5}} \bigg\}
\end{split}
\end{align}
Similarly, we also give explicit upper bound of $\Lambda_{6}$ and $\Lambda_{7}$. 
\begin{align} \label{omega6}
\begin{split}
    \Lambda_6 =\ & \frac{1}{2}\left(\sum_{j=1}^{m} \frac{\EE X_j^4}{4!} \frac{\zeta^4}{B_m^4} + \frac{\EE X_j^6}{6! } \frac{\zeta^6}{B_m^6} + \frac{|U_{j,m}(\zeta)|^2}{2 |1 + \theta_{2,j,m}(\zeta) U_{j,m}(\zeta)|^2} \right)^2\\
    \leq\ & \frac{1}{2} \left(\frac{K_{4}}{4!} \frac{\zeta^4}{m} + \frac{K_{6}}{6!} \frac{\zeta^6}{m^2} + \frac{1}{2(1 - U^{m}(c))^2} \sum_{j=1}^{m} |U_{j,m}(\zeta)|^2 \right)^2\\
    \leq\ & \frac{1}{2} \left(\frac{K_{4}}{4!} \frac{\zeta^4}{m} +  \frac{K_{6}}{6!} \frac{\zeta^6}{m^2} + \frac{U_m(\zeta) }{2(1 - U^{m}(c))^2} \right)^2
\end{split}
\end{align}
with $U_m(\zeta)$ defined in \eqref{eqn:U_m_bound}.
To have an upper bound for $\Lambda_{7}$, we first need an upper bound for $|\theta_{3,m}|$. Recall Equation \eqref{eqn:theta_3_m} that 
\begin{align*}
\begin{split}
    |\theta_{3,m}| \leq\ & \sup_{\zeta \in [0, c \sqrt{\log(m)}]} \exp \left(\frac{K_{4}}{4!} \frac{\zeta^4}{m} + \frac{K_{6}}{6!} \frac{\zeta^6}{m^2} + \sum_{j=1}^{m} \frac{|U_{j,m}(\zeta)|^2}{2 |1 + \theta_{2,j,m}(\zeta) U_{j,m}(\zeta)|^2} \right)\\
    \leq\ & \exp \left(\frac{c^4 K_{4}}{4!} \frac{\log(m)^2}{m} + \frac{c^6 K_{6}}{6!} \frac{\log(m)^3}{m^2} + \frac{\sup_{\zeta \in [0, c \sqrt{\log(m)}]} U_m(\zeta)}{2(1 - U^{m}(c))^2} \right)\\
    =\ & \exp \left(\frac{c^4 K_{4}}{4!} \frac{\log(m)^2}{m} + \frac{c^6 K_{6}}{6!} \frac{\log(m)^3}{m^2} + \frac{U_m (c \sqrt{\log(m)}{})}{2(1 - U^{m}(c))^2} \right)
\end{split}
\end{align*}
Having the upper bound on $|\theta_{3,m}|$, $\Lambda_{7}$ yields the following upper bound. 
\begin{align} \label{omega7}
\begin{split}
    \Lambda_{7} =\ & \frac{|\theta_{3,m}|}{6} \left(\sum_{j=1}^{m} \frac{\EE X_j^4}{4!} \frac{\zeta^4}{B_m^4} + \frac{\EE X_j^6}{6! } \frac{\zeta^6}{B_m^6} + \frac{|U_{j,m}(\zeta)|^2}{2 |1 + \theta_{2,j,m}(\zeta) U_{j,m}(\zeta)|^2} \right)^3\\
    \leq\ & \frac{1}{6} \exp \left(\frac{c^4 K_{4}}{4!} \frac{\log(m)^2}{m} + \frac{c^6 K_{6}}{6!} \frac{\log(m)^3}{m^2} + \frac{U_m (c \sqrt{\log(m)}{})}{2(1 - U^{m}(c))^2} \right) \\
    &\times \left(\sum_{j=1}^{m} \frac{\EE X_j^4}{4!} \frac{\zeta^4}{B_m^4} + \frac{\EE X_j^6}{6! } \frac{\zeta^6}{B_m^6} + \frac{|U_{j,m}(\zeta)|^2}{2 |1 + \theta_{2,j,m}(\zeta) U_{j,m}(\zeta)|^2} \right)^3\\
    \leq\ & \frac{1}{6} \exp \left(\frac{c^4 K_{4}}{4!} \frac{\log(m)^2}{m} + \frac{c^6 K_{6}}{6!} \frac{\log(m)^3}{m^2} + \frac{U_m (c \sqrt{\log(m)}{})}{2(1 - U^{m}(c))^2} \right) \\
    & \qquad \qquad \qquad \qquad \times \left(\frac{K_{4}}{4!} \frac{\zeta^4}{m} + \frac{K_{6}}{6!} \frac{\zeta^6}{m^2} + \frac{U_m(\zeta)}{2(1 - U^{m}(c))^2} \right)^3
\end{split}
\end{align}
Combining \eqref{omega5}, \eqref{omega6} and \eqref{omega7}, we conclude that 
\begin{align*}
    &\bigg|f_{S_m} - \ex^{-\zeta^2/2} \left(1 + \frac{\lambda_{4} \zeta^4}{24 m}\right)  \bigg|
    \leq \exp \left( - \frac{\zeta^2}{2} \right) \times \bigg\{ \Lambda'_5 + \Lambda'_6 + \Lambda'_7 \bigg\}
\end{align*}
where 
\begin{align*}
    \Lambda'_5 =\ & \frac{\zeta^4}{8m} \cdot \left(\frac{1}{(1 - U^{m}(c))^2} + 1 \right) + \frac{K_{6}}{6!} \frac{\zeta^6}{m^2} + \frac{1}{(1 - U^{m}(c))^2} \cdot \bigg\{\frac{K_{6} \zeta ^6}{48 m^2} + \frac{K_{8} \zeta ^8}{640 m^3} + \frac{K_{10} \zeta ^{10}}{17280 m^4} + \frac{K_{12} \zeta ^{12}}{1036800 m^{5}} \bigg\}\\
    \Lambda'_6 =\ & \frac{1}{2} \left(\frac{K_{4}}{4!} \frac{\zeta^4}{m} +  \frac{K_{6}}{6!} \frac{\zeta^6}{m^2} + \frac{U_m(\zeta) }{2(1 - U^{m}(c))^2} \right)^2\\
    \Lambda'_7 =\ & \frac{1}{6} \exp \left(\frac{c^4 K_{4}}{4!} \frac{\log(m)^2}{m} + \frac{c^6 K_{6}}{6!} \frac{\log(m)^3}{m^2} + \frac{U_m (c \sqrt{\log(m)}{})}{2(1 - U^{m}(c))^2} \right) \times \left(\frac{K_{4}}{4!} \frac{\zeta^4}{m} + \frac{K_{6}}{6!} \frac{\zeta^6}{m^2} + \frac{U_m(\zeta)}{2(1 - U^{m}(c))^2} \right)^3
\end{align*}
This completes the proof of this lemma. 
\end{proof}

\subsection{Proof of Lemma \ref{lemma: omega4}}
\label{sec:proof_omega4}
\begin{proof}

We would like to give an upper bound on $|f_{S_{m}}(t)|$. The idea is inspired by Theorem 2.2 in \cite{shevtsova2012moment}.

\begin{lemma}[\cite{shevtsova2012moment}]
\label{lemma: cos}
Let $\theta_0$ be unique root of the Equation \eqref{eqn: theta_0} within the interval $(0, 2\pi)$. For any $x \in \RR$, and $\theta_0 < \theta \leq 2 \pi$, then
\begin{align*}
    \cos(x) \leq 1 - a(\theta) x^2 + b(\theta) |x|^3
\end{align*}
where 
\begin{align*}
    a(\theta) =\ & 3 \cdot \frac{1 - \cos(\theta)}{\theta^2} - \frac{\sin(\theta)}{\theta}\\
    b(\theta) =\ & 2 \cdot \frac{1 - \cos(\theta)}{\theta^3} - \frac{\sin(\theta)}{\theta^2}
\end{align*}
\end{lemma}
Now, let $\Tilde{X}_{i}$ be an independent copy of $X_{i}$, then we have: 
\begin{align*}
    |f_{S_{m}}(\zeta)|^2 =\ & \prod_{i=1}^{m} \left|f_{X_{i}} \left(\frac{\zeta}{B_{m}} \right)\right|^2 = \prod_{i=1}^{m} \EE \cos \frac{\zeta(X_{i} - \Tilde{X}_{i})}{B_{m}}
\end{align*}
By Lemma \ref{lemma: cos} and relation $\EE (X_{i} - \Tilde{X}_{i})^2 = 2 \kappa_{2}$, we obtain
\begin{align*}
    |f_{S_{m}}(\zeta)|^2 \leq\ & \prod_{i=1}^{m} \left(1 - a(\theta) \frac{\zeta^2 \EE (X_{i} - \Tilde{X}_{i})^2 }{B_{m}^2} + b(\theta) \frac{|\zeta|^3 \EE |X_{i} - \Tilde{X}_{i}|^3 }{B_{m}^3} \right)\\
    \leq\ & \prod_{i=1}^{m} \left(1 - a(\theta) \frac{2 \zeta^2 \kappa_{2}}{B_{m}^2} + b(\theta) \frac{|\zeta|^3 \EE |X_{i} - \Tilde{X}_{i}|^3 }{B_{m}^3} \right)
\end{align*}
By AM-GM inequality, we have
\begin{align*}
    |f_{S_{m}}(\zeta)|^2 \leq\ & \left(1 - \frac{1}{m} \sum_{i=1}^{m} a(\theta) \frac{2 \zeta^2 \kappa_{2}}{B_{m}^2} + b(\theta) \frac{|\zeta|^3 \EE |X_{i} - \Tilde{X}_{i}|^3 }{B_{m}^3} \right)^m\\
    \leq\ & \left(1 + \frac{2}{m} \left(- a(\theta) \zeta^2  + b(\theta) |\zeta|^3 \frac{\sum_{i=1}^{m} \EE |X_{i} - \Tilde{X}_{i}|^3 }{2 B_{m}^3} \right) \right)^m    
\end{align*}
This implies that 
\begin{align*}
    |f_{S_{m}}(\zeta)|^2 \leq\ & \exp\left(- 2a(\theta) \zeta^2  + b(\theta) |\zeta|^3 \frac{\sum_{i=1}^{m} \EE |X_{i} - \Tilde{X}_{i}|^3 }{B_{m}^3}  \right)
\end{align*}
For convenience, we choose a specific $\theta = \theta_0$. We have upper bound
\begin{align*}
    |f_{S_{m}}(\zeta)|^2 \leq\ & \exp\left(- \zeta^2  + \frac{\cos(\theta_0) - 1 + \theta_0^2/2}{\theta_0^3} \cdot \frac{\sum_{i=1}^{m} \EE |X_{i} - \Tilde{X}_{i}|^3 }{B_{m}^3} \cdot |\zeta|^3\right)
\end{align*}
By Lemma \ref{lemma:char_property}, the characteristic function is strictly decreasing between $c \sqrt{\log (m)}$ and $\pi B_{m}$. Therefore, for any $\zeta \in [c \sqrt{\log (m)}, \pi B_{m}]$, we have
\begin{align*}
    |f_{S_{m}} (\zeta)| \leq |f_{S_{m}} (c \sqrt{\log (m)})| \leq \exp\left(- \frac{c^2}{2} \log(m)  + \frac{\cos(\theta_0) - 1 + \theta_0^2/2}{2 \theta_0^3} \cdot \frac{\sum_{i=1}^{m} \EE |X_{i} - \Tilde{X}_{i}|^3 }{B_{m}^3} \cdot c^3 (\log (m))^{3/2} \right)
\end{align*}
This further implies that 
\begin{align*}
        \Lambda_{4} \leq\ & \frac{1}{\pi B_{m}} \int_{c \cdot \sqrt{\log(m)}}^{\pi B_m} \ex^{-\zeta^2/2} \left(1 + \frac{\lambda_{4} \zeta^4}{24 m} \right) \\
        &+ \exp\left(- \frac{c^2}{2}  \log(m)  + \frac{\cos(\theta_0) - 1 + \theta_0^2/2}{2 \theta_0^3} \cdot \frac{\sum_{i=1}^{m} \EE |X_{i} - \Tilde{X}_{i}|^3 }{B_{m}^3} \cdot c^3 (\log (m))^{3/2} \right) \\
        =\ & \frac{1}{\pi B_{m}} \int_{c \cdot \sqrt{\log(m)}}^{\pi B_m} \ex^{-\zeta^2/2} \left(1 + \frac{\lambda_{4} \zeta^4}{24 m} \right) \\
        &+ \exp\left(- \frac{c^2}{2}  \log(m)  + \frac{\cos(\theta_0) - 1 + \theta_0^2/2}{2 \theta_0^3} \cdot \frac{8 m \EE |X_i|^3}{\kappa_2^{3/2} m^{3/2}} \cdot c^3 (\log (m))^{3/2} \right) \\
        =\ & \frac{1}{\pi B_{m}} \int_{c \cdot \sqrt{\log(m)}}^{\pi B_m} \ex^{-\zeta^2/2} \left(1 + \frac{\lambda_{4} \zeta^4}{24 m} \right) \\
        &+ \exp\left(- \frac{c^2}{2}  \log(m)  + \frac{\cos(\theta_0) - 1 + \theta_0^2/2}{2 \theta_0^3} \cdot 8 c^3 K_3 \cdot \left(\frac{\log^3 (m)}{m} \right)^{1/2} \right) 
\end{align*}
This gives an upper bound on $\Lambda_{4}$. 

\end{proof}

\section{Supplementary Figures and Tables}
\label{sec:supp_fig}

The supplemental figures include 946 figures with corresponding raw data tables, which is the privacy accounting of all the composed mechanisms $[ \widetilde{\mathbf{M}}_k, \widetilde{\mathbf{M}}_l ]$. 
All figures follow the same pattern as in Figures \ref{fig:trade_off_paths_bypass} and \ref{fig:eps_delta_paths_bypass}. 
For brevity, we provide the Figures and Tables in the folder ``results'' in GitHub repository (see \url{https://github.com/BuxinSu/Exact-Privacy-Accounting-for-2020-U.S.-Census.git}). We also include all privacy budget allocations $\rho_i$ used in the 2020 Census DHC File in the folder ``privacy budget allocation''.

\end{document}